\documentclass[acmsmall]{acmart}

\usepackage[utf8]{inputenc} 
\usepackage[T1]{fontenc} 
\usepackage{microtype}

\usepackage[shortlabels]{enumitem}
 \usepackage{wrapfig}
 
\DeclareSymbolFont{arrows}{U}{FdSymbolC}{m}{n}

\setcopyright{none}

\usepackage{booktabs}   
\usepackage{subcaption} 
 \usepackage{ stmaryrd }                       

\usepackage{relsize}
\usepackage{mathpartir}
\usepackage{amsmath}
\usepackage{amsthm}
\usepackage{listings}
\usepackage{xspace}
\usepackage{definitions}
\usepackage{multirow,bigdelim}
\usepackage{pbox}
\usepackage{courier}
\usepackage{soul}
\usepackage{centernot}

\newcommand{\aref}[2]{#2}
\newcommand{\afull}[2]{#2}

\definecolor{ferngreen}{rgb}{0.31, 0.47, 0.26}
\newcommand{\kjx}[1]{{}}
\newcommand{\scd}[1]{{}}

\newcommand{\sdcomment}[1]{{}}
\newcommand{\secomment}[1]{{}}
\newcommand{\jncomment}[1]{{}}

\newcommand{\se}[1] {#1} 

\definecolor{amber}{rgb}{1.0, 0.75, 0.0}
\definecolor{amethyst}{rgb}{0.6, 0.4, 0.8}

\newcommand{\ponders}[3]{\marginpar{\tiny\itshape\raggedright\textcolor{#2}{\textbf{#1:} #3}}\ignorespaces}
\renewcommand{\ponders}[3]{{#3\ignorespaces}}  

\newcommand{\TODO}[1]{} 
\newcommand{\sophia}[1]{{#1}}
\newcommand{\toby}[1]{} 
\newcommand{\susan}[2][]{{#2}\xspace}

\newcommand{\sdN}[1]{{{#1}}}

\renewcommand{\sophia}[2][]{}

\newcommand{\sdfootnote}[1]{}

\begin{document}

\title{Reasoning about External Calls -- Extended Version}

\titlenote{This extended version of our work contains an appendix with full formal definitions, examples, more explanations, and worked out proofs. Citations should use the original version of the paper, \cite{ExternalCallsOOPSLA22}.} 



 
\author{Sophia Drossopoulou}
\orcid{0000-0002-1993-1142}             
\affiliation{
  \institution{Imperial College London}            
  \country{UK} 
}
\email{scd@imperial.ac.uk}          

\author{Julian Mackay}
\orcid{0000-0003-3098-3901}             
\affiliation{
  \institution{Kry10 Ltd, and Victoria University of Wellington}           
   \country{NZ}   
}
\email{julian.mackay@ecs.vuw.ac.nz}         

\author{Susan Eisenbach}
\orcid{0000-0001-9072-6689}             
\affiliation{
  \institution{Imperial College London}           
  \country{UK} 
}
\email{susan@imperial.ac.uk}         

\author{James Noble}
\orcid{0000-0001-9036-5692}             
\affiliation{
  \institution{Creative Research \& Programming}           
  \city{Wellington}
  \postcode{6012}
   \country{NZ}   
}
\email{kjx@acm.org}         




\newcommand{\paragraphsd}[1]{\vspace{.02cm}{\textit{#1}}}





   \begin{abstract}

 In today's complex software, internal trusted code  is tightly intertwined  
 with external untrusted code. 
%
%
%
%
To reason about internal code, programmers must reason about 
the potential effects of calls to external code, even though that
code is not trusted and may not even be available.
%
%
%
The effects of external calls can be \emph{limited}
if internal code 
is programmed defensively,
limiting potential effects by limiting access to the capabilities necessary to cause those effects.
%
%
%

This paper
addresses the specification and verification of internal code
that relies on
encapsulation and object capabilities to limit the effects of external calls.
We propose new assertions for access to capabilities,
new specifications for limiting effects,
and a Hoare logic to verify that a module satisfies its specification, even while making external calls.
We illustrate the approach though a running example with mechanised proofs,
and prove soundness of the Hoare logic.

\end{abstract}


\maketitle


\section{Introduction}
\label{s:intro}

\begin{minipage}{.64\textwidth}
This paper addresses reasoning about \emph{external calls} --- when
trusted internal code calls out to untrusted, unknown external code:
In the code sketch to the right, 
an internal module, $M_{intl}$, has two methods. 
Method \prg{m2} takes an untrusted parameter \prg{untrst},
and calls an unknown external method \prg{unkn} passing itself as an argument. 
The challenge is: 
What effects will that external call have?
What if \prg{untrst} calls back into $M_{intl}$? 
\end{minipage}
\hfill
\begin{minipage}{.35\textwidth}
\begin{lstlisting}[mathescape=true, language=Chainmail,frame=none,numbers=none]
 module $M_{intl}$        
   method m1 ..
      ...  $\mbox{trusted code}$ ...  
   method m2(untrst:external) 
      ... $\mbox{trusted code}$ ...
      $\mbox{\red{untrst.unkn(this)}}$   
      ... $\mbox{trusted code}$ ...
\end{lstlisting}
\end{minipage}

%
%
%

 {\em{External calls  need not have  arbitrary   effects.}}
If the programming language supports encapsulation (\eg no address forging, privacy,
 \etc) then internal modules can be  written 
so  that effects are  

 \begin{customquote}
$\bullet$\ \ \emph{Precluded}, \ie  guaranteed to \emph{never happen}.
 E.g., a  correct  
 implementation of the DAO  \cite{Dao}  can ensure that  
 its balance never falls below the sum of the balances of its subsidiary accounts, \\ \emph{or}

 $\bullet$\ \  \emph{Limited}, \ie  they  \emph{may happen}, but
 only in well-defined circumstances.
E.g., while the DAO does not preclude that a signatory's balance will decrease, it  does ensure that the balance decreases only
as a direct consequence of calls from the signatory.
 
 \end{customquote}

\noindent
Precluded effects are special case of limited effects,
and have been studied extensively in the context of object invariants   \cite{staticsfull,DrossoFrancaMuellerSummers08,BarDelFahLeiSch04,objInvars,MuellerPoetzsch-HeffterLeavens06}.
In this paper, we tackle the more general, 
and more subtle case of reasoning about limited effects for external calls.

\paragraphsd{The Object Capability Model}
combines the capability model
of operating system security \cite{levy:capabilities,CAP}
with pure object-oriented programming
\cite{selfpower,selfexp95,agha_actors_1987}.  Capability-based
operating systems reify resources
as \textit{capabilities} ---
unforgeable, distinct, duplicable, attenuable, communicable bitstrings
which both denote a resource and grant rights over that resource.
Effects can only be caused by invoking capabilities:
controlling effects reduces to controlling capabilities.

Mark Miller's \cite{MillerPhD}
\textit{object}-capability model treats
object references
as capabilities. 
Building on early object-capability languages such as E
\cite{MillerPhD,ELang} and Joe-E \cite{JoeE}, 
a range of recent programming languages and web systems
\cite{CapJavaHayesAPLAS17,CapNetSocc17Eide,DOCaT14} including Newspeak
\cite{newspeak17},
AmbientTalk \cite{ambientTalk}
Dart \cite{dart15}, Grace \cite{grace,graceClasses},
JavaScript (aka Secure EcmaScript \cite{miller-esop2013}),
and Wyvern \cite{wyverncapabilities} have adopted the object
capability model.
Security  and encapsulation 
is encoded in the relationships between the objects, and the interactions between them.
 As argued in  \cite{capeFTfJP}, 
 object capabilities
 make it possible to write secure
  programs,
  but cannot by themselves guarantee that any particular program
  will be secure.
%



\paragraphsd{Logics for Reasoning with Capabilities}
 {
have been proposed in
 \cite{BirkedalL:caps-mmio-conf,ddd,vmsl-pldi2023,irisWasm23},   
 {primarily in object-based settings -- our work operates in a class-based setting.}
 More importantly, in these  works,  capabilities are statically embedded within objects, and authority is conferred through the object’s interface. In contrast, our approach models capabilities in a manner that is, in some sense, the dual of theirs: capabilities are not intrinsic to objects but may be 
 {dynamically accessible to the callee—typically via ambient authority in the execution environment.}}
 
 {This distinction carries important implications for reasoning  about   effects. 
 In  \cite{BirkedalL:caps-mmio-conf,ddd,vmsl-pldi2023,irisWasm23}, one can argue that if an object is exported with an interface that omits certain capabilities, then the corresponding effects are  ruled out. 
 However, this argument hinges on the assumption that the exporting site has precise knowledge of the object’s embedded capabilities.
 Our model requires no such assumption. Instead, we reason about potential effects based on the capabilities available in the receiving context.  
{As long as the target context lacks access to specific capabilities}, we can guarantee that the corresponding effects are unobservable —regardless of the object’s interface.
 }


{Our earlier work}  
{\cite{FASE,OOPSLA22}   follows a
similar dynamic model of object capabilities, though without addressing external effects.}
Specifically, \cite{FASE}  introduces “holistic specifications” to capture a module’s emergent behavior, while
\cite{OOPSLA22}  proposes
a logic for proving that modules devoid of external calls conform to their holistic specifications.


\paragraphsd{This paper contributes}
(1) \emph{protection assertions} to limit access to object-capabilities, 
(2) a specification language to define how limited access to capabilities should limit effects, 
(3) a Hoare logic to reason about external calls and to prove that modules satisfy their 
specifications,
(4) proof of soundness,
(5) a worked illustrative example {with a mechanised proof in Coq}.

  \paragraphsd{Structure of this paper:}
Sect.\ \ref{s:outline}   outlines the main ingredients of our approach in terms of an {example}.
Sect.\ \ref{sect:underlying} outlines a simple object-oriented language used for our work. 
Sect.\ \ref{s:assertions} and 
\ref{sect:spec}  give syntax and semantics of assertions  and  specifications.
Sect.\ \ref{sect:proofSystem} develops a Hoare logic  
 to prove external calls, and modules' adherance to  specifications, and summarises our  proof
  of  the running example 
 (Coq  at \cite{CoqOOPSLA25}).
Sect.\ \ref{sect:sound:proofSystem} outlines 
proof of soundness of
the Hoare logic. 
 Sect.\ \ref{sect:conclusion} concludes. 
Full  technical details can be found in the appendices 
\afull{in the extended version of this work \cite{externalCallsFull}.}{of this paper.}

\section{The problem and our approach}
\label{s:outline}  
 \newcommand{\pwd}{key}

\renewcommand{\password}{key\xspace}

We introduce the problem  through an example, and outline our
approach.  We work with a  small, class-based object-oriented, sequential language similar to Joe-E \cite{JoeE} with modules,   module-private fields
({accessible} only from   methods {from} the same module),
and unforgeable, un-enumerable addresses.
We distinguish between  \emph{\internalO  objects} --- instances of our internal module $M$'s classes ---
and \emph{\externalO  objects} defined in
\emph{any} number of external modules $\overline M$.%
\!\footnote{We use the notation $\overline z$ for a sequence of $z$, \ie for $z_1,z_2,...z_n$ }
{\prg{Private} methods  {may only be} called by objects of the same
  module,  while \prg{public}  methods  may be {called} by \emph{any}
  object with a reference to the method receiver, {and with
  actual arguments of  dynamic types that match} the declared formal parameter types.}%
\!\footnote{As in Joe-E, we leverage  module-based privacy to restrict propagation of capabilities, and reduce the need for reference monitors etc, \cf Sect 3 in  \cite{JoeE}.}   

\label{s:concepts}
 
We are concerned with guarantees made in an \emph{open} setting; 
Our internal module
$M$ must be programmed so that 
  execution of $M$  together with \emph{any} unknown, arbitrary, external modules $\overline M$
will satisfy these guarantees --
without relying on any assumptions about $\overline M$'s code.\footnote{This is a critical distinction from e.g.\
cooperative approaches such as rely/guarantee
\cite{relyGuarantee-HayesJones-setss2017,relyGuarantee-vanStaden-mpc2015}.}
All we can rely on, is the guarantee that external code interacts with the internal code only through public methods; such a guarantee may be given by the programming language or by the underlying platform.\footnote{Thus our approach is also applicable to inter-language safety.}

\subsection{\prg{Shop} -- Illustrating Limited Effects}  
\label{sec:how}
\label{sec:shop}

Consider below the internal module \Mshop, 
with classes \prg{Item}, \prg{Shop}, \prg{Account}, and \prg{Inventory}. 
Classes \prg{Inventory} and \prg{Item} are straightforward: we elide
their details. 
\prg{Account}s hold a balance and have a \password. 
Access to an \prg{Account},  allows one  to pay money into it, 
and  access to an \prg{Account}  and its \prg{Key}, allows one to withdraw money from it.
A \prg{Shop} has an \prg{Account},
and a public method \prg{buy} to allow a 
\prg{buyer} --- an \prg{external}   object --- to buy
and pay for an \prg{Item}:
 

\begin{lstlisting}[mathescape=true, language=Chainmail, frame=none,
  numbers=left,
  numberstyle=\tiny,
  numbersep=5pt,
  xleftmargin=2em
]
module M$_{shop}$
  ...   
  class Shop
    field accnt:Account, invntry:Inventory, clients:external     
     
    public method buy(buyer:external, anItem:Item)
      int price = anItem.price
      int oldBlnce = this.accnt.blnce
      $\red{\mbox{buyer.pay(this.accnt, price)}}$      
      if (this.accnt.blnce == oldBlnce+price)  
         this.send(buyer,anItem)
      else
         buyer.tell("you have not paid me") 
                       
    private method send(buyer:external, anItem:Item)   ....
\end{lstlisting}

\noindent
The sketch below shows a possible heap snippet.

\begin{minipage}{.23\textwidth}
\includegraphics[width=30mm]{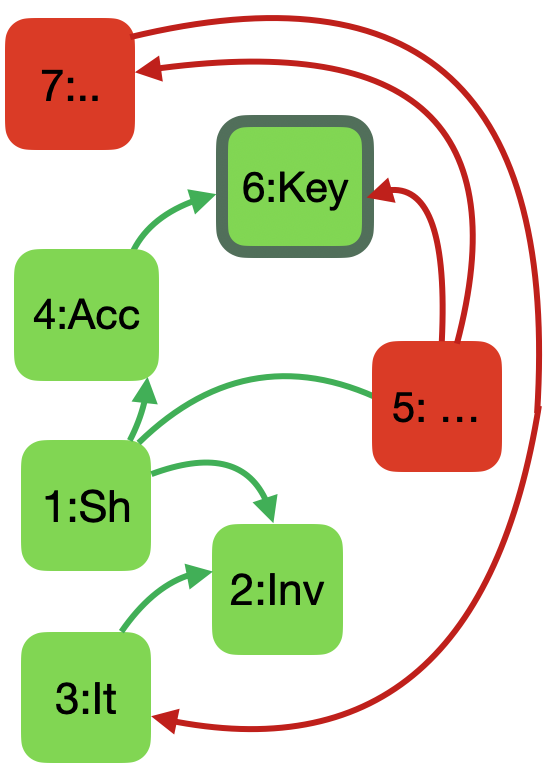}
\end{minipage}
\hfill
\hspace{0.01\textwidth}
\begin{minipage}{.73\textwidth}
External objects are red; internal objects are green.
Each object has a number,
and an abbreviated class name: $o_1$, $o_2$ and $o_5$ are a \prg{Shop}, an \prg{Inventory}, and an external object. 
Curved arrows
indicate field values: $o_1$ has three fields, pointing to 
$o_4$, $o_5$ and $o_2$.
Fields denote direct access. The transitive closure of direct access
gives indirect (transitive) access: $o_1$ has direct access to $o_4$, and indirect access to $o_6$.
Object $o_6$ ---
highlighted with a dark outline
---
is the key capability that allows withdrawal from $o_4$.
The critical point in our code is the external call on line 9,   where the \prg{Shop} asks the \prg{buyer} to pay the price of that item,
by calling  \prg{pay} on \prg{buyer} and passing the \prg{Shop}'s account as an argument.
\end{minipage}

As \prg{buyer} is an external object, the module \Mshop has no method specification for \prg{pay}, and no 
certainty about what its implementation might do. 
What are the possible effects of that external call?
At line 9, the \prg{Shop} hopes
the buyer will have deposited the price into its account, 
but needs to be certain the \prg{buyer} cannot
have emptied that account instead. 
Can the \prg{Shop}  be certain?

 Indeed, if

\begin{enumerate}[(A)]
\item   Prior to the call of  \prg{buy}, the \prg{buyer}  has no eventual access to the account's \password, \ \ \ \emph{----} \ and
\item  \Mshop ensures that 
\begin{enumerate}[(a)]
\item access to keys is not leaked to external objects, \ \ \ \emph{----} and
\item   funds cannot be withdrawn unless the external entity responsible for the withdrawal (eventually) has  access to the account's \password,
\end{enumerate}

\ \emph{----} \ then

\item  The external  call on line 9 can never result in a decrease in the shop's account balance.
\end{enumerate}

 \noindent
The remit of this paper is to provide specification and verification tools that support arguments like the one above.
This gives rise to the following two challenges:\  1$^{st}$:  A specification language which describes access to capabilities and limited effects,\ 
2$^{nd}$: A  Hoare Logic for adherence to such specifications.


\subsection{1$^{st}$ Challenge: Specification Language} 

We want to give a formal meaning to the guarantee that for some effect, $E$, and an object $o_c$ which is the capability for $E$:

\vspace{.05cm}

  \begin{minipage}{.05\textwidth}
   \textbf{(*)}
\end{minipage}
\hfill
\begin{minipage}{.95\textwidth}
\begin{flushleft}
$E$  (\eg the account's balance decreases)  can be caused only  by external objects calling methods on internal objects, \\
and only if the causing object has access  to $o_c$ (\eg the key).
\end{flushleft}
\end{minipage}

\vspace{.05cm}


\noindent 
The first task is to describe that effect  $E$ took place: if we  find  some assertion $A$ (\eg  balance is $\geq$ some value $b$)
which is invalidated by $E$, then, (*) can be described by something like:

\vspace{.05cm}

  \begin{minipage}{.05\textwidth}
   \textbf{(**)}
\end{minipage}
\hfill
\begin{minipage}{.95\textwidth}
\begin{flushleft}
If $A$ holds, \  and \    if no external access to  $o_c$ \ \ \ \ then\  \ \  \ $A$ holds in the future. 
\end{flushleft}
\end{minipage}

\vspace{.05cm}

\noindent 
We next make more precise that ``no external access to  $o_c$'', and that ``$A$ holds in the future''.

In a first attempt, we could say that ``no external access to  $o_c$'' means  that no external object exists, nor will any external objects be created.
This is too strong, however: it defines away the problem we are aiming to solve.

In a second attempt, we could say that ``no external access to  $o_c$'' means that no external object has access to $o_c$, nor will ever get access to $o_c$. This is also too strong, as it would preclude $E$ from ever happening, while our remit is that $E$ may happen but only under certain conditions. 

This discussion indicates that the lack of external access to $o_c$ is not a global property, and that the future in which  $A$ will hold is not permanent. 
Instead, they are both defined \emph{from the perspective of the current point of execution}.

Thus:

\vspace{.05cm}

  \begin{minipage}{.05\textwidth}
   \textbf{(***)\ \ }
\end{minipage}
\hfill
\begin{minipage}{.9\textwidth}
  If $A$ holds, \  and \ if no external object  reachable from the current point of execution  has access to $o_c$,   
\   and\   if no  internal objects pass $o_c$ to external objects,  \\
  then  \\
 $A$ holds in  \emph{the future scoped by the current point of execution}.  
\end{minipage}

\vspace{.1cm}

\noindent We will shortly formalize ``reachable from the current point of execution'' as
\emph{protection} in \S \ref{sect:approach:protection},
and then  ``future scoped by the current point of execution''
as  \emph{scoped invariants} in \S \ref{sect:approach:scoped}.
Both of these definitions are in terms of the ``current point of execution'':

\noindent\emph{The Current Point of Execution}
is characterized by the heap, and   the activation frame of the currently executing method. 
Activation frames (frames for short) consist of a variable map and a continuation -- the statements remaining to be executed in that method.
Method calls push frames onto the stack of frames; method returns pop frames off.
The frame on top of the stack (the most recently pushed frame) 
belongs to the currently executing method.

Fig.\ \ref{f:CurrentPoint} illustrates the current point of execution.  The left pane, $\sigma_1$, shows a state with the same heap as earlier, and where the top frame is $\phi_1$ -- it could be the state before a call to \prg{buy}.
  The middle pane, $\sigma_2$, is a state where we have pushed  $\phi_2$ on top of the stack of $\sigma_1$ -- it could be a state during execution of \prg{buy}.
   The right pane, $\sigma_3$, is a state where we have pushed  $\phi_3$ on top of the stack of $\sigma_2$ -- it could be a state during execution of \prg{pay}. 

States whose top frame has a  receiver (\prg{this}) which is an internal object are called  \emph{internal states}, the other states are called  \emph{external states}. In Fig \ref{f:CurrentPoint}, states $\sigma_1$ and $\sigma_2$ are internal, and  $\sigma_3$ is external.

\begin{figure}[th]
\begin{tabular}{|c|c|c|}
\hline
\resizebox{3,5cm}{!}{
\includegraphics[width=\linewidth]{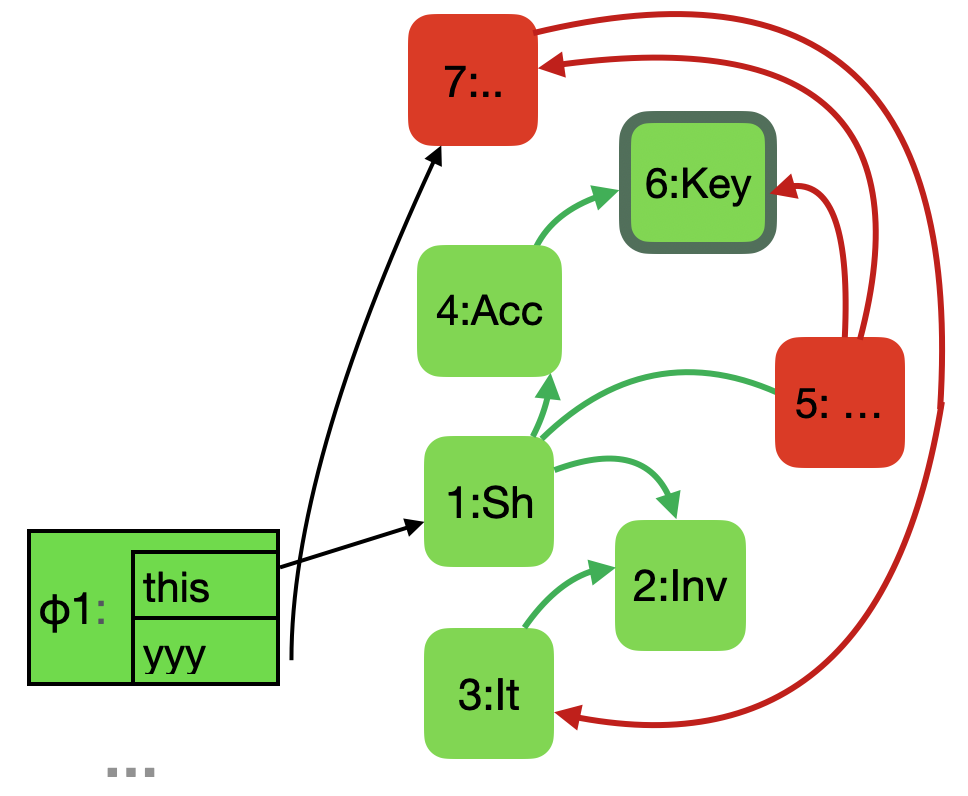}
} 
&
\resizebox{3,5cm}{!}{
\includegraphics[width=\linewidth]{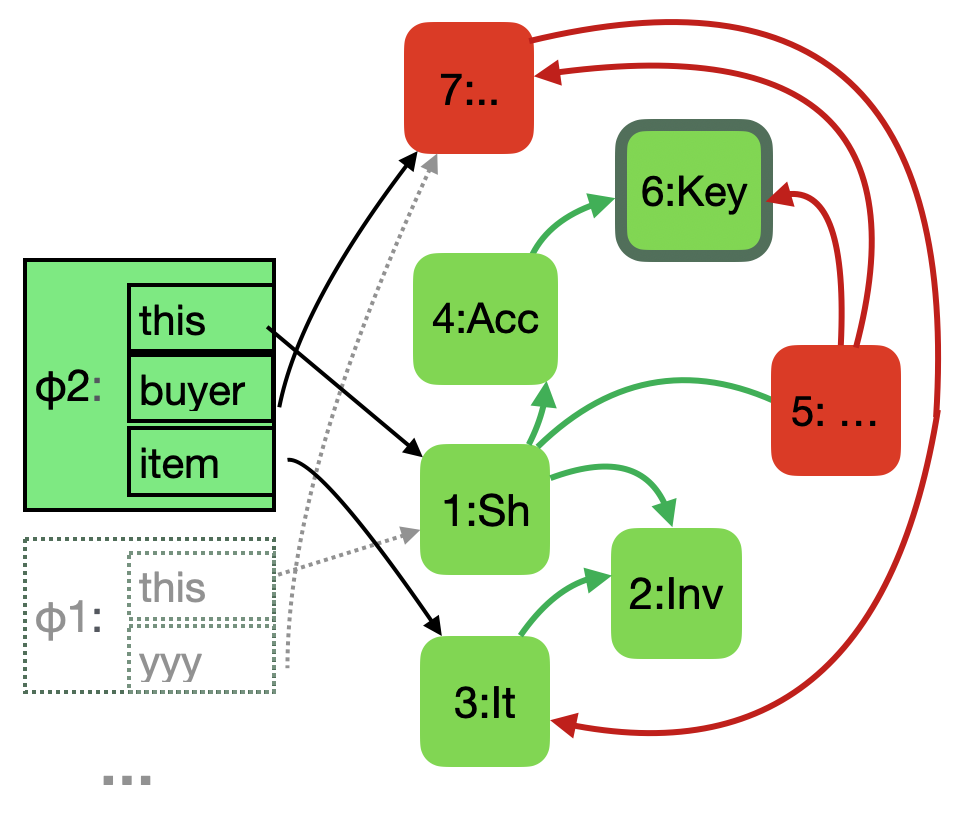}
}
&
\resizebox{3.5cm}{!}{
\includegraphics[width=\linewidth]{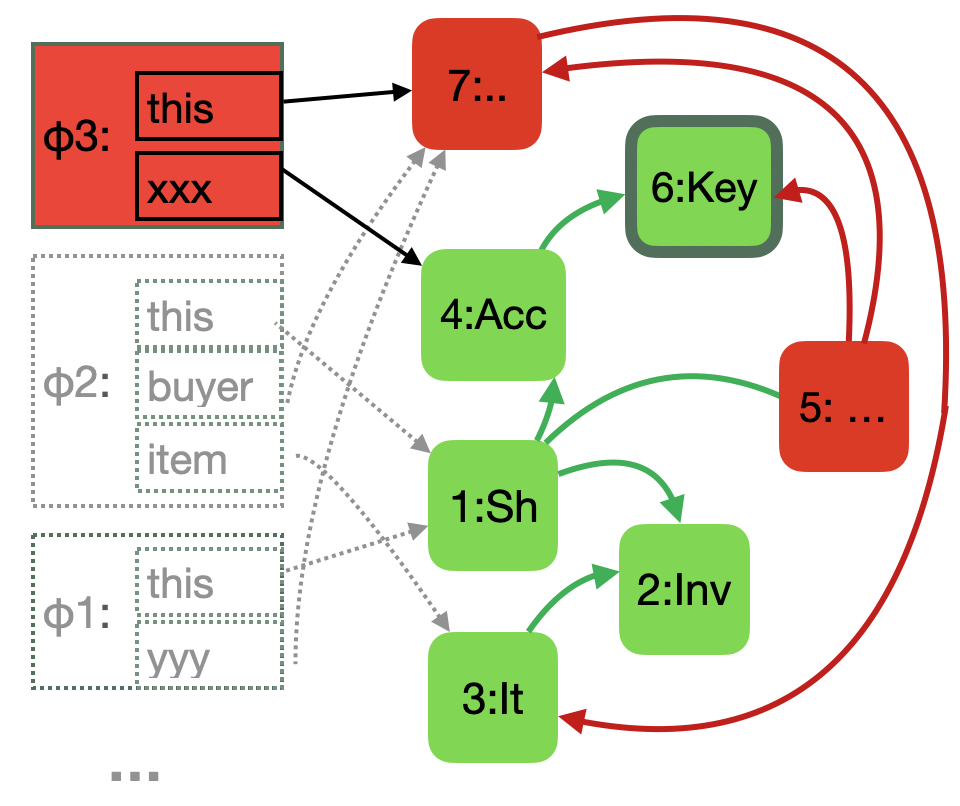}
}
\\
\hline
$\sigma_1$, top frame $\phi_1$  
&
$\sigma_2$, top frame $\phi_2$ 
&
$\sigma_3$, top frame $\phi_3$ 
\\
\hline 
\end{tabular}
\caption{\textit{The current point of execution} before \prg{buy}, during \prg{buy}, and during \prg{pay}.  Frames $\phi_1$, $\phi_2$ are green  as their receiver (\prg{this}) is internal;  $\phi_3$ is red as its receiver is external. Continuations are omitted.}
 \label{f:CurrentPoint}
 \end{figure}

\subsubsection{Protection}
\label{sect:approach:protection}

 \begin{description}
\item[Protection] 
Object $o$ is \emph{protected  from} object $o'$, formally $\protectedFrom {o} {o'}$,  
if no external object indirectly accessible from $o'$ has direct access to $o$.
Object $o$ is \emph{protected}, formally ${\inside{\prg{\it{o}}}}$,  
if no external object indirectly accessible from the current frame has direct access to $o$,
 and  if the receiver is external then $o$ is not an argument.\footnote{An object has direct access to another object if it has a field pointing to the latter;
 it has indirect access to another object if there exists a sequence of field accesses (direct references) leading to the other object; 
 an object is indirectly accessible from the frame if one of the frame's variables has indirect access to it.} 
 More in Def.\ \ref{def:chainmail-protection-from}.   
\end{description}

Fig.\ \ref{fig:ProtectedBoth} illustrates \emph{protected} and \emph{protected
from}. Object $o_6$ is not protected in states $\sigma_1$ and
$\sigma_2$, but \emph{is} protected in state $\sigma_3$.
This is so, because the external object $o_5$ is indirectly accessible from the top frame in $\sigma_1$ and in  $\sigma_2$, but not
from the top frame in $\sigma_3$ --
in general, calling a method (pushing a frame) can only ever \emph{decrease} the set of indirectly accessible objects. 
Object $o_4$  is  protected in  states $\sigma_1$ and  $\sigma_2$,  and not  protected in state $\sigma_3$
because though neither object $o_5$ nor $o_7$ have direct access to $o_4$, in state $\sigma_3$ the receiver is external and $o_4$ is one of the arguments.

\begin{figure}[htb]
\begin{tabular}{|c|c|c|c|}
\hline
 heap
&
$\sigma_1$   
&
$\sigma_2$ 
&
$\sigma_3$ 
\\
\hline 
$... \models \protectedFrom {o_6} {o_4}$
&
$\sigma_1   \not\models \inside{o_6}$
&
$\sigma_2   \not\models \inside{o_6}$
&
$\sigma_3 \models \inside{o_6}$
\\
$... \not\models  \protectedFrom {o_6} {o_5}$
&
$\sigma_1  \models \inside{o_4}$
&
$\sigma_2 \models \inside{o_4}$
&
$\sigma_3 \not \models \inside{o_4}$
\\
\hline  
\end{tabular}
\caption{\textit{Protected from and Protected}. -- continuing from Fig, \ref{f:CurrentPoint}.
 }
   \label{fig:ProtectedBoth}
 \end{figure}

If a protected object $o$ is never passed to external objects (\ie never leaked)  then $o$ will remain protected during the whole execution of the current method,
including during any nested calls.
This is the case even if $o$ was not protected before the call to the current method.
We define  \emph{scoped invariants} to  describe 
{property  preservation within the current  and all its nested calls}.

\subsubsection{Scoped Invariants}
 
\label{sect:approach:scoped}
We build on the concept of history invariants \cite{liskov94behavioral,usinghistory,Cohen10} and define:

\begin{description}
\item[{Scoped invariants}]  
{$\TwoStatesN  {\overline{x:C}}  {A}$} expresses that if an external {state} $\sigma$ 
 has objects $\overline x$ of class $\overline C$, and satisfies $A$, then all  {external} states which are part of
the \emph{scoped  future} of $\sigma$   will  {also} satisfy  {$A$}. 
The scoped future contains all  states which can be reached through any program execution steps, including further method calls and returns, but stopping just before returning  from the call active in $\sigma$ \footnote{{Here lies the difference to history invariants, which consider \emph{all} future states, including returning from the call active in $\sigma$.}}  --  \cf Def  \ref{def:shallow:term}.
\end{description}

Fig.\ \ref{fig:illusrPreserve} shows the  states of an unspecified execution starting at internal state $\sigma_3$ and terminating at internal state $\sigma_{24}$.
It distinguishes between steps within the same method ($\rightarrow$),
method calls ($\uparrow$), and method returns ($\downarrow$). 
The scoped future of $\sigma_6$ consists of $\sigma_6$-$\sigma_{21}$, but does not contain $\sigma_{22}$ onwards, since  scoped future stops before returning. 
  Similarly, the scoped future of $\sigma_9$ consists of $\sigma_9$, $\sigma_{10}$,  $\sigma_{11}$,  
  $\sigma_{12}$,   $\sigma_{13}$,  and $\sigma_{14}$, and does not include, \eg,  $\sigma_{15}$, or $\sigma_{18}$.
 
\begin{figure}[htb]
\begin{tabular}{|c|}
\hline  
\hspace{2cm}
\resizebox{9cm}{!}{
\includegraphics[width=\linewidth]{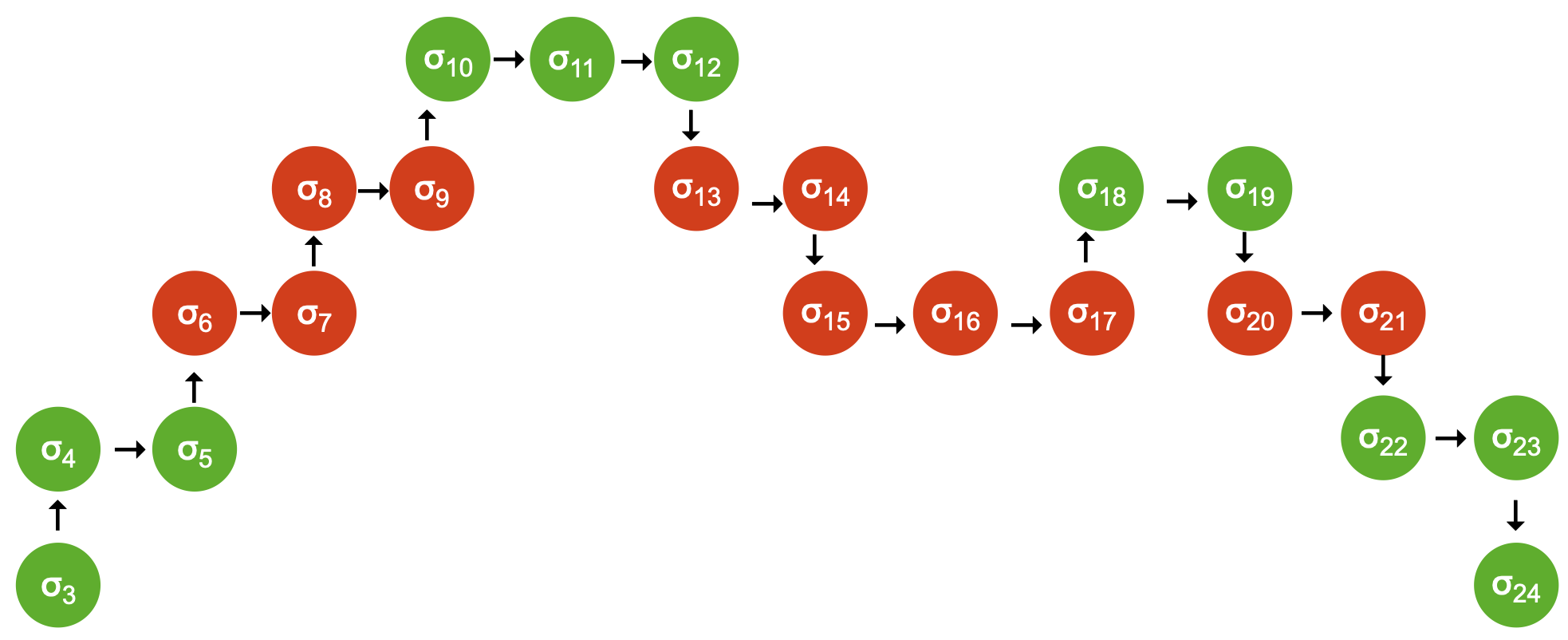}  
}
\hspace{2cm}
\\
\hline
\end{tabular}
   \caption{ \textit{Execution}.
Green disks represent internal states; red disks external states.
}
 \label{fig:illusrPreserve} 
 \end{figure}
 
The scoped invariant   ${\TwoStatesN {\overline {x:C}} {A_0}}$   guarantees that  if $A_0$ holds in $\sigma_8$, then it will also hold in $\sigma_9$,    $\sigma_{13}$, and $\sigma_{14}$;
 it   doesn't have to hold in $\sigma_{10}$, $\sigma_{11}$,  and $\sigma_{12}$ as these are internal states.
 Nor does it have have to hold at $\sigma_{15}$ as it is not part of $\sigma_9$'s scoped future.

    
\label{s:bank}

\begin{example}
\label{s:bankSpecEx}
The following scoped invariants\\
$\strut \SPSP\ \ \   S_1\ \  \triangleq \ \ \TwoStatesN {\prg{a}:\prg{Account}}  {\inside{\prg{a}}} $ 
\hspace{1.1cm}
$\strut  \SPSP\ \ \   S_2\ \  \triangleq \ \ \TwoStatesN  {\prg{a}:\prg{Account}}  {\inside{\prg{a.key}}} $ 
\\
$\strut  \SPSP\ \ \   S_3\ \  \triangleq \ \ \TwoStatesN{ \prg{a}:\prg{Account},\prg{b}:\prg{int} } {\inside{\prg{a.key}} \wedge \prg{a.\balance} \geq \prg{b} } $ 

\noindent
 guarantee that   accounts are not leaked  ($S_1$), \ \ keys are not leaked  ($S_2$), \ \  and that the balance does not decrease unless there is unprotected access to the key  ($S_3$).
%
\end{example} 

\noindent
This example illustrates three crucial properties of our invariants:
 \begin{customquote}
 \vspace{.05cm}
\noindent
\textbf{\emph{Conditional}}: They 
are \emph{preserved}, but unlike object invariants, they do not 
always hold.
\ \Eg   \prg{buy} cannot assume $\inside{a.\prg{key}}$ holds on entry, but   guarantees that if it holds on entry, then  it will still hold on exit.

\vspace{.05cm}
\noindent
\textbf{\emph{Scoped}}:  
They 
 are preserved during  execution of a specific method but not beyond its return. It is, in fact, expected that the invariant will eventually cease to hold after its completion. For instance, while $\inside{a.\prg{key}}$ may currently hold, it is possible that an earlier (thus quiescent) method invocation frame has direct access to $a.\prg{key}$ -- without such access, $a$ would not be usable for payments. Once control flow returns to the quiescent method (\ie enough frames are popped from the stack) $\inside{a.\prg{key}}$  will no longer hold.

 \vspace{.05cm}
\noindent
\textbf{\emph{Modular}}: They 
describe  externally observable effects (\eg  \prg{key} stays protected) across whole modules,
rather than  the individual methods (\eg\, \prg{set}) making up a module's interface.
Our example specifications will
characterize \emph{any}
module defining accounts with a  
 \balance~and a \prg{\password} -- even as ghost fields --  irrespective of their APIs. 
 \end{customquote}
 
 \begin{example}
 We   now use the features from the previous section to specify methods. 

{\sprepostShort
		{\strut \ \ \ \ S_4} 
		{   \protectedFrom{\prg{this}.\prg{\myAccount}.\prg{key}} {\prg{buyer}} \wedge \prg{this}.\prg{\myAccount}.\prg{\balance}=b
		 }
		{\prg{public Shop}} {\prg{buy}} {\prg{buyer}:\prg{external}, \prg{anItem}:\prg{Item} }
		{ 
		  \prg{this}.\prg{\myAccount}.\prg{\balance}\geq b
		} 
		}

\noindent
$S_4$  guarantees that if the  key was protected from \prg{buyer} before the call, then the balance will not decrease\footnote{We ignore the ... for the time being.}.
 It does \emph{not} guarantee \prg{buy} will only be called when $\protectedFrom{\prg{this}.\prg{\myAccount}.\prg{key}} {\prg{buyer}}$ holds: 
as a  public method,  \prg{buy}  can be invoked by external code that ignores all specifications.
\end{example}


\begin{example}
\label{e:versions}
We illustrate the meaning of our specifications using three versions 
($\ModA$,  $\ModB$, and $\ModC$) of the \Mshop module \cite{OOPSLA22};
these all share the same \prg{transfer} method to withdraw money: 

\begin{lstlisting}[mathescape=true, language=Chainmail, frame=none,numbers=none]
module $\ModA$      
  class Shop   ... as earlier ...
  class Account
    field blnce:int 
    field key:Key
    public method transfer(dest:Account, key':Key, amt:nat)
       if (this.key==key')  this.blnce-=amt; dest.blnce+=amt
    public method set(key':Key)
       if (this.key==null)  this.key=key'
\end{lstlisting}

\noindent Now consider modules \ModB and \ModC, which differ from \ModA
only in their \prg{set} methods. Whereas \ModA 's key is fixed once it is \prg{set},
\ModB allows any client to \prg{set} an account's key at any time, while
\ModC requires the existing key in order to replace it.

\vspace*{2mm}

\begin{tabular}{lll}
\begin{minipage}[b]{0.40\textwidth}

\begin{lstlisting}[mathescape=true, language=Chainmail, frame=none,numbers=none]
module $\ModB$
  ... as earlier ...
  public method set(key':Key)
    this.key=key'
\end{lstlisting}
\end{minipage}
&\ \ \  \ \   &%
\begin{minipage}[b]{0.48\textwidth}
\begin{lstlisting}[mathescape=true, language=chainmail, frame=none,numbers=none]
module $\ModC$
  ... as earlier ...
  public method set(key',key'':Key)
    if (this.key==key')  this.key=key''
\end{lstlisting}
\end{minipage} 
\end{tabular}

Thus, in all three modules, the key is a capability which \emph{enables} the withdrawal of the money. 
Moreover, in $\ModA$ and $\ModC$, the key capability
is a necessary precondition for withdrawal of money, while in 
 in $\ModB$ it is not. 
Using $\ModB$, it is possible to start in a state where the account's key is unknown, modify the key, and then withdraw the money. 
Code   such as 
\\ 
$\ \strut \hspace{.2in} $ \prg{k=new Key;  acc.set(k); acc.transfer(rogue\_accnt,k,1000)} 
\\ 
is enough to drain  \prg{acc} in \ModB without knowing the \password.\footnoteSD{CAREFUL: we had 
$\ \strut \hspace{.01in} $ \prg{an\_account.set(42); an\_account.transfer(rogue\_accnt,42)} but this was type incorrect!}
Even though  \prg{transfer} in  \ModB is ``safe'' when considered in isolation, it is not safe when considered in conjunction with other methods from the same module. 

$\ModA$ and $\ModC$ satisfy $S_2$ and $S_3$, while $\ModB$ satisfies neither. 
So if $\ModB$ was required to satisfy either $S_2$ or $S_3$ 
then it would  be rejected by our inference system as not safe.
None of the three versions satisfy $S_1$ because \prg{pay}  could leak 
an \prg{Account}.

\end{example}
 

  \subsection{2$^{nd}$ Challenge:  A Hoare Logic for Adherence to Specifications}  
 \label{sec:howSecond}

\subsubsection{Hoare Quadruples} Scoped invariants require quadruples, rather than  classical triples.
Specifically, \\
  $\strut \ \hspace{4cm} {\TwoStatesN  {\overline{x:C}}  {A}}$\\
 asserts that if an external {state} $\sigma$ 
 satisfies  $\overline {x:C} \wedge A$, then all its \emph{scoped} external future  states will   also  satisfy  $A$. 
For example, if $\sigma$ was an external state executing a call to \prg{Shop::buy}, then a \emph{scoped} external future  state
 could be reachable during execution of the   call \prg{pay}.
This implies that we consider not only states at termination but also external states reachable
 \emph{during} execution of  statements. 
To  capture this, we extend   traditional Hoare triples to quadruples of  form\\
 $\strut \ \hspace{4cm} \quadruple {A} {\, stmt\, }{A'} {A''}$\\  
 promising that if a state satisfies $A$ and executes $stmt$, any terminating state will satisfy $A'$, and 
 and  any intermediate external states reachable during execution of $stmt$ satisfy    $A''$ -- \cf Def. \ref{def:hoare:sem}.
 
\vspace{.1cm}

We assume an  underlying   Hoare logic  of  triples  
$ M \vdash_{ul} \{ \ A\ \} {\ stmt\ }\{\ A'\ \} $,
which does not {have} the concept of protection -- 
that is, the assertions $A$ in the $\vdash_{ul}$-triples do not mention protection.
We then embed  the $\vdash_{ul}$-logic into the quadruples logic 
( $\vdash_{ul}$-triples  whose   statements do not contain method calls give rise to quadruples in our logic -- see rule below).
We  extend assertions $A$ so they may mention protection and add rules about protection
 (\eg newly created objects are protected -- see rule below), and
 add   usual conditional and substructural rules.
 More in Fig.\ref{f:protection}  and \aref{Fig. 15}{\ref{f:substructural:app}}.

\noindent
\small
\noindent
\begin{center}
$  
\begin{array}{lcr}
\inferruleNoName 
	{  
	  M \vdash_{ul} \{ \ A\ \} {\ stmt\ }\{\ A'\ \}   \ \ \ \ stmt\  \mbox{calls no methods}  
	}
	{\hprovesN{M}  {A} {\ stmt \ } {A'} {A''} } 
&   &
\inferruleNoName 
	{ 
	 	
	}  	 
	{	 
 	\hprovesN  {M}  
	                {  true  }  
 			   {\  u = \prg{new}\ C \ }
 			   {\  \inside{u}\  }  { \ A \ }
	}
\end{array}
$
\end{center}

\normalsize

\subsubsection{Well-formed modules and External Calls} A module is well-formed, if  its specification is well-formed,   its public methods preserve   the module's scoped invariants, and  all  methods satisfy their specifications - \cf  Fig.  \ref{f:wf}.
{\Eg to prove  that  \prg{Shop::buy} satisfies {$S_4$}, taking   $stmts_{b}$ 
for the    body of \prg{buy},  we  have to prove:}
\\
%
%
 \vspace{.05cm}
  \begin{minipage}{.05\textwidth}
   \textbf{(0)}\ \ 
\end{minipage}
\hfill
\begin{minipage}{.95\textwidth}
\begin{flushleft}
$\{\  \   \external{\prg{buyer}} \ \wedge\ 
  \protectedFrom {\prg{this.\myAccount.key}} {\prg{buyer} } 
 \ \wedge\ \prg{this.\myAccount.\balance}= b  \ \  \}$\\
$\ \ \ \ \ \ \ \ \ \ \ \ {\ \ stmts_{b}   \ } $\\
$  \{\  \ \  {\prg{this.\myAccount.\balance}} \geq  b \  \  \} \ \ ||\ \  \{ \ ... \ \} $ 
\end{flushleft}
\end{minipage}


 \label{sec:howThird}

{The proof}  {proceeds using}   {our Hoare quadruples logic}. 
The treatment of external calls is of special interest. 
{Here,} the challenge is 
the external call  {on line 9}. 
We need to establish the Hoare quadruple:

 \vspace{.05cm}
  \begin{minipage}{.05\textwidth}
   \textbf{(1)}\ \ 
\end{minipage}
\hfill
\begin{minipage}{.95\textwidth}
\begin{flushleft}
$\{\  \   \external{\prg{buyer}} \ \wedge\ 
 {\color{red} {\protectedFrom {\prg{this.\myAccount.key}}  {\prg{buyer} } }}
 \ \wedge\ \prg{this.\myAccount.\balance}= b  \ \  \}$\\
$\ \ \ \ \ \ \ \ \ \ \ \ {\ \prg{buyer.pay(this.accnt,price)}   \ } $\\
$  \{\  \ \  {\prg{this.\myAccount.\balance}} \geq  b \  \  \} \ \ ||\ \  \{ \ ... \ \} $ 
\end{flushleft}
\end{minipage}
\\
\noindent
which says that  if the shop's account's key is protected from \prg{buyer}, then the account's balance will not decrease after the call {of
\prg{pay}.}
 \vspace{.05cm}
 
To prove \textbf{(1)}, we aim to use $S_3$, but this is not straightforward: 
 $S_3$  requires 
 {\color{red} {$\inside{\prg{this.\myAccount.key}}$}}, 
 which is not provided by the precondition of \textbf{(1)} .
 More alarmingly,  
$\inside{\prg{this.\myAccount.key}}$ may \emph{not hold} at the time of the call.
For example, in state $\sigma_2$ (Fig. \ref{fig:ProtectedBoth}), which could initiate the call to \prg{pay}, we have $\sigma_2 \models \protectedFrom{o_4\prg{.key}} {o_7}$, but $\sigma_2 \not\models \inside{o_4\prg{.key}}$.

Fig. \ref{fig:ProtectedBoth} provides insights into addressing this issue. Upon entering the call, in state $\sigma_3$, 
we find that $\sigma_3 \models \inside{o_4\prg{.key}}$. More generally, if $\protectedFrom{\prg{this.\myAccount.key}}{\prg{buyer}}$ holds before the call to \prg{pay}, then $\inside{\prg{this.\myAccount.key}}$ holds upon entering the call.
This is because any objects indirectly accessible during \prg{pay} must have been indirectly accessible from the call's
receiver (\prg{buyer}) or  arguments (\prg{this.\myAccount} and \prg{price}) when \prg{pay} was called.
 
In general, if   $\protectedFrom{\prg{x}}{\prg{y}_i}$ holds  for all 
 $\prg{y}_i$, before a call $\prg{y}_0.m(\prg{y}_1, ..., \prg{y}_n)$, then $\inside{\prg{x}}$ holds  upon entering the call. 
Here we have  $\protectedFrom{\prg{this.accnt.key}}{\prg{buyer}}$ by precondition. We also  have that \prg{price} is a scalar and therefore 
$\protectedFrom{\prg{this.accnt.key}}{\prg{price}}$.  And the type information gives that
all fields transitively accessible from an \prg{Account} are scalar or internal; this gives that
$\protectedFrom{\prg{this.accnt.key}}{\prg{this.accnt}}$. 
This enables the application of $S_3$ in \textbf{(1)}. The corresponding Hoare logic rule is shown in Fig. \ref{f:external:calls}.

\subsection{Summary}

In  an open world, external objects can execute arbitrary code, invoke any public internal methods, access any other external objects, and even collude with each another.
The external code may be written in the same or a different programming language than the internal code -- all we need is that the platform protects direct external read/write of  the internal private fields, while allowing indirect manipulation through calls of public methods.

The conditional and scoped nature of our invariants is critical to their applicability.
Protection is a local condition, constraining accessible objects rather than imposing a structure across the whole heap.
Scoped invariants are likewise local: they do not preclude some effects from the whole execution of a program,
rather the effects are precluded only in some local contexts.
While $a.\prg{\balance}$ may decrease in the future, this can only happen in contexts where an external object has direct access to  $a.\prg{key}$. 
Enforcing these local conditions is the responsibility of the internal module:
precisely because these conditions are local, they can be enforced locally within a module,
irrespective of all the other modules in the open world.

\renewcommand{\LangOO}{\ensuremath{{\mathcal{L}}_{ul}}\xspace }

\section{The Underlying Programming Language \LangOO}  
\label{sect:underlying}

\subsection{\LangOO Syntax and Runtime Configurations}
\label{sub:Loo} 
This work is based on \LangOO, a {minimal}, imperative, sequential,  class based, typed, object-oriented language. 
We believe, however, that the work can easily be adapted to any capability safe language with some form of encapsulation,
and that it can also support inter-language security, provided that the 
platform offers means to protect a module’s private state; cf capability-safe hardware as in Cheri \cite{davis2019cheriabi}.
Wrt to encapsulation and  capability safety,  \LangOO supports private fields, private and public methods, unforgeable addresses, and no ambient authority (no static methods, no address manipulation).
To reduce the complexity of our formal models, as is usually done, \eg \cite{IgaPieWadTOPLAS01,DietlDrossopoulouMueller07a,ParBiePOPL05},  \LangOO lacks many common languages features, omitting static fields and methods, interfaces,
inheritance, subsumption, exceptions, and control flow.  
{In our examples, we use numbers and booleans -- these can be encoded.}
 
Fig. \ref{f:loo-syntax} shows the   \LangOO syntax. {Statements, $stmt$, are three-address instructions,   method calls, or empty, $\epsilon$.}  
Expressions, $\re$, are ghost code;  as such, they may appear in assertions but not in statements, and have no side-effects \cite{ghost,ghost:spirit}.
Expressions  may contain fields, $\re.f$, or ghost-fields, ${\re_0.gf(\overline {\re})}$.
The meaning of $\re$ is module-dependent; \eg$\prg{a}.\prg{\balance}$   is a field lookup  in \ModA , but in a module which stores balances in a table it would be a recursive lookup through that table  -- \cf example   in \S \aref{A.3 }{app:BankAccount:ghost}. 
  \footnote
{For convenience,   $\re.gf$ is short for $\re.gf()$. Thus,  $\re.gf$ may be simple field lookup  in some modules, or  ghost-field  in others. }
In line with most tools, we support ghost-fields, but they are not central to our work.

\LangOO states, $\sigma$,  consist of a  heap $\chi$ and a stack. 
{A stack  is a sequence of frames, $\phi_1\!\cdot\!...\!\cdot\! \phi_n$.}
A  frame, $\phi$, consists of a local variable map and a continuation, \ie {the}  
statements to be executed.
The top frame, \ie  the frame most recently pushed onto the stack,  in a state $(\phi_1\!\cdot\!...\!\cdot\! \phi_n, \chi)$ is $\phi_n$.

\begin{figure}[t]
\footnotesize{
$
 \begin{array}{ll}
 \begin{syntax}
\syntaxElement{Mdl}{Module Def.}
		{
		\syntaxline{\overline{C\ \mapsto\ CDef}}\endsyntaxline
		}
\endSyntaxElement\\
\syntaxElement{CDef}{Class Def.}
		{
		 \prg{class}\ C\ 
		\{\  \overline{fld}; \overline{mth};\  \overline{gfld};\  \}		
		}
\endSyntaxElement 
\\
\syntaxElement{mth}{Method Def.}
		{
		\syntaxline
		{ {p}\  \prg{method}\ m\ (\overline{x : T}){:T}\{\ s\ \} }
		\endsyntaxline
		}
\endSyntaxElement
\end{syntax}
 &    
\begin{syntax}
\syntaxElement{fld}{Field Def.}
		{\syntaxline
			{\prg{field}\ f\ :\ T}
		\endsyntaxline}
\endSyntaxElement 
\\
\syntaxElement{T}{Type}
		{
		\syntaxline
				{C}
		\endsyntaxline
		}
\endSyntaxElement
\\
\syntaxElement{p}{Privacy}
		{
		\syntaxline
		{\prg{private}}
		{\prg{public}}
		\endsyntaxline
		}
\endSyntaxElement 
\end{syntax}

 \end{array}
 $
\\
\[
\begin{syntax}
\syntaxElement{stmt}{ } 
		{
				\syntaxline
				{{x:=y}}
				{{x:=v}}
				{x:=y.f}
				{x.f:=y}
				{x:=y_0.m(\overline{y})}
				{x:=\prg{new} \ {C} }
				{ stmt;\ stmt }
				  { \epsilon }
			       \endsyntaxline
		}
\endSyntaxElement
\\
\syntaxElement{gfld}{Ghosts}
		{\syntaxline
			{\prg{ghost}\ gf(\overline{x : T})\{\ \re\ \} : T}
		\endsyntaxline}
\endSyntaxElement\\
\syntaxElement{{\re}}{{{ }}} 
		{
		\syntaxline
				{x}
				{v}
				{\re.f}
				{{\re.gf(\overline {\re})}}
		\endsyntaxline
		}
\endSyntaxElement
\\
\end{syntax}
\]
\\
 $
 \begin{array}{lcl}
 \begin{syntax}
\syntaxElement{\sigma}{Program {State}}
		{
		\syntaxline
		{( \overline \phi, \chi )}
		\endsyntaxline
		}
\endSyntaxElement 
\\
\syntaxElement{\phi}{Frame}
		{
		\syntaxline
		{  (\  \overline{x\mapsto v};\ s \ ) }
		\endsyntaxline
		}
\endSyntaxElement\\
\syntaxElement{\chi}{Heap}
		{(\  \overline{\alpha \mapsto o}\ )}
\endSyntaxElement
\end{syntax}
 & \hspace{1cm} & 
\begin{syntax}
\syntaxElement{{C, f, m, gf, x, y}}{ }
		{{Identifier}}
\endSyntaxElement\\
\syntaxElement{o}{Object}
		{(\ C;\  \overline{f \mapsto v} \ )}
\endSyntaxElement\\
\syntaxElement{v}{Value}
		{
		\syntaxline
				{\alpha} 
				{\nul}
		\endsyntaxline
		}
\endSyntaxElement
\end{syntax}

 \end{array}
 $
 }
\caption{\LangOO Syntax. We use $x$, $y$, $z$ for variables, \ $C$, $D$ for class identifiers, $f$ for field identifier, ${gf}$ for ghost field identifiers, $m$ for method identifiers, $\alpha$ for addresses.
}
\label{f:loo-syntax}
\end{figure}

\paragraphsd{Notation.} We adopt the following unsurprising notation:
\label{s:notation}
\begin{itemize}
\item
{An object is uniquely identified by the address that points to it. We shall be talking of objects $o$, $o'$ when talking less formally, and of addresses, $\alpha$, $\alpha'$, $\alpha_1$, ...  when more formal.}
\item
$x$, $x'$, $y$, $z$, $u$, $v$, $\va$  are {variables}; \ 
$dom(\phi)$ and $Rng(\phi)$ indicate the variable map in $\phi$; \ \ $dom(\sigma)$ and $Rng(\sigma)$ indicate the variable map in the top frame in $\sigma$
\item
$\alpha \in \sigma$ means that $\alpha$ is defined in the heap of $\sigma$, and $x\in \sigma$ means that $x\in dom(\sigma)$.
Conversely,  $\alpha\notin\sigma$ and $x\notin\sigma$ 
 have the obvious meanings.
$\interpret{\sigma}{\alpha}$  is $\alpha$; and $\interpret{\sigma}{x}$  is the value to which  $x$  is mapped in the top-most frame of $\sigma$'s stack, 
and $\interpret{\sigma}{e.f}$ looks up in $\sigma$'s heap the value of $f$ for the object  $\interpret{\sigma}{e}$.
\item 
$\phi[x \mapsto \alpha]$ updates  the variable map  of $\phi$,  
and  $\sigma[x \mapsto \alpha]$ updates the top frame of $\sigma$. \
$A[\re/x]$ is textual substitution where we replace all occurrences of $x$ in $A$ by $\re$. 
\item
As usual, $\overline q$ stands for  sequence $q_1$, ... $q_n$, where $q$ can be an address, a variable,    a frame, an update or a substitution.
Thus,   $\sigma[\overline{x \mapsto \alpha}]$ and $A[ \overline{e/y}]$ 
have the expected meaning.
\item
$\phi.\prg{cont}$ is the continuation of frame $\phi$, and  $\sigma.\prg{cont}$ is the continuation in the top frame.
\item
$text_1 \txteq text_2$ expresses that $text_1$ and $text_2$ are  the same text.
\item
We define the depth of a stack as $\DepthFs {\phi_1...\phi_n} \triangleq n$. For states, $\DepthSt {(\overline \phi, \chi)} \triangleq  \DepthFs {\overline \phi}$.
The  operator $\RestictTo  \sigma k$ truncates the stack up to the $k$-th frame: 
 $\RestictTo {(\phi_1...\phi_k...\phi_n,\chi)} {k}  \triangleq   (\phi_1...\phi_k,\chi)$
\item
{ $\vs(stmt)$ returns the variables which appear in $stmt$. For example, $\vs(u:=y.f)$=$\{u,y\}$.}
\end{itemize}

\subsection{\LangOO Execution}
\label{sect:execution}

{Fig. \ref{f:loo-semantics} describes \LangOO execution}  by a small steps operational semantics with shape  $\leadstoOrig  {\Mtwo} {\sigma}   {\sigma'}$.
  $\Mtwo$ stands for one or more modules, where a
  module,  $M$, maps class names to class definitions. 
  {
  The functions $\class{\sigma}{x}$, $\meth{\Mtwo}{C}{m}$,
  { $fields(\Mtwo,C)$,}
    $\Same {x} {y} {\sigma}{\Mtwo}$, and $\Formals {\sigma}  {\Mtwo}$,
return the class of $x$, the method $m$ for class $C$, {the fields for class $C$,} whether $x$ and $y$ belong to the same module, and 
 the formal parameters of the method currently executing in $\sigma$ -- \cf Defs
\aref{A.2 -- A.7}{\ref{def:class-lookup} -- \ref{def:params}}. 
Initial states, $\initial{\sigma}$, contain a single frame 
with single variable \prg{this} pointing to a single object 
in the heap 
and a continuation, \cf \aref{A.8}{def:initial}.
}

\begin{figure}[bt]
\begin{minipage}{\textwidth}
\footnotesize{
\begin{mathpar}
\infer
	{
	{\sigma.\prg{cont}  \txteq  x := y.f; stmt} \ \ \ \ \ \ \ 
	 x \notin \Formals {\sigma} {\Mtwo} \\  
	 \Same {\prg{this}}  {y}  {\sigma} {\Mtwo}
	}
	{\exec{\Mtwo}{\sigma}{\sigma[x\mapsto  \interpret{\sigma}{y.f} \} ][\prg{cont} \mapsto stmt ] }}
	\quad(\textsc{Read})
	\and
\infer
	{
	\sigma .\prg{cont}  \txteq  x.f := y; stmt \ \ \ \ \ \ \ 
	\Same {\prg{this}}  {x}  {\sigma} {\Mtwo} 
	}
	{\exec{\Mtwo}{\sigma}{\sigma[\interpret{\sigma}{x}.f \mapsto\ \interpret{\sigma}{y} ][\prg{cont}\mapsto stmt]}}
	{}
	\quad(\textsc{Write})
	\and
\infer
	{
	\sigma.\prg{cont}\  \txteq\  x := {\prg{new}\ C }; stmt \ \ \ \ \ \ \ 
	 x \notin \Formals {\sigma} {\Mtwo} \\ 
	\textit{fields}(\Mtwo,C)=\overline{f} \\
	\alpha \mbox{ fresh in } \sigma
	}
	{\exec{\Mtwo}{\sigma}{\sigma[x\mapsto \alpha][\alpha  \mapsto  (\ C;\  \overline{f\, \mapsto \, \nul} \  ) ] [\prg{cont}\mapsto stmt] }}
	\quad(\textsc{New})
\and
\infer
	{
	   {\phi_n}.\prg{cont}  \txteq   u := y_0.m(\overline{y}); \_ \ \ \ \ \ \ \ 
	   u \notin \Formals {( {\overline{\phi}\cdot\phi_n},  \chi)} {\Mtwo} 	   \\
    \meth{\Mtwo}{\class{(\phi_n,\chi)}{y_0}}{m} = p \ C\!::\!m(\overline{x : T}){:T}\{\, stmt\, \}\ \ \ \\
        	{{p=\prg{public}  \vee   \Same{\prg{this}} {y_0} {(\phi_n,\chi)}{\Mtwo} }} 
	}
	{\exec{\Mtwo}{ ( {\overline{\phi}\cdot\phi_n},  \chi) }{{(\overline{\phi}\cdot\phi_n\cdot(\  \prg{this}\, \mapsto\, \interpret{\phi_n}{y_0},\overline{x\, \mapsto\, \interpret{\phi_n}{y}}; \ stmt\ )},\chi)}}
	\quad(\textsc{Call})
	\and
\infer
	{
		\phi_{n+1}.\prg{cont} \txteq  {\epsilon} \\  
	\phi_n.\prg{cont}   \txteq  {x := y_0.m(\overline y)}; stmt  
			}
	{\exec{\Mtwo}{({\overline{\phi} \cdot \phi_n \cdot \phi_{n+1}}, \chi) }{({\overline{\phi}\cdot \phi_{n}[x \mapsto \interpret{\phi_{n+1}}{\prg{res}}][\prg{cont} \mapsto {stmt} ]}, \chi)}}
	\quad(\textsc{Return})
\end{mathpar}
 }
\end{minipage}
 \caption{\LangOO operational Semantics}
\label{f:loo-semantics}
\end{figure}

The semantics is unsurprising:  
The  top frame's continuation {(${\sigma.\prg{cont}}$)} contains the statement to be  executed next.  
We dynamically enforce a simple form of module-wide privacy: 
Fields may be read or written only if they {belong to an object (here $y$)} whose class comes from the same module as the  class of the object 
reading or writing the fields ($\prg{this}$). \footnote{More fine-grained privacy, \eg C++ private fields or ownership types, would provide all
the guarantees needed in our work.} 
Wlog, {to simplify some proofs} we  require, as in Kotlin, that method bodies do not assign to formal parameters.

Private methods may be called only if the class of 
the callee ({the object whose method  is being called -- here $y_0$}) 
comes from the same module as the  class of the caller (here $\prg{this}$).
Public methods may always be called.
When a method is called, a new frame is pushed onto the stack; this frame  maps \prg{this} and the formal parameters to  the values for the receiver and other arguments, and the continuation to the body of the method. 
Method bodies are expected to store their return values in the {implicitly defined} variable \prg{res}\footnote{For ease of presentation, we omit assignment to \prg{res} in methods returning \prg{void}.}. 
  When the continuation is  empty ($\epsilon$), the frame is popped and the value of \prg{res}
 from the popped frame  is stored  in the variable map of the top frame.

{Thus, when $\leadstoOrig {\Mtwo}{\sigma}   {\sigma'} $ is within the same method we have  $\DepthSt {\sigma'}$= $\DepthSt {\sigma}$;\  when it is a call we have
 $\DepthSt {\sigma'}$= $\DepthSt {\sigma}+1$; \ and when it is a return we have  $\DepthSt {\sigma'}$= $\DepthSt {\sigma}-1$.}
Fig. \ref{fig:illusrPreserve}  from \S \ref{s:outline} 
distinguishes 
  ${\leadstoN}$ {execution} steps into: 
steps within the same  call ($\rightarrow$),\   entering a method  ($\uparrow$),\    returning from a method  ($\downarrow$).
Therefore $\leadstoOrig {\Mtwo}{\sigma_8}   {\sigma_9} $ is a step within the same call, 
$\leadstoOrig {\Mtwo}{\sigma_9}   {\sigma_{10}} $ is a method entry with $\leadstoOrig {\Mtwo}{\sigma_{12}}   {\sigma_{13}} $
the corresponding return. 
In general,  $\leadstoOrigStar  {\Mtwo} {\sigma}   {\sigma'}$ may involve {any}  number of  calls or returns: \eg
$\leadstoOrigStar  {\Mtwo} {\sigma_{10}}   {\sigma_{15}}$,   involves no calls and two returns.

%

\subsection{Fundamental  Concepts}
\label{s:auxiliary}

The novel features of our assertions — protection and scoped invariants  
— are both defined in terms of the current point of execution.
Therefore, for the semantics of our   assertions we need to represent calls and returns, scoped execution, and (in)directly accessible objects.
 
 \subsubsection{Method Calls and Returns} These  are characterized through pushing/popping   frames :
$ \PushSF  {\phi} {\sigma}$ pushes 
frame $\phi$ onto the stack of $\sigma$, while
${\sigma\, \popSymbol}$   pops the top frame 
and updates the continuation and variable map.


\begin{definition}
\label{def:push:frame}
Given a state $\sigma$, and a frame $\phi$,  we define
\begin{itemize}
\item
 $ \PushSF  {\phi} {\sigma} \ \triangleq \ ({\overline{\phi}\cdot\phi}, \chi)$ \ \ \  if \ \ \  $\sigma=(\overline{\phi}, \chi)$.
\item
$ \ \sigma\, \popSymbol \ \ \  \triangleq\   { (\overline{\phi}\cdot (\phi_n[\prg{cont}\mapsto stmt][x \mapsto \interpret {\phi_n}{\prg{res}}]), \chi)}$ \ \ \  if \\
 $\strut \hspace{1.5cm} \sigma=(\overline{\phi}\cdot\phi_n\cdot\phi_{n+1}, \chi)$, and $\phi_n(\prg{cont})\txteq x:= y_0.m(\overline y); stmt $
\end{itemize}
 \end{definition}

 \noindent Consider Fig. \ref{fig:illusrPreserve}  again: $\sigma_8 = \PushSF  {\phi} {\sigma_7}$ for some $\phi$, {and}  $\sigma_{15}$=$\sigma_{14} \popSymbol$.

 \subsubsection{Scoped Execution}
 \label{sect:bounded}

In order to give semantics to scoped invariants (introduced in \S  \ref{sect:approach:scoped} and to be fully defined  in Def.  \ref{def:necessity-semantics}), we need a new definition of execution, called \emph{scoped execution}.

 \renewcommand{\EarlierS}[2]{\DepthSt{#1} \leq \DepthSt{#2}}
 \renewcommand{\NotEarlierS}[2]{\DepthSt{#1} \not\leq \DepthSt{#2}} 
 
\begin{definition}[Scoped Execution] Given modules $\Mtwo$, and states $\sigma$, $\sigma_1$, $\sigma_n$, and $\sigma'$, we define:
\label{def:shallow:term}
 
\begin{itemize}

  \item
{  $\leadstoBounded  {\Mtwo} {\sigma}   {\sigma'} \ \ \   \,   \ \ \ \triangleq \ \ \  \leadstoOrig {\Mtwo} {\sigma} {\sigma'} \, \wedge\, 
 \EarlierS {\sigma}  {\sigma'} $}
  \item
{  $\leadstoBoundedStar {\Mtwo}  {\sigma_1}  {\sigma_n}  \ \ \,  \ \    \ \triangleq  \ \ \  {\sigma_1} = {\sigma_n}\ \ \vee \ \  \exists \sigma_2,...\sigma_{n-1}.\forall i\!\in\! [1..n)[\  \leadstoOrig {\Mtwo}  {\sigma_i}  {\sigma_{i+1}}\  \wedge\  \EarlierS{\sigma_1} {\sigma_{i+1}} \ ]$ }
\item
  $\leadstoBoundedStarFin {\Mtwo}  {\sigma}  {\sigma'}  \  \,  \ \  \ \triangleq  \ \ \  \leadstoBoundedStar {\Mtwo}  {\sigma}  {\sigma'}  \ \wedge\ \
 {\DepthFs \sigma = \DepthFs {\sigma'} \ \ \wedge \ \ \sigma'.\prg{cont}=\epsilon  } $
 \end{itemize}
\end{definition}

Consider    Fig. \ref{fig:illusrPreserve} :
Here $\EarlierS {\sigma_8} {\sigma_9}$
and thus $\leadstoBounded   {\Mtwo} {\sigma_8} {\sigma_9}$.
Also,  $\leadstoOrig {\Mtwo} {\sigma_{14}}  {\sigma_{15}}$  \
  but  $\NotEarlierS {\sigma_{14}} {\sigma_{15}} $
  (this step returns from the active call in $\sigma_{14}$),
  and hence   $\notLeadstoBounded  {\Mtwo}  {\sigma_{14}}   {\sigma_{15}}$. 
Finally, even though $\DepthSt {\sigma_8} = \DepthSt {\sigma_{18}}$
 and $\leadstoOrigStar {\Mtwo} {\sigma_8}  {\sigma_{18}}$, we have  
 $\notLeadstoBoundedStar {\Mtwo} {\sigma_8}   {\sigma_{18}}$:
This is so, because the execution $\leadstoOrigStar {\Mtwo} {\sigma_8}  {\sigma_{18}}$ goes through the step
$\leadstoOrig {\Mtwo} {\sigma_{14}}  {\sigma_{15}}$ and  $\NotEarlierS {\sigma_{8}} {\sigma_{15}} $
 (this step returns from the active call in  $\sigma_8$).

\vspace{.1cm}
{The relation $\boundedTrans$ contains more than the transitive closure of  $\bounded$.
\Eg, ${\leadstoBoundedStar  {\Mtwo}  {\sigma_9}  {\sigma_{13}}}$, even though ${\leadstoBounded   {\Mtwo}  {\sigma_{9}}  {\sigma_{12}}}$  and ${\notLeadstoBoundedStar   {\Mtwo}  {\sigma_{12}}  {\sigma_{13}}}$.} 
Lemma \ref{l:params:do:not:change} says that the value of the parameters does not change during  execution of the same method. 
Appendix \aref{B}{\ref{app:aux}} discusses proofs, and further properties.


\begin{lemma}
\label{l:params:do:not:change} 
 
For all $\Mtwo$, $\sigma$, $\sigma'$:
 $\strut \hspace{0.3cm} \leadstoBoundedStar {\Mtwo}  {\sigma}  {\sigma'} \ \wedge  \ \DepthFs{\sigma}=\DepthFs {\sigma'}\ \ \ \Longrightarrow\ \  \ \forall x\in \Formals {\Mtwo} {\sigma}.[ \ \interpret \sigma x = \interpret {\sigma'} x \ ]$

\end{lemma}

\subsubsection{Reachable  Objects, Locally Reachable Objects, and Well-formed States}

To define protection (no external object indirectly accessible from the top frame
has access to the protected object, \cf \S~\ref{sect:approach:protection}) we first define
reachability. 
%
%
An object
$\alpha$ is
 \emph{locally reachable}, i.e. $\alpha \in  \LRelevantO   \sigma $, if it is reachable from the top frame on the stack of $\sigma$.
 
\begin{definition} We define 
\begin{itemize}
\item
{{$\Relevant {\alpha} {\sigma}  \ \ \ \ \ \triangleq \ \  \  \{ \ \alpha' \, \mid \ \exists n\!\in\!\mathbb{N}.\exists f_1,...f_n. [\ \interpret {\sigma} {\alpha.f_1...f_n} = \alpha'\ ] \ \}$}}.
\item
$ \LRelevantO   \sigma  \ \  \ \triangleq \ \  \  \{ \ \alpha \ \mid \ { \exists x\in dom(\sigma) \wedge \alpha \in \Relevant {\interpret  {\sigma} x}
{\sigma} \ \} } $.
\end{itemize}
\end{definition}


In well-formed states, $\Mtwo \models \sigma$,    the value of a parameter in  any callee   (${\RestictTo  \sigma {k}}$) is also the 
 value of some variable in the caller (${\RestictTo  \sigma {k\!-\!1}}$),
and any address reachable from any frame (${\LRelevantO   {\RestictTo  \sigma {k}} }$) is reachable from some formal parameter of that frame.

\begin{definition}[Well-formed states]
\label{def:wf:state}
 For modules $\Mtwo$, and  states $\sigma$, $\sigma'$:

$\Mtwo \models \sigma \ \ \triangleq \ \  \forall k\in \mathbb{N}. [ \  1<k\leq \DepthFs {\sigma} \ \Longrightarrow $\\
$\strut \hspace{2.25cm}[\ \ \forall   x \in   \Formals {\RestictTo  \sigma {k}} {\Mtwo}.[\ \exists y. \ \interpret {\RestictTo  \sigma {k}}  {x} = \interpret {\RestictTo  \sigma {k\!-\!1}}  {y} \ ]$ $\hspace{0.3cm} \ \ \  \wedge$  
\\
$\strut \hspace{2.5cm}  {\LRelevantO   {\RestictTo  \sigma {k}} } = \bigcup_{z \in {\Formals {\RestictTo  {\sigma} {k}} {\Mtwo}}}  
{\Relevant { \interpret   {\RestictTo  \sigma {k}}  {z}}\sigma }  \  \ \ \ \ \ ] $
\end{definition}

Lemma  \ref{lemma:relevant} 
says that 
(\ref{wf:preserve}) execution preserves well-formedness, and 
(\ref{oneLR}) any object which is locally reachable after pushing a frame was locally reachable before pushing that frame.
 
\begin{lemma}
\label{lemma:relevant}
\label{l:wf:state}
\label{lemma:push:N}
For all modules $\Mtwo$, states $\sigma$, $\sigma'$,   and frame $\phi$:
\begin{enumerate}
\item
\label{wf:preserve}
$\Mtwo \models \sigma \ \wedge \ {\exec{\Mtwo} {\sigma} {\sigma'}}  \ \    \Longrightarrow \ \ \Mtwo \models \sigma' $
\item
\label{oneLR}
{$ \sigma'= \PushSF {\phi} {\sigma}   \ \wedge  \   \Mtwo \models \sigma' 
 \ \  \Longrightarrow\ \ \LRelevantO {\sigma'}\  \subseteq \LRelevantO   {\sigma}$}

\end{enumerate}
\end{lemma}


 \label{s:underlying}
   
\section{Assertions} 
\label{s:assertions}
\label{sub:SpecO}

Our assertions are    standard  (\eg properties of the values of expressions,  connectives, quantification \etc)  or  about protection (\ie ${\protectedFrom{{\re}} {{\re}}} $ and 
$ {\inside {{\re}}} $).

\begin{definition}
\label{def:assert:syntax}
Assertions, $A$,  are defined as follows:

\label{f:chainmail-syntax}
$
\begin{syntax}

\syntaxElement{A}{}
		{
		\syntaxline
				{{\re}}
				{{\re} : C}
				{\neg A}
				{A\ \wedge\ A}
				{\all{x:C}{A}}
				{\external{{\re}}}
 				{\protectedFrom{{\re}} {{\re}}} 
				 {\inside {{\re}}} 
		\endsyntaxline
		}
\endSyntaxElement\\
\end{syntax}
$
\footnote{Addresses in assertions 
as \eg  in  $\alpha.blnce > 700$, 
are useful when giving semantics to universal quantifiers 
\cf Def. \ref{def:chainmail-semantics}.(\ref{quant1}), {when the local map changes \eg upon call and return, and in general,} for scoped invariants, \cf Def. \ref{def:necessity-semantics}.}

\vspace{.05cm}

{$\fv(A)$ returns the free variables in $A$; for example, $\fv(a\!:\!Account \wedge \forall b:int.[a.\balance = b])$=$\{ a \}$.} 

\end{definition}

\forget{
\noindent
\textbf{NOTES}  \notesep 
 Assertions  may contain addresses; \eg   $\alpha.bal > 700$. 
{While addresses make little sense in user-written assertions, they are useful when giving semantics to universal quantifiers 
\cf Def. \ref{def:chainmail-semantics}.(\ref{quant1}), {when the local map changes \eg upon call and return, and in general,} for two-state invariants, \cf Def. \ref{def:necessity-semantics}.(2).}
\notesep The syntax does  not distinguish between fields and ghost fields. For instance, $\prg{a}.\prg{\balance}$ may, in some modules (\eg in \ModA), be a field lookup, while in others (e.g. when the balance is defined though an entry in a lookup table) may involve executing a ghost function. 
}

\begin{definition}[Shorthands] 
{We write $\internal{\re}$ for $\neg (\external {\re})$}, and
$\extThis$. resp. {$\intThis$} for $\external{\prg{this}}$ resp. $\internal{\prg{this}}$. 
Forms such as $A \rightarrow A'$,  $A \vee A'$, and $\exists x:C.A$  can be encoded.
\end{definition}

\label{def:chainmail-semantics-all}
\label{dup:def:chainmail-semantics}
\noindent
Satisfaction  of Assertions by a module and a state is expressed  through \ $\satisfiesA{M}{\sigma}{{A}}$ \  and defined by cases on the shape of $A$, in definitions \ref{def:chainmail-semantics}  and 
 \ref{def:chainmail-protection}.
 {$M$} is used 
 to look up the definitions of ghost fields, and to find class definitions to determine whether an object is  external.
 
\footnoteSD{say why we split the def into three defs.} 
\noindent

\subsection{Semantics of Assertions  -- First Part}
\label{sect:semantics:assert:standard}

To determine satisfaction of an expression, we    use the evaluation relation, $\eval{M}{\sigma}{e}{v}$,
which says that the expression $e$ evaluates
to value $v$ in the context of state $\sigma$ and module $M$.
{Ghost fields may be recursively defined, thus evaluation of $\re$ might}
 not  terminate. Nevertheless, the logic of  assertions 
 remains classical because recursion is restricted to expressions. 
\footnoteSD{
The semantics of $\hookrightarrow$ {is} unsurprising 
(see {the appendices 
\cite{necessityFull}).}
We have taken this approach from \citeasnoun{FASE}, which also contains a mechanized Coq proof that assertions are classical \cite{coqFASE}. } 

\begin{definition}[Satisfaction 
of Assertions -- first part] 
\label{def:chainmail-semantics}
We define satisfaction of an assertion $A$ by a 
state $\sigma$ with 
 module $M$ as:
\begin{enumerate}
\item
\label{cExpr}
$\satisfiesA{M}{\sigma}{{\re}}\ \ \ \triangleq \ \ \   \eval{M}{\sigma}{{\re}}{\true}$
\item
\label{cClass}
$\satisfiesA{M}{\sigma}{{{\re}} : C}\ \ \ \triangleq \ \ \   \eval{M}{\sigma}{{\re}}{\alpha}\   \wedge \ \class{\alpha} {\sigma}= C$
\item
$\satisfiesA{M}{\sigma}{\neg A}\ \ \ \triangleq \ \ \   {M},{\sigma}\not\models{A}$
\item
$\satisfiesA{M}{\sigma}{A_1\ \wedge\ A_2}\ \ \ \triangleq \ \ \   \satisfiesA{M}{\sigma}{A_1} \   \wedge \ \satisfiesA{M}{\sigma}{A_2}$
\item
\label{quant1}
$\satisfiesA{M}{\sigma}{\all{x:C}{A}} \ \ \ \triangleq \ \ \   
\forall \alpha.[\   \satisfiesA {M}{\sigma} {\alpha:C}  \ \Longrightarrow   \ \satisfiesA{M}{\sigma} {A[\alpha/x]} \ ] $

\item
\label{cExternal}
$\satisfiesA{M}{\sigma}{\external{{\re}}} \ \ \ \triangleq \ \ \  \exists C.[\ \satisfiesA{M}{\sigma}{{{\re}} : C} \ \wedge \ C \notin M \ ]$
\end{enumerate}
\end{definition}

Note that while execution takes place in the context of one or more modules, $\Mtwo$, satisfaction of assertions considers \emph{exactly one} module  $M$ -- the internal module. 
{$M$} is used  to look up the definitions of ghost fields, and to 
 determine whether objects are  external.

\subsection{Semantics of Assertions - Second Part}  

\label{sect:protect}

In \S \ref{sect:approach:protection} we 
introduced protection -- we will now formalize this concept. 

 An object is protected from another object, $\protectedFrom{{\alpha}} {{\alpha_{o}}}$, if 
the two objects are not equal, and no external object reachable from $a_o$ has a field pointing to  $\alpha$.
This ensures that the last element on any path leading from $\alpha_o$ to $\alpha$
in an internal object.

An object is protected,  $\inside{{\alpha}}$,  if no external object
reachable from any of the current frame's arguments has a field
pointing to $\alpha$; and furthermore, if the receiver is external,
then
no parameter to the current method call directly refers to $\alpha$.
This  ensures that no external object reachable from the current
receiver or arguments can ``obtain'' $\alpha$, where obtain $\alpha$ is either direct access through a field,  
 or by virtue of the method's receiver 
 being able to see all the arguments.

\begin{definition}[Satisfaction 
of Assertions  -- Protection] 
\label{def:chainmail-protection-from}
\label{sect:semantics:assert:prtFrom}
 \label{def:chainmail-protection}
-- continuing definitions in \ref{def:chainmail-semantics}:
\begin{enumerate}
 \item
 \label{cProtected}
$\satisfiesA{M}{\sigma} {\protectedFrom \re  {\re_o}}$ $\ \  \triangleq\ \ $ 
$\exists \alpha, \alpha_{o}. [$ 
  \begin{enumerate}
 \item
 $ \eval{M}{\sigma}{{\re}}{\alpha}\ \ \ \wedge\ \ \eval{M}{\sigma}{{\re_o}}{\alpha_o} 
\ \ \ \wedge\ \  \alpha\neq \alpha_0$,
 \item
$\forall \alpha' \in {\Relevant {\alpha_o} {\sigma}}.  \forall f.[\ \ \satisfiesA{M}{\sigma}
{\external {\alpha'}  \ \ \Longrightarrow \ \    \interpret {\sigma} {\alpha'.f} \neq \alpha  }   \ ] $  $  \strut \hspace{.25cm} ] $.
\end{enumerate}

\item
 \label{sect:semantics:assert:prt}
$\satisfiesA{M} {\inside \re} $ $\ \  \triangleq\ \  $ 
$\exists \alpha. [$  
\begin{enumerate}
\item
$ \eval{M}{\sigma}{{\re}}{\alpha}$,
 \item 
 $\satisfiesA {M}{\sigma}{\extThis}\ \ \Longrightarrow\ \ \forall x\!\in\! \sigma.\ \satisfiesA{M}{\sigma}{x\neq \alpha}$,
 \item
$\forall \alpha'\in {\LRelevantO {\sigma}}.\forall f.[\  
 \satisfiesA{M}{\sigma}{ \external {\alpha'}  \ \ \Longrightarrow \ \  
  \interpret {\sigma} {\alpha'.f} \neq \alpha  }  \ ]$ 
 \strut \hspace{.25cm} ].   
  \end{enumerate} 
    \end{enumerate} 
 \end{definition} 
 
 We illustrate ``protected'' and ``protected from'' in Fig.   \ref{fig:ProtectedBoth} in \S \ref{s:outline}.
   and    Fig.    \aref{12 in App. C}{\ref{fig:ProtectedFrom} in App. \ref{appendix:assertions}}.
In general,  $\protectedFrom{{\alpha}} {{\alpha_{o}}}$ ensures that $\alpha_o$ will get access to $\alpha$ only if another object 
 grants that access.
Similarly, $\inside \alpha$ ensures that during execution of the current method, no external object will get direct access to $\alpha$ unless some internal object grants that access\footnote{This is in line with the motto "only connectivity begets connectivity" from \cite{MillerPhD}.}.
Thus, protection together with protection preservation  (\ie no internal object gives access) guarantee
lack of eventual external access.

\footnoteSD{JAMES' comment: If is possible that ``we'' do not know the complete heap (eg we only know about the green stuff.) how do we know whether an object is protected. The answer is that we do not know that it is protected, but we do know that our code guaranrtees preservation of protectedness.
}  
 
 \subsubsection{Discussion} 
Lack of  eventual 
direct access is a central concept in the verification of code with calls to and callbacks  from untrusted code.
It has already been over-approximated in several different ways, \eg
2nd-class \cite{rompf-second-class-oopsla2016,rompf-dont-pop-second-class-ecoop2022}
or borrowed (``2nd-hand'') references
\cite{boyland-promises-icse1998,boyland-aliasburying-spe2001},
 textual modules \cite{OOPSLA22},
information flow \cite{ddd}, runtime
checks \cite{secure-io-fstar-popl2024},
abstract data type exports \cite{vmsl-pldi2023},
  separation-based invariants 
Iris \cite{iris-wasm-pldi2023,cerise-jacm2024},
-- more in  \S~\ref{sect:related}.
In general, protection is applicable in more situations (i.e.\ is less
restrictive) than most of these approaches,
 although more restrictive than the ideal ``lack of eventual access''.

\noindent
\begin{flushleft}
\begin{tabular}{@{}lr@{}}
  \begin{minipage}{.85\textwidth}
   {An alternative definition might consider $\alpha$ as protected from $\alpha_o$, 
if   any path from $\alpha_o$ to $\alpha$ goes through at least one internal object.
With this definition, $o_4$ would be protected from $o_1$ in the heap shown here.
However,  $o_1$ can make a call to $o_2$, and  this call could  return $o_3$. 
Once $o_1$ has direct access to $o_3$, it can also get direct access to $o_4$. 
The example justifies our current definition.  
}
\end{minipage}
& 
\begin{minipage}{.18\textwidth}
\resizebox{2cm}{!}{
\includegraphics[width=\linewidth]{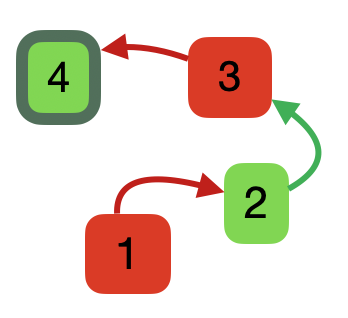}
} 
\end{minipage}
\end{tabular}
\end{flushleft}

 \subsection{Preservation of Assertions}
 \label{s:preserve}
 
Program logics require some form of framing, \ie conditions under which  satisfaction of  assertions is preserved across program execution. 
This is the subject of  the current subsection.

Def. \ref{def:as}  turns an assertion  $A$ 
to the equivalent variable-free from by replacing all free variables from $A$ by their values in $\sigma$. 
 Then,   Lemma \ref{l:assrt:unaffect}  says  that 
satisfaction of an assertion is not affected by replacing free  variables by their values, nor by changing the sate's continuation.

\begin{definitionAndLemma} $ $ ~
\label{def:as}
\label{lemma:addr:expr}
\label{l:assrt:unaffect}
For all $M$, $\sigma$,  $stmt$,   $A$, and $\overline x$ where  $\overline x = \fv{(A)}$:

\begin{itemize}  
\item

${\as  \sigma A} \ \triangleq \ \ A[\overline {{\interpret {\sigma} {x}}/x}]$\  
\item
$\satisfiesA{M}{\sigma}{A}   \ \ \ \Longleftrightarrow\ \ \ \satisfiesA{M}{\sigma}{{\as {\sigma} A} }    \ \ \ \Longleftrightarrow\ \ \  \satisfiesA{M}{\sigma[\prg{cont}\mapsto stmt]}{A}$ 
\end{itemize}
\end{definitionAndLemma}

%
%
%

 \noindent
We now move to assertion preservation across method call and return.  

\subsubsection{{Stability}} 
\label{s:preserve:call:ret}
In most program logics, satisfaction of  variable-free assertions  is preserved when pushing/popping frames
-- \ie immediately after entering a method or  returning from it.
But this is not  so for our assertions, where protection depends 
on the heap but also 
on the range of 
the top frame. \Eg  Fig. \ref{fig:ProtectedBoth}:\  
$\sigma_2 \not\models \inside {o_6}$, but after pushing a frame, we have $\sigma_3  \models \inside {o_6}$.

{Assertions} which do  not contain  $\inside {\_}$   are called $\Stable {\_}$, 
while assertions which do  not contain $\inside {\_}$ in \emph{negative} positions are called $\Pos {\_}$. 
Sect \ref{sect:Compare:stable:enc} shows some examples.
Lemma \ref{l:stbl:preserves} says that $\Stable{\_}$ assertions are  preserved when pushing/popping frames,
and $\Pos {\_}$ assertions are preserved when pushing  {internal} frames.
\Cf    \A\ \aref{D}{\ref{app:preserve}} for   definitions and proofs. 

\begin{lemma}
For all  states $\sigma$, frames $\phi$,   all assertions $A$ with  $\fv(A)=\emptyset $
\label{l:preserve:asrt}
\label{l:stbl:preserves} 
\begin{itemize}
\item 
$\Stable{A} \  \ \  \Longrightarrow  \  \ \  [\ \ M, \sigma \models A \ \ \Longleftrightarrow \ \  M,{\PushSLong \phi \sigma} \models A\ \ ]$
\item 
\label{l:preserve:asrt:two}
$\Pos{A}   \ \wedge    \ {M\cdot\Mtwo \models {\PushSLong \phi \sigma}}\  \wedge\  M, {\PushSLong \phi \sigma} \models  \intThis\  \wedge  \ M, \sigma \models A $
$\  \ \ \Longrightarrow \ \  \ M,{\PushSLong \phi \sigma} \models A\ $
\end{itemize}
\end{lemma}

While $Stb^+$ assertions \emph{are} preserved  when pushing  internal frames,   they  are \emph{not} necessarily preserved when pushing  external frames  
nor when popping frames   (\cf Ex. \ref{ex:pop:does:not:preserve}). 

{ 
\begin{example}
\label{push:does:not:preserve}
Fig. \ref{fig:ProtectedBoth} illustrates that \\ 
-- \textit{$Stb^+$  not necessarily preserved by External Push}: 
Namely, $\sigma_2 \models \inside {o_4}$,   pushing   frame $\phi_3$ with an external receiver 
 and $o_4$ as argument gives $\sigma_3$,  we have $\sigma_3 \not\models \inside {o_4}$.
%
\\
\label{ex:pop:does:not:preserve}
-- \textit{$Stb^+$  not necessarily preserved by Pop}:  
Namely,  $\sigma_3 \models \inside {o_6}$, returning from $\sigma_3$ would give  $\sigma_2$, and  we have  $\sigma_2 \not\models \inside {o_6}$.
\end{example}
}

{We work with  $Stb^+$  assertions   (the  $Stbl$ requirement is too strong). 
But  we need to address the lack of preservation of $Stb^+$ assertions  {for external method calls and returns}.
We do the former   through    \emph{adaptation}   ($\pushSymbolInText$ in Sect \ref{s:viewAndProtect}), and the latter through  
\emph{\strong satisfaction} (\S \ref{s:scoped:valid}).
  }

\subsubsection{{Encapsulation}} 
 \label{sect:Compare:stable:enc}
 
 As external code is unknown, it could, in principle, have unlimited effect and  invalidate any assertion, and thus make  reasoning about  external calls impossible.
However, because fields are private, assertions which  read internal fields only,  cannot be invalidated by external execution steps.
Reasoning about external calls relies on such   \emph{encapsulated} assertions. 

Judgment $M\ \vdash \encaps{A}$,  
{defined in} 
\aref{D.4}{\ref{d:encaps:sytactic}},  
{requires that} $A$ looks up the contents of
internal objects only,  
{that $A$} does not contain  $\inside {\_}$  in negative positions,  
{nor does it} contain  $\protectedFrom {\_} {\_}$  
{in any position}.  
Lemma \ref{lem:encap-soundness} says that 
$M\ \vdash \encaps{A}$  
{implies}
that any external scoped execution step which involves $M$ and any set of other modules $\Mtwo$  
{preserves} satisfaction of $A$. 

\begin{lemma}[Encapsulation] 
\label{d:encaps}  For all modules $M$, and assertions $A$:  
\label{lem:encap-soundness}

\begin{itemize}
\item
$ \proves{M}{\encaps{A}}   \ \Longrightarrow\ \ \forall \Mtwo, \sigma, \sigma'.[   \ \  \satisfiesA{M}{\sigma}{(A  \wedge \externalexec)}\  \wedge\ { \leadstoBounded {M\madd\Mtwo}  {\sigma}{\sigma'}} 
        \  \Longrightarrow\  
    {M},{\sigma'}\models {\as \sigma A} \ \  ]$
 \end{itemize}
  
 \end{lemma}
 
 \noindent 
 The example below compares the judgments $\Stable{\_}$, \ \ $\Pos {\_}$, \ and \ $\encaps{\_}$. 

 $
  \begin{array}{c}
  \begin{array}{ll}
  \ \ \  &
 \begin{array}{|c|c|c|c|c|c|}
 \hline
          & \ \ \prg{z.f}\geq 3  \ \ &  \ \   \inside{\prg{x}} \ \ & \ \ \neg(\ \inside{\prg{x}}\ ) \ \ & \ \ \protectedFrom {\prg{y}} {\prg{x}} \ \ &  \ \ \neg (\ \protectedFrom {\prg{y}} {\prg{x}}\ ) \ \ \\
   \hline
 \Stable  {\_}  &  \checkmark &  \times & \times &  \checkmark  &  \checkmark  \\
 \hline 
  \Pos {\_ } &  \checkmark &   \checkmark  & \times &  \checkmark  &  \checkmark  \\ 
  \hline 
  \encaps{\_ } &   \checkmark &  \checkmark & \times &   \times   &  \times \\ 
  \hline
 \end{array}
  \end{array}
    \end{array}
 $

\section{ Specifications}
\label{sect:spec}    
 
\subsection{\textbf{Syntax, Semantics, and Examples of Specifications}}

\begin{definition} [Specifications Syntax]     We define the syntax  of  specifications, $S$:
\label{f:holistic-syntax}
\[
\begin{syntax}
\syntaxElement{S}{ }
		  {\syntaxline
				{\  \TwoStatesN {\overline {x:C}} {A}\  }
 				{\ \mprepostN{A}{p\ C}{m}{y}{C}{A} {A}\ } 
				{\ S\, \wedge \, S\ }
		 \endsyntaxline
 		}
\endSyntaxElement\ 
\syntaxElement{p}{ } 
 	 {\syntaxline
                                  {\    \prg{private} \ } 	
				 {\   \prg{public} \ } 	
		 \endsyntaxline
 		}
\endSyntaxElement 
\end{syntax}
\]

\end{definition}

In Def. \ref{f:holistic-wff}  later on we describe  well-formedness of $S$, but  we first discuss  semantics and some examples.
%
\label{ssect:sem}
 %
We use quadruples involving states: 
${\satAssertQuadruple  \Mtwo  M     {A} \sigma {A'} {A''} }$ 
  says that   if $\sigma$ satisfies $A$, then any terminating scoped execution of its continuation (${\leadstoBoundedStarFin { \Mtwo\madd M}{\sigma}  {\sigma'} }$) will satisfy $A'$, and any intermediate reachable external state 
  (${\leadstoBoundedStar  {\Mtwo\madd M}{\sigma}  {\sigma''}}$) will satisfy the ``mid-condition'', $A''$.

\begin{definition} \label{def:hoare:sem}
\label{def:shallow:spec:sat:state}
For modules $\Mtwo$, $M$, state $\sigma$, and assertions $A$, $A'$ and  $A''$, we define:
\begin{itemize}
\item
$ {\satAssertQuadruple  \Mtwo  M     {A} \sigma {A'} {A''} } \ \ \triangleq \ \ \forall \sigma',\sigma''.[
$  \\
$\strut \hspace{.2cm} M,  \sigma \models  {A}   
  \  \ \Longrightarrow \ \   [ \ \  {\leadstoBoundedStarFin { \Mtwo\madd M}{\sigma}  {\sigma'} }\ \ \Longrightarrow\ \   M,  \sigma' \models  {A'}  
 \ \ \ \  ] \ \ \ \wedge$\\ 
 $\strut   \   \hspace{2.5cm}  [ \ \   {\leadstoBoundedStar  {\Mtwo\madd M}{\sigma}  {\sigma''} }\ \  \ \Longrightarrow\   \   M,  \sigma'' \models  {(\extThis \rightarrow {\as \sigma {A''}} )}\ \ \  ] \ \ \ \ \ ]$ 
\end{itemize} 
\end{definition}

{\begin{example}
Consider ${\satAssertQuadruple  {...}  {...}     {A_1} {\sigma_{4}} {A_2} {A_3} }$   
for Fig. \ref{fig:illusrPreserve}.
It  means  that  if  $\sigma_4$ satisfies $A_1$, 
then $\sigma_{23}$ will satisfy $A_2$, while $\sigma_6$-$\sigma_9$,\ $\sigma_{13}$-$\sigma_{17}$, and $\sigma_{20}$-$\sigma_{21}$ will satisfy $A_3$.
It does not imply anything about $\sigma_{24}$ because $\notLeadstoBoundedStar {...} {\sigma_4} {\sigma_{24}}$.
Similarly, if $\sigma_8$ satisfies $A_1$ 
then $\sigma_{14}$ will satisfy $A_2$, and  $\sigma_{8}$, $\sigma_{9}$, $\sigma_{13}$, $\sigma_{14}$  will satisfy $A_3$, while making no claims about  $\sigma_{10}$, $\sigma_{11}$, $\sigma_{12}$, nor  about $\sigma_{15}$ onwards.
 \end{example}}

{Now  we} 
 define    $\satisfies{M}{\TwoStatesN {\overline {x:C}} {A}}$ 
to mean that  if an external state $\sigma$ satisfies $A$, then all future external states reachable from $\sigma$—including nested  
 calls and returns but  {\emph{stopping} before}   returning from the active call in $\sigma$— also satisfy $A$. 
 And  $\satisfies{M} { \mprepostN {A_1}{p\ D}{m}{y}{D}{A_2} {A_3} }$ means that scoped execution of a call to $m$ from $D$   in  states satisfying $A_1$ leads to final states satisfying $A_2$ (if it terminates),
 and to intermediate external states satisfying $A_3$.

\begin{definition}  [Semantics of  Specifications]
We define $\satisfies{M}{{S}}$ by cases over $S$:  

\label{def:necessity-semantics}

\begin{enumerate}
 \item
\label{def:necessity-semantics-first}
 $\satisfies{M}{\TwoStatesN {\overline {x:C}} {A}} \ \  \ \triangleq   \ \ \ {\forall   \Mtwo,  \sigma.[\ {\satAssertQuadruple  \Mtwo  M    {\extThis \wedge \overline {x:C} \wedge A} \sigma {A} {A} }\ ].}$
  \item
   \label{def:necessity-semantics-second}
 $\satisfies{M} { \mprepostN {A_1}{p\ D}{m}{y}{D}{A_2} {A_3} }\  \ \ \   \triangleq  $ \\ 
$\strut \ \   \forall   \Mtwo,  \sigma, y_0,\overline y.[\ 
 \ \sigma.\prg{cont}\txteq {u:=y_0.m(y_1,..y_n)} \ \ \Longrightarrow \ \ $\\
$\strut  \ \ \   \ \ \ \ \ \ \ \ \   \ \ \  \ \ 
\ \ \ \ \ \ \ \ \ {\satAssertQuadruple  { \Mtwo} {M} { y_0\!:\!D, \overline {y\!:\!D}   \wedge   A[y_0/\prg{this}]}  {\ \sigma\ }   {A_2[u/res,y_0/\prg{this}] }{A_3 } } \  \ \  ]  $   
 \item
 $\satisfies{M}{S\, \wedge\, S'}$\ \ \  \ \ \  $\triangleq$  \  \ \  \   $\satisfies{M}{S}\ \wedge \ \satisfies{M}{S'}$
\end{enumerate}
\end{definition}

Fig. \ref{fig:illusrPreserve} in  \S \ref{sect:approach:scoped}  illustrated  the meaning of ${\TwoStatesN {\overline {x:C}} {A_0}}$. 
Moreover, $M_{good} \models S_2 \wedge S_3 \wedge S_4$, and  $M_{fine} \models S_2 \wedge S_3 \wedge S_4$,
 while $M_{bad} \not\models S_2$.
More examples  below and in  
\A ~\aref{E.1.1}{\ref{app:spec}}.

{
 \begin{example}[Scoped Invariants and Method Specs]
 \label{example:twostate}
 \label{example:mprepostl}
 $S_5$  says 
  that   non-null keys are immutable:
 \\
 \begin{tabular}{lcll}
$\strut \ \ \ \ \ \ \ \ S_5$ & $\triangleq$   & ${\TwoStatesN {\prg{a}:\prg{Account},\prg{k}:\prg{Key}}  {\prg{null}\neq \prg{k}=\prg{a.\password}}} $  \end{tabular}
\\
$S_9$    guarantees that \prg{set} preserves the protectedness of any account, and any key.  \\
   {\sprepost
		{\strut \ \ \ \ \ \ \ \ \ S_9} 
		{  a:\prg{Account}, a':\prg{Account}\wedge  \inside{a}\wedge  \inside{a'.\prg{key}} }
		{\prg{public Account}} {\prg{set}} {\prg{key'}:\prg{Key}}
		{   \inside{a}\wedge  \inside{a'.\prg{key}}  }
		{   \inside{a}\wedge  \inside{a'.\prg{key}} }
}

\noindent
Note that  $a$, $a'$ are disjoint from \prg{this} and the formal parameters of \prg{set}. 
In that sense, $a$ and $a'$ are universally quantified; a call of \prg{set} will preserve protectedness for \emph{ all} accounts and their keys. 

\end{example}

\subsection{\textbf{Well-formed Specifications}} We now define what it means for a specification to be well-formed:

\begin{definition}
 {\emph{Well-formedness}} of specifications,  $\vdash S$,  is   defined by cases on $S$:
\label{f:holistic-wff}

\begin{itemize}
\item
  $\strut \  \vdash {\TwoStatesN {\overline {x:C}} {A}} \ \ \ \triangleq\  \ \ \fv(A)\subseteq\{  \overline x \}\,\wedge\, {M \vdash \encaps {{\overline {x: C}} \wedge A}} $.
    \item
 $\strut \ \vdash {\mprepostN{\overline{x:C'} \wedge A}{p\ C}{m}{y}{C}{A'} {A''}}\ \ \ \triangleq$\\
$\strut \hspace{0.7cm}  [ \ \ {\small{   {\prg{res},\prg{this}\!\notin\! \overline{x}, \overline{y}}\ \wedge\  {\fv(A)\!\subseteq\! \overline x, \overline y, \prg{this}}\     \wedge\    \fv(A')\!\subseteq \! \overline{x}, \overline{y}, \prg{this}, \prg{res}\   \wedge\   \fv(A'')\!\subseteq\!  {\overline{x}} }} $\\
    $\strut \hspace{0.7cm} \ \ \  \wedge\  \Pos {A } \, \ \wedge \ \,  \Pos {A'} \, \ \wedge \  \,  M \vdash \encaps  {\overline {x: C'} \wedge A''}\ \ \  ]$ 
  \item  $\strut \   \vdash S\, \wedge \, S' \ \ \triangleq \ \  \vdash S\, \ \wedge\,  \ \vdash S'  $.  
\end{itemize}

\end{definition}

Def \ref{f:holistic-wff}'s  requirements about  free variables are relatively straightforward -- \cf  \S \aref{E.1.1}{\ref{wff:spec:free:more}}.

Def \ref{f:holistic-wff}'s  requirements about encapsulation are motivated by  Def.   \ref{def:necessity-semantics}. If  $\overline {x:C}\wedge A$ in the scoped invariant  were not encapsulated,  then it could be invalidated by some external code, and it would be impossible to ever satisfy Def.   \ref{def:necessity-semantics}(\ref{def:necessity-semantics-first}). 
Similarly, if a method specification's mid-condition, $A''$, could be invalidated by some external code, then it would be impossible to ever satisfy Def.   \ref{def:necessity-semantics}(\ref{def:necessity-semantics-second}).

Def \ref{f:holistic-wff}'s  requirements about stability are motivated by our Hoare logic rule for internal calls,   {\sc{[Call\_Int]}}, Fig \ref{f:internal:calls}. The requirement    $\Pos {A}$ for the method's precondition  gives that $A$ is preserved when an internal frame is pushed, \cf Lemma \ref{l:preserve:asrt}.
The requirement     $\Pos {A'}$ for the method's postcondition gives,  in the context of \strong satisfaction,  that $A'$ is preserved when an internal frame is popped, \cf Lemma \aref{G.42}{\ref{l:calls:return:deep}}. This is crucial for soundness of  {\sc{[Call\_Int]}}.

{In  \S \A ~\aref{E.1.2}{\ref{wff:spec:encaps:more}} we discuss the encapsulation requirements, and why these do not restrict expressiveness.}

\subsection{\textbf{Discussion}}  

\paragraphSD{Difference with Object and History Invariants.}  Our scoped invariants are similar to, but different from, history invariants  and object invariants.
We compare through an example:
%
 %
 
\begin{flushleft}
\begin{tabular}{@{}lr@{}}
  \begin{minipage}{.80\textwidth}
State $\sigma_a$ is making a call  transitioning to $\sigma_b$,    execution of $\sigma_b$'s continuation   eventually results in $\sigma_c$, and $\sigma_c$ returns  to $\sigma_d$. 
Suppose all four states are external, and the module guarantees $\TwoStatesN {\overline{x:Object}} {A}$, and $\sigma_a \not\models A$, but $\sigma_b \models A$. 
Scoped invariants 
ensure  $\sigma_c \models A$, but allow   $\sigma_d \not\models A$.\end{minipage}
& 
\begin{minipage}{.18\textwidth}
\resizebox{2cm}{!}{
\includegraphics[width=\linewidth]{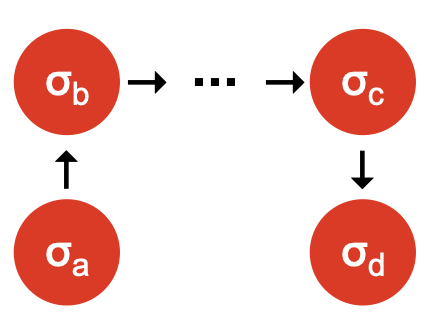}
} 
\end{minipage}
\end{tabular}
\end{flushleft}

{\emph{History  invariants}} \cite{liskov94behavioral,usinghistory,Cohen10}, instead, consider {all  future states including any method returns}, and therefore {would  require that   $\sigma_d \models A$. Thus, they are,}  for our purposes,  both
 \emph{unenforceable} and overly \emph{restrictive}.\ \  \emph{Unenforceable}: \ Take $A \txteq \inside{\prg{acc.key}}$,  assume  in $\sigma_a$ a path to an external object which has access to $\prg{acc.key}$, assume that path is unknown in $\sigma_b$: then, the transition from $\sigma_b$ to $\sigma_c$ cannot eliminate that path—hence, $\sigma_d \not\models \inside{\prg{acc.key}}$.\ \  \emph{Restrictive}:\ Take $A \txteq \inside{\prg{acc.key}}\wedge a.\prg{blnce}\geq b$; then,  requiring  $A$   to hold in all states from $\sigma_a$ until termination would prevent all future withdrawals from $a$, rendering the account useless.

{\emph{Object invariants}}  \cite{Meyer92,MeyerDBC92,BarDelFahLeiSch04,objInvars,MuellerPoetzsch-HeffterLeavens06}, on the other hand, expect 
invariants to hold in all (visible) states,
here would require,  \eg that $\sigma_a \models A$. Thus, they  are 
\emph{inapplicable} for us: They would require, \eg, that for all 
 $\prg{acc}$, in all (visible) states, $\inside{\prg{acc.key}}$, and thus prevent \emph{any} withdrawals from \emph{any} account in \emph{any} state.

\paragraphSD{Difference between Postconditions and Invariants.}
In all  method specification examples so far, the post-condition and   mid-condition were identical.
However, this need not be so. 
Assume a method \prg{tempLeak} defined in \prg{Account}, with an external argument \prg{extArg}, and  method body:
\\
$\strut \hspace{1cm} \prg{extArg.m(this.key); this.key:=new Key}$} 
\\
Then, the assertion   $ \inside{\prg{this.key}}$  is  invalidated by the external call \prg{extArg.m(this.key)}, but is  established by \prg{this.key:=new Key}.
Therefore, $ \inside{\prg{this.key}}$  is a valid post-condition but not a valid   mid-condition.
The specification of \prg{tempLeak} could be\\
$
{\sprepost
		{\strut \ \ \ \ \ \ \ \ \ S_{\prg{tempLeak}}} 
		{  \ \prg{true}\  }
		{\prg{public Account}} {\prg{tempLeak}} {\prg{extArg}:\prg{external}}
		{  \  \inside{\prg{this.key} }\  }
		{  \  \prg{true}\  }
}
$

\footnoteSD{First bullet: This means that we require all objects to satisfy even if not locally relevant. Second Bullet: notice that we are asking for globally relevant objects}  
\footnoteSD{{TODO: Make an example that demonstrates the difference if in the second bullet we had asked for locally relevant objects ${\overline o}$.}}
\footnoteSD{TODO: explain why we did not require the stronger $\leadstoFin{M_{ext}\!\circ \!M}{\sigma}{\sigma'}$ rather than $\leadstoBoundedStar {M_{ext}\!\circ \!M}{\sigma}  {\sigma'}$.}

\newcommand{\paragraphSDD}[1]{\vspace{.02cm}{\textit{#1}}}
 
 \paragraphSD{Expressiveness} 
In  \S \aref{E.2}{\ref{app:expressivity}} we argue the expressiveness of our approach  through a sequence of capability patterns studied in related approaches from the literature  
 \cite{OOPSLA22,dd,VerX,irisWasm23,ddd} and written in our specification language.
These approaches 
 are based on temporal logics \cite{VerX,OOPSLA22}, or on extensions of Coq/Iris \cite{dd,irisWasm23,ddd}, and do not
offer Hoare logic  rules  for external calls.

\section{Hoare Logic} 
\label{sect:proofSystem}
 We   develop  an inference system for adherence to our specifications.
 We distinguish three phases:
 
\vspace{.05cm}
\textit{\textbf{First Phase:}} We assume  an underlying Hoare logic, ${M \vdash_{ul}  \{A\} {\ stmt\ } \{A'\} }$, and extend it to a logic  
  ${\hproves{M}  {A} {\ stmt\ }{A'} }$  with 
the expected meaning, \ie 
(*) execution of statement $stmt$ in a state satisfying 
$A$ will lead to a state satisfying  
$A'$.
These triples only apply to   $stmt$'s  that  do not contain method calls  (even internal) -- this is so, because method calls may make further calls to   external methods.
In our extension we  introduce   judgements   which 
talk about protection.

\vspace{.05cm}

\textit{\textbf{Second Phase:}} We develop a logic of quadruples ${\hprovesN{M}  {A} {\ stmt\ }{A'} {A''}}$. These promise  (*) and  
in addition,   that (**) any intermediate external states reachable during execution of that $stmt$
 satisfy the \emph{mid-condition}  $A''$.  
 We incorporate the triples from the first phase,       
introduce  mid-conditions, give the usual substructural rules, and deal with method calls. 
For internal   calls we use the methods' specs. 
For external   calls, we   use 
 the module's invariants. 
 
 \vspace{.05cm}
 
\textit{\textbf{Third Phase:} } We prove adherence  to  our specifications. 
For method specifications we require that the body maps the precondition to the postcondition and preserves the method's  mid-condition. 
For module invariants we require that they  are preserved by all public methods of the module.

 \vspace{.1cm}

\noindent
\textit{Preliminaries:}  
 The judgement    ${\promises M S}$ expresses that $S$ is part of $M$'s specification.  
In particular, it allows   \emph{safe  renamings}. 
These renamings are   a convenience, akin to the Barendregt convention, and  allow simpler Hoare rules  -- \cf Sect. \ref{sect:wf},
Def. \aref{F.1}{\ref{d:promises}}, and Ex. \aref{F.2}{\ref{e:rename}}. 
We also require an underlying Hoare logic with judgements $M \vdash_{ul} \{ A \} stmt \{ A' \}$ 
-- \cf Ax. \aref{F.3}{\ref{ax:ul}}.

\subsection{First Phase: Triples}
\label{s:hoare:first}

In  Fig. \ref{f:underly}  we introduce our triples, of the form ${   \hproves{M}  {A} {stmt}  {A'}}$. 
These promise, as expected, that any execution of $stmt$ in a state   satisfying $A$ leads to a state satisfying $A'$.

{\small{
\begin{figure}[tht]
$
\begin{array}{c}
\begin{array}{lcl}
\inferruleSD{\hspace{0.5cm} [\sc{embed\_ul}]}
	{stmt  \ \mbox{contains no method call}   
        \\
        \Stable{A}  \ \  \Stable{A'}  \ \ \ \ \  \ \ { M \vdash_{ul} \{ \ A\ \} {\ stmt\ }\{\ A'\ \}}
	 	 }
	{\hproves{M}  {A} {\ stmt\ }{A'} } 
& &
\inferruleSD{\hspace{0.5cm} [\sc{prot-new}]}
	{
	\sdN{u \txtneq   x	 }
	}
	{	 
 	\hproves{M} 
 						{ \sdN{true} }  
 					{\  u = \prg{new}\ C \ }
 						  {  \inside{u}\  \wedge  \protectedFrom{u}{x}} 
 	} 
\\ \\
	{\inferruleSD{\hspace{0.5cm} [\sc{prot-1}]}
	{   stmt \ \mbox{ is free of  meth. cals, or assignment to $z$}
	\\
	{\hproves{M}  {{A \wedge} \re\!=\!z} {\ stmt\ }{ \re\!=\!z} }
	}
	{\hproves{M} 
						{\  {A \wedge} \inside{\re}  \ }
						{\  stmt \ }
						{\  \inside{\re}\ }
	}
}
& &
      {\inferruleSD{\hspace{0.5cm} [\sc{prot-2}]}
	{ stmt \mbox{ is either $x:=y$ or $x:=y.f$}\ \ \ \ \    z,z' \txtneq x 
		\\
	{\hproves{M} 
						{  {A \wedge} z\!=\!\re \wedge z'\!\!=\!\! \re' }
						{\ stmt }
						{   z\!=\! \re \wedge z'\!\!=\!\! \re' }
	}
	}
	{
	\hproves{M} 
						{\ {A \wedge} \protectedFrom{\re}{\re'}\ }
						{\ stmt }
						{\ \protectedFrom{\re}{\re'}\ }
	}
}
\\ \\
      {\inferruleSD{\hspace{0.5cm} [\sc{prot-3}]}
	{\sdN{x \txtneq   z	 } }
	{\hproves{M} 
						{\ \protectedFrom{y.f}{z}\ }
						{\ x =y.f\ }
						{\ \protectedFrom{x}{z}\ }
	}
}
& &
        {\inferruleSD{\hspace{0.5cm} [\sc{prot-4}]}
	{ \ A\txteq\protectedFrom{x}{z}  \wedge   \protectedFrom{x}{y'} \ \ \ \  \ A'\txteq\protectedFrom{x}{z}\ }
	{\hproves{M} 
						{  \ A \  }
						{\ y.f=y'\ }
						{  \ A'\   }
	}
}
\end{array}
\end{array}
 $
\caption{Embedding the Underlying Hoare Logic, and Protection}
\label{f:protection}
\label{f:underly}
\end{figure}
}}

With rule {\sc{embed\_{ul}}} in Fig. \ref{f:underly},  any triple $ \{ A \} stmt \{ A' \} $  whose statement does not contain a method call, and which 
can be proven in the underlying Hoare logic, can also be proven in our logic.  
In \textsc{Prot-1}, we see that  protection of an object $o$ is preserved by internal code which does not call any methods: namely any heap modifications will
ony affect internal objects, and this will not expose new, unmitigated external access to $o$.
 \textsc{Prot-2}, \textsc{Prot-3} and \textsc{Prot-4} describe the preservation of relative protection.
Proofs of soundness for these rules can be found in App.\aref{G.5}{\ref{s:app:protect:lemmas}}.

\subsection{Second Phase: Quadruples}
\label{s:hoare:second}

\subsubsection{Mid-conditions and Substructural Rules}
We now introduce  quadruple rules.  
Rule {\sc{mid}} embeds  triples  ${\hproves{M}  {A} {\ s\ }{A'} }$  into quadruples ${\hprovesN{M}  {A} {\ s\ }{A'} {A''} }$.
This is sound, because $stmt$ is guaranteed not to contain method calls (by lemma \aref{F.5}{ \ref{l:no:meth:calls}}).

\begin{center}
$
\begin{array}{c}
\inferruleSD{[\sc{Mid}]}
	{\hproves{M}  {A} {\ stmt\ }{A'} }
	{\hprovesN{M}  {A} {\ stmt\ }{A'} {A''} }
  \end{array}
 $
 \end{center}
 
Substructural quadruple rules appear in  Fig.\aref{15}{ \ref{f:substructural:app}}, and are as expected: 
Rules   {\sc{sequ}} and {\sc{consequ}} are  the usual rules for statement sequences and consequence, adapted to quadruples.
Rule {\sc{combine}} combines two quadruples for the same statement into one.
Rule  {\sc{Absurd}} allows us to deduce anything out of \prg{false} precondition, and  {\sc{Cases}} allows for case analysis.
These  rules  apply to \emph{any} statements -- even those containing method calls.

\subsubsection{Adaptation}

  \label{s:viewAndProtect}
 
In the outline of the Hoare proof of the external call in  \S  \ref{sec:howThird},  we saw that an assertion of the form $\protectedFrom x {\overline {y}}$ at the call site may imply $\inside{x}$ at entry to the call.
More generally,  the $\pushSymbolInText$ operator  adapts an assertion from the view of the callee to that of the caller,
 and is used in the Hoare logic for method calls. It is defined below.

\begin{definition}
\label{def:push}
[The $\pushSymbolInText$  operator]  

$
\begin{array}{c}
\begin{array}{l}
\begin{array}{rclcrcl}
  \PushAS y {\inside \re} & \triangleq &  \protectedFrom \re {\overline {y} }
  & \ \ \  \ &
  \PushAS y   {(A_1  \wedge  A_2)} & \triangleq &  (\PushAS y  { A_1})  \wedge  ( \PushAS y  {A_2} )  
\\ 
 \PushAS y {(\protectedFrom \re {\overline {u}})} &  \triangleq& \protectedFrom \re {\overline {u}} 
  & &
 \PushAS y  {(\forall x:C.A)} & \triangleq & \forall x:C.({\PushAS y A} )  
  \\  
  \PushAS y  {(\external \re)} &  \triangleq & {\external \re}  
  & & 
  \PushAS y  {(\neg A)} &  \triangleq & \neg( {\PushAS y A} )  
    \\
     \PushAS y  {\re} &  \triangleq&   \re 
    & &
    \PushAS y  {(\re:C)} &  \triangleq&   \re:C 
 \end{array}
\end{array}
\end{array}
$
\label{f:Push}
\end{definition}

Only the first equation in  Def.  \ref{def:push}  is interesting: for $\re$ to be {protected at  a} callee with arguments $\overline y$, it should be protected from   
these arguments -- thus
  $\PushAS y {\inside \re} = \protectedFrom {\re} {\overline {y}}$. 
The notation $\protectedFrom {\re} {\overline {y}}$   stands for $\protectedFrom \re {y_0}\  \wedge\  ...  \wedge \protectedFrom \re {y_n}$, assuming that $\overline y$=${y_0, ... y_n}$.

Lemma \ref{lemma:push:ass:state}  states that   
indeed, $\pushSymbolInText$ adapts assertions from the callee to the caller, and is the counterpart to the  
$\pushSymbol$.
{In particular:\ \ 
 \ (\ref{l:push:stbl}):\    $\pushSymbolInText$ turns an assertion into a stable assertion.
\ (\ref{lemma:push:ass:state:one}):\ If the caller,   $\sigma$,  satisfies  $\PushSLong  {Rng(\phi)} {A}$, then  the callee,   $\PushSLong {\phi} {\sigma}$, satisfies $A$.
\ \ (\ref{lemma:push:ass:state:two}): \ When returning from external states,  an assertion implies its adapted version.
 \ \ (\ref{lemma:push:ass:state:three}): \ When calling from external states, an assertion implies its adapted version. 
}

{
\begin{lemma} 
\label{lemma:push:ass:state}
For  states  $\sigma$, assertions $A$, 
so that $Stb^+(A)$ and $\fv(A)=\emptyset$,  
frame $\phi$,  variables $y_0$, $\overline y$: 

\begin{enumerate}
 \item
\label{l:push:stbl}
$\Stable{\,  \PushASLong {(y_0,\overline y)} A\, }$
\item
 \label{lemma:push:ass:state:one}
$M, \sigma \models \PushASLong  {Rng(\phi)} {A}\ \  \ \ \ \ \  \ \ \    \Longrightarrow  \ \ \ \ M,  \PushSLong {\phi} {\sigma}   \models A$
\item
 \label{lemma:push:ass:state:two}
$M,  \PushSLong {\phi} {\sigma}   \models  A  \wedge \extThis    \ \  \ \ \  \  \Longrightarrow  \ \ \ \ M, \sigma \models \PushASLong  {Rng(\phi)} {A}$
 \item
 \label{lemma:push:ass:state:three}
$M, \sigma  \models  A  \wedge \extThis  \ \wedge \ M\cdot\Mtwo \models \PushSLong {\phi} {\sigma}   \ \  \ \ \  \  \Longrightarrow  \ \ \ \ M, \PushSLong {\phi} {\sigma} \models \PushASLong  {Rng(\phi)} {A}$
\end{enumerate}
\end{lemma}
}

{
Proofs 
 in \A\ \aref{F.5}{\ref{appendix:adaptation}}. Example \ref{push:does:not:imply}
 demonstrates the need for    \prg{extl}  
  requirement in  
  (\ref{lemma:push:ass:state:two}). 
}

\begin{example}[When returning from   internal states, $A$ does not imply $\PushASLong {Rng(\phi)} {A}$]  
\label{push:does:not:imply}   
In  Fig. \ref{fig:ProtectedBoth} we have
 $\sigma_2= \PushSLong {\phi_2} {\sigma_1}$, and    $\sigma_2 \models \inside{o_1}$,  and $o_1\!\in \!Rng(\phi_2)$,
 but $\sigma_1 \not\models \protectedFrom {o_1} {o_1}$. 
%
%
\end{example}

\subsubsection{Calls}
\label{s:calls}
is described in Fig. \ref{f:internal:calls}. {\sc{Call\_Int}}  
 for internal methods, whether public or private; \ 
and {\sc{Call\_Ext\_Adapt}} and {\sc{Call\_Ext}\_Adapt\_Strong} for  external methods.

\begin{figure}[htb]
{\small{
$\begin{array}{c}
 \inferruleSD{\hspace{4.7cm} [\sc{Call\_Int}]}
	{
	   	\begin{array}{c}
		\promises  M {\mprepostN{A_1}{p\ C}{m}{x}{C}{A_2} {A_3}}  \\
		{A_1' = A_1[y_0,\overline y/\prg{this}, \overline x]  \ \ \ \ \ \ \ \ \  A_2' = A_2[y_0,\overline y,u/\prg{this}, \overline x,\prg{res}]}  
		          	\end{array}
		}
	{  \hprovesN {M} 
						{ \  y_0:C,\overline {y:C} \wedge  {A_1'}\ }  
						 { \ u:=y_0.m(y_1,.. y_n)\    }
					         { \ {A_2'} \ } 
						{  \ A_3 \ }  
}
%
\\
 \\ 
 \inferruleSD{\hspace{4.7cm} [\sc{Call\_Ext}\_Adapt]}
 	{ 
	 \promises M   {\TwoStatesN {\overline {x:C}} {A}} 
        }
	{   \hprovesN{M} 
						{ \    { \external{y_0}} \,     \wedge \,  \overline{x:C}\  \wedge\ {\PushASLong {{(y_0,\overline {y})}}  {A}}  \ } 
						{ \ u:=y_0.m(y_1,.. y_n)\    }
						{ \   {\PushASLong {{(y_0,\overline {y})}}  A}  \ }
						{\  A \   }
	}	
\\
 \\ 
{
 \inferruleSD{\hspace{4.7cm} [\sc{Call\_Ext}\_Adapt\_Strong]}
 	{ 
	 \promises M   {\TwoStatesN {\overline {x:C}} {A}} 
        }
	{   \hprovesN{M} 
						{ \    { \external{y_0}} \,     \wedge \,  \overline{x:C}\ \wedge  A   \wedge\ {\PushASLong {{(y_0,\overline {y})}}  {A} }\  }   
						{ \ u:=y_0.m(y_1,.. y_n)\    }
						{ \   A \wedge {\PushASLong {{(y_0,\overline {y})}}  {A} } \  }  
						{\  A \   }	
}
}

\end{array}
$
}}
\caption{Hoare Quadruples for Internal and External Calls -- here $\overline y$ stands for $y_1, ... y_n$}
\label{f:internal:calls}
\label{f:external:calls}
\label{f:calls}
\end{figure}

For  internal calls, we  start, as usual,  by looking up the method's specification, and 
{substituting the formal  by the actual parameters parameters ($\prg{this}, \overline{x}$  by $y_0,\overline{y}$).} 
  {\sc{Call\_Int}} is as expected:   we  require the precondition, and guarantee the postcondition and mid-condition.
 {\sc{Call\_Int}} 
 {is} applicable whether the method is public or private.

For external calls, 
we consider the module's invariants. 
If the module promises to preserve $A$, \ie if  $\promises M   {\TwoStatesN {\overline {x:D}} {A}}$, {and  if its adapted version, $ \PushASLong {(y_0, \overline y)}{A}$},  holds before the call, then it also holds after  the call \ ({\sc{Call\_Ext\_Adapt}}).
{If, in addition, the un-adapted version also holds before the call, then it also holds after the call  ({\sc{Call\_Ext\_Adapt\_Strong}}).}


\vspace{.1cm}

Notice that   internal calls, {\sc{Call\_Int}},   require   the \emph{un-adapted} 
 method  precondition (\ie $A_1'$), while   external calls, both {\sc{Call\_Ext\_Adapt}} and {\sc{Call\_Ext\_Adapt\_Strong}},  require the 
 \emph{adapted} 
 invariant (\ie $ \PushASLong {(y_0, \overline y)}{A}$). 
{This is sound, because  internal callees preserve  
 $\Pos{\_}$-assertions} 
 -- \cf Lemma \ref{l:preserve:asrt}. 
On the other hand, 
 {external callees do not necessarily preserve  $\Pos{\_}$-assertions} -- \cf Ex. \ref{push:does:not:preserve}. 
Therefore, in order to guarantee that $A$ holds upon entry to the callee, we need to know that $ \PushASLong {(y_0, \overline y)}{A}$ held at the caller site -- \cf Lemma \ref{lemma:push:ass:state}.


Remember that {popping frames does not necessarily preserve}
 $\Pos{\_}$  assertions 
-- \cf Ex. \ref{ex:pop:does:not:preserve}.
Nevertheless, {\sc{Call\_Int}} guarantees the unadapted version, $A$,  upon return from the call. 
This is sound, because of our 
 \emph{\strong satisfaction} of assertions -- more in Sect.  \ref{s:scoped:valid}.

%
%
%

\paragraphsd{
Polymorphic Calls.} Our rules do  not \emph{directly} address scenaria where  the receiver may be
   internal or external, andthe choice about this is made at runtime. 
However, such scenaria  are \emph{indirectly} supported, through  our rules of consequence and  case-split.
More   in  \A\ \aref{H.6}{\ref{app:polymorphic}}.

\begin{example}[Proving external calls]
We continue our discussion from \S \ref{sec:howThird} on how to establish the Hoare triple    \textbf{(1)} :

 \vspace{.05cm}
  \begin{minipage}{.05\textwidth}
   \textbf{(1?)}\ \ 
\end{minipage}
\hfill
\begin{minipage}{.95\textwidth}
\begin{flushleft}
$\{\  \   \external{\prg{buyer}} \ \wedge\ 
   {\protectedFrom {\prg{this.\myAccount.key}}  {\prg{buyer} } }
 \ \wedge\ \prg{this.\myAccount.\balance}= b  \ \  \}$\\
$\ \ \ \ \ \ \ \ \ \ \ \ {\ \prg{buyer.pay(this.accnt,price)}   \ } $\\
$  \{\  \ \  {\prg{this.\myAccount.\balance}} \geq  b \  \  \} \ \ ||\ \  \{\ \inside{\prg{a.\password}}\wedge  \prg{a.\balance}\!\geq\!{\prg{b}}   \ \}  $ 
\end{flushleft}
\end{minipage}
 
\vspace{.03cm}
\noindent
We use $S_3$, which says that $\TwoStatesN{ \prg{a}:\prg{Account},\prg{b}:\prg{int} } {\inside{\prg{a.key}} \wedge \prg{a.\balance} \geq \prg{b} }$. 
We can apply rule {\sc{Call\_Ext\_Adapt}}, by taking  $y_0 \triangleq \prg{buyer}$,  and $\overline {x : D}\triangleq \prg{a}:\prg{Account},\prg{b}:\prg{int}$, 
and $A \triangleq  \inside{\prg{a.\pwd}}\wedge \prg{a.\balance}\geq \prg{b}$, \ \ 
and $m \triangleq \prg{pay}$,\ and $\overline {y} \triangleq \prg{this.accnt},\prg{price}$,
and provided that we can establish that\\
\strut \ \ \   \textbf{(2?)}  $ {\small{\strut \ \ \ \protectedFrom {\prg{this.\myAccount.key}} {(\prg{buyer},\prg{this.\myAccount},\prg{price})}}}$\\
holds. Using type information, we get that all fields transitively accessible from \prg{this.\myAccount.key}, or \prg{price} are internal or scalar. This implies\\
\strut \ \ \   \textbf{(3)}  $ {\small{\strut \ \ \ \protectedFrom {\prg{this.\myAccount.key}} {\prg{this.\myAccount}} \wedge  \protectedFrom {\prg{this.\myAccount.key}} {\prg{price}}}} $\\
Using then    Def. \ref{def:push},  we can indeed establish that\\
\strut \ \ \   \textbf{(4)} $ {\small{\strut \ \ \ \protectedFrom {\prg{this.\myAccount.key}} {(\prg{buyer},\prg{this.\myAccount},\prg{price})}  \ = \   \protectedFrom {\prg{this.\myAccount.key}} {\prg{buyer}}}}$\
Then, by application of the rule of consequence, \textbf{(4)}, and the rule {\sc{Call\_Ext\_Adapt}}, we can establish \textbf{(1)}.\ \ \ 
More details in \S  \aref{H.3}{\ref{l:buy:sat}}.
\end{example}

\subsection{Third Phase: Adherence to Module Specifications}
\label{sect:wf}

In Fig. \ref{f:wf} we  define the judgment $\vdash M$, which says that  
$M$ has been proven to be well formed.

\begin{figure}[thb]
$
\begin{array}{l}
\begin{array}{lcl}
\inferruleSDNarrow 
{~ \strut  {\sc{WellFrm\_Mod}}}
{  \vdash \SpecOf {M}
  \hspace{1.2cm}  M \vdash \SpecOf {M}
}
{
\vdash M  
}
& \hspace{0.7cm} &
\inferruleSDNarrow 
{~ \strut   {\sc{Comb\_Spec}}}
{  
M \vdash S_1 \hspace{1.2cm}  M \vdash S_2
}
{
M \vdash S_1 \wedge S_2
}
\end{array}
\\
\inferruleSD 
{~ \strut \hspace{6.5cm} {\sc{method}}}
{  
 \prg{mBody}(m,D,M)=p \ (\overline{y:D})\{\  stmt \ \}       
    \\
  {\hprovesN{M} { \ \prg{this}:\prg{D}, \overline{y:D}\, {\wedge\, A_1}\  } 
  {\ stmt\ } {\ A_2\, \sdN{\wedge\, \PushASLong {\prg{res}} {A_2}} \ }   {A_3} } 
}
{
M \vdash {\mprepostN {A_1}{p\ D}{m}{y}{D}{A_2} {A_3} }
}
\\
\inferruleSD 
{~ \strut \hspace{6.5cm} {\sc{invariant}}}
{
\begin{array}{l}
\forall  D,  m. \ \ \  \prg{mBody}(m,D,M)=\prg{public} \ (\overline{y:D})\{\  stmt \ \}      \ \ \Longrightarrow  
  \\
 \begin{array}{l}
   \hprovesN {M}  
{ \ \prg{this}:\prg{D}, \overline{y:D},\,   \overline{x:C} \, \wedge\,  A \, \sdN{\wedge\, \PushASLong {(\prg{this},\overline y)} {A}}\, }  
  	{\ stmt\ }   
	 {\  A\, \wedge\, \PushASLong {\prg{res}} {A} \ }  
{\ A \ }  
 \end{array}
 \end{array}
}
{
M \vdash \TwoStatesN{ \overline{x:C}} {A}
}
\end{array}
$
\caption{Methods' and Modules' Adherence to Specification}
\label{f:wf}
\end{figure}

\footnoteSD{Julian's reflections on the invariant and method call rules:

 I think it somewhat makes sense for the Invariant rule to be asymmetric since
calls to internal code are asymmetric from the perspective of protection. There is no change in specific protection when entering 
internal code, but there may be when returning to external code from internal code. 

I'd have to think about it more, but I don't believe there is any decrease in general protection when entering internal code,
but there may be when exiting internal code. The adaptation in the post-condition is there to address the potential
for protection to be lost on return.

Magic wand adaptation (i.e. A -* x) expresses what must hold before crossing the boundary from internal to external code
(i.e. before an internal method return or call to an external method) for the adapted assertion to still be satisfied after crossing 
that boundary. I think there are a few things that are useful to consider here:
it is good to note that if you remove protection from the assertion language, the magic wand would be the identity
(i.e. A -* x = A) since only protection is modified when passing values between objects via method calls
more importantly, the way we define protection is based around fields and external/internal objects, and passing a value
to an internal object (via a method call) modifies neither a field nor grants any external object access to any other object
According to my understanding of the intended semantics of the magic wand, entering internal code from frame should
not affect any sort of protection.
My impression is that part of the intent of the magic wand is to be able to express ``what needs to hold before a call
so that <x>  holds within the new frame (whether that frame is internal or external)?''. I still need to think about how that
fits in with my above description of magic wand adaptation.
}

{\sc{WellFrm\_Mod}} and {\sc{Comb\_Spec}} say that $M$ is well-formed  if its specification is well-formed (according to Def. \ref{f:holistic-wff}), and if $M$ satisfies all conjuncts of the specification.
{\sc{Method}} says that  
a module satisfies a method specification if the 
body satisfies the corresponding pre-, post- and \midcond. 
 In  the postcondition we also ask that $\PushASLong {\prg{res}}  {A}$, so that \prg{res} does not leak any of the values that $A$ promises will be protected. 
{\sc{Invariant}} says that  a module satisfies {a} 
specification $\TwoStatesN{ \overline{x:C}} {A}$,  if the method body of each public method
 has $A$ as its  pre-, post- and \midcond. 
{Moreover, the precondition is strengthened by $\PushASLong {(\prg{this},\overline y)} {A}$ -- this is sound because
 the caller is external, and by Lemma   \ref{lemma:push:ass:state}, part (\ref{lemma:push:ass:state:three}).}

\paragraphsd{Barendregt} In  {\sc{method}} we \sdN{implicitly} require   the free variables  in a method's precondition  not to overlap with variables in its body, unless they are the receiver or one of the parameters ($\sdN{\vs(stmt) \cap \fv(A_1) \subseteq   \{ \prg{this}, y_1, ... y_n \} }$).  And in {\sc{invariant}} we require   the free variables in $A$ (which are a subset of  $\overline x$) not to overlap with the variable  in $stmt$ ($ \sdN{ \vs(stmt)\,  \cap\, \overline x\, =\, \emptyset}$).
This can easily be achieved through renamings, \cf Def. \aref{H.3}{\ref{d:promises}}.

\newcommand{\sdsp}{\strut \ \ \ \ \ }

\begin{example}[Proving a public method] Consider the proof that \prg{Account::set} from $M_{fine}$ satisfies $S_2$. 
Applying rule {\sc{invariant}}, we need to establish:\\
\label{e:public}
{\small{ \vspace{.05cm}
  \begin{minipage}{.05\textwidth}
  \textbf{(5?)}\ \ 
\end{minipage}
\hfill
\begin{minipage}{.95\textwidth}
\begin{flushleft}
$\{ \  \   ...\  \prg{a}:\prg{Account}\ \wedge\,  {\inside{\prg{a.key}}} \, \wedge \, \protectedFrom {\prg{a.key}} { (\, \prg{key'},\prg{key''}\, )} \  \} \ $\\
$\ \ \ \ \ \ \ \ \ \ \  \prg{body\_of\_set\_in\_Account\_in}\_ M_{fine}\   $\\
$  \{\  \    {\inside{\prg{a.key}}}\ \wedge\ {\PushASLong {\prg{res}\ } {\ \inside{\prg{a.key}}}} \ \   \} \ \ \  || \ \ \ 
	\{\ \    {\inside{\prg{a.key}}}\ \  \}  $ 
\end{flushleft}
\end{minipage}
}}
\vspace{.03cm}
\noindent
Given the conditional statement in \prg{set}, and with the obvious treatment of conditionals (\cf Fig. \aref{15}{\ref{f:substructural:app}}), among other things, we  need to prove for the \prg{true}-branch that:\\
\vspace{.01cm}
{\small{  \begin{minipage}{.05\textwidth}
  \textbf{(6?)}\ \ 
\end{minipage}
\hfill
\begin{minipage}{.95\textwidth}
\begin{flushleft}
$\{ \  \   ... \   {\inside{\prg{a.key}}} \, \wedge \, \protectedFrom {\prg{a.key}} { (\, \prg{key'},\prg{key''}\, )}  \wedge  \,  \prg{this.key}=\prg{key'}\  \} \ $\\
$\ \ \ \ \ \ \ \ \ \ \  \prg{this.key := key''}\   $\\
$  \{\  \  {\inside{\prg{a.key}}}\   \   \} \ \ \  || \ \ \ \{  \ \  {\inside{\prg{a.key}}}\  \  \}  $ 
\end{flushleft}
\end{minipage}
}}
\\
 \vspace{.03cm}
\noindent
We can apply case-split  (\cf Fig. \aref{15}{\ref{f:substructural:app}}) on whether \prg{this}=\prg{a}, and thus a proof of \textbf{(7?)} and \textbf{(8?)}, would give us a proof of \textbf{(6?)}:\\
 \vspace{.03cm}
{\small{  \begin{minipage}{.05\textwidth}
   \textbf{(7?)}\ \ 
\end{minipage}
\hfill
\begin{minipage}{.95\textwidth}
\begin{flushleft}
$\{ \  \   ... \   {\inside{\prg{a.key}}} \, \wedge \, \protectedFrom {\prg{a.key}} { (\, \prg{key'},\prg{key''}\, )} \wedge  \,  \prg{this.key}\!=\!\prg{key'}\ \wedge\ \prg{this}\!=\!\prg{a} \ \  \} \ $\\
$\ \ \ \ \ \ \ \ \ \ \   \prg{this.key := key''}\    $\\
$  \{\  \      {\inside{\prg{a.key}}} \   \   \} \ \ \  || \ \ \ \{  \ \  {\inside{\prg{a.key}}} \ \  \}  $ 
\end{flushleft}
\end{minipage}
}}
\\
and also
\\
 \vspace{.03cm}
{\small{  \begin{minipage}{.05\textwidth}
   \textbf{(8?)}\ \ 
\end{minipage}
\hfill
\begin{minipage}{.95\textwidth}
\begin{flushleft}
$\{ \  \   ...   {\inside{\prg{a.key}}} \, \wedge \, \protectedFrom {\prg{a.key}} { (\, \prg{key'},\prg{key''}\, )}  \wedge  \,  \prg{this.key}\!=\!\prg{key'}\ \wedge\ \prg{this}\!\neq\!\prg{a} \ \ \} \ $\\
$\ \ \ \ \ \ \ \ \ \ \  \prg{this.key := key''}\   $\\
$  \{\  \   {\inside{\prg{a.key}}}  \  \  \}\ \ \  || \ \ \ \{  \ \  {\inside{\prg{a.key}}} \ \  \}  $ 
\end{flushleft}
\end{minipage}
}}
 
 \vspace{.03cm}
\noindent
If  $ \prg{this.key}\!\!=\!\!\prg{key'} \wedge\ \prg{this}\!\!=\!\!\prg{a}$, then $\prg{a.key}\!\!=\!\!\prg{key'}$. But   $\protectedFrom {\prg{a.key}} {\prg{key'}}$ and   
 {\sc{Prot-Neq}} from Fig. \aref{16}{\ref{f:protection:conseq:ext}} give $\prg{a.key}\!\neq\!\prg{key'}$. So,  by contradiction (\cf Fig. \aref{15}{\ref{f:substructural:app}}), we can prove    \textbf{(7?)}.
If  $\prg{this}\!\neq \!\prg{a}$, then we  obtain from the underlying Hoare logic that the value of \prg{a.key} did not change. Thus, by rule {\sc{Prot\_1}}, we obtain  \textbf{(8?)}. \ \ \ More details in \S \aref{H.5}{\ref{l:set:sat}}.
 
 \vspace{.045cm}
\noindent
On the other hand, \prg{set} from $M_{bad}$ cannot be proven to satisfy $S_2$, because it 
requires  proving\\
\vspace{.03cm}
{\small{  \begin{minipage}{.05\textwidth}
   \textbf{(???)}\ \ 
\end{minipage}
\hfill
\begin{minipage}{.95\textwidth}
\begin{flushleft}
$\{ \  \   ...  \  {\inside{\prg{a.key}}} \, \wedge \, \protectedFrom {\prg{a.key}} { (\, \prg{key'},\prg{key''}\, )}  \   \} \ $\\
$\ \ \ \ \ \ \ \ \ \ \  \prg{this.key := key''}\   $\\
$  \{\  \    \{ {\inside{\prg{a.key}}}\   \   \}\ \ \  || \ \ \ \{  \ \  {\inside{\prg{a.key}}} \ \  \}  $ 
\end{flushleft}
\end{minipage}
}}
\\
and without the condition \prg{this.key}=\prg{key'} there is no way we can prove \textbf{(???)}.

\subsection{Our Example Proven} 
\label{sect:example:proof:short}
Using our Hoare logic, we have developed a mechanised proof in Coq, that, indeed, $M_{good} \vdash S_2 \wedge S_3$.
This proof is  part of the accompanying artifact. 
 
 Our proof models  \LangOO, the assertion language, the specification language, and the Hoare logic from \S \ref{s:hoare:first},  \S   \ref{s:hoare:second},  \S  \ref{sect:wf},  \S \aref{F}{\ref{app:hoare}} and Def. \ref{def:push}.
In keeping with   the start of  \S \ref{sect:proofSystem}, our proof assumes the existence of an underlying Hoare logic,  
and several, standard, properties of that underlying logic, the assertions logic (\eg equality of objects implies equality of field accesses) and of type systems
(\eg  fields of objects of different types cannot be aliases of one another).
All assumptions  are clearly indicated in the associated artifact.
 Appendix \aref{H}{\ref{s:app:example}}  
outlines   that proof. 


\end{example}

 \section{Soundness} 
 \label{sect:sound:proofSystem}
We now give a synopsis of the proof of soundness of the logic from \S \ref{sect:proofSystem}, and outline the 
 two most interesting aspects: deep satisfaction, and summarized execution.

\paragraphSD{\Strong Satisfaction} 
\label{s:scoped:valid}
We are faced with the challenge that assertions are not always preserved when  the top frame is popped (\cf Ex. \ref{ex:pop:does:not:preserve}),  yet we must be able to argue that method return preserves the method body's postconditions.
 For this, we introduce a  ``deeper'' notion of assertion satisfaction, which an assertion be satisfied not only from the view of the top frame, but also from the view of every frame starting from the $k$-th frame onwards:   $ \satDAssertFrom M  \sigma k   A$   says that   $\forall j. [\  k\!\leq\! j \leq\! \DepthSt \sigma \ \Rightarrow \ M, \RestictTo \sigma j \models A \ ]$.
Accordingly, we introduce \emph{\strong specification satisfaction},  ${\satDAssertQuadrupleFrom \Mtwo  M  \sigma   {A} {A'} {A''} } $, which promises for all $k\!\leq\! \DepthSt \sigma$,  
if $ \satDAssertFrom M  \sigma k   A$, and    if 
  scoped execution of $\sigma$'s continuation leads to final state $\sigma'$ and intermediate external state $\sigma''$, then
 $ \satDAssertFrom M  {\sigma'} k   {A'}$, and  $ \satDAssertFrom M  {\sigma''} k   {A''}$
 -  \cf    App.\aref{G.3}{\ref{s:shallow:deep:appendix}}.
 
 Here how \strong satisfaction addresses this problem: Assume    state  $\sigma_1$  right before entering a call, $\sigma_2$ and $\sigma_3$ at start and end of the call's body, and   $\sigma_4$ upon return. If a  pre-condition holds at $\sigma_1$,  then it  holds for a $k\leq \DepthSt {\sigma_1}$; hence,  if the postcondition holds for $k$ at $\sigma_3$, and because $\DepthSt {\sigma_3}= \DepthSt {\sigma_1}\!+\!1$, it also holds for $\sigma_4$.    
\Strong satisfaction is stronger than shallow (\ie  specification satisfaction as in Def. \ref{def:shallow:spec:sat:state}). 
 
\begin{lemma}
For all $\Mtwo$, $M$, $A$, $A'$, $A''$, $\sigma$:  
\begin{itemize}
\item
 $ {\satDAssertQuadrupleFrom \Mtwo  M  \sigma   {A} {A'} {A''} } \ \Longrightarrow \ \
  {\satAssertQuadruple  \Mtwo  M   {A}  \sigma  {A'} {A''} } $\end{itemize}
\end{lemma}

\paragraphSD{Soundness of the Triples Logic}
\label{sect:prove:triples:sound}
We require the assertion logic,  $M\vdash A$, and  the  underlying Hoare logic,  $M\ \vdash_{ul}\  \triple A {stmt} {A'}$,   to be be  sound. 
Such  sound logics do exist.
Namely, one can build an assertion logic, $M\vdash A$, by extending a logic which does not talk about protection, through the addition of structural rules which talk about protection; this extension preserves soundness - \cf  App. \aref{G.1}{\ref{s:expectations}}. 
Moreover,  since the assertions $A$ and $A'$ in $M \vdash_{ul}\  \triple A {stmt} {A'}$ may, but need not, talk about protection, 
one can take a Hoare logic from the literature as the $ \vdash_{ul}$-logic.

We then prove   soundness of the rules about protection from Fig. \ref{f:protection}, and, based on this, 
we prove soundness of the inference system for triples  -- \cf \A\ \aref{G.5}{\ref{s:sound:app:triples}}.

\begin{Theorem}
\label{l:triples:sound}
For module  $M$   such that  $\vdash M$, and for any assertions $A$,  $A'$, $A''$ and statement  $stmt$:
\begin{center}
$M\ \vdash\  \triple A {stmt} {A'}  \ \ \ \  \Longrightarrow  \ \ \ \ \satisfiesD {M} {\quadruple {A} {stmt} {A'} {A''}}$
\end{center}
\end{Theorem}

\paragraphSD{Summarised Execution}
\label{s:summaized}
Execution of an external call may consist of any number of external
transitions, interleaved with calls to public internal methods, which in
turn may make any number of further internal calls (public or private),  
and these, again may call external methods.
For the   proof of soundness,  internal and external transitions use different arguments.
 For  external transitions we consider small steps  and  argue in terms of  preservation of  encapsulated properties,
while for internal calls, we use large steps, and appeal to the method's specification.
Therefore, we define  \emph{sumarized} executions, where  internal calls are collapsed into one, large step, \eg below:

\label{sect:termExecs}


\begin{tabular}{lll}
\begin{minipage}{.41\textwidth}
%
 \resizebox{4.1cm}{!}{
\includegraphics[width=\linewidth]{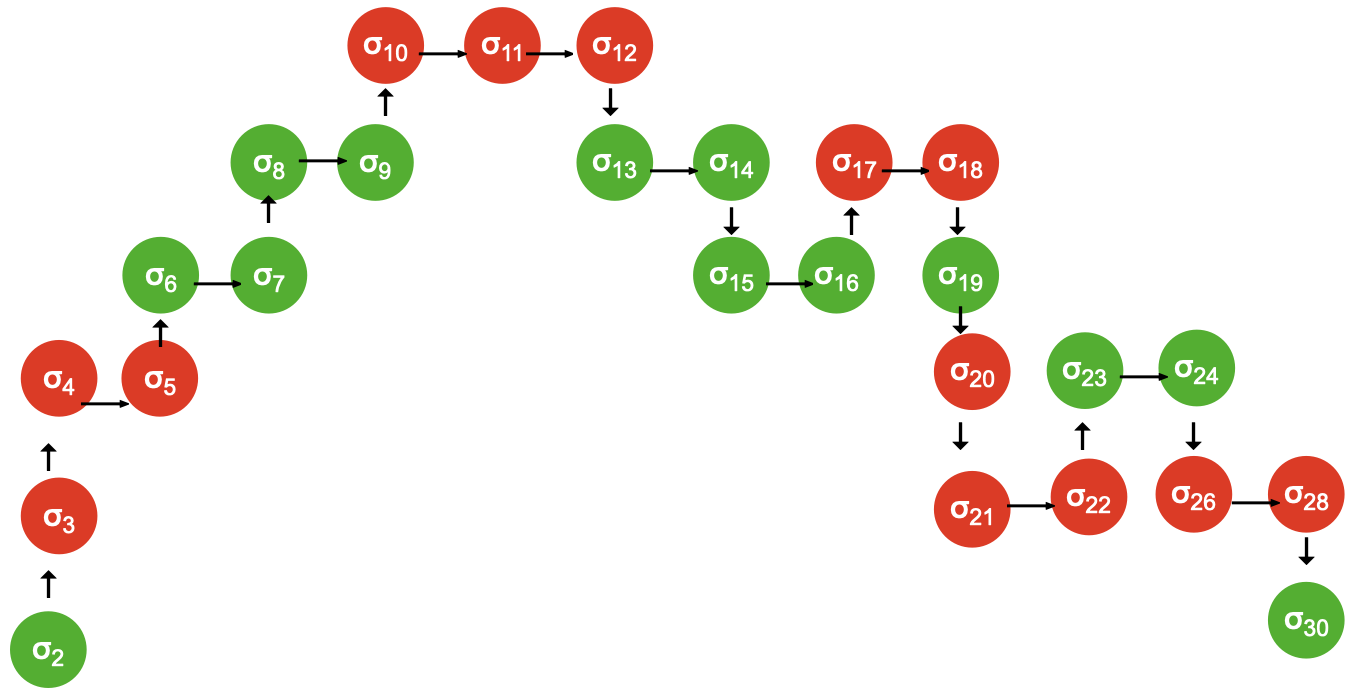}
} 
 \end{minipage}
&  \begin{minipage}{.14\textwidth}
summarized\\
$\strut \ \ \ \ \ $ to 
\end{minipage}
 &
\begin{minipage}{.41\textwidth}
~ \\
~ \\
\resizebox{4.1cm}{!}
{
\includegraphics[width=\linewidth]{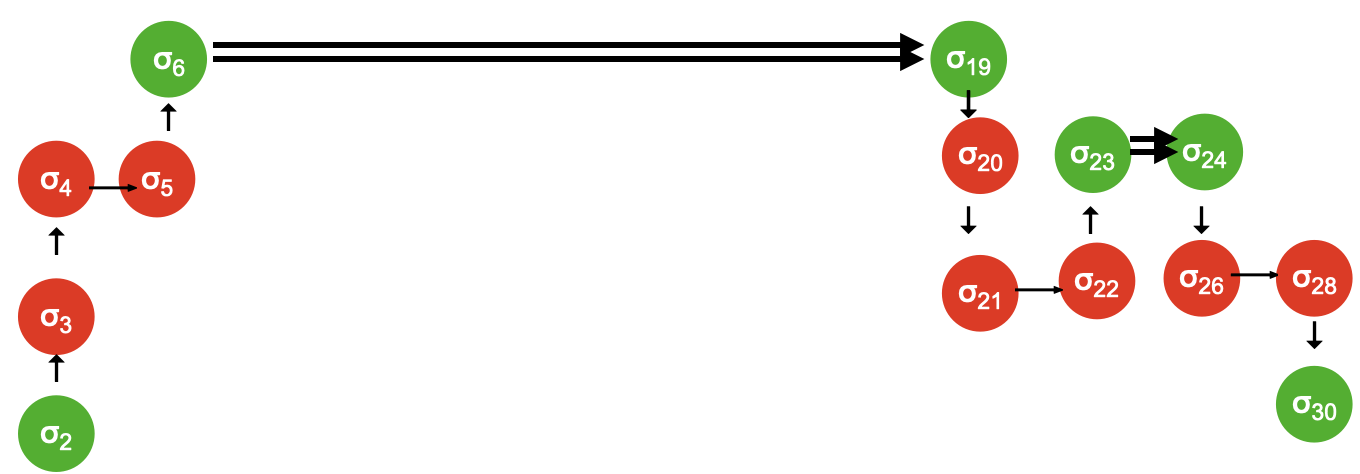}
} \end{minipage}
\end{tabular}

\vspace{.05cm}

Lemma \aref{G.28}{\ref{lemma:external_breakdown:term}} 
says that any terminating execution 
 starting in an external state  consists of a  sequence of  external states interleaved with terminating executions
  of public methods.
Lemma  \aref{G.29}{ \ref{lemma:external_exec_preserves_more}} says that such an execution preserves an encapsulated assertion $A$  
provided that all these finalising internal executions  
also preserve $A$.
%


\paragraphSD{ Soundness of the Quadruples Logic}
Proving soundness of our quadruples  requires  induction on the execution  in  some cases, and  induction on the derivation of the quadruples in others.  We address this   through  a well-founded ordering that combines both, \cf 
\label{sect:prove:wellfounded}
\label{sect:prove:sound:quadruples}
  Def.  \aref{G.22}{\ref{def:measure}}  and  Lemma \aref{G.23}{\ref{lemma:normal:two}}. 
  Finally, in \S \aref{G.16}{\ref{s:app:proof:sketch;quadruples}}, we prove soundness:

\begin{theorem}
\label{t:quadruple:sound}
\label{thm:soundness}
For module  $M$,   assertions $A$, $A'$, $A''$,   state  $\sigma$, and specification $S$:

\begin{enumerate}[(A)]
\item
 $:\strut \   \vdash M  \ \ \ \wedge \ \ \  M\ \vdash\  \quadruple {A} {stmt} {A'} {A''}  \ \ \ \ \ \ \ \Longrightarrow \ \ \ \ \ \  \ \ \  M\ \modelsD\  \quadruple {A} {stmt} {A'} {A''}$
 \item
  $:\strut \  \  \proves{M}{S}\ \ \ \ \ \ \Longrightarrow\ \ \ \ \ \  \ \ \ {M} \modelsD {S}$
 
\end{enumerate}

\end{theorem}

%
%

%


  \section{Conclusion: Summary, Related Work and Further Work}
 \label{sect:related}
\label{sect:conclusion}
  \paragraphSD{Our motivation}
comes from the OCAP approach to security, whereby object capabilities guard against unsanctioned effects.
Miller \cite{miller-esop2013,MillerPhD} advocates
 \emph{defensive consistency}: whereby 
 {``An object is defensively
  consistent when it can defend its own invariants and provide correct
  service to its well behaved clients, despite arbitrary or malicious
  misbehaviour by its other clients.''}  Defensively consistent
modules  are  hard to design 
 and  verify, but
make it much
easier to make guarantees about systems composed of multiple components
\cite{Murray10dphil}.

 \paragraphSD{Our work} aims to elucidate such guarantees. We want to formalize and prove  that 
\cite{permissionAuthority}:
\begin{quote}
\emph{Lack of eventual access implies that certain properties will be preserved, even in the presence of external calls.}
\end{quote}
For this, we had  to  model the concept of  lack of eventual access,  determine the temporal scope of the preservation, and  develop a Hoare logic framework to formally prove such guarantees.

For lack of eventual access,  we introduced protection, 
a property of all  the paths of all external objects accessible from the current stack frame.
For the  temporal scope of preservation, we developed scoped invariants, which ensure that a given property holds as long as we have not returned from the current method.
  (top of current stack has not been popped yet). 
 For our Hoare logic, we introduced an adaptation operator, which translates assertions between the caller’s and callee’s frames. 
 Finally, to prove the soundness of our approach, we developed the notion of \strong satisfaction,  which mandates that an assertion must be satisfied from a particular stack frame onward. 
 Thus, most concepts in this work are  \emph{scope-aware}, as they depend  on the current stack frame.
 
 With these concepts, we 
 developed a specification language for modules limiting effects, a Hoare Logic for proving external calls, protection, and adherence to specifications, and have proven it sound. 

\paragraphSD{Lack of Eventual Access}   Efforts to restrict ``eventual access'' have been extensively explored, with
 Ownership Types  being a prominent example \cite{simpleOwnership,existOwn}.
These types enforce encapsulation boundaries to safeguard internal implementations, thereby ensuring representation independence and defensive consistency
\cite{ownalias,NobPotVitECOOP98,Banerjee:2005}.
Ownership is fundamental to key systems like Rust’s memory safety  
\cite{RustPL2,RustBelt18},
Scala's Concurrency \cite{ScalaCapabilities,ScalaLightweightAffine},
Java's heap analyses \cite{PotterNC98,HillNP02,MitECOOP06}, 
and it  plays a critical role in program verification
\cite{BoyLisShrPOPL03,hypervisor} including Spec$\#$
\cite{BarLeiSch05,BarDelFahLeiSch04}, universes
\cite{DieDroMue07,DietlMueller05,LuPotPOPL06},
Borrowable Fractional Ownership \cite{borrow-fract-vmcai2024},
and recently has been integrated into languages like OCAML \cite{ocaml-ownership-icfp2024,funk-ownership-oopsla2024}.

\sdN{Ownership types are closely related to the notion of protection: both are scoped relative to a frame. However, ownership requires an object to control some part of the path, while protection demands that module objects control the endpoints of paths. }

\sdN{In future work we want to   explore how to express protection within Ownership Types, with the primary challenge being how to accommodate capabilities accessible to some external objects while still   inaccessible to others. }
Moreover,  tightening some
rules in our current Hoare logic (e.g.\ Def. \ref{def:chainmail-protection-from})
may lead to a native  Hoare logic of ownership.
Also, recent approaches like
%
%
the Alias
Calculus \cite{meyer-alias-calculus-scp2015,meyer-auto-alias-sncs2020},
Reachability
Types \cite{romf-reachability-types-oopsla2021,rompf-poly-reachability-popl2024}
and Capturing
Types \cite{odersky-capturing-types-toplas2023,scoped-effects-oopsla2022,odersky-reach-prog2024}
abstract fine-grained method-level descriptions of 
references and aliases flowing into and out of methods and fields,
and likely accumulate enough information to express 
protection. Effect exclusion
\cite{fx-exclusion-icfp2023} directly prohibits nominated
effects, but within a closed, fully-typed world.

\paragraphSD{Temporal scope of the guarantee} Starting with loop invariants\cite{Hoare69,Floyd67}, property preservation at various granularities and durations has been widely and successfully adapted and adopted \cite{Hoare74,liskov94behavioral,usinghistory,Cohen10,Meyer92,MeyerDBC92,BarDelFahLeiSch04,objInvars,MuellerPoetzsch-HeffterLeavens06,DrossoFrancaMuellerSummers08}.
In our work, the temporal scope of the preservation guarantee includes all nested calls, until termination of the currently executing method, but not beyond. 
We compare with object and history invariants in \S \ref{sect:bounded}.

Such guarantees are maintained by the module as a whole.
\citet{FASE}  proposed ``holistic specifications'' which take an external
perspective across the interface of a module. 
\citet{OOPSLA22} builds upon this work, offering a specification
language based on \emph{necessary} conditions and temporal operators.
Neither of these systems support any kind of external calls.
\sdN{Like \cite{FASE,OOPSLA22} we propose ``holistic specifications'', albeit without temporal logics, and with sufficient conditions.
In addition, we introduce protection, and develop a Hoare logic for protection and external calls.}

\paragraphSD{Hoare Logics} were first developed in Hoare's seminal 1969 paper \cite{Hoare69}, and have inspired a plethora of influential further developments and tools. We shall discuss a few only.

\sdN{Separation logics  \cite{IshtiaqOHearn01,Reynolds02}  
 reason  about disjoint memory regions. 
Incorporating Separation Logic's powerful framing mechanisms  will pose several challenges: 
We have no specifications and no footprint for external calls. 
Because protection is ``scope-aware'', expressing it as a predicate would require quantification over all possible paths and variables within the current stack frame.
We may also require a  new separating conjunction \susan{operator}.
}
Hyper-Hoare Logics \cite{hyper-hoare-pldi2024,compositional-hypersafety-oopsla2022}  reason about the execution of several programs, and  \sdN{could thus be applied to our problem, if extended to model all possible  
sequences of calls of internal public methods.}

 Incorrectness Logic
\cite{IncorrectnessLogic}
under-approximates  postconditions, and thus
reasons about the presence of bugs, rather than their absence.
%
Our work, like classical Hoare Logic, over-approximates postconditions,
 and differs from Hoare and Incorrectness Logics
by tolerating interactions between verified code and unverified components.
Interestingly, even though in earlier  works
  \cite{FASE,OOPSLA22} {we} employed \emph{necessary} conditions for effects (\ie under-approximate pre-conditions), we can, instead, employ \emph{sufficient} conditions for the lack of effects (over-approximate postconditions).
Incorporating our work into  Incorrectness Logic might require  under-approximating  eventual access, while protection over-approximates it.

Rely-Guarantee
\cite{relyGuarantee-HayesJones-setss2017,relyGuarantee-vanStaden-mpc2015}
and Deny-Guarantee \cite{DenyGuarantee} 
distinguish between assertions guaranteed by a thread, and those a
thread can reply upon. 
Our Hoare quadruples are (roughly) Hoare triples plus 
the ``guarantee'' portion of rely-guarantee.
When a
specification includes a guarantee, that guarantee must be maintained
by every ``atomic step'' in an execution
\cite{relyGuarantee-HayesJones-setss2017}, rather than just at method
boundaries as in visible states semantics
\cite{MuellerPoetzsch-HeffterLeavens06,DrossoFrancaMuellerSummers08,considerate}.
In concurrent reasoning,  
this is because shared state may be accessed
by another co{\"o}perating thread at any time:
while in our case, it is because unprotected
state may be accessed by an untrusted component within the same
thread.  

\paragraphSD{Models and Hoare Logics for the interaction with the the external world}
\citet{Murray10dphil} made the first attempt to formalise defensive
consistency, 
to tolerate interacting with any untrustworthy object,
although without a specification language for describing effects
(i.e.\ when an object is correct).

 \citet{CassezFQ24} propose one approach to reason about external calls.
Given that external callbacks are necessarily restricted to the module's public interface,
external callsites are replaced  with a
generated \texttt{externalcall()} method that  nondeterministically invokes any method in that interface.
\citet{iris-wasm-pldi2023}'s Iris-Wasm is similar.
WASM's
modules are very loosely coupled: a module
has its own byte memory
and object table.
Iris-Wasm ensures models 
can only be
modified via their explicitly exported interfaces.

\citet{ddd}  designed OCPL, a logic
that separates internal implementations (``high values'')
from interface objects
(``low values''). %
OCPL supports defensive
consistency 
(called ``robust safety'' after the
security literature \cite{Bengtson})
by ensuring
low values can never leak high values, a 
and 
prove 
object-capability patterns, such as
sealer/unsealer, caretaker, and membrane.
%
%
RustBelt \cite{RustBelt18}
developed this approach to prove Rust memory safety using Iris \cite{iris-jfp2018},
and combined with RustHorn \cite{RustHorn-toplas2021} for the safe subset,
produced RustHornBelt \cite{RustHornBelt-pldi2022} that verifies
both safe and unsafe Rust programs. 
Similar techniques were extended to C \cite{RefinedC-pldi2021}.

\citet{dd} deploy step-indexing, Kripke worlds, and 
  public/private state machines to model problems including the 
DOM wrapper and a mashup application.
Their distinction between public and private transitions 
is similar  to our
distinction between internal and external objects.
VMSL is an Iris-based separation logic for
virtual machines to assure defensive consistency
\cite{vmsl-pldi2023}.
%
%
Cerise uses Iris invariants to support proofs of programs with outgoing calls and callbacks,
on capability-safe CPUs \cite{cerise-jacm2024}.
\cite{BirkedalL:caps-mmio-conf}  {supports several different wrappers on the same object.}
{Besides the distinction between the object-based and class-based approach}, our work differs from
{\cite{iris-wasm-pldi2023,vmsl-pldi2023,dd,ddd,cerise-jacm2024,BirkedalL:caps-mmio-conf,RustHornBelt-pldi2022}}
in that {they follow a ``static'' capability model whereby authority is conferred though an object's interface, whereas we follow a ``dynamic'' model whereby authority may be stored in the callee's context.}
 
{\citet{stack-safety-csf2023} extend memory safety} arguments to ``stack
safety'', 
ensuring that method calls and returns are well bracketed,
and preserving   the integrity of
caller and callee.  
{\citet{schaeferCbC}    adopt}
an information-flow approach to ensure confidentially by construction.

\paragraphSD{Smart Contracts} also pose the problem of external calls.
Rich-Ethereum \cite{rich-specs-smart-contracts-oopsla2021}
assumes  instance-private and unaliased contract fields ,
and  employs a 
``finished'' flag to manage  callbacks. 
Scilla \cite{sergey-scilla-oopsla2019},
a minimalist functional alternative to Ethereum, helps prevent common contract errors.


The VerX tool can verify
specifications for Solidity contracts automatically \cite{VerX}.
VerX's specification language is based on temporal logic.
It 
is restricted to ``effectively call-back free'' programs
\cite{Grossman,relaxed-callbacks-ToDES},
delaying any callbacks until the
incoming call to the internal object has finished.
 
\textsc{ConSol} \cite{consolidating-pldi2024}
provides a specification langauge for smart contracts,
checked at runtime \cite{FinFel01}.
SCIO* \cite{secure-io-fstar-popl2024}, implemented in
F*, supports both
verified and unverified code.
Both \textsc{Consol} and SCIO* are 
similar to gradual verification techniques 
\cite{gradual-verification-popl2024,Cok2022} that
insert dynamic checks between verified and unverified code,
and contracts for general access control 
\cite{DPCC14,AuthContract,cedar-oopsla2024}.

%

  \forget{
\paragraph{Future work} THERE AREN"T REALLY QUESTIONS ABOVE includes the questions mentioned above, and also
\sdN{the investigation of more ways to approximate eventual access, footprint and framing operators}, and   
the application of these techniques to
languages that rely on lexical nesting for access
control such as Javascript \cite{ooToSecurity},
rather than public/private annotations;
languages that support ownership types that can be leveraged for
verification
\cite{leveragingRust-oopsla2019,RustHornBelt-pldi2022,verus-oopsla2023};
and languages from the
functional tradition such as OCAML, which is gaining imperative
features such as ownership and uniqueness \cite{funk-ownership-oopsla2024,ocaml-ownership-icfp2024}. 
%
%
We expect our techniques can be incorporated into existing program
verification tools \cite{Cok2022}, especially those attempting
gradual verification \cite{gradual-verification-popl2024},
thus paving the way towards practical verification for
the open world.
}

\paragraphSD{Future work}  includes 
 the application of our techniques to
 languages that rely on lexical nesting for access control, 
{\eg}   Javascript \cite{ooToSecurity},
to languages that support ownership types,  
{\eg} Rust, 
 leveraged for verification
\cite{leveragingRust-oopsla2019,RustHornBelt-pldi2022,verus-oopsla2023},
and to
{features} such as 
uniqueness\cite{funk-ownership-oopsla2024,ocaml-ownership-icfp2024}. 
{These} different language paradigms may lead to a 
{refined view of the meaning
of protection}. 
{We also want to investigate tighter integrations with the underlying Hoare logic and   its framing operators.}
%



 \section*{Acknowledgements}
We are grateful for the feedback from Lars Birkedal, David Swasey, Peter Müller, Derek Dreyer, Dominique Devriese, and the anonymous OOPSLA reviewers. 
We are especially thankful to Reviewer B, and even more so to Reviewer C, whose 
 sustained engagement across multiple rounds of discussion went well beyond 
 expectations and was instrumental in helping us clarify parts of the work that were initially unclear.

\appendix
\section{Appendix to Section \ref{sect:underlying} -- The programming language \LangOO}
\label{app:loo}

We introduce \LangOO, a simple, typed, class-based, object-oriented language.

\subsection{Syntax}

The syntax of \LangOO is given in Fig. \ref{f:loo-syntax}\footnote{{Our motivating example is provided in a slightly richer syntax for greater readability.}}.
To reduce the complexity of our formal models, as is usually done, CITE - CITE,  \LangOO lacks many
common languages features, omitting static fields and methods, interfaces,
inheritance, subsumption, exceptions, and control flow.  
 \LangOO 
and which may be defined recursively.

\LangOO modules ($M$) map class names ($C$) to class definitions ($\textit{ClassDef}$).
A class definition consists of 
a list of field definitions, ghost field definitions, and method definitions.
{Fields, ghost fields, and methods all have types, $C$; {types are
    classes}.
    Ghost fields may be optionally 
 annotated as \texttt{intrnl}, requiring the argument to have an internal type, and the 
body of the ghost field to only contain references to internal objects. This is enforced by the
limited type system of \LangOO.}
A program state ($\sigma$) is a pair of of a stack and a heap.
The stack is a a stack is a non-empty list of frames ($\phi$), and the heal ($\chi$)
is a map from addresses ($\alpha$) to objects ($o$). A frame consists of a local variable
map and a continuation .\prg{cont} that represents the statements that are yet to be executed ($s$).
A statement is either a field read ($x := y.f$), a field write ($x.f := y$), a method call
($u :=y_0.m(\overline{y})$), a constructor call ($\prg{new}\ C$), 
  a sequence of statements ($s;\ s$),
  or empty ($\epsilon$).

\LangOO also includes syntax for expressions $\re$ that may 
be used in writing
specifications or the definition of ghost fields.

\subsection{Semantics}
\LangOO is a simple object oriented language, and the operational semantics 
(given in Fig. \ref{f:loo-semantics} and discussed later)
do not introduce any novel or surprising features. The operational 
semantics make use of several helper definitions that we 
define here.

{
We provide a definition of reference interpretation in Definition \ref{def:interpret}
\begin{definition}
\label{def:interpret}
For a frame $\phi= (\overline {x \mapsto v}, s)$, and a program state $\sigma = (\overline \phi \cdot \phi,, \chi)$, we   define:
\begin{itemize}
\item
$\interpret{\phi}{x}\ \triangleq\ v_i$\ \ \ if \ \ \ $x=x_i$
\item
 $\interpret{\sigma}{x}\ \triangleq\  \interpret{\phi}{x}$
\item
$\interpret{\sigma}{\alpha.f}\ \triangleq\ v_i $ \ \ if \ \ $\chi(\alpha)=(\_; \  \overline {f \mapsto v})$, and $f_i=f$
\item
$\interpret{\sigma}{x.f}\ \triangleq\ \interpret{\sigma}{\alpha.f}$ where $\interpret{\sigma}{x}=\alpha$
\item
$\phi.\prg{cont} \ \triangleq\ s$ 
\item
$\sigma.\prg{cont} \ \triangleq\ \phi.\prg{cont}$\
\item
$\phi[\prg{cont}\mapsto s'] \ \triangleq\ (\overline {x \mapsto v}, s')$
\item
$\sigma[\prg{cont}\mapsto s'] \ \triangleq \ (\ {\overline \phi}\cdot \phi[\prg{cont}\mapsto s'],\  \chi\ )$ 
\item
$\phi[\prg{x'}\mapsto v'] \ \triangleq\ ( \ (\overline {x \mapsto v})[\prg{x'}\mapsto v'],\ s \ )$
\item
$\sigma[\prg{x'}\mapsto v'] \ \triangleq\ (\ (\overline {\phi} \cdot (\phi[\prg{x'}\mapsto v']), \ \chi)$ 
\item
$\sigma [\alpha \mapsto o ] \ \triangleq\ (\ (\overline {\phi} \cdot \phi), \ \chi [\alpha \mapsto o ]\ )$ 
\item
$\sigma [\alpha.f' \mapsto v' ] \ \triangleq\ \sigma [\alpha \mapsto o ] $\ \ \  if \ \  
$\chi(\alpha)=(C, {\overline {f \mapsto v}})$, and $o=(\ C;  ({\overline {f \mapsto v}})[f' \mapsto v' ]\ )$ 
\end{itemize}
\end{definition}
}
That is, a variable $x$, or a field access on a variable $x.f$ 
has an interpretation within a program state of value $v$
if $x$ maps to $v$ in the local variable map, or the field
$f$ of the object identified by $x$ points to $v$.

Definition \ref{def:class-lookup} defines the class lookup function an object 
identified by variable $x$.
\begin{definition}[Class Lookup]
\label{def:class-lookup}
For program state $\sigma = ({\overline {\phi}}\cdot\phi, \chi)$, class lookup is defined as 
$$\class{\sigma}{x}\ \triangleq\ C \ \ \ \ \ \mbox{if} \ \ \  \chi(\interpret{\sigma}{x})=(C,\_ )$$
\end{definition}

Module linking is defined for modules with disjoint definitions:

\begin{definition}
\label{def:linking}
For all modules $\Mtwo$ and $M$, if the domains of $\Mtwo$ and $M$ are disjoint, 
we define the module linking function as $M\cdot \Mtwo\ \triangleq\ M\ \cup\ M'$.
\end{definition}
That is,  their linking is the union of the two if their domains are disjoint.

Definition \ref{def:meth-lookup} defines the method lookup function for a method
call $m$ on an object of class $C$.
\begin{definition}[Method Lookup]
\label{def:meth-lookup}
For module $\Mtwo$, class $C$, and method name $m$, method lookup is defined as 
$$\meth{\Mtwo}{C}{m}\ \triangleq\ { pr}\  \prg{method}\ m\ (\overline{x : T}) {:T}\{\ s\ \}  $$
if there exists an $M$ in $\Mtwo$, so that $M(C)$ contains the definition ${ pr}\  \prg{method}\ m\ (\overline{x : T}) {:T}\{\ s\ \} $
\end{definition}

Definition \ref{def:fields-lookup} looks up all the field identifiers in a given class
\begin{definition}[Fields Lookup]
\label{def:fields-lookup}
For modules $\Mtwo$,and  class $C$, fields lookup is defined as 
$$fields(\Mtwo,C) \ \triangleq\  \  \{ \ f  \ | \  \exists  M\in\Mtwo. s.t.  M(C) \mbox{contains the definition}  \prg{field}\ f: T\ \} $$
\end{definition}

We define what it means for two objects to come from the same module
\begin{definition}[Same Module]
\label{def:same:module}
For program state $\sigma$,  modules $\Mtwo$, and variables $x$ and $y$, we defone
$$\Same {x} {y} {\sigma}{\Mtwo}\ \triangleq\ \exists C, C', M[ \ M\in \Mtwo \wedge C, C'\in M \wedge  \class{\sigma}{x}=C \wedge \class{\sigma}{y} =C'\ ]$$
\end{definition}

As we already said in \S \ref{s:underlying}, we forbid assignments to a method's parameters. 
To do that, the following function returns the  identifiers of the formal parameters of the currently active method.

\begin{definition}
For program state $\sigma$:
\label{def:params}

$\Formals \sigma \Mtwo \ \ \triangleq \ \  \overline x \ \ \ \mbox{such that} \ \  \  \exists \overline {\phi},\,\phi_k, \, \phi_{k+1}, \,C,\,p.$\\
$\strut \hspace{3.2cm} [   \ \ \sigma =  (\overline {\phi}\cdot{\phi_k}\cdot{\phi_{k+1}}\,, \chi) 
\  \ \wedge\  \ \phi_k.\prg{cont}=\_:=y_0.m(\_); \_ \  \  \wedge\ \ $\\
$\strut \hspace{3.2cm} \class{(\phi_{k+1},\chi)}  {\prg{this} }\  \ \wedge\ \ \meth{\Mtwo} {C} {m} = p \ C::m(\overline{x : \_}){:\_}\{\_\} \ ] $
\end{definition}

While the small-step operational semantics of \LangOO is given in Fig. \ref{f:loo-semantics},
specification satisfaction is defined over an abstracted notion of 
the operational semantics that models the open world. 

An \emph{Initial} program state contains a single frame 
with a single local variable \prg{this} pointing to a single object 
in the heap of class \prg{Object}, and a continuation.
\begin{definition}[Initial Program State]
\label{def:initial}
A program state $\sigma$ is said to be an initial state ($\initial{\sigma}$)
if and only if
\begin{itemize}
\item
$\sigma\ =\  ( ((\prg{this}\ \mapsto\ \alpha), s); \  (\alpha \mapsto (\prg{Object}, \emptyset)$
\end{itemize} 
for some address $\alpha$ and some statement $s$.
\end{definition}


We provide a semantics for expression evaluation is given in Fig. \ref{f:evaluation}. 
That is, given a module $M$ and a program state $\sigma$, expression $e$ evaluates to $v$
if $\eval{M}{\sigma}{e}{v}$. Note, the evaluation of expressions is separate from the operational
semantics of \LangOO, and thus there is no restriction on field access.

{
\paragraph{Lemmas and Proofs}
}

\beginProof{l:wf:state}
{The first assertion is proven by unfolding the definition of $\_ \models \_ $.

The second assertion is proven by case analysis on the execution relation $\exec {\_} {\sigma} {\sigma'}$. 
The assertion gets established when we call a method, and is preserved through all the execution steps, because we do not allow assignments to the formal parameters.
 
}
\completeProof

We now prove lemma \ref{l:var:unaffect}:

\beginProof{l:var:unaffect} 
\begin{itemize}
\item
We first show that $\leadstoBoundedThree {\Mtwo} {\sigma} {\sigma\bd}  {\sigma'}   \ \wedge \ k<\DepthSt \sigma\bd  \ \ \Longrightarrow \ \  \interpret {\RestrictTo \sigma k} y =  \interpret {\RestrictTo {\sigma'} k} y$
This follows easily from the operational semantics, and the definitions.

\item
By induction on the earlier part, we  obtain that $\leadstoBoundedStar {\Mtwo}  {\sigma}  {\sigma'}  \ \wedge \ k<\DepthSt \sigma  \ \ \Longrightarrow \ \  \interpret {\RestrictTo \sigma k} y =  \interpret {\RestrictTo {\sigma'} k} y$

\item

We now show that $\leadstoBoundedStarFin {\Mtwo}  {\sigma}  {\sigma'} \ \wedge \ y\notin \vs(\sigma.\prg{cont}) \ \ \Longrightarrow \ \  \interpret \sigma y =  \interpret {\sigma'} y$ by  induction on the number of steps, and  using the earlier lemma.
\end{itemize}
\completeProof

{Lemma \ref{lemma:relevant:more} states that initila states are well-formed, and that (\ref{threeLR}) a pre-existing object, locally reachable after any number of scoped execution steps, was locally reachable at the first step.

\begin{lemma}
\label{lemma:relevant:more}
For all modules $\Mtwo$, states $\sigma$, $\sigma'$,   and frame $\phi$:
\begin{enumerate}
\item
\label{lemma:relevant:more:one}
$\initial {\sigma}  \ \    \Longrightarrow \ \ \Mtwo \models \sigma $
\item
\label{threeLR}
{${\leadstoBoundedStar {\Mtwo}  {\sigma}    {\sigma'}}  \ \  \Longrightarrow\ \ 
dom(\sigma) \cap \LRelevantO {\sigma'} \subseteq   \LRelevantO {\sigma}$
}

\end{enumerate}
\end{lemma}

{Consider Fig.  \ref{fig:illusrPreserve} . 
Lemma \ref{lemma:relevant:more}, part \ref{threeLR}  promises that any objects locally reachable in $\sigma_{14}$ which already existed in $\sigma_{8}$, were locally reachable in $\sigma_{8}$. However, the lemma is only  applicable to scoped execution, and as 
$\notLeadstoBoundedStar {\Mtwo} {\sigma_8}  {\sigma_{17}}$, 
the lemma does not promise that  objects locally reachable in $\sigma_{17}$ which already existed in $\sigma_{8}$, were locally accessible in $\sigma_{8}$ -- namely it could be that objects are made globally reachable upon method return, during the step from $\sigma_{14}$ to $\sigma_{15}$.}
}

\vspace{1cm}
Finally, we define the evaluation of expressions, which, as we already said, represent ghost code.

\begin{figure}[hbp]
\begin{minipage}{\textwidth}
\footnotesize{
\begin{mathpar}
\infer
		{}
		{\eval{M}{\sigma}{v}{v}}
		\quad(\textsc{E-Val})
		\and
\infer
		{}
		{\eval{M}{\sigma}{x}{\interpret{\sigma}{x}}}
		\quad(\textsc{E-Var})
		\and
\infer
		{
		\eval{M}{\sigma}{\re}{\alpha}
		}
		{
		\eval{M}{\sigma}{\re.f}{\interpret{\sigma}{\alpha.f}}
		}
		\quad(\textsc{E-Field})
		\and
\infer
		{
		\eval{M}{\sigma}{\re_0}{\alpha}\\
		\overline {\eval{M}{\sigma}{\re}{v}}\\
		M(\class{\sigma}{\alpha})\ \mbox{contains}\ \prg{ghost}\ gf(\overline{x : T)}\{\re\} : T' \\
		\eval{M}{\sigma}{[\overline {v/x}]\re}{v}
		}
		{
		\eval{M}{\sigma}{\re_0.gf(\overline{\re})}{v}
		}
		\quad(\textsc{E-Ghost})
\end{mathpar}
}
\caption{\LangOO Expression evaluation}
\label{f:evaluation}
\end{minipage}
\end{figure}

\subsection{ $M_{ghost}$ Accounts expressed through ghost fields}

\label{app:BankAccount:ghost}

We   revisit the bank account example, to demonstrate the use of ghost fields.
In Fig. \ref{f:ex-bank-short}, accounts belong to banks, and their \prg{\balance} is  kept in a ledger. 
Thus, $\prg{account.\balance}$ is a ghost field which involves a recursive search through that ledger.

\begin{figure}[h]
\begin{lstlisting}[language=chainmail, mathescape=true, frame=lines]
module $M_{ghost}$ 
  
  class Shop  ...
  
  class Account
    field bank: Bank
    field key:Key
    public method transfer(dest:Account, key':Key, amt:nat)
      if (this.key==key') 
         this.bank.decBalance(this,amt); 
         this.bank.incBalance(dest.amt);
    public method set(key':Key)
      if (this.key==null) this.key=key'
    ghost balance(): int
        res:=bank.balance(this)
        
  class Bank
    field ledger: Ledger
    method incBalance(a:Account, amt: nat)
        this.ledger.decBalance(a,amt)
    private method decBalance(a:Account, amt: nat)
        this.ledger.decBalance(a,amt)   
    ghost balance(acc):int 
        res:=this.ledger.balance(acc)
        
   class Ledger
      acc:Acc
      bal:int
      next:Ledger
      ghost balance(a:Acc):int 
        if this.acc==a 
          res:=retrun bal
        else
          res:=this.next.balance(a)
      
\end{lstlisting}
\caption{$M_{ghost}$  -- a module with ghost fields}
\label{f:ex-bank-short}
\end{figure}

\clearpage
\section{Appendix to Section \ref{s:auxiliary}  --  Fundamental Concepts }
\label{app:aux}

Lemma \ref{lemma:orig:to:bounded:front} says, essentially,  that \scoped executions describe the same set of executions as those  starting at an initial state\footnote{An \emph{Initial} state's heap contains a single object of class \prg{Object}, and
its  stack   consists of a single frame, whose local variable map is a mapping from \prg{this} to the single object, and whose continuation is  any statement.
(See Def. \ref{def:initial})}.   
For instance, revisit  Fig. \ref{fig:illusrPreserve} , and assume that $\sigma_6$ is an initial state.
We have  $\leadstoOrigStar {\Mtwo} {\sigma_{10}}  {\sigma_{14}}$ and $ \notLeadstoBoundedStar {\Mtwo}  {\sigma_{10}} {\sigma_{14}}$, but also 
 $\leadstoBoundedStar  {\Mtwo}  {\sigma_{6}}   {\sigma_{14}}$. 

 \begin{lemma}
\label{lemma:orig:to:bounded:front}
For all modules $\overline M$, state  $\sigma_{init}$,  $\sigma$, $\sigma'$, where
$\sigma_{init}$ is  initial:
\begin{itemize} 
\item 
\label{otbOne}
$\leadstoBoundedStar  {\Mtwo}  {\sigma} {\sigma'} \ \ \Longrightarrow \  \
\leadstoOrigStar {\Mtwo} {\sigma}  {\sigma'}$
\item 
\label{otbTwo}
$\leadstoOrigStar {\Mtwo} {\sigma_{init}}  {\sigma'}  \ \  \Longrightarrow\ \
\leadstoBoundedStar  {\Mtwo}  {\sigma_{init}} {\sigma'}$.
\end{itemize}
\end{lemma}

\vspace{.05cm}

{Lemma \ref{l:var:unaffect} says that scoped execution does not affect the contents of variables in earlier frames.}
and that 
the interpretation of a variable remains unaffected by
scoped execution of statements  which do not mention that variable. More  in Appendix~\ref{app:aux}.

\begin{lemma}
\label{l:var:unaffect}
For any modules $\Mtwo$, states $\sigma$, $\sigma'$,  variable $y$, and number $k$:
\begin{itemize}
\item
\label{carInFrame}
{$\leadstoBoundedStar {\Mtwo}  {\sigma}  {\sigma'}  \ \wedge \ k<\DepthSt \sigma  \ \ \Longrightarrow \ \  \interpret {\RestrictTo \sigma k} y =  \interpret {\RestrictTo {\sigma'} k} y$
}
\item
$\leadstoBoundedStarFin {\Mtwo}  {\sigma}  {\sigma'} \ \wedge \ y\notin \vs(\sigma.\prg{cont}) \ \ \Longrightarrow \ \  \interpret \sigma y =  \interpret {\sigma'} y$
\end{itemize}
\end{lemma}

\begin{figure}[htb]
\begin{tabular}{|c|c|c|}
\hline 
\resizebox{3.5cm}{!}{
\includegraphics[width=\linewidth]{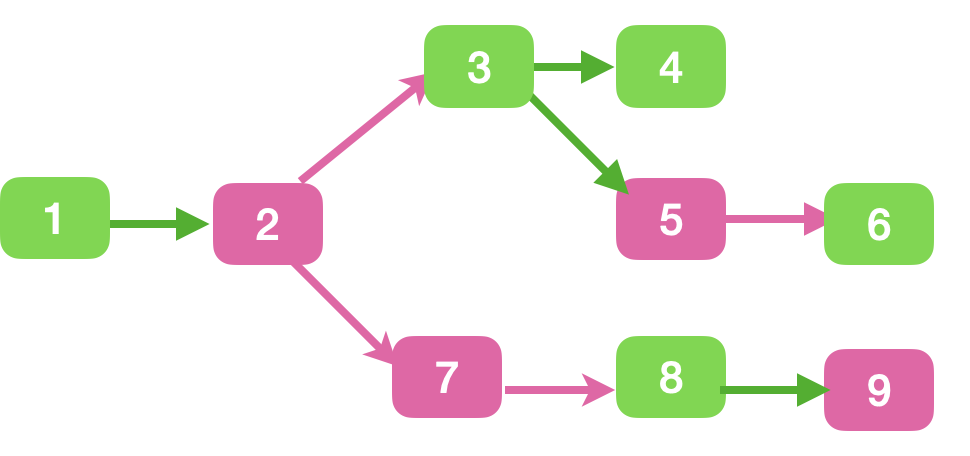}
} 
&
\resizebox{5cm}{!}{
\includegraphics[width=\linewidth]{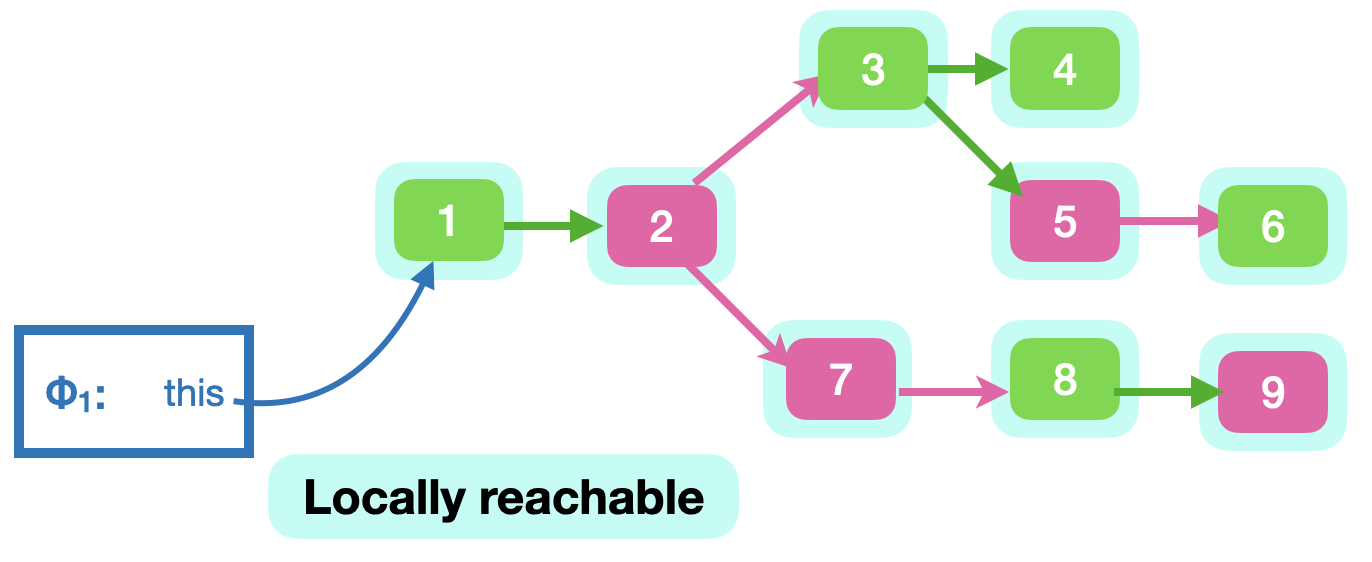}
} 
&
\resizebox{5cm}{!}{
\includegraphics[width=\linewidth]{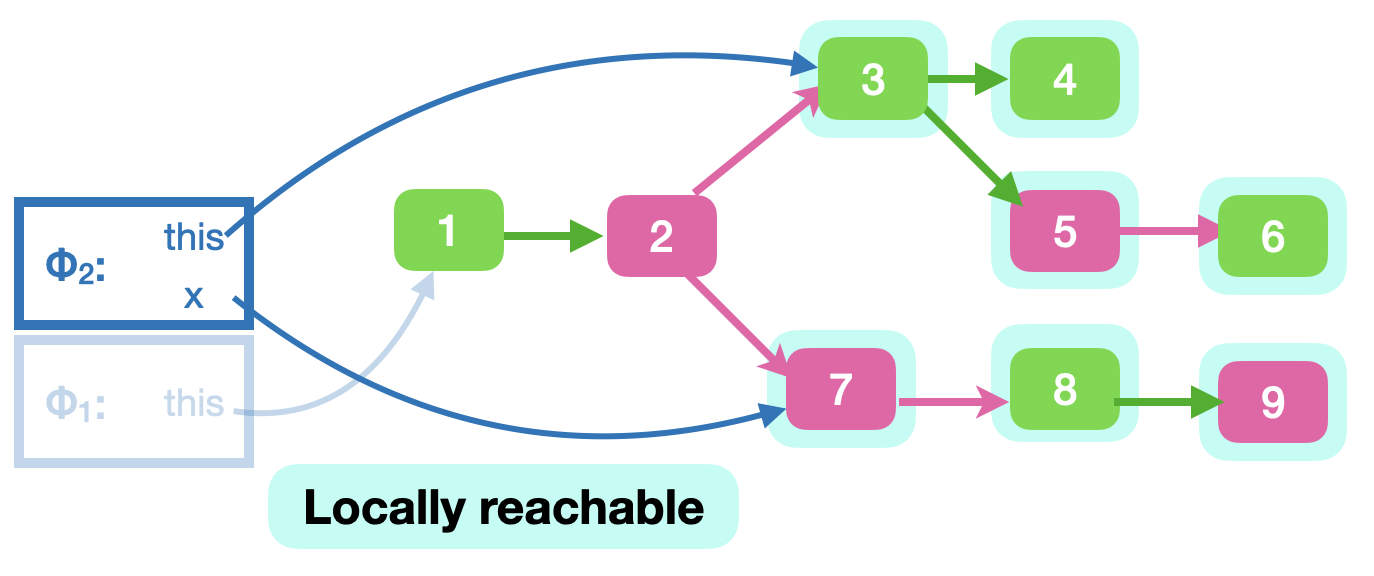}
} 
\\
\hline
 a heap
&
Locally Reachable from $\phi_1$
&
Locally Reachable from $\phi_2$
\\
\hline \hline
\end{tabular}
   \caption{-Locally Reachable Objects} 
   \label{fig:LReachable}
 \end{figure}

Fig. \ref{fig:LReachable} illustrates local reachability:  In the middle pane the top frame is $\phi_1$ which maps \prg{this} to $o_1$; all objects are locally reachable. 
In the right pane the top frame is $\phi_2$, which maps \prg{this} to $o_3$, and $x$ to $o_7$; now $o_1$ and $o_2$ are no longer locally reachable.

 \beginProof{lemma:orig:to:bounded:front}

\begin{itemize}
\item

 By unfolding and folding the definitions.
\item
By unfolding and folding the definitions, and also, by the fact that $\DepthSt {\sigma_{init}}$=1, \ie minimal.
\end{itemize}

\completeProof

\vspace{1cm}

 \beginProof{l:var:unaffect}

\begin{itemize}
\item

We unfolding the definition of $\leadstoBounded  {\Mtwo}  {\sigma}  {\sigma'}$  $\leadstoBounded  {\Mtwo}  {\sigma}  {\sigma'}$ and the rules of the operational semantics. 
\item
Take $k=\DepthSt {\sigma}$.
We unfold the definition from \ref{def:shallow:term}, and obtain that
$\sigma = \sigma'$ or, $\exists \sigma_1,...\sigma_{n1}.\forall i\!\in\! [1..n)[\  \leadstoOrig {\Mtwo}  {\sigma_i}  {\sigma_{i+1}}\  \wedge\  \EarlierS{\sigma_1} {\sigma_{i+1}}\ \wedge \sigma=\sigma_1\ \wedge\ \sigma'=\sigma_n  \ ]$

Consider the second case. 
Take any  $i\!\in\! [1..n)$. Then, by Definition, $k \leq \DepthSt {\sigma}$. 
If $k= \DepthSt {\sigma_i}$, then we are executing part of $\sigma.prg{cont}$, and because 
 $y\notin \vs(\sigma.\prg{cont})$, we get $\interpret {\RestrictTo \sigma i} y =  \interpret {\RestrictTo {\sigma_{i+1}} k} y$.
 If $k= \DepthSt {\sigma_i}$, then we apply the bullet from above, and also obtain 
 $\interpret {\RestrictTo \sigma i} y =  \interpret {\RestrictTo {\sigma_{i+1}} k} y$
 
 This gives that  $\interpret {\RestrictTo {\sigma} k} y= \interpret {\RestrictTo {\sigma'} k} y$.
 Moreover, because   $\leadstoBoundedStarFin {\Mtwo}  {\sigma}  {\sigma'}$  we obtain that
 $\DepthSt {\sigma} = \DepthSt {\sigma'} = k$.
 Therefore, we have that  $\interpret   {\sigma} y= \interpret   {\sigma'} y$.
 
 \end{itemize}
 \completeProof

\vspace{1cm}

 We also prove that in well-formed states ($\models \sigma$), all objects locally reachable from a given frame also locally reachable from the frame below.
 
\begin{lemma}
\label{rel:smaller}
  $\models \sigma\ \  \wedge k< \DepthSt \sigma\ \ \  \ \Longrightarrow\ \ \  \ \LRelevantO   {\RestrictTo \sigma {k+1}} \subseteq \LRelevantO  {\RestrictTo \sigma {k}}$
\end{lemma}

\begin{proof}
By unfolding the definitions: Everything that is in $ \RestrictTo \sigma {k+1}$ is reachable from its frame, and everything that is reachable from the frame of  $ \RestrictTo \sigma {k+1}$ is also reachable from the frame of  $ \RestrictTo \sigma {k}$. We then apply that $\models \sigma$
 
\end{proof}

 \vspace{1cm}

 \beginProof {lemma:relevant}
 
 \begin{enumerate}
\item
By unfolding and folding the definitions. 
Namely, everything that is locally reachable in $\sigma'$ is locally reachable through the frame $\phi$, and everything in the frame $\phi$ is locally reachable in $\sigma$.
  
  \item
 We require that $\models \sigma$ -- as we said earlier, we require this implicitly.
  Here we apply induction on the execution. 
  Each step is either a method call (in which case we apply the bullet from above), or a return statement (then we apply lemma \ref{rel:smaller}), or the creation of a new object (in which case reachable set is the same as that from previous state plus the new object),
  or an assignment to a variable (in which case the locally reachable objects in the new state are a subset of the locally reachable from the old state), or a an assignment to a field. 
  In the latter case, the locally reachable objects are also a subset of the locally reachable objects from the previous state.

 \end{enumerate}

\completeProof

 \clearpage
\section{Appendix to Section \ref{s:assertions} -- Assertions}
\label{appendix:assertions}

Figure \ref{fig:ProtectedFrom} illustrates ``protected from'' and ``protected''. 
In the first row  we  highlight in yellow  the objects protected from other objects. Thus, all objects except $o_6$ are protected from $o_5$ (left pane);\ all objects expect $o_8$ are protected from $o_7$ (middle pane);\ and all objects except $o_3$, $o_6$, $o_7$, and $o_8$ are protected from $o_2$ (right pane). 
Note  that $o_6$ is not protected from $o_2$, because  $o_5$ is reachable from $o_2$, is external, and has direct access to $o_6$.

In the third row of   Figure \ref{fig:Protected} we show three states: 
 $\sigma_1$ has  top frame $\phi_1$, which has  one variable, \prg{this}, pointing to $o_1$, while 
 $\sigma_2$ has  top frame $\phi_2$; it has two  variables,   \prg{this} and \prg{x} pointing to $o_3$ and  $o_7$, and 
 $\sigma_3$ has  top frame $\phi_3$; it has two  variables,  \prg{this} and \prg{x}, pointing to $o_7$ and $o_3$.  
%
We also   highlight the protected objects with a yellow halo.
 Note that $o_3$ is protected in $\sigma_2$, but is not protected in $\sigma_3$. This is so, because $\interpret {\sigma_3} {\prg{this}}$  is external, and  $o_3$ is an argument to the call. As a result, during the call, $o_7$ may obtain direct access to $o_3$. 
 
\begin{figure}[tbh]
\begin{tabular}{|c|c|c|}
\hline  
\resizebox{4.5cm}{!}{
\includegraphics[width=\linewidth]{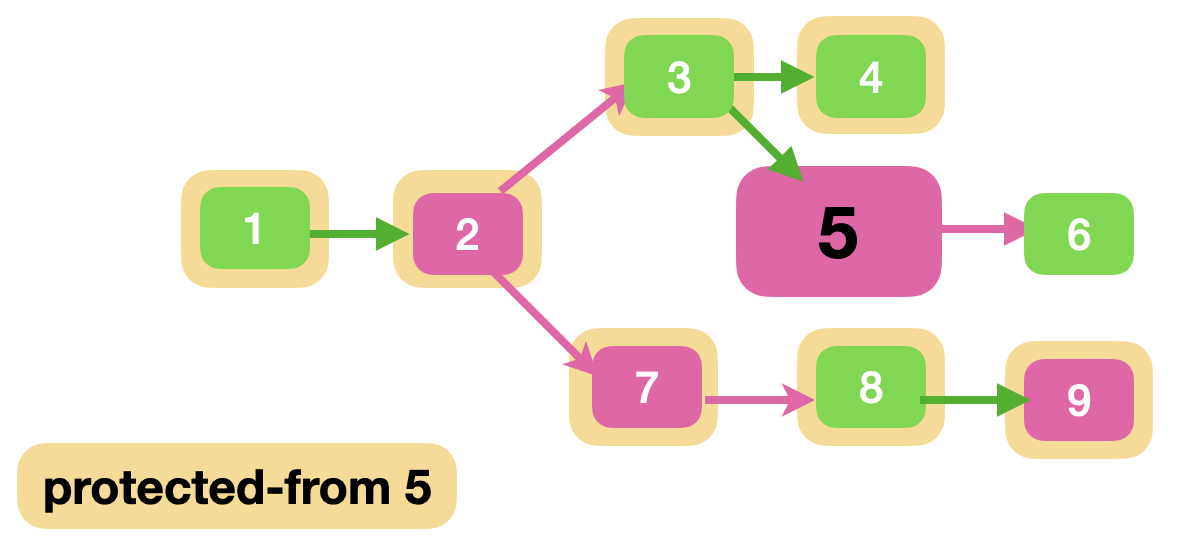}
} 
&
\resizebox{4.5cm}{!}{
\includegraphics[width=\linewidth]{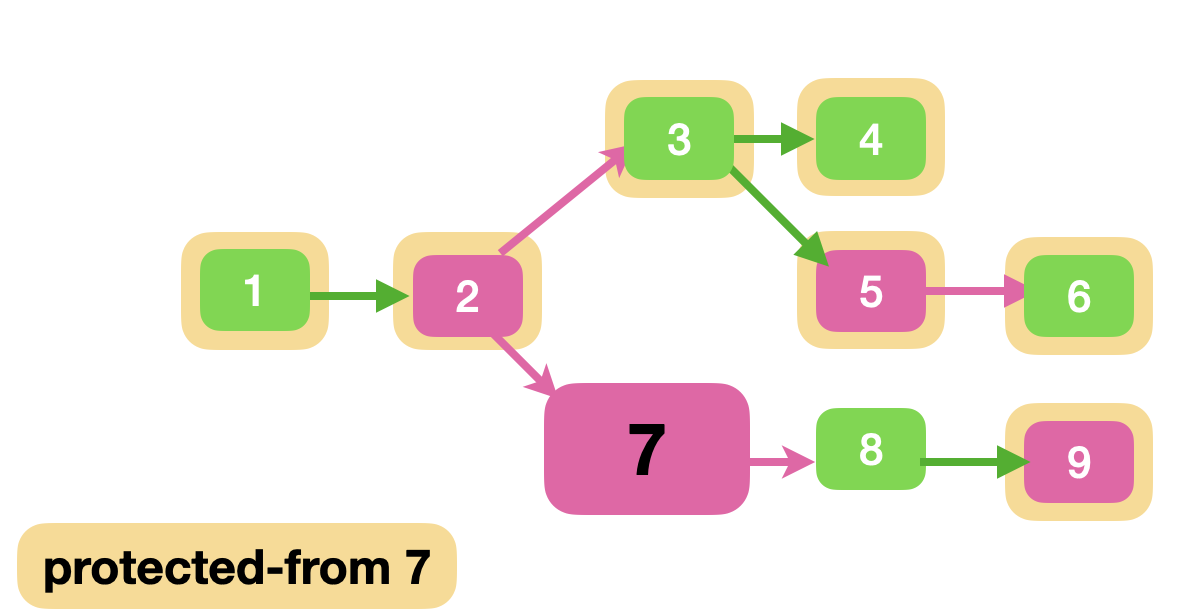}
} 
&
\resizebox{4.5cm}{!}{
\includegraphics[width=\linewidth]{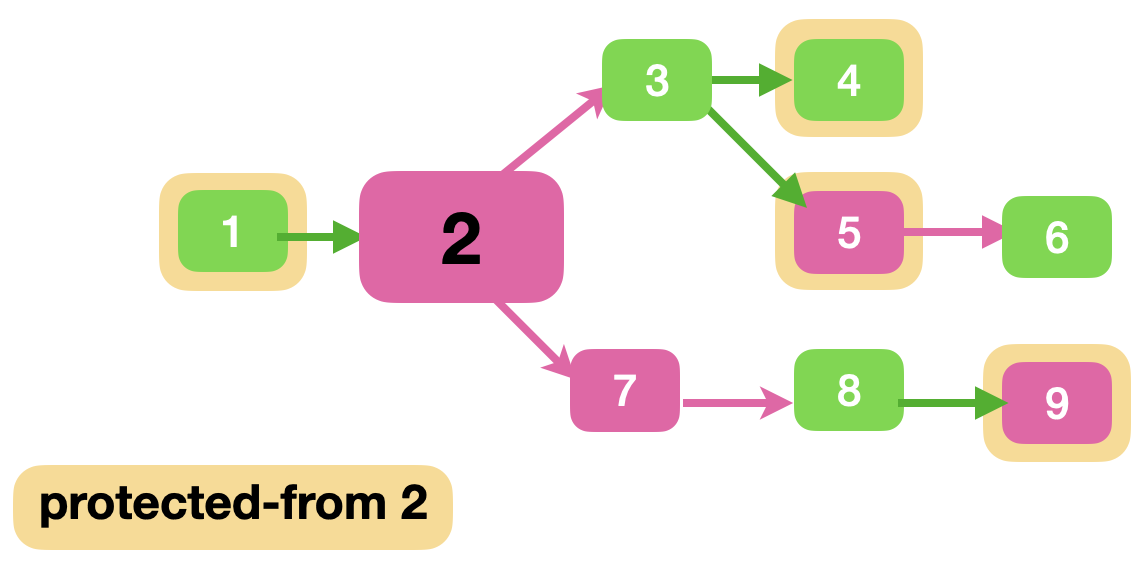}
} 
\\
\hline
protected from $o_5$
&
protected from $o_7$
&
protected from $o_2$
\\
\hline  \hline
\resizebox{4.5cm}{!}{
\includegraphics[width=\linewidth]{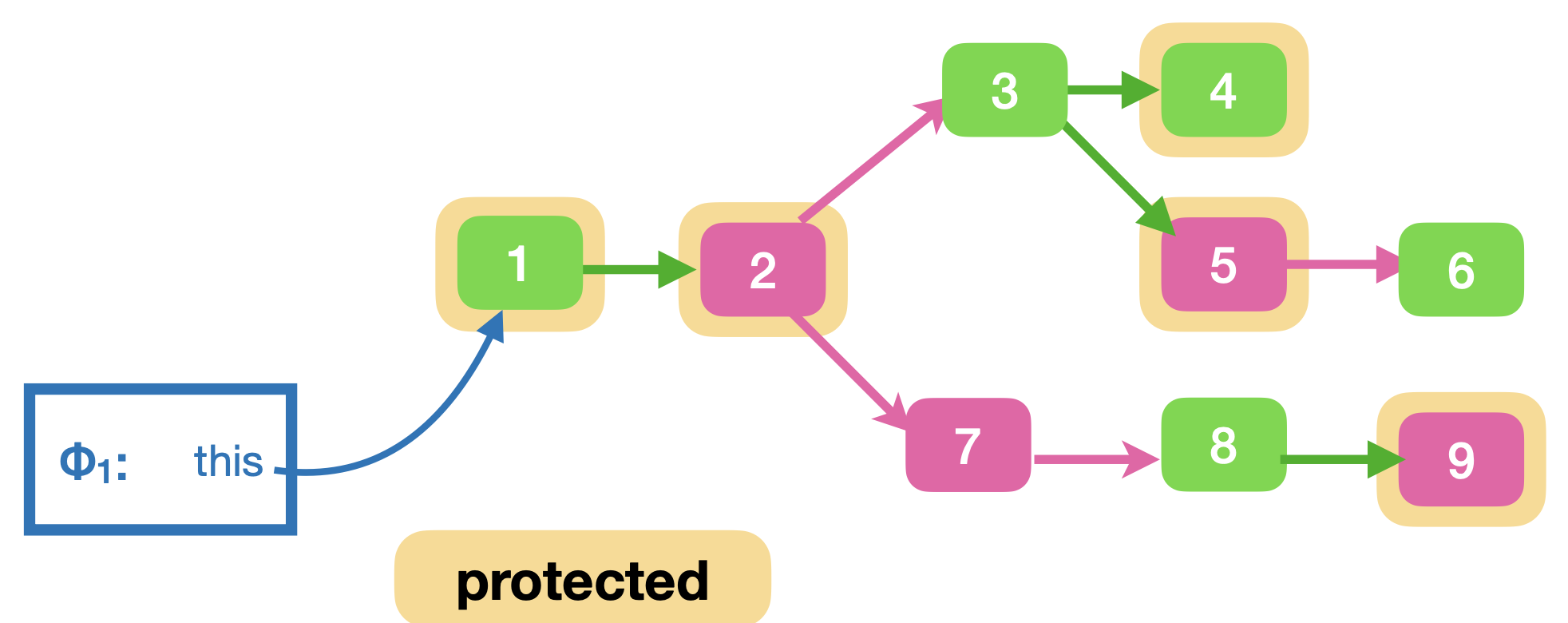}
} 
&
\resizebox{4.5cm}{!}{
\includegraphics[width=\linewidth]{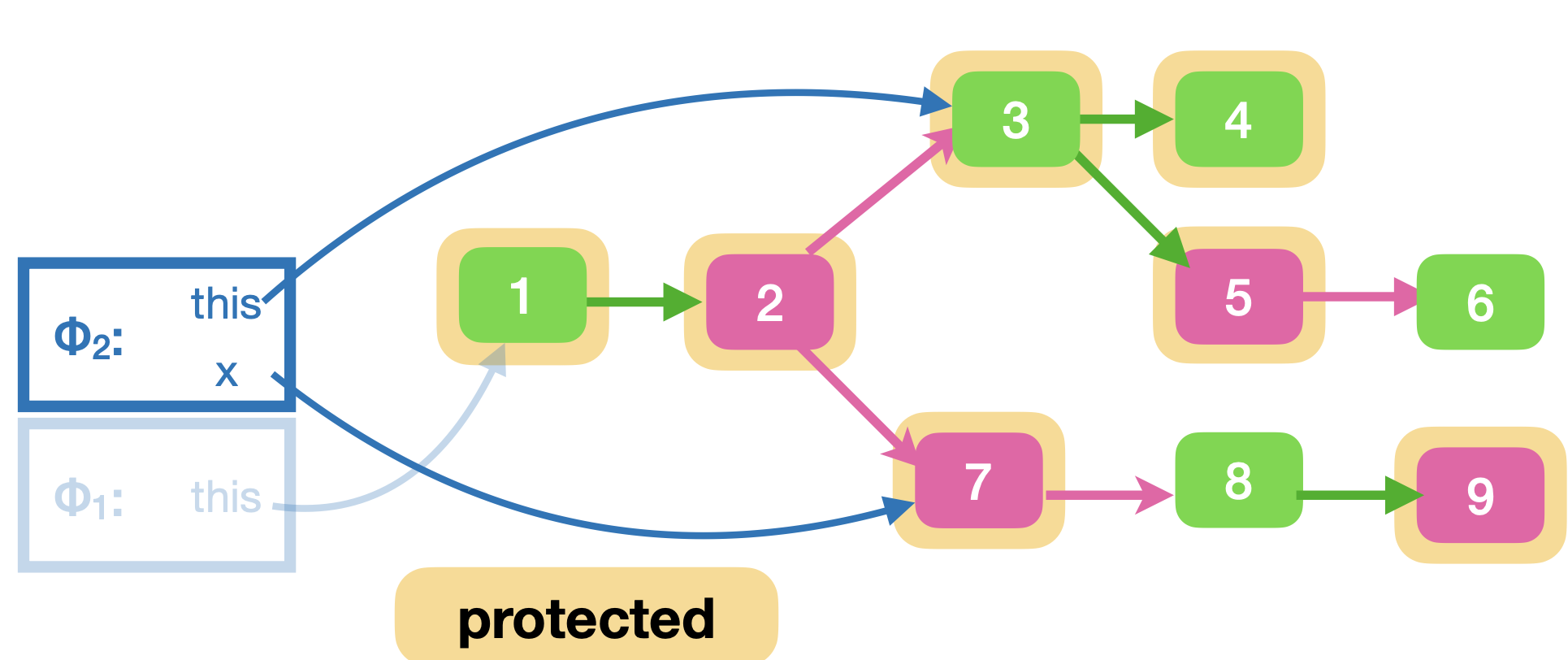}
} 
&
\resizebox{4.5cm}{!}{
\includegraphics[width=\linewidth]{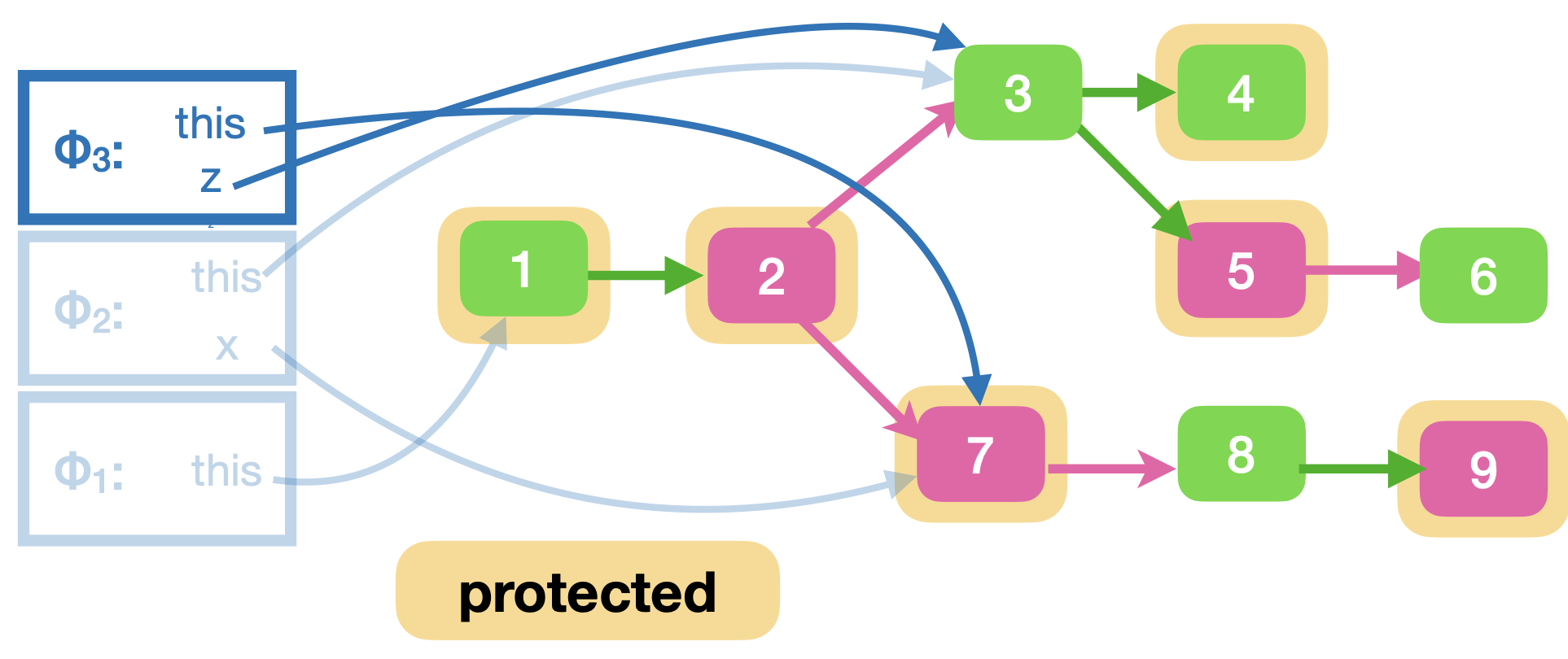}
} 
\\
\hline
protected in $\sigma_1$ & 
protected   in $\sigma_2$ & 
protected  in $\sigma_3$  
\\
\hline \hline
\end{tabular}
   \caption{ Protection. Pink objects are external, and green objects are internal.}
   \label{fig:ProtectedFrom}
    \label{fig:Protected}
 \end{figure}

\sdN{In order to prove \ref{lemma:addr:expr} from the next appendix, we first formulate and prove the following auxiliary lemma, which allows us to replace any variable $x$ in an extended expression $\re$, by its interpretation

\begin{lemma}
\label{aux:lemma:vars:eval}
For all extended expressions $\re$,  addresses $\alpha$ and  variables $x$, so that $x \in dom(\sigma):$

\begin{itemize}
\item
$\eval{M}{\sigma}{{\re}} {\alpha} \ \ \ \Longleftrightarrow \ \ \ \ \eval{M}{\sigma}{{\re[{\interpret \sigma x}/x] }} {\alpha}$
\end{itemize}
\end{lemma}
}

\sdN{
Note that in the above we require that $x\in dom( \sigma)$, in order to ensure that the replacement $[{\interpret \sigma x}/x]$ is well-defined.
On the other hand, 
we do not require  that $x\in \fv(\re)$, because 
if  $x\notin \fv(\re)$, then ${\re[{\interpret \sigma x}/x] } \txteq \re$ and the guarantee from above becomes a tautology. 
}

\sdN{
\noindent
\textbf{Proof of Lemma 
\ref{aux:lemma:vars:eval}} 
 The proof goes by induction on the structure of $\re$ -- as defined in Def. \ref{f:chainmail-syntax} -- and according to the expression evaluation rules from Fig. \ref{f:evaluation}.
\noindent
\textbf{End of Proof}
}

\clearpage

\section{Appendix to Section \ref{s:preserve} -- Preservation of Satisfaction }
\label{app:preserve}

\beginProof{lemma:addr:expr}

We first prove that for  any $M$ $A$, $\sigma$

\begin{enumerate}
\item
To show that  $ \satisfiesA{M}{\sigma}{A}\ \ \ \ \Longleftrightarrow \ \ \ \ \ \satisfiesA{M}{\sigma}{A[{\interpret \sigma x}/x]} $ 

The proof goes by induction on the structure of $A$,   application of  Defs.  \ref{def:chainmail-semantics}, \ref{def:chainmail-protection-from}, and  \ref{def:chainmail-protection}.

\item
To show that $ \satisfiesA{M}{\sigma}{A}   \ \ \ \Longleftrightarrow\ \ \ \satisfiesA{M}{\sigma[\prg{cont}\mapsto stmt]}{A}$ 

The proof goes by induction on the structure of $A$,   application of  Defs.  \ref{def:chainmail-semantics}, \ref{def:chainmail-protection-from}, and  \ref{def:chainmail-protection}.

\end{enumerate}

The lemma itself then follows form (1) and (2) proven above.

\completeProof

In addition to what is claimed in Lemma  \ref{lemma:addr:expr}, it  also holds that 
\begin{lemma}
$\eval{M}{\sigma}{{\re}}{\alpha}  \ \ \Longrightarrow\ \  [ \ \satisfiesA{M}{\sigma}{A} \  \Longleftrightarrow\   \  \satisfiesA{M}{\sigma}{A[\alpha/\re]} \  \  ]$
\end{lemma}

\begin{proof} by induction on the structure of $A$,   application of  Defs.  \ref{def:chainmail-semantics}, \ref{def:chainmail-protection-from}, and  \ref{def:chainmail-protection}, and , and auxiliary lemma \ref{aux:lemma:vars:eval}.

\end{proof}

\subsection{Stability}


We first give complete definitions for the concepts of $  \Stable {\_]}$ and $\Pos {\_}$

\vspace{.2cm}

\begin{definition}
\label{def:Basic}
[$\Stable{\_}$] assertions: 

$
\begin{array}{l}
 \begin{array}{c}
  \Stable {\inside{\re}}  \triangleq  false \\
    \Stable {\protectedFrom \re {\overline {u}}} =  
  \Stable  {\internal \re} =  
    \Stable {\re}=   
     \Stable {\re:C}\   \triangleq \    true
 \end{array}
  \\
 \begin{array}{lcl}
 \Stable  {A_1  \wedge  A_2}\  \triangleq\     \Stable  { A_1}  \wedge    \Stable  {A_2}    &
\lclSPACE  &  
 \Stable  {\forall x:C.A} =\Stable  {\neg A} \   \triangleq\   \Stable A
 \end{array}
 \end{array}
$
\label{f:Basic}
 \end{definition}

 \begin{definition}
[$\Pos{\_}$] assertions: 

$
\begin{array}{l}
 \begin{array}{c}
  \Pos {\inside{\re}} =  \Pos {\protectedFrom \re {\overline {u}}} =  
  \Pos  {\internal \re} =   
    \Pos {\re}=   
     \Pos {\re:C}\   \triangleq \    true
     \\
 \Pos  {A_1  \wedge  A_2}\  \triangleq\     \Pos  { A_1}  \wedge    \Pos  {A_2}  
  \\ 
 \Pos  {\forall x:C.A}   \triangleq\   \Pos A
\\ 
  \Pos {\neg A}  \triangleq \Stable A 
\end{array}
 \end{array}
 $
 \label{def:Pos}
\end{definition}

%

The definition  of $\Pos{\_}$ is  less  general than would be   possible. \Eg $(\inside {x} \rightarrow  x.f=4) \rightarrow xf.3=7$  does not satisfy our definition of $\Pos {\_}$.
We have given these less general definitions in order to simplify both our defintions and our proofs.

\beginProof {l:preserve:asrt} 
Take any  state  $\sigma$, frame  $\phi$,  assertion  $A$,

\begin{itemize}
\item 
To show\\
  $\Stable{A}\ \wedge \  \fv(A)=\emptyset \ \ \  \Longrightarrow \ \  \ \  [\ \ M, \sigma \models A \ \ \Longleftrightarrow \ \  M,{\PushSLong \phi \sigma} \models A\ \ ]$

By induction on the structure of the definition of $\Stable{A}$.

\item 
To show\\
 $\Pos{A}\ \wedge \  \fv(A)=\emptyset \  \wedge \     {M\cdot\Mtwo \models {\PushSLong \phi \sigma}}\ \wedge\ 
  \ M, \sigma \models A \  \wedge \  M, {\PushSLong \phi \sigma} \models  \intThis \ \ \Longrightarrow$ \\
  $\strut \hspace{2cm}  \ \  M,{\PushSLong \phi \sigma} \models A $

By induction on the structure of the definition of $\Pos{A}$.
The only interesting case is when $A$ has the form $\inside {\re}$. Because 
$ fv(A)=\emptyset$, we know that $\interpret {\sigma} {\re}$=$\interpret {\PushSLong \phi \sigma}  {\re}$. Therefore, we assume that 
 $\interpret {\sigma} {\re} = \alpha$ for some $\alpha$, assume that $ M,  \sigma  \models \inside \alpha$, and want to show that  $ M,{\PushSLong \phi \sigma} \models \inside \alpha$. 
 From $   {M\cdot\Mtwo \models {\PushSLong \phi \sigma}}$ we obtain that
 $Rng(\phi) \subseteq Rng(  \sigma)  $. 
 From this, we obtain that
  $  \LRelevantO {\PushSLong \phi \sigma} \subseteq \LRelevantO  {\sigma}$.
  The rest follows by unfolding and folding Def. \ref{def:chainmail-protection}.

  \end{itemize}
 
\completeProof

\subsection{Encapsulation}

{
Proofs of adherence to {\SpecLang specifications  hinge on the expectation that some, 
specific, assertions cannot be invalidated unless some 
} internal (and thus known) computation took place. 
{We call such assertions   \emph{encapsulated}.}
}
We define the  judgement,  $M\ \vdash  \encaps{A}$, in terms of the judgment  $M; \Gamma \vdash \encaps A  ; \Gamma'$ from Fig. \ref{f:encaps:aux}.
This judgements ensures   that any objects whose fields  are read  in the validation of $A$ are internal, 
that $\protectedFrom {\_} {\_}$  does not appear in $A$,  and that protection assertions (ie $\inside{}$ or $\protectedFrom {\_} {\_}$) do not appear in negative positions in $A$. The second environment in this judgement, $\Gamma'$, is used to keep track of any variables introduces in that judgment, \eg we would have that\\
$\strut \hspace{1cm} M_{good}, \emptyset\ \vdash\ \encaps{a:\prg{Account}\wedge  k:\prg{Key} \wedge \inside{a.\prg{key}} \wedge a.\prg{key}\neq k}; \ (a:\prg{Account}, k:\prg{Key}$.

We assume a type judgment $M; \Gamma \vdash e :  \prg{intl}$ which says that in the context of $\Gamma$, the expression $e$ belongs to a class from $M$.
We also assume that the judgement $M; \Gamma \vdash e :  \prg{intl}$ can deal with ghostfields -- namely, ghost-methods have to be type checked in the contenxt of $M$ and therefor they will only read the state of internal objects.
Note that it is possible for $M; \Gamma \vdash \encaps {\re}; \Gamma'$ to hold and 
$M; \Gamma \vdash  e : \prg{intl}$ not to hold -- \cf rule {\sc{Enc\_1}}.

\begin{figure}[thb]
$
\begin{array}{l}
\begin{array}{lclcl}
\inferruleSDNarrow 
{~ \strut  {\sc{Enc\_1}}}
{  
\begin{array}{l}
M; \Gamma \vdash \re : \prg{intl} \\
M; \Gamma \vdash \encaps{\re};\  \Gamma'
\end{array}
}
{
M; \Gamma \vdash \encaps{\re.f};\  \Gamma'
}
& &
\inferruleSDNarrow 
{~ \strut  {\sc{Enc\_2}}}
{  
\begin{array}{l}
  \\
M; \Gamma \vdash \encaps{\re};\  \Gamma'
\end{array}
}
{
M; \Gamma \vdash \encaps{\re: C};\  (\Gamma', \re:C)
}
& &
\inferruleSDNarrow 
{~ \strut  {\sc{Enc\_3}}}
{   
\begin{array}{l}
M; \Gamma \vdash \encaps{A};\ \Gamma'  \\
 A \mbox{\ does\ not\ contain\ $\inside{\_}$}
\end{array}
}
{
M; \Gamma \vdash \encaps{ \neg A};\  \Gamma'  
}
\\ \\
\inferruleSDNarrow 
{~ \strut  {\sc{Enc\_4}}}
{  
\begin{array}{l}
M; \Gamma \vdash \encaps{A_1};\ \Gamma''   \\
  M; \Gamma'' \vdash \encaps{ A_2};\ \Gamma' 
  \end{array} 
}
{
M; \Gamma \vdash \encaps{A_1 \wedge A_2};\  \Gamma'
}
& &
\inferruleSDNarrow 
{~ \strut  {\sc{Enc\_5}}}
{  
\\
M; \Gamma, {x:C} \vdash \encaps {A};\ \Gamma' 
}
{
M; \Gamma \vdash \encaps {\forall {x:C}. A};\  \Gamma'
}
& & 
\inferruleSDNarrow 
{~ \strut  {\sc{Enc\_6}}}
{  
\\
M; \Gamma \vdash \encaps{\re};\  \Gamma'
}
{
M; \Gamma \vdash \encaps{\re: \prg{extl} };\  \Gamma'
}
\end{array}
\\ \\
\inferruleSDNarrow 
{~ \strut  {\sc{Enc\_7}}}
{  
M; \Gamma \vdash \encaps{\re}; \Gamma' 
}
{
M; \Gamma \vdash \encaps{\inside{\re}}; \ \Gamma' 
}
\end{array}
$
\caption{The judgment $M; \Gamma \vdash \encaps  {A}; \Gamma'$}
\label{f:encaps:aux}
\label{f:encaps}
\end{figure}

\begin{definition}[An assertion $A$ is \emph{encapsulated} by module $M$] $~$ \\
\label{d:encaps:sytactic}
\begin{itemize}
\item 
$M \vdash \encaps{A}  \ \   \triangleq  \ \  \exists \Gamma.[\ M; \emptyset \vdash \encaps{A}; \Gamma\ ]$ \ \  as defined in Fig. \ref{f:encaps}.
 \end{itemize}
  \end{definition}

To motivate the design of our judgment $M; \Gamma \vdash \encaps{A}; \Gamma'$,  we first give a semantic notion of encapsulation:

\begin{definition}  An assertion $A$ is semantically encapsulated by module $M$:
\label{d:encaps:sem}

\begin{itemize}
\item
$
    M\ \models \encaps{A}\ \   \triangleq  \   
     \forall \Mtwo, \sigma, \sigma'.[   \ \  \satisfiesA{M}{\sigma}{(A  \wedge \externalexec)}\  \wedge\ { \leadstoBounded {M\madd\Mtwo}  {\sigma}{\sigma'}} 
        \  \Longrightarrow\  
    {M},{\sigma'}\models {\as \sigma A} \ \  ]
  $
\end{itemize}
\end{definition}

\noindent
\textbf{More on Def. \ref{d:encaps:sem}} {If the definition \ref{d:encaps:sem} or in lemma \ref{d:encaps} we had used the more general execution, $\leadstoOrig  {M\madd\Mtwo}  {\sigma}{\sigma'}$, rather than the scoped execution,  $\leadstoBounded {M\madd\Mtwo}  {\sigma}{\sigma'}$,
 then fewer assertions would have been encapsulated.}
Namely, assertions like    $\inside {x.f}$ would not be encapsulated.
Consider, \eg, a heap $\chi$, with objects $1$, $2$, $3$ and $4$, where  $1$, $2$ are external, and $3$, $4$ are internal, and  $1$ has fields pointing to $2$ and $4$, and $2$ has a field pointing to $3$, and $3$ has a field $f$ pointing to $4$. Take  state $\sigma$=$(\phi_1\!\cdot\!\phi_2,\chi)$, where $\phi_1$'s receiver is $1$,  $\phi_2$'s receiver is $2$,   and there are no local variables. 
We have  $...\sigma\models \externalexec \wedge \inside {3.f}$. 
We  return from the most recent all, 
getting  $\leadstoOrig  {...}  {\sigma}{\sigma'}$ where $\sigma'=(\phi_1,\chi)$; and have   $...,\sigma'\not\models  \inside {3.f}$.

\begin{example}
\label{ex:not:encaps}
For an assertion $A_{bal}\  \triangleq\ a:\prg{Account}\wedge a.\prg{balance}=b$, 
and modules \ModB and  \ModC  from \S~\ref{s:outline}, we have  \ \ \ $\ModB\ \models\ \encaps{ A_{bal} }$, \ \ \ and \ \ \ $\ModB\ \models\ \encaps{ A_{bal} }$.
\end{example}

\begin{example} Assume   further modules, $\ModD$ and $\ModE$,  which  use ledgers mapping  accounts to their balances, and export functions that update this map. In  $\ModD$ the ledger is  part of the {internal} module, 
while in $\ModE$ it is part of the  {external} module.
Then  \ \ $\ModD \ \not\models\encaps{ A_{bal}} $, \ \  and \ \ $\ModE  \models \encaps{ A_{bal}} $.
Note that in both $\ModD$ and $\ModE$, the term \prg{a.balance} is a ghost field. 
\end{example}

\begin{note} Relative protection 
is not encapsulated, (\eg $M \not\models {\encaps{\protectedFrom{x}{y}}}$), even though    absolute protection is
(\eg $M \models \encaps{\inside{x}}$).
Encapsulation of an assertion does not imply encapsulation of its negation; 
 for example,  $M \not\models {\encaps{\neg\inside{x}}}$.
\end{note}

\noindent
\textbf{More on Def. \ref{d:encaps:sytactic}} This definition is less permissive than necessary. 
For example $M \not\vdash \encaps{\neg ( {\neg \inside {x}})}$ even though 
 $M  \models \encaps{\neg ( {\neg \inside {x}})}$.
 Namely, 
$\neg (\neg \inside {x}) \equiv  \inside {x}$ and $M \vdash {\encaps{  \inside {x}}}$.
A more permissive, sound, definition, is not difficult, but not the main aim of this work.
We gave this, less permissive definition, in order to simplify the definitions and the proofs.

 \vspace{.1cm}

 \beginProof{lem:encap-soundness} This says that $M \vdash {\encaps A}$ implies that $M \vdash {\encaps A}$.
 \\
 We fist prove that\\
 $\strut \ \ \ \ $  (*) Assertions $A_{poor}$ which do not contain $\inside {\_}$ or $\protectedFrom {\_} {\_}$ are preserved by any external step.\\
 Namely, such an assertion only depends on the contents of the fields of internal objects, and these are not modified  by external steps.
 Such an $A_{poor}$ is defined essentially through
\\ 
 $
\begin{syntax}
\syntaxElement{\ \ \ \  A_{poor}}{}
		{
		\syntaxline
				{{\re}}
				{{\re} : C}
				{\neg A_{poor}}
				{A_{poor}\ \wedge\ A_{poor}}
				{\all{x:C}{A_{poor}}}
				{\external{{\re}}}
		\endsyntaxline
		}
\endSyntaxElement\\
\end{syntax}
$
\\
 We can prove (*) by induction on the structure of $A_{poor}$ and case analysis on the execution step.
  
\vspace{.05cm}   
\noindent
We then prove Lemma \ref{lem:encap-soundness} by induction on the structure of $A$. 

\noindent 
--- The cases {\sc{Enc\_1}}, {\sc{Enc\_2}},  and {\sc{Enc\_6}} 
are  straight application of (*). 

\noindent 
--- The case {\sc{Enc\_3}} also follows from (*), because any $A$ which satisfies $\encaps {A}$ and which does not contain $\inside {\_}$
is an $A_{poor}$ assertion.

\noindent 
--- The cases {\sc{Enc\_4}} and {\sc{Enc\_5}}  follow  by induction hypothesis.

\noindent 
--- The case {\sc{Enc\_7}} is more interesting. \\
We assume that $\sigma$ is an external state,   that $\leadstoBounded {...} {\sigma} {\sigma'}$, and that $.. \sigma \models \inside \re$.
By definition, the latter means that\\
 $\strut \ \ \ \ $  (**) no locally reachable external object in $\sigma$ has a field ponting to $\re$, \\
 $\strut \ \ \ \ \ \ \ \ \ \ \ $ nor is $\re$ one of the variables.\\
We proceed by case analysis on the step $\leadstoBounded {...} {\sigma} {\sigma'}$.
\\
- If that step was an assignment to a local variable $x$, then this does not affect $\inside {\as \sigma \re}$ because in $\interpret \sigma {\re}$=$\interpret {\sigma'} {\as \sigma {\re}}$,
and $... \sigma' \models x \neq re$.
\\
- If that step was an assignment to an external object's field, of the form $x.f :=y$ then this does not affect $\inside {\as \sigma \re}$ either.
This is so, because $\encaps {\re}$ gives that $\interpret \sigma {\re}$=$\interpret \sigma' {\as \sigma {\re}}$ -- namely the evaluation of $\re$ does not read $x$'s fields, since $x$ is external. 
And moreover, the assignment $x.f :=y$ cannot create a new, unprotected  path to $\re$ (unprotected means here that the penultimate element in that path is external), because then we would have had in $\sigma$  an unprotected path from $y$ to $\re$.
\\
- If that step was a method call, then we apply lemma \ref{lemma:push:N} which says that all objects reachable in $\sigma'$ were already reachable in $\sigma$.
\\
- Finally, we do not consider method return (\ie the rule {\sc{Return}}),   because we are looking at $\leadstoBounded {...} {\_} {\_}$ execution steps rather than $\leadstoOrig  {...}  {\_}{\_}$ steps.
\hspace{3cm}
\completeProof


\clearpage

\section{Appendix to Section \ref{sect:spec} -- Specifications}
\label{app:spec}

 \begin{example}[Badly Formed Method Specifications]
$S_{9,bad\_1}$ is not a well-formed specification, because $A'$ is not a formal parameter, nor free in the precondition. 

   {\sprepost
		{\strut \ \ \ \ \ \ \ \ \ S_{9,bad\_1} }
		{  a:\prg{Account} \wedge  \inside{a} }
		{\prg{public Account}} {\prg{set}} {\prg{key'}:\prg{Key}}
		{   \inside{a}\wedge  \inside{a'.\prg{key}}  }
		{  true }
}
		
 {\sprepost
		{\strut \ \ \ \ \ \ \ \ \ S_{9,bad\_2} }
		{  a:\prg{Account} \wedge  \inside{a} }
		{\prg{public Account}} {\prg{set}} {\prg{key'}:\prg{Key}}
		{   \inside{a}\wedge  \inside{a'.\prg{key}}  }
		{  \prg{this}.\balance \ }

}
\end{example}

{ \begin{example}[More Method Specifications]
\label{ex:spesMore}
$S_7$ below  guarantees that
\prg{transfer} does not affect the balance of accounts different  from the receiver or argument, and  if the key supplied is not that of the receiver, then no account's balance is affected.  \
$S_8$ guarantees that if the key supplied is that of the receiver, the correct amount is transferred from the receiver to the destination.
 $S_9$ guarantees that \prg{set} preserves the protectedness of a key.

\small{
{\sprepost
		{\strut \ \ \ \ S_7} 
		{ a:\prg{Account}\wedge  a.\prg{\balance}=b \wedge
		(\prg{dst}\neq a\neq\prg{this} \vee \prg{key'}\neq a.\prg{\password})}
	               {\prg{public Account}} {\prg{transfer}} {\prg{dst}:\prg{Account},\prg{key'}:\prg{Key},\prg{amt}:\prg{nat}}
		{ a.\prg{\balance}=b}
		{\sdred{ a.\prg{\balance}=b}}
}
\\
{\sprepostLB
		{\strut \ \ \ \ \ \ \ \ \ S_8} 
		{  
		\prg{this}\neq \prg{dst}\wedge \prg{this}.\prg{\balance}=b \wedge  \prg{dst}.\prg{\balance}=b' }
		  {\prg{public Account}}
		  	  {\prg{transfer}} 
		   	 {\prg{dst}:\prg{Account},\prg{key'}:\prg{Key},\prg{amt}:\prg{nat} }
		{\prg{this}.\prg{\balance}=b-\prg{amt} \wedge \prg{dst}.\prg{\balance}=b'+\prg{amt} } 
		{    \prg{this}.\prg{\balance}=b \wedge  \prg{dst}.\prg{\balance}=b'   }
}
\\
{\sprepost
		{\strut \ \ \ \ \ \ \ \ \ S_9} 
		{  a:\prg{Account}\wedge
		 \inside{a.\prg{\password}}}
		{\prg{public Account}} {\prg{set}} {\prg{key'}:\prg{Key}}
		{ \inside{a.\prg{\password}}}
		{ \inside{a.\prg{\password} } }		
}
}
\end{example}
}
 
%

\forget{
 \label{example:twostatesarisfy}
\se{We revisit the modules and specifications from Sect. \ref{s:bankSpecEx}, and Example \ref{ex:spacesMore} :}

\begin{tabular}{lllllllll}
$\ModA  \not\models S_1$  &   $\ModA  \models S_2$  &     $\ModA \models S_3$    & $\ModA \models S_5$\\
 $\ModB \not\models S_1$  &   $\ModB \not\models S_2$     &  $\ModB  \not\models S_3$   & $\ModB \not\models S_5$ \\
 $\ModC  \not\models S_1$    & $\ModC \models S_2$ &   & $\ModC \not\models S_3$   & $\ModC \not\models S_5$ 
\end{tabular}
\end{example}

 \begin{example}
 \label{example:mprepost:sat:one}
 For  
 Example \ref{example:mprepostl}, we have
  $\ModA \models S_6$ and $\ModB \models S_6$ and  $\ModC \models S_6$.
Also,  $\ModA \models S_7$ and $\ModB \models S_7$ and  $\ModC \models S_7$.
However,   $\ModA  \models S_8$, while $\ModB  \not\models S_8$.
\end{example}

 \begin{example}
\label{example:mprepost:sat:two}
 For  
any   specification  $S \triangleq {\mprepost{A}{p\ C}{m}{x}{C}{A'} }$ and any module  $M$ which does not have a class $C$  with a method $m$ with formal parameter  types ${\overline C}$, we have that $M \models S$.
Namely, if a method were to be called with that signature on a $C$  from $M$, then execution would be stuck, and the requirements from Def. \ref{def:necessity-semantics}(3) would be trivially satisfied.
Thus,   $\ModC \models S_8$. 
\end{example}
}

\subsection{Examples of Semantics of our Specifications}

\begin{example}
 \label{example:mprepost:sat:three}
We  revisit the specifications given in Sect. \ref{s:bankSpecEx},  the three  modules from Sect. \ref{s:bank}, and Example \ref{ex:spesMore}

\begin{tabular}{lllllllll}
$\ModA  \models S_1$  &   $\ModA  \models S_2$ &   $\ModA \models S_3$    & $\ModA \models S_5$\\
 $\ModB \models S_1$  &   $\ModB \not\models S_2$   &  $\ModB  \not\models S_3$   & $\ModB \not\models S_5$ \\
 $\ModC  \models S_1$    & $\ModC \models S_2$ & $\ModC \not\models S_3$   & $\ModC \not\models S_5$ 
\end{tabular}
\end{example}

 \begin{example}
 \label{example:mprepost:sat:four}
 For  
 Example \ref{example:mprepostl}, we have
  $\ModA \models S_7$ and $\ModB \models S_7$ and  $\ModC \models S_7$.
Also,  $\ModA \models S_8$ and $\ModB \models S_8$ and  $\ModC \models S_8$.
However,   $\ModA  \models S_9$, while $\ModB  \not\models S_9$.
\end{example}

 \begin{example}
\label{example:mprepost:sat:five}
 For  
any   specification  $S \triangleq {\mprepost{A}{p\ C}{m}{x}{C}{A'} }$ and any module  $M$ which does not have a class $C$  with a method $m$ with formal parameter  types ${\overline C}$, we have that $M \models S$.
Namely, if a method were to be called with that signature on a $C$  from $M$, then execution would be stuck, and the requirements from Def. \ref{def:necessity-semantics}(3) would be trivially satisfied.
Thus,   $\ModC \models S_8$. 
\end{example}

\subsubsection{Free variables in well-formed specifications}
\label{wff:spec:free:more}

We now discuss the requirements about free variables in well-formed specifications as defined in Def. \ref{f:holistic-wff}.
 In scoped invariants, $A$ may only mention   variables introduced by the quantifier, $\overline x$. In method specifications, the precondition, $\overline{x:C'} \wedge A$, may only mention the receiver, \prg{this}, the formal parameters, $\overline y$, and the explicitly introduced  variables, $\overline x$; it may \emph{not}  mention the result  \prg{res}. The postcondition, $A'$, may mention these variables, and in addition, may mention the result, \prg{res}. The mid-condition, $A''$ is about a state which has at least one more frame than the current method's,
and therefore it  may not mention \prg{this}, nor $\overline{y}$, nor \prg{res}.

\section{Expressiveness} 

\label{app:expressivity}

\begin{figure}[tbh]
\begin{lstlisting}[language = Chainmail, mathescape=true, frame=lines]
module DOM 
   class Node
      field cnt: int
      field parent: Node 
      
   class Proxy
      field hght: nat
      field node: Node
	
      public method set(newCnt:int, up:nat): void
            if this.hght >= up
               setPrivate(newCnt, up)
            else
               return
	       
      private method setPrivate(newCnt:int, up:nat): void
            if up==0 then
               this.node.cnt := nwCnt
            else
               setPrivate(newCnt, up-1)
\end{lstlisting}
\caption{The DOM module  -- classes \prg{Node} and \prg{Proxy}}
\label{fig:DoMCode}
\end{figure}

We argue the expressiveness of our approach by comparing with example specifications  proposed in \cite{OOPSLA22,dd,irisWasm23}.

\subsection{The DOM}  
\label{ss:DOM}

\subsubsection{The Problem} This is the motivating example in \cite{dd}. It
deals with a tree of DOM nodes: Access to a DOM node
gives access to all its \prg{parent}s, 
with the ability to modify the node's contents   -- where  \prg{parent} 
and \prg{cnt} are fields in class \prg{Node}. Since the top nodes of the tree
usually contain privileged information, while the lower nodes contain
less crucial third-party information, we must be able to limit 
 access given to third parties to only the lower part of the DOM tree. 
 
To do this,   \citet{dd}  propose   
  a \prg{Proxy} class, which has a field \prg{node} pointing to a \prg{Node}, and a field height (\prg{hght}), which restricts the range of \prg{Node}s which may be modified through the use of the particular \prg{Proxy}. Namely,   a \prg{Proxy}  may modify 
the \prg{cnt} of all the ancestors of its  \prg{node}, up to the   \prg{hght}-th ancestor of that \prg{node}. 

A possible implementation of such \prg{Node} and \prg{Proxy} classes is shown in Fig.  \ref{fig:DoMCode}, 
while the creation of a tree, three proxies, and  passing the proxies to the external world is shown in Fig \ref{fig:DoMCodeClient}.

\begin{figure}[tbh]
\begin{lstlisting}[language = Chainmail, mathescape=true, frame=lines]
module DOM 
   ... as before ...
   class Example
      public method demo(untrst:external) : void 
          //  create a tree of $\prg{Node}s$
         nd1 := new Node; nd1.cnt := 1
         nd2 := new Mode; nd2.parent:= nd1; nd2,ctns:=2;
         nd3 := new Mode; nd3.parent:= nd1; nd3,ctns:=3;
         nd4 := new Mode; nd4.parent:= nd2; nd4,ctns:=4;
         nd5 := new Mode; nd5.parent:= nd2; nd5,ctns:=5;
         
         //  create three $\prg{Proxy}s$
         prx10 := newProxy; prx10.height:=2; prx10.node:=nd4;
         prx11 := newProxy; prx10.height:=1; prx10.node:=nd4; 
         prx12 := newProxy; prx10.height:=0; prx10.node:=nd4;
          
         // External calls:         
         untrst.meth_A();       // no modification in the tree
         
         untrst.meth_B(nd2);    // no modification in the tree     
         
         untrst.meth_C(prx12);  // nd4.cnt may have changed; 
                                // all else has stayed the same
         
         nd4.ctns:=4            // re-establish old value of nd4.ctns
         untrst.meth_A( );      // nd4.cnt may have changed again; 
                                // all else has stayed the same
         
         untrst.meth_C(prx12);  // nd4.cnt, nd2.cnt, and nd1.ctns may have changed; 
                                // all else has stayed the same
      
        // Revocation:
         prx12.height := 1;     // revocation
         nd1.ctns:=1            // re-establigh the old value of nd1.ctns
         untrst.meth_C(prx12);  // nd1.ctns does not change     
\end{lstlisting}
\caption{The DOM module continued -- class \prg{Example} }
\label{fig:DoMCodeClient}
\end{figure}
In Fig. \ref{f:DOM:Tree:Diagrs}  we show a diagram of  a tree consisting of the five nodes, \prg{nd1}, \prg{nd2},  \prg{nd3}, \prg{nd4} and \prg{nd5}, 
and three proxies, \prg{prx10}, \prg{prx11},  and \prg{prx12}. These are created in lines 6-15 in the code in Fig. \ref{fig:DoMCodeClient}.
Moreover, in Fig. \ref{f:DOM:Diagrs}  we revisit the diagram, and show highlight that three objects have capability on \prg{nd4.cnt},
while two objects have capability on \prg{nd2.cnt}, and finally, only one object has capability on  \prg{nd1.cnt}.

We now consider the remaining code in Fig. \ref{fig:DoMCodeClient}:

\begin{itemize}
\item 
The first external call, in line 18, does not affect the DOM tree or is \prg{cnt}, because \prg{untrst} has no access to it.
 \item 
The second external call, in line 20, also does not affect the DOM tree or is \prg{cnt} -- even though \prg{untrst} has access to a \prg{Node}, because the class has no public method, none of its fields may be modified. 
 \item 
The third external call, in line 22, may affect the contents of \prg{nd4.ctns}, because we have passed the capability \prg{prx12} in the call. 
From then on, the capability is no longer protected.
 \item 
Before the fourth external call, on line 25, we restore the value of \prg{nd4.cnt} to what it was earlier. Nevertheless, even though the fourth external call, on 
line 26, is identical to the from line 19, it may result in a modification of \prg{nd4.ctns}. This is so, because \prg{prx12} is no longer protected.
 \item 
The fifth external call, on line 29, may result in the modification of \prg{nd4.ctns}, \prg{nd2.ctns}, and \prg{nd1.ctns}, because we have passed the capability \prg{prx10}.
 \item 
We now revoke the capability of \prg{prx12}, on line 34, so that it can only affect \prg{nd4.cnt} and \prg{nd2.cnt}. We re-establish the value of \prg{nd1.cnt} on line 35, and on line 36 we make an external call. We know that this call cannot affect the value of \prg{nd1.cnt}.
\footnote{Another example of a revocation would be to set the   \prg{parent} of \prg{nd4} to be \prg{nd3}; then, \prg{prx11} 
would have capability on \prg{nd4} and \prg{nd3}.
On the other hand, setting the \prg{parent} of \prg{nd4}  to be \prg{nd1} would augment the capability of \prg{prx11};
such an augmentation would, however be forbidden by $S_{dom_4}$, if the corresponding capability was had already been exported to the external world.}
\end{itemize}

\begin{figure}[th] 
\resizebox{7cm}{!}{
\includegraphics[width=\linewidth]{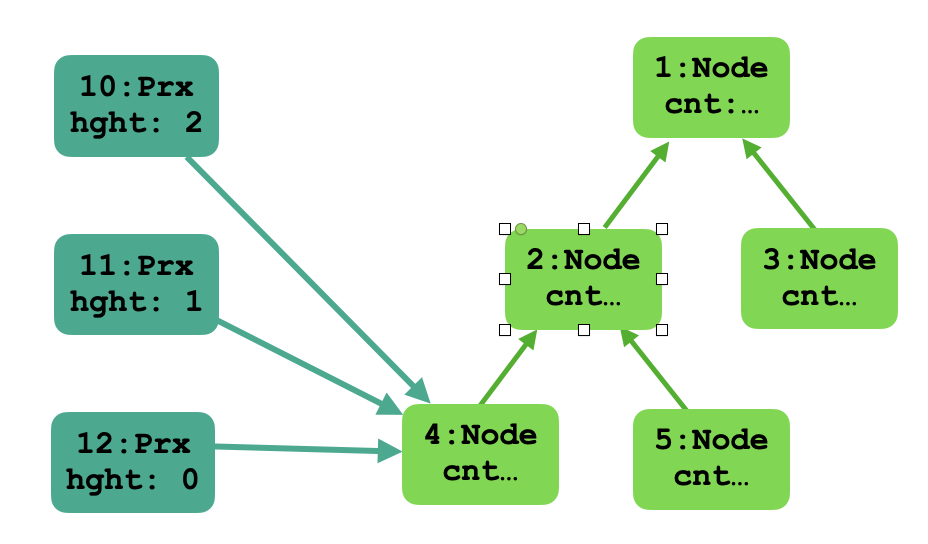}
} 
\caption{DOM tree and Proxies, created as in the code from Fig \ref{fig:DoMCodeClient}}
 \label{f:DOM:Tree:Diagrs}
 \end{figure}

\begin{figure}[th] 
\begin{tabular}{|c|c|c|}
\hline
\resizebox{3,8cm}{!}{
\includegraphics[width=\linewidth]{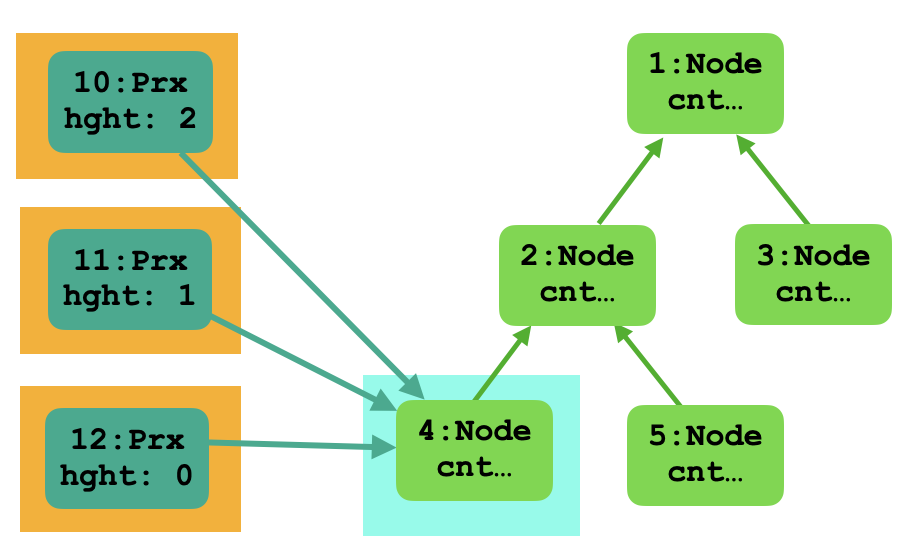}
}
&
\resizebox{3.8cm}{!}{
\includegraphics[width=\linewidth]{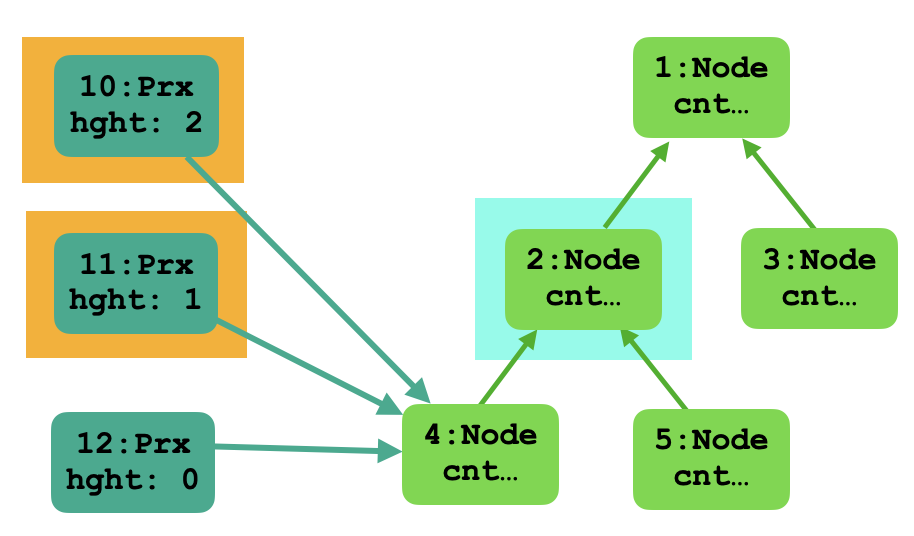}
}
&
\resizebox{3,8cm}{!}{
\includegraphics[width=\linewidth]{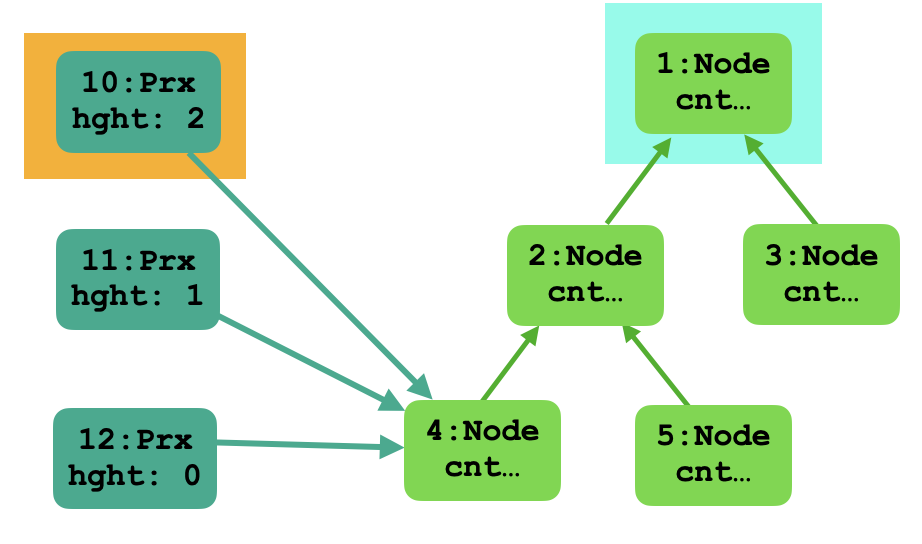}
}
\\
\hline
\prg{prx10}, \prg{prx10}, and \prg{prx12} are the  
&
\prg{prx10} and  \prg{prx11}  are the 
&
\prg{prx10}  is the\ 
\\
capabilities for \prg{nd4.cnt} 
&
capabilities for \prg{nd2.cnt}  
&
capability for \prg{nd1.cnt}  
\\
\hline
\end{tabular}
\caption{Capabilities for the DOM tree from Fig. \ref{f:DOM:Tree:Diagrs}. In particular, notice the many-to-many relation between capabilities and effects: For example, there are three capabilities on \prg{nd4.cnt}. And \prg{prx10} has capabilities on \prg{cnt} of three objects.}
 \label{f:DOM:Diagrs}
 \end{figure}

\subsubsection{Specifications}
We now discuss possible specifications. For this, we will use  the predicate $may\_modify \subseteq \prg{Proxy} \times \prg{Node}$, which says that 
\prg{prx} has \emph{modification-capabilities} on \prg{nd}, where \prg{prx} is
a  \prg{Proxy} and \prg{nd} is a \prg{Node}, if \prg{nd} is the \prg{prx}$k^{th}$  parent
of   \prg{pr.node} where $k \leq \prg{nd}.\prg{hght}$.
The formal definition is as follows:
\\
$\strut \SPSP   may\_modi\!f\!y(\prg{prx}, \prg{nd}) \triangleq \exists k:\mathbb{N}. [ \  \prg{prx}.\prg{node}.\prg{parent}^k=\prg{nd}\ \wedge k\leq   \prg{prx}.\prg{hght}]$
\\
Thus, a \prg{Proxy} object \prg{prx}, which satisfies $may\_modi\!f\!y(\prg{prx},\prg{nd})$ is a capability which may modify the node $\prg{nd}$.

\vspace{.1cm}

\noindent
We now use $may\_modify$ in writing the specifications.
Below, $S_{dom\_1}$ mandates  that nodes, $\prg{nd}$ are not leaked. 
Moreover, $S_{dom\_2}$  and $S_{dom\_3}$   mandate that the \prg{parent} node of a \prg{Node} cannot be modified by external code: 
{Similarly, $S_{dom\_4}$ mandates  that proxies which $may\_modi\!f\!y(\prg{prx},\prg{nd})$ are not leaked.}
{Finally, $S_{dom\_5}$ mandates   that proxies which $may\_modi\!f\!y(\prg{prx},\prg{nd})$ are not leaked, and that when a node $nd$ is protected, and  all proxies that can  $may\_modi\!f\!y(\prg{prx},\prg{nd})$ are protected, then  the $cnt$ of $nd$ cannot be modified by external code.}
\\
$\strut \SPSP  S_{dom\_1}\ \  \triangleq \ \ \TwoStatesN{ \prg{nd}:\prg{Node}}{\  \inside{\prg{nd}}\  ] \ }$ 
\\
\\
$\strut \SPSP  S_{dom\_2}\ \  \triangleq \ \ \TwoStatesN{ \prg{nd}:\prg{Node},\prg{nd'}:\prg{Node}}{\  \prg{nd.parent}=\prg{nd}'   \ }$ 
\\
$\strut \SPSP  S_{dom\_3} \  \triangleq \ \ \TwoStatesN{ \prg{nd}:\prg{Node} }{\  \prg{nd.parent}=\prg{null} \ }$ 
\\
\noindent
\\
$\strut \SPSP  S_{dom\_4}\ \  \triangleq \ \ \TwoStatesN{ \prg{nd}:\prg{Node}}{\  \forall \prg{prx}:\prg{Proxy}.[ \ may\_modi\!f\!y(\prg{prx}, \prg{nd}) \rightarrow \inside{\prg{prx}}\  ] \ }$ 
\\
\\
$\strut \SPSP  S_{dom\_5}\ \  \triangleq \ \  \forall{ \prg{nd}:\prg{Node}.\forall \prg{val}:\prg{Object} }$.\\
$\strut \SPSP\strut  \SPSP\strut \SPSP	\{  \  \inside{\prg{nd}} \wedge \forall \prg{prx}:\prg{Proxy}.[ \ may\_modi\!f\!y(\prg{prx}, \prg{nd} ) \rightarrow \inside{\prg{prx}}\  ]  \wedge \prg{nd.cnt} = \prg{val} \ \}  $

 Note that $S_{dom\_5}$ is strictly stronger than $S_{dom\_4}$. 
The module shown in Fig. \ref{fig:DoMCode} and \ref{fig:DoMCodeClient} does satisfy  all the  specifications from above. 

The careful reader might worry that the method \prg{demo} in class \prg{Example} in Fig. \ref{fig:DoMCodeClient} 
breaks specifications $S_{dom\_1}$, and $S_{dom\_4}$, since it does export \prg{nd1}  and \prg{prx12}. 
However, there is no external state between the creation of   \prg{nd1} and its export, and so our scoped invariant is preserved. 
Similarly, there is no external state between the creation of   \prg{prx10} and its export, and so our scoped invariant is preserved.

 To prove that the external calls in Fig. \ref{fig:DoMCodeClient} preserve the stated properties  we need to apply the specifications  $S_{dom\_2}$, 
 and $S_{dom\_5}$, but we do not need   $S_{dom\_1}$ or $S_{dom\_3}$.

\subsubsection{Other Specifications}

\citet{OOPSLA22} specify this as:
 
 \begin{lstlisting}[language = Chainmail, mathescape=true, frame=lines]
DOMSpec $\triangleq$ from nd : Node $\wedge$ nd.property = p  to nd.property != p
  onlyIf $\exists$ o.[ $\external {\prg{o}}$ $\wedge$ 
     $( \exists$ nd':Node.[ $\access{\prg{o}}{\prg{nd'}}$ ]  $\vee$ 
     $\,\;\exists$ pr:Proxy,k:$\mathbb{N}$.[$\, \access{\prg{o}}{\prg{pr}}$ $\wedge$ nd.parent$^{\prg{k}}$=pr.node.parent$^{\prg{pr.height}}$ ] $\,$ ) $\,$ ]
\end{lstlisting}

\prg{DomSpec} states that the \prg{property} of a node can only change if
some external object presently has 
access to a node of the DOM tree, or to some \prg{Proxy} with modification-capabilties
to the node that was modified.
The assertion $\exists {o}.[\ \external {\prg{o}} \wedge \access{\prg{o}}{\prg{pr}}\ ]$ is the contrapositive of our  $\inside{pr}$, but is is weaker than that, because it does not specify the frame from which $o$ is accessible.
Therefore, $\prg{DOMSpec}$ is a stronger requirement than $S_{dom\_1}$.

\subsection{DAO}
The Decentralized Autonomous Organization (DAO)~\cite{Dao}  is a well-known Ethereum contract allowing 
participants to invest funds. The DAO famously was exploited with a re-entrancy bug in 2016, 
and lost \$50M. Here we provide specifications that would have secured the DAO against such a 
bug. 
\\ 
$\strut \SPSP  S_{dao\_1}\ \  \triangleq \ \ \TwoStatesN{ d:\prg{DAO}}{\ \forall p:\prg{Participant}. [\ d.ether \geq d.balance(p) \ ]   \ }$ 
\\
$\strut \SPSP  S_{dao\_2}\ \  \triangleq \ \ \TwoStatesN{ d:\prg{DAO}}{\ \ d.ether \geq \sum_{p \in d.particiants} d.balance(p)\  \ }$

The specifications above say the following:
\\
\begin{tabular}{ll}
\begin{minipage}{.10\textwidth}
$\strut \SPSP  S_{edao\_1}$
\end{minipage}
&
\begin{minipage}{.85\textwidth}
guarantees that the DAO holds more ether than the balance  of any of its  participant's.
\end{minipage}
\\
\\
\begin{minipage}{.10\textwidth}
$\strut \SPSP  S_{dao\_2}$ 
\end{minipage}
&
\begin{minipage}{.85\textwidth}
guarantees that that the DAO holds more ether than the sum  of the balances held by DAO's participants.
\end{minipage}
\end{tabular}

$S_{dao\_2}$  is stronger than $S_{dao\_1}$. They would both have precluded the DAO bug. Note that these specifications  do not mention capabilities. 
They are, essentially, simple class invariants and could have been expressed with the techniques proposed already by \cite{MeyerDBC92}.
The only difference is that $S_{dao\_1}$ and $S_{dao\_2}$ are two-state invariants, which means that we require that they are \emph{preserved},
\ie if they hold in one (observable) state they have to hold in all successor states,
while class invariants are one-state, which means they are required to hold in all (observable) states.
\footnote{This should have been explained somewhere earlier.}

\vspace{0.5cm}
We now compare with the specification given in \cite{OOPSLA22}.
\prg{DAOSpec1} in similar to  $S_{dao\_1}$: iy
says that no participant's balance may ever exceed the ether remaining 
in DAO. It is, essentially, a one-state invariant.

\begin{lstlisting}[language = Chainmail, mathescape=true, frame=lines]
DAOSpec1 $\triangleq$ from d : DAO $\wedge$ p : Object
            to d.balance(p) > d.ether
            onlyIf false
\end{lstlisting}
\prg{DAOSpec1}, similarly to $S_{dao\_1}$,   in that it enforces a class invariant of \prg{DAO}, something that could be enforced
by traditional specifications using class invariants.

 \cite{OOPSLA22}  gives one more   specification: 
 
 \begin{lstlisting}[language = Chainmail, mathescape=true, frame=lines]
DAOSpec2 $\triangleq$ from d : DAO $\wedge$ p : Object
            next d.balance(p) = m
            onlyIf $\calls{\prg{p}}{\prg{d}}{\prg{repay}}{\prg{\_}}$ $\wedge$ m = 0 $\vee$ $\calls{\prg{p}}{\prg{d}}{\prg{join}}{\prg{m}}$ $\vee$ d.balance(p) = m
\end{lstlisting}

 \prg{DAOSpec2} states that if after some single step of execution, a participant's balance is \prg{m}, then 
either 
\begin{description}
\item[(a)] this occurred as a result of joining the DAO with an initial investment of \prg{m}, 
\item[(b)] the balance is \prg{0} and they've just withdrawn their funds, or 
\item[(c) ]the balance was \prg{m} to begin with
\end{description}

%
%

\subsection{ERC20}

The ERC20 \cite{ERC20} is a widely used token standard describing the basic functionality of any Ethereum-based token 
contract. 
This functionality includes issuing tokens, keeping track of tokens belonging to participants, and the 
transfer of tokens between participants. Tokens may only be transferred if there are sufficient tokens in the 
participant's account, and if either they (using the \prg{transfer} method) or someone authorised by the participant (using the \prg{transferFrom} method) initiated the transfer. 

For an $e:\prg{ERC20}$, the term $e.balance(p)$  indicates the number of tokens in   participant $p$'s  account at $e$.
The 
assertion $e.allowed(p,p')$ expresses that participant $p$ has been authorised to spend moneys from $p'$'s account at $e$.
 
The security model in Solidity is not based on having access to a capability, but on who the caller of a method is. 
Namely, Solidity supports the  construct \prg{sender} which indicates the identity of the caller.
Therefore, for Solidity, we adapt our approach in two significant ways:
we change the meaning of $\inside{\re}$ to express that $\re$ did not make a method call.
Moreover, we introduce a new, slightly modified form of two state invariants of the form $\TwoStates{\overline {x:C}}{A}{A'}$ which expresses that any execution which satisfies $A$, will preserve $A'$.


We specify the guarantees of   ERC20  as follows:
\\
\\
$\strut \SPSP  S_{erc\_2}\ \  \triangleq \ \ \TwoStatesLB{ e:\prg{ERC20},p,p':\prg{Participant},n:\mathbb{N}} 
 {\ \forall p'.[\,(e.allowed(p',p) \rightarrow   \inside{p'}\, ] \ } { \ e.balance(b)=n \ } $ 
\\
$\strut \SPSP  S_{erc\_3}\ \  \triangleq \ \ \TwoStatesLB{ e:\prg{ERC20},p,p':\prg{Participant}}  {\ \forall p'.[\,(e.allowed(p',p) \rightarrow   \inside{p'}\, ] \ } { \ \neg (e.allowed(p'',p) \ } $ 

The specifications above say the following:
\\
\begin{tabular}{ll}
\begin{minipage}{.10\textwidth}
$\strut \SPSP  S_{erc\_1}$
\end{minipage}
&
\begin{minipage}{.85\textwidth}
guarantees that the the owner of an account is always authorized on that account -- this specification is expressed using the original version of two-state invariants.
\end{minipage}
\\
\\
\begin{minipage}{.10\textwidth}
$\strut \SPSP  S_{erc\_2}$ 
\end{minipage}
&
\begin{minipage}{.85\textwidth}
guarantees that any execution which does not contain calls from a participant $p'$ authorized on $p$'s account will not affect the balance of $e$'s account. Namely, if the execution starts in a state in which $ e.balance(b)=n$, it will lead to a state where $ e.balance(b)=n$ also holds.
\end{minipage}
\\
\\
\begin{minipage}{.10\textwidth}
$\strut \SPSP  S_{erc\_3}$ 
\end{minipage}
&
\begin{minipage}{.85\textwidth}
guarantees that any execution which does not contain calls from a participant $p'$ authorized on $p$'s account will not affect who else is authorized on that account. That is, if the execution starts in a state in which $ \neg (e.allowed(p'',p)$, it will lead to a state where $ \neg (e.allowed(p'',p)$ also holds.
\end{minipage}
\end{tabular}


\vspace{1cm}

We compare with the specifications given in \cite{OOPSLA22}:
 Firstly, \prg{ERC20Spec1} 
says that if the balance of a participant's account is ever reduced by some amount $m$, then
that must have occurred as a result of a call to the \prg{transfer} method with amount $m$ by the participant,
or the \prg{transferFrom} method with the amount $m$ by some other participant.
\begin{lstlisting}[language = Chainmail, mathescape=true, frame=lines]
ERC20Spec1 $\triangleq$ from e : ERC20 $\wedge$ e.balance(p) = m + m' $\wedge$ m > 0
              next e.balance(p) = m'
              onlyIf $\exists$ p' p''.[$\calls{\prg{p'}}{\prg{e}}{\prg{transfer}}{\prg{p, m}}$ $\vee$ 
                     e.allowed(p, p'') $\geq$ m $\wedge$ $\calls{\prg{p''}}{\prg{e}}{\prg{transferFrom}}{\prg{p', m}}$]
\end{lstlisting}
Secondly, \prg{ERC20Spec2} specifies under what circumstances some participant \prg{p'} is authorized to 
spend \prg{m} tokens on behalf of \prg{p}: either \prg{p} approved \prg{p'}, \prg{p'} was previously authorized,
or \prg{p'} was authorized for some amount \prg{m + m'}, and spent \prg{m'}.
\begin{lstlisting}[language = Chainmail, mathescape=true, frame=lines]
ERC20Spec2 $\triangleq$ from e : ERC20 $\wedge$ p : Object $\wedge$ p' : Object $\wedge$ m : Nat
              next e.allowed(p, p') = m
              onlyIf $\calls{\prg{p}}{\prg{e}}{\prg{approve}}{\prg{p', m}}$ $\vee$ 
                     (e.allowed(p, p') = m $\wedge$ 
                      $\neg$ ($\calls{\prg{p'}}{\prg{e}}{\prg{transferFrom}}{\prg{p, \_}}$ $\vee$ 
                              $\calls{\prg{p}}{\prg{e}}{\prg{allowed}}{\prg{p, \_}}$)) $\vee$
                     $\exists$ p''. [e.allowed(p, p') = m + m' $\wedge$ $\calls{\prg{p'}}{\prg{e}}{\prg{transferFrom}}{\prg{p'', m'}}$]
\end{lstlisting}

\prg{ERC20Spec1} is related to $S_{erc\_2}$. Note that \prg{ERC20Spec1} is more API-specific, as it expresses the precise methods which caused the modification of the balance.

\subsection{Wasm, Iris, and the stack}

In \cite{irisWasm23}, they consider inter-language safety for Wasm. They develop Iris-Wasm, a mechanized higher-order separation logic mechanized in Coq and the Iris framework. Using Iris-Wasm, with the aim to
specify and verify individual modules separately, and then compose them modularly in a simple host language
featuring the core operations of the WebAssembly JavaScript Interface. They develop a 
logical relation that enforces robust safety: unknown, adversarial code can only aﬀect other modules through
the functions that they explicitly export. 
They do not offer however a logic to deal with the effects of external calls.

As a running example, they use a \prg{stack} module, which is an array of values, and exports functions to inspect the stack contents or modify its contents. 
Such a setting can be expressed in our language through a \prg{stack} and a \prg{modifier} capability.
Assuming a predicate $Contents(\prg{stack},\prg{i},\prg{v})$, which expresses that the contents of \prg{stack} at index \prg{i} is \prg{v}, we can specify the stack through
 
 $$\strut \SPSP  S_{stack}\ \  \triangleq \ \ \TwoStatesN{ s:\prg{Stack},i:\mathbb{N},\prg{v}:\prg{Value}} 
 {\ \inside{\prg{s.modifier}} \ \wedge \ Contents(\prg{s},i,\prg{v})\  }$$

 In that work, they provide a tailor-made proof that indeed, when the stack makes an external call, passing only the inspect-capability, the contents will not change. 
 However, because the language is essentially functional, they do not consider the possibility that the external call might already have stored the modifier capability.
 Moreover, the proof does not make use of a Hoare logic.  
 
 \subsection{Sealer-Unsealer pattern} 
 The sealer-unsealer pattern, proposed by \citet{JamesMorris}, is a security  pattern  to enforce data
abstraction while interoperating with untrusted  code. He proposes a function
\prg{makeseal} which generating pairs of functions (\prg{seal}, \prg{unseal} ), such that \prg{seal} takes a value $v$ and returns a low-integrity value $v'$.
The function \prg{unseal} when given $v'$ will return $v$. But there is no other way to obtain $v$ out of $v'$ except throughthe use of the \prg{usealer}.
Thus, $v'$ can securely be shared with untrusted code.
This pattern has been studied by \citet{ddd}.

We formulate this pattern here. 
As we are working with an object oriented rather than a functional language, we assume the existence of a class 
\prg{DynamicSealer} with two methods, \prg{seal}, and \prg{unseal}. 
And we define a predicate $Sealed(v,v',us)$ to express that $v$ has been sealed into $v'$ and can be unsealed using $us$.

Then, the scoped invariants 

 $$\strut \SPSP  S_{sealer\_1}\ \  \triangleq \ \ \TwoStatesN{ \prg{v},\prg{v}',\prg{us}: \prg{Object} }
 {\   \inside{\prg{us}} \ \wedge\ Sealed(\prg{v},\prg{v}',\prg{us}) }$$  
 
 $$\strut \SPSP  S_{sealer\_2}\ \  \triangleq \ \ \TwoStatesN{ \prg{v},\prg{v}',\prg{us}: \prg{Object} }
 {\ \inside{\prg{v}} \ \wedge \ \inside{\prg{us}} \ \wedge\ Sealed(\prg{v},\prg{v}',\prg{us}) }$$  

\noindent
 expresses that the unsealer is not leaked to external code ($S_{sealer\_1}$), and that if the external world has no access to the high-integrity value $\prg{v}$ nor to the its unsealer  \prg{us}, then it will not get access to the value  ($S_{sealer\_2}$). 
 
 \subsection{Access Security for Hotel Electronic Keys}

A hotel giving access to rooms through cards with electronic keys is a case study in  \cite{AccessHL}. 
This paper\footnote{at the time of writing not yet officially published, but  fresh in arXive}, 
argue that security should be understood in terms of necessary (rather than sufficient) pre-conditions -- this is as we already proposed in the FASE and our OOPSLA 22 paper.
In contrast to our OOPSLA'22 and the current paper, rather than fall back on a the negation of sufficient conditions, t\cite{AccessHL}  propose  a Hoare logic for necessary conditions.
They propose what the call \emph{accessibility tuples}, of the form\\
$ \strut \hspace{4cm}  <\,P \, >\, C \, <\, Q \, >$\\
which express that for any states $\sigma$,$\sigma'$, if execution of command $C$ in state $\sigma$  leads to   state $\sigma'$, and if  $\sigma'$ satisfies $Q$,  then $\sigma$ must satisfy $P$. Moreover, in contrast to our work, the \cite{AccessHL} accessibility tuples are in terms of two conditions and a command, whereas our tuples are only in terms of two conditions -- we quantify universally over the commands.
 
 The Access Security for Hotel Electronic Keys problem as proposed in \cite{AccessHL},  in a hotel that uses cards to control access to its room, the room's door has a battery-powered lock which has an electronic door key, \prg{dk}. A card holds two electronic keys, \prg{ck1} and \prg{ck2}. Such a card can be used to update the lock's key: if  \prg{ck1} is the same as  \prg{dk}, then  \prg{dk} is set to  \prg{ck2}. From then on, \prg{ck2} may be used to open the door.
 A possible implementation in our toy language looks, essentially, as follows
 
 \begin{lstlisting}[language = Chainmail, mathescape=true, frame=lines]
module Hotel
   class Card
      field ck1, ck2: Object
      private init(k1, k2: Object)
          this.ck1:= k1;
          this.ck2:=k2
	
    class Lock
      field dk: Object
      field is_open: bool
      private init(k1, k2: Object)
          this.ck1:= k1;
          this.ck2:=k2
                
      public open(card:Card) : void
          if (this.dk == card.ck1)
              this.dk := card.ck2
              is_open := true
          else 
              is_open := (this.dk == card.ck2)
 \end{lstlisting}
  
 The code for the electronic keys given in \cite{AccessHL} is essentially a snippet -- no modules, classes, or methods.  The paper does not give a specification for the code either, but we assume that the specification would be, roughly
\\
$~ \strut \ \ \ \  (*) \ \ <  \prg{c.ck1}== \prg{l.dk} \vee \prg{c.ck2}== \prg{l.dk}\ > \ \prg{l.open(c)}\ < \ \prg{l.is\_open}\ >$
\\
where $\prg{c}$ is a \prg{Card} and  $\prg{l}$ is a \prg{Lock}. Notice, that in this work, the specification is over one method, nothing prevents another method in the nodule -- or a sequence of methods for that matter -- to break the necessity requirement that the door cannot be opened unless the card had access to the two keys. 

In contrast, our specification is holistic, and would  not be concerned over whether the opening is effected through a call to a method \prg{open}, or any other method from that module. However, the spec does assume that opening can only be effected through possession of a card with the appropriate electronic keys,

$~  (**) \ \  \TwoStatesN{ \prg {cd}:\prg{Card}, \  \prg {lk}:\prg{Lock}}
{  \  \   \neg  \prg{lk.isOpen} \ \wedge\ ( \prg{lk.dk}\in\{ \prg{cd}.\prg{ck1} , \prg{cd}.\prg{ck2} \}  \ 
  \ \rightarrow \  \inside{\prg {cd}}\ ) \ \  }$

 The implementation also satisfies the specification $(*\!*\!*)$ below,  which says that regardless of the state of the lock, no card containing either of the the lock's keys may be leaked:
 
  $~  (*\!*\!*) \ \  \TwoStatesN{ \prg {cd}:\prg{Card}, \  \prg {lk}:\prg{Lock}}
{  \  \    \prg{lk.dk}\in\{ \prg{cd}.\prg{ck1} , \prg{cd}.\prg{ck2} \}  \ 
  \ \rightarrow \  \inside{\prg {cd}}\  \  }$
  
  Note, that while  $(*\!*\!*)$  and  $(\!*)$  share the part about $( \prg{lk.dk}\in\{ \prg{cd}.\prg{ck1} , \prg{cd}.\prg{ck2} \}  \ 
  \ \rightarrow \  \inside{\prg {cd}}\ )$ neither of the two is stronger than the other.
  
\subsection{Access through any element out of a set of  Electronic Keys}

Another case study proposed in  \cite{AccessHL}, is that access should only be granted through possession of any out of a set of keys.

 \begin{lstlisting}[language = Chainmail, mathescape=true, frame=lines]
module Vault
      field keys: List< Object >
      field isOpen: boolean
      
      private init(keys:List<Object>)
          ...
      
      public open(key:Object) : bool
      	  keysList := keys
          while keysList.notEmpty() do
            	  firstKey := keysList.top();
          	  if (firstKey == this.key) then 
          	     isOpen:= true 
	              return true
         	  keysList.pop();
          return false                   
 \end{lstlisting}
 
The paper  \cite{AccessHL} is concerned with the verification of the body of \prg{open}, and in particular the \prg{while}-loop,. And so, they could prove that\footnote{To be precise, their proof  is only  for the body of \prg{open} because their logic lacks support function calls. However, that lack is not a severe limitation; function calls can be easily added.} a call to \prg{v.open(k)} returns \prg{true} only if \prg{k} is one of the keys of \prg{v}:
\\
$~ \strut \ \ \ \  (*) \ \ <\  \prg{v.keys.contains(k)} \ > \ \prg{v.open(k)}\ < \ \prg{ret}==\prg{true}\  >$
\\
where \prg{v} is a \prg{Vault} and \prg{contains} is assumed to be a method in \prg{List} checking whether its argument is in the list.

With our approach, we would use a standard, sufficient condition Hoare  logic to obtain that
$~ \strut \ \ \ \  (*\!*) \ \ \{ \neg \prg{v.keys.contains(k)} \wedge \neg \prg{v.isOpen} \ \}$\\
$~ \strut \ \ \ \  \ \ \ \ \ \ \ \ \ \ \ \ \ \ \prg{l.open(k)}\ $\\
$~ \strut \ \ \ \ \ \ \  \ \ \ \ \  \{ \ \neg \prg{v.keys.contains(k)} \wedge \neg \prg{v.isOpen} \ \}$
\\
The triple from $(*\!*)$ together with the fact that \prg{open} is the only public method would allow us to prove that if the vault is not open, and none of its keys are externally accessible, then it remains not open, and its keys remain externally in-accessible, \ie \\
$~ \strut \ \ \ \  (*\!*\!*) \ \   \TwoStatesN{ \prg {v}:\prg{Vault} }
{  \  \   \neg  \prg{v.isOpen} \ \wedge\ \forall \prg{k}:\prg{Object}.[  \ \prg{v.keys.contains(k)} \ \rightarrow \  \inside{\prg {k} } \ ] \ \  }$

 The implementation also satisfies   spec, $(*\!*\!*\!*)$ below,  which says that regardless of the state of the vault, none of its keys are leaked:
\\
 $~ \strut \ \ \ \  (*\!*\!*\!*) \ \   \TwoStatesN{ \prg {v}:\prg{Vault} }
{  \  \    \forall \prg{k}:\prg{Object}.[  \ \prg{v.keys.contains(k)} \ \rightarrow \  \inside{\prg {k} } \ ] \ \  }$
  
 %
%
%
%
%
%


\clearpage

\section{Appendix to Section \ref{sect:proofSystem} } \label{app:proof}
\label{app:hoare}

\subsection{Preliminaries: Specification Lookup,  Renamings, Underlying Hoare Logic}

Definition \ref{d:promises} is broken down as follows:  $S_1 \txtin  S_2$ says that $S_1$ is textually included in $S_2$; \ \ $S \thicksim S'$ says that $S$ is a safe renaming of $S'$; \ \   $\promises M S$ says that $S$ is a safe renaming of one of the specifications given for $M$. 
 
In particular, a safe renaming of  ${ \TwoStatesN {\overline {x:C}} {A}  }$ can replace any of the variables $\overline x$.  
A safe renaming  of ${\mprepostN{A_1}{p\ D}{m}{y}{D}{A_2} {A_3}}$  can replace  the formal parameters ($\overline y$) by actual parameters  ($\overline {y'}$) but requires the actual parameters  not to include \prg{this}, or \prg{res}, (\ie $\prg{this}, \prg{res}\notin \overline{y'}$). -- 
Moreover, it can replace  the free variables which do not overlap with the formal parameters or the receiver ( $\overline{x}=\fv(A_1)\setminus\{{\overline y},\prg{this}\}$).

\begin{definition}
For a module $M$ and a specification $S$, we define:
\label{d:promises}
\begin{itemize}
\item
$S_1 \txtin  S_2  \ \ \ \triangleq\ \ \   S_1 \txteq  S_2$, or  $S_2 \txteq  S_1 \wedge S_3$, or $S_2\txteq S_3 \wedge S_1$,  or   $S_2 \txteq S_3 \wedge  S_1 \wedge S_4$ for some $S_3$, $S_4$.
\item
$S  \thicksim  S'$ \ \ \  is defined by cases

\begin{itemize}
\item
$ { \TwoStatesN {\overline {x:C}} {A}  }   \thicksim  { \TwoStatesN {\overline {x':C}} {A'[\overline{x'/x}]} } $
\item
$ {\mprepostN{A_1}{p\ D}{m}{y}{D}{A_2} {A_3}}  \thicksim
 {\mprepostN{A_1'}{p\ D}{m}{y'}{D}{A_2'} {A_3'}} $
  \\
 $\strut \hspace{2cm}    \ \  \triangleq\ \ \  A_1 = A_1'[\overline{y/y'}][\overline{x/x'}], \ \ A_2 = A_2'[\overline{y/y'}][\overline{x/x'}], \ \ A_3 = A_3'[\overline{y/y'}][\overline{x/x'}], \ \ \wedge$\\
 $\strut \hspace{2cm}\ \ \ \ \ \ \ \ \ \prg{this}, \prg{res}\notin \overline{y'}, \ \ \overline{x}=\fv(A_1)\setminus\{{\overline y},\prg{this} \} $
  \end{itemize} 
  
 \item  
 $\promises M  S \ \ \ \triangleq
 \ \ \ \exists S'.[ \ \ S'  \txtin \SpecOf M\ \ \wedge\ \ S' \thicksim S \ \ ]$
  \end{itemize} 
  \end{definition} 
  
The restriction on renamings of method specifications that  the actual parameters should not to include \prg{this}  or \prg{res} 
 is necessary because \prg{this} and \prg{res}  denote different objects from the point of the caller  than from the point of the callee.
It means that we are not able to verify a method call whose actual parameters include \prg{this} or \prg{res}. 
This is not a serious restriction: we can encode any such method call by preceding it with assignments to fresh local variables, \prg{this'}:=\prg{this}, and   \prg{res'}:=\prg{res}, and using \prg{this'} and \prg{res'} in the call.

\begin{example}
\label{e:rename}
The specification  from  Example \ref{example:mprepostl} can be renamed as 

\small
{
   {\sprepost
		{\strut \ \ \ \ \ \ \ \ \ S_{9r}} 
		{  a1:\prg{Account}, a2: \prg{Account}\wedge  \inside{a1}\wedge  \inside{a2.\prg{key}} }
		{\prg{public Account}} {\prg{set}} {\prg{nKey}:\prg{Key}}
		{   \inside{a1}\wedge  \inside{a2.\prg{key}} }
		{{   \inside{a1}\wedge  \inside{a2.\prg{key}} }}		
}}

\end{example}

\begin{axiom}
\label{ax:ul}
{Assume   Hoare logic with judgements 
\ $M \vdash_{ul} \{ A \} stmt \{ A' \}$, 
with  $\Stable{A}$,  $\Stable{A'}$. }
\end{axiom}

\subsection{Types}
\label{types}

The rules in Fig. \ref{f:types} allow triples to talk about the types 
Rule {\sc{types-1}} promises that types of local variables do not change.
Rule {\sc{types-2}} generalizes {\sc{types-1}} to any statement, provided that  there already exists a triple for that statement.

\begin{figure}[tht]
$
\begin{array}{c}
 \begin{array}{lcl}
\inferrule[\sc{types-1}]
	{  stmt \ \mbox{contains no method call} \\
	stmt  \ \mbox{contains   no assignment to $x$}}
	{\hproves{M}  {x:C} {\ stmt\ }{x:C} }
\\
\\
\inferrule[\sc{types-2}]
	{ \hprovesN{M}  {A} {\ s\ }  {A'} {A''}  }
	{\hprovesN{M}  {x:C \wedge A} {\ s\ }  {x:C\wedge A'} {A''}}
\end{array}
\end{array}
 $
\caption{Types}
\label{f:types}
\end{figure}

In {\sc{types-1}} we restricted to statements which do not contain method calls  in order to make lemma   \ref{l:no:meth:calls}  valid.

\subsection{Second Phase - more}

in Fig. \ref{f:substructural:app}, we    extend the Hoare Quadruples Logic with substructural rules, rules for conditionals, case analysis, and a contradiction rule.
For the conditionals we assume the obvious operational. semantics, but do not define it in this paper

\begin{figure}[htb]
$
\begin{array}{c}
\begin{array}{lcl}
\inferruleSD{[\sc{combine}]}
	{  \begin{array}{l}
	\hprovesN{M}  {A_1} {\ s\ } {A_2}  {A} \\ 
	\hprovesN{M}  {A_3} {\ s\ } {A_4} {A}
	\end{array}
	}
	{ \hprovesN{M}  {A_1 \wedge A_3 }{\ s\ } {A_2 \wedge A_4} {A} }
& &
\inferruleSD{[\sc{sequ}]}
	{  \begin{array}{l} 
	\hprovesN{M}  {A_1} {\ s_1\ } {A_2}  {A}  \\ 
	\hprovesN{M}  {A_2} {\ s_2\ } {A_3} {A}
	\end{array}
	}
	{   \hprovesN{M}  {A_1   }{\ s_1; \, s_2\ } {  A_3} {A} }
\end{array}
\\ \\
\inferruleSD{ \hspace{3cm} [\sc{consequ}]}
	{
	 { \hprovesN  {M}  {A_4} {\, s\, } { A_5} {A_6}  }
	 \hspace{1.4cm} 
	 M \vdash A_1 \rightarrow A_4 
	 \hspace{1.4cm} 
	{ M \vdash A_5   \rightarrow  A_2  }
	 \hspace{1.4cm}   
	{  M \vdash A_6 \rightarrow A_3 }
	}
	{   \hprovesN{M}  {A_1 }{\ s\ } {A_2} {A_3} }
  \end{array}
 $
 
 $
\begin{array}{c}
\inferruleSD{\hspace{2.5cm} [\sc{If\_Rule}]}
	{
	 \begin{array}{c}
	  \hprovesN {M}   
		{\  A \wedge Cond \  }
		{\    stmt_1   \ }
 		{\ A' \ }
		{\ A'' \ }
	\\
	    \hprovesN {M}   
		{\  A \wedge \neg Cond \  }
		{\    stmt_2   \ }
 		{\ A' \ }
		{\ A'' \ }	
	\end{array}
	}	
 	{  	
	\hprovesN {M}   
		{\  A \  }
		{\  \prg{if}\ Cond\ \prg{then}\ stmt_1\ \prg{else}\ stmt_2 \ \ }
		{\ A' \ }
		{\ A'' \ }
}
\\
\\
\begin{array}{lcl}
{
\inferruleSD{\hspace{0.5cm} [\sc{Absurd}]}
	{	
	}	 
 	{  	
	\hprovesN {M}   
		{\  false \  }
		{\  \ stmt \ \ }
		{\ A' \ }
		{\ A'' \ }
}
} & &
{
\inferruleSD{\hspace{0.5cm} [\sc{Cases}]}
	{ \begin{array}{l}
	\hprovesN {M}   
		{\  A \wedge A_{1}  \  }
		{\  \ stmt \ \ }
		{\ A' \ }
		{\ A'' \ }
		\\
		\hprovesN {M}   
		{\   A \wedge A_{2} \  }
		{\  \ stmt \ \ }
		{\ A' \ }
		{\ A'' \ }
	\end{array}	
	}	 
 	{  	
	\hprovesN {M}   
		{\  A \wedge (A_1 \vee A_2) \  }
		{\  \ stmt \ \ }
		{\ A' \ }
		{\ A'' \ }
}
}
\end{array}
\end{array}
$
 \vspace{-.5cm}
\caption{Hoare Quadruples -    substructural rules, and conditionals }
\label{f:substructural:app}
\end{figure}

\subsection{Extend the semantics and Hoare logic to accommodate scalars and conditionals}
\label{s:app:scalars}

{We extend the notion of protection to also allow it to apply to scalars. }

\begin{definition}[Satisfaction  of Assertions  -- Protected From]
\label{def:chainmail-protection-from-ext}
extending the definition of Def 
\ref{def:chainmail-protection-from}. We use $\alpha$ to range over addresses, $\beta$  to range over scalars, and   $\gamma$ to range over addresses or scalars.

\noindent
We define  $\satisfiesA{M}{\sigma}{\protectedFrom{{\gamma}} {{\gamma_{o}}}} $ as:
\begin{enumerate}
\item
\label{cProtectedNew}
 $\satisfiesA{M}{\sigma}{\protectedFrom{{\alpha}} {{\alpha_{o}}}}   \ \ \ \triangleq $ 
  \begin{itemize}
 \item
$\alpha\neq \alpha_0$,
 \ \ \ \  and
 \item
$\forall n\in\mathbb{N}. \forall f_1,...f_n..
[\ \ \interpret{\sigma}{\alpha_{o}.f_1...f_n}=\alpha \ \ \  \Longrightarrow \ \ \  \satisfiesA{M}{\sigma}{ {\interpret{\sigma}{\alpha_{o}.f_1...f_{n-1}}}:C} \ \wedge \ C\in M\ \ ]$
\end{itemize}
\item
 $\satisfiesA{M}{\sigma}{\protectedFrom{{\gamma}} {{\beta_{o}}}}   \ \ \ \triangleq  \ \ \ true$
 \item
 $\satisfiesA{M}{\sigma}{\protectedFrom{{\beta}} {{\alpha_{o}}}}   \ \ \ \triangleq  \ \ \ false$
  \item
$\satisfiesA{M}{\sigma}{\protectedFrom{{\re}} {{\re_{o}}}} \ \ \ \triangleq $ \\
  $\exists \gamma, \gamma_{o}. [\  \ \eval{M}{\sigma}{{\re}}{\gamma}\ \wedge \eval{M}{\sigma}{{\re_0}}{\gamma_0} \  \wedge \ 
  \satisfiesA{M}{\sigma}{\protectedFrom{{\gamma}} {{\gamma_{o}}}}
 \ \  ]$
 \end{enumerate}
 \end{definition}

{The definition from above gives rise to further cases of  protection; we supplement the triples from 
Fig. \ref{f:protection} with some further inference rules, given   in Fig. \ref{f:protection:conseq:ext}.}
Namely, any expression $\re$ is protected from a scalar (rules {\sc{Prot-In}}, {\sc{Prot-Bool}} and {\sc{Prot-Str}}).
Moreover, if starting at some $\re_o$ and following any   sequence of fields $\overline f$ we reach   internal objects or  scalars (\ie never reach an external object), then any $\re$ is protected from $\re_o$ (rule {\sc{Prot\_Intl}}).

\begin{figure}[htb]
\begin{mathpar}
\inferrule
	{M \vdash \re_o : \prg{int} \rightarrow \protectedFrom{\re}{\re_o} }
	{}
	\quad[\textsc{Prot-Int}]
	\and
\inferrule
	{M \vdash \re_o : \prg{bool} \rightarrow \protectedFrom{\re}{\re_o} }
	{}
	\quad[\textsc{Prot-Bool}]
	\and
\inferrule
	{M \vdash \re_o : \prg{str} \rightarrow \protectedFrom{\re}{\re_o} }
	{}
	\quad[\textsc{Prot-Str}]
\and
 {
\inferrule
	{M \vdash  \protectedFrom{\re}{\re_o}  \wedge {\re} \internal\,  \rightarrow \re \neq e'} 
	{}
	\quad[\textsc{Prot-Neq]}
  }
  \and
  {
\inferrule
	{M \vdash  A \rightarrow \forall \overline f.[\ \re_o.\overline f\internal\, \vee\, \re_o.\overline f:\prg{int} \,  \vee\,    \re_o.\overline f:\prg{bool}\,  \vee\,  \re_o.\overline f:\prg{str}\ ] }
	{M \vdash  A \rightarrow \protectedFrom {\re} {\re_o} }
	\quad[\textsc{Prot-Intl]}
  } 
\end{mathpar}
\caption{Protection for Scalar and Internal Types}
\label{f:protection:conseq:ext}
\end{figure}

 %
%

\begin{lemma}
\label{l:no:meth:calls}
If ${\hproves{M}  {A} {\ stmt\ }{A'} }$, then $stmt$ contains no method calls.
\end{lemma}

\begin{proof}
By induction on the rules in Fig. \ref{f:underly}.

\end{proof}

\subsection{Adaptation}
\label{appendix:adaptation}
 
 \newcommand{\SP}{$\strut \ \ \ \ $}

 We now discuss the proof of Lemma \ref{lemma:push:ass:state}.

 \vspace{0.5cm}
 
 \beginProofSub{lemma:push:ass:state}{l:push:stbl}
$~$ \\
To Show: \ \ \  $\Stable{\,  \PushASLong {(y_0,\overline y)} A\, }$
\\
By structural induction on $A$.\\
\completeProofSub

\vspace{1cm}

For parts \ref{lemma:push:ass:state:one},  \ref{lemma:push:ass:state:two}, and  \ref{lemma:push:ass:state:three}, we first prove the following auxiliary lemma:

\begin{auxLemma}
\label{l:push:pop:aux}
For all $\alpha$,   $\overline {\phi_1}$, $\overline {\phi_2}$, $\overline {\phi_2}$, $\phi$ and $\chi$\\
$\strut ~ \ \ \ \ \ (L1)\ \ \    M, (\overline {\phi_1},\chi) \models \protectedFrom \alpha {Rng(\phi)} \ \Longrightarrow \ M, (\overline {\phi_2}\cdot \phi,\chi) \models \inside \alpha$
\\
$\strut ~ \ \ \ \ \ (L2)\ \ \    M, (\overline {\phi_1}\cdot\phi,\chi) \models \inside \alpha   \wedge \extThis \ \ \Longrightarrow \ \ M, (\overline {\phi_2},\chi) \models \protectedFrom \alpha {Rng(\phi)} $\\
$\strut ~ \ \ \ \ \ (L3)\ \ \    M, (\overline {\phi_1}\cdot \phi_1,\chi) \models \inside \alpha   \wedge \extThis \ \ \wedge Rng(\phi)\subseteq Rng(\phi_1)\ \ \  \Longrightarrow \ \ M, (\overline {\phi_2},\chi) \models \protectedFrom \alpha {Rng(\phi)} $
\\\end{auxLemma}

\begin{proof}
$~$ \\
We first prove (L1): \\
~ \\
We define $\sigma_1 \triangleq (\overline {\phi_1},\chi)$, and  $\sigma_2 \triangleq (\overline {\phi_2}\cdot \phi,\chi) $.\\
The above definitions imply that: \\
\SP (1)\ \ $\forall \alpha',\forall \overline f.[\  \interpret {\sigma_1} {\alpha'.\overline f} =  \interpret {\sigma_2} {\alpha'.\overline f}\ ]$\\
\SP (2)\ \ $\forall \alpha'.[\  \Relevant {\alpha'} {\sigma_1} = \Relevant {\alpha'} {\sigma_2}\ ]$\\
\SP (3)\ \ $\LRelevantO {\sigma_2} = \bigcup_{\alpha'\in Rng(\phi)} \Relevant {\alpha'} {\sigma_2} $.\\
We now assume that\\
\SP (4)\ \ $M, \sigma_1 \models \protectedFrom \alpha {Rng(\phi)}$.\\
and want to show that\\
\SP (A?)\ \ $M, \sigma_2 \models \inside \alpha$\\
From (4) and  by definitions, we obtain that\\
\SP (5)\ \ $\forall \alpha'\in Rng(\phi).\forall \alpha''\in \Relevant {\alpha'} {\sigma_1}.\forall f.[ \   M, \sigma_1 \models \alpha'':\prg{extl}\ \rightarrow \alpha''.f  \neq \alpha\ ]$, \ \ \ \ and also\\
\SP (6)\ \ $\alpha \notin Rng(\phi)$\\
From (5) and (3) we obtain:\\
\SP (7)\ \  $\forall \alpha' \in \LRelevantO {\sigma_2}.\forall f.[ \   M, \sigma_1 \models \alpha':\prg{extl}\ \rightarrow \alpha'.f  \neq \alpha\ ]$\\
From (7) and (1) and (2) we obtain:\\
\SP (8) \ \  $\forall \alpha' \in \LRelevantO {\sigma_2}.\forall f.[ \   M, \sigma_2 \models \alpha':\prg{extl}\ \rightarrow \alpha'.f  \neq \alpha\ ]$\\
From (8), by definitions, we obtain\\
 \SP (10)\ \ $M, \sigma_2 \models \inside \alpha$\\
which is (A?).\\
 This completes the proof of (L1).
 \\
  $\strut ~ \ $\\
  We now prove (L2): \\
  ~ \\
 We define $\sigma_1 \triangleq (\overline {\phi_1}\cdot \phi,\chi)$, and  $\sigma_2 \triangleq (\overline {\phi_2},\chi) $.\\
The above definitions imply that: \\
\SP (1)\ \ $\forall \alpha',\forall \overline f.[\  \interpret {\sigma_1} {\alpha'.\overline f} =  \interpret {\sigma_2} {\alpha'.\overline f}\ ]$\\
\SP (2)\ \ $\forall \alpha'.[\  \Relevant {\alpha'} {\sigma_1} = \Relevant {\alpha'} {\sigma_2}\ ]$\\
\SP (3)\ \ $\LRelevantO {\sigma_1} = \bigcup_{\alpha'\in Rng(\phi)} \Relevant {\alpha'} {\sigma_1}$.\\
We   assume that\\
\SP (4)\ \  $M, \sigma_1 \models \inside \alpha \wedge \extThis$.\\
and want to show that\\
\SP (A?)\ \ $M, \sigma_2 \models \PushASLong  {Rng(\phi)} {A}$.\\
From (4), and unfolding the definitions, we obtain:\\
\SP (5)\ \  $\forall \alpha'\in \LRelevantO {\sigma_1}.\forall f.[ \   M, \sigma_1 \models \alpha':\prg{extl}\ \rightarrow \alpha'.f  \neq \alpha\ ]$, \ \ \ and\\
\SP (6)\ \ $\forall \alpha'\in Rng (\phi). [ \ \alpha'\neq \alpha \ ]$.\\
From(5), and using (3) and (2) we obtain:
\\
\SP (7)\ \  $\forall \alpha'\in Rng(\phi).\forall \alpha'' \in\Relevant {\alpha'} {\sigma_2}.\forall f.[ \   M, \sigma_2 \models \alpha'':\prg{extl}\ \rightarrow \alpha''.f  \neq \alpha\ ]$\\
From (5) and (7) and by definitions, we obtain
\\
\SP (8)\ \  $\forall \alpha'\in Rng (\phi).[ \   \models \alpha \protectedFrom \alpha {\alpha'}\ ]$.\\
From (8) and definitions we obtain (A?).\\
This completes the proof of (L2). 
 \\
  $\strut ~ \ $\\
  We now prove (L3): \\
  ~ \\
 We define $\sigma_1 \triangleq (\overline {\phi_1}\cdot \phi_1,\chi)$, and  $\sigma_2 \triangleq (\overline {\phi_2},\chi) $.\\
The above definitions imply that: \\
\SP (1)\ \ $\forall \alpha',\forall \overline f.[\  \interpret {\sigma_1} {\alpha'.\overline f} =  \interpret {\sigma_2} {\alpha'.\overline f}\ ]$\\
\SP (2)\ \ $\forall \alpha'.[\  \Relevant {\alpha'} {\sigma_1} = \Relevant {\alpha'} {\sigma_2}\ ]$\\
\SP (3)\ \ $\LRelevantO {\sigma_1} = \bigcup_{\alpha'\in Rng(\phi_1)} \Relevant {\alpha'} {\sigma_1}$.\\
We   assume that\\
\SP (4a)\ \  $M, \sigma_1 \models \inside \alpha \wedge \extThis$, and
\SP (4b)\ \ $Rng(\phi) \subseteq Rng(\phi_1)$\\
We  want to show that\\
\SP (A?)\ \ $M, \sigma_2 \models \PushASLong  {Rng(\phi)} {A}$.\\
From (4a), and unfolding the definitions, we obtain:\\
\SP (5)\ \  $\forall \alpha'\in \LRelevantO {\sigma_1}.\forall f.[ \   M, \sigma_1 \models \alpha':\prg{extl}\ \rightarrow \alpha'.f  \neq \alpha\ ]$, \ \ \ and\\
\SP (6)\ \ $\forall \alpha'\in Rng (\phi_1). [ \ \alpha'\neq \alpha \ ]$.\\
From(5), and   (3) and (2) and (4b) we obtain:
\\
\SP (7)\ \  $\forall \alpha'\in Rng(\phi).\forall \alpha'' \in\Relevant {\alpha'} {\sigma_2}.\forall f.[ \   M, \sigma_2 \models \alpha'':\prg{extl}\ \rightarrow \alpha''.f  \neq \alpha\ ]$ \\
From(6), and   (4b) we obtain:
\\
\SP (8)\ \ $\forall \alpha'\in Rng (\phi_1). [ \ \alpha'\neq \alpha \ ]$.\\
From (8) and definitions we obtain (A?).\\
This completes the proof of (L3). 

\end{proof}

\beginProofSub{lemma:push:ass:state}{lemma:push:ass:state:one}
$~$ \\
To Show: \ \ \  $(*)\ \ \ M, \sigma \models \PushASLong  {Rng(\phi)} {A}\  \ \ \ \ \  \ \ \    \Longrightarrow  \ \ \ \ M,  \PushSLong {\phi} {\sigma}   \models A$
\\ $~$ \\
 By  induction on the structure of $A$. For the case where $A$ has the form $\inside {\alpha.\overline f}$, we use lemma \ref{l:push:pop:aux},(L1), taking $\overline {\phi_1} = \overline { \phi_2}$, and $\sigma \triangleq (\overline {\phi_1},\chi).$
\\
\completeProofSub

\vspace{1cm}

\beginProofSub{lemma:push:ass:state}{lemma:push:ass:state:two}
$~$ \\
To Show \ \ \  $(*)\ \ \  M,  \PushSLong {\phi} {\sigma}   \models  A  \wedge \extThis    \ \  \ \  \Longrightarrow  \ \ \ \ M, \sigma \models \PushASLong  {Rng(\phi)} {A}$ 
\\
$~$ \\
We apply induction on the structure of $A$. For the case where $A$ has the form $\inside {\alpha.\overline f}$, we apply lemma \ref{l:push:pop:aux},(L2), using    $\overline {\phi_1} = \overline { \phi_2}$, and $\sigma \triangleq (\overline {\phi_1},\chi).$

\completeProofSub

\vspace{1cm}
\beginProofSub{lemma:push:ass:state}{lemma:push:ass:state:three}
$~$ \\
To Show:\ \ \   (*) \ \  $M, \sigma  \models  A  \wedge \extThis  \ \wedge \ M\cdot\Mtwo \models \PushSLong {\phi} {\sigma}   \ \  \ \ \  \  \Longrightarrow  \ \ \ \ M, \PushSLong {\phi} {\sigma} \models \PushASLong  {Rng(\phi)} {A}$
\\ 
$~$ \\
By induction on the structure of $A$. 
 For the case where $A$ has the form $\inside {\alpha.\overline f}$, we want to apply lemma \ref{l:push:pop:aux},(L3). We take  $\sigma$ to be $ (\overline {\phi_1}\cdot\phi_1, \chi)$, and $\overline {\phi_2}=\overline {\phi_1}\cdot\phi_1\cdot \phi$. Moreover,  $M\cdot\Mtwo \models \PushSLong {\phi} {\sigma}$ gives  that $Rng(\phi)\subseteq \LRelevantO {\sigma_2}$. Therefore, (*) follows by application of lemma \ref{l:push:pop:aux},(L3).\\
\completeProofSub

\clearpage

\section{Appendix to Section \ref{sect:sound:proofSystem} -- Soundness of the Hoare Logics}

\subsection{Expectations}
\label{s:expectations}

\begin{axiom}
\label{lemma:axiom:enc:assert:ul}
\label{ax:ul:sound}
We require a sound logic of assertions ($M \vdash A$), and a sound Hoare logic , \ie that for all $M$, $A$, $A'$, $stmt$:
\begin{center}
$M \vdash A   \ \ \ \  \Longrightarrow  \ \ \ \  \forall \sigma.[\ M, \sigma \models A\ ]$.\\
%
%
{$M\ \vdash_{ul}\  \triple A {stmt} {A'}  \ \ \ \  \Longrightarrow  \ \ \ \ \satisfies  {M} { \triple A {stmt} {A'}}$ }
 \end{center}
\end{axiom}

The expectation that $M\ \vdash_{ul}\  \triple A {stmt} {A'} $ is sound is not onerous: 
since the assertions $A$ and $A'$ do not talk about protection, many Hoare logics from the literature could be taken.

On the other hand, in the logic  $M \vdash A$ we want to allow the assertion $A$   to talk about protection. 
Since protection is a novel concept, the   literature offers no such logics.
Nevertheless, such a logic can be constructed by extending and underlying assertion logic $M \vdash_{ul} A$  which does not talk about protection.
We show such an extension in Fig \ref{assert:logic:extend}.

\begin{figure}[htb]
$
\begin{array}{c}
\inferruleSD{[\sc{Ext-1}]}
	{  M \vdash_{ul} A  	}
	{ M \vdash_{ul} A   }
\\ \\
\begin{array}{lcl}
\inferruleSD{[\sc{Ext-2}]}
	{  
	M \vdash A \rightarrow A'
	 
	}
	{   M \vdash (A \wedge \protectedFrom {\re} {\re'}) \rightarrow (A' \wedge \protectedFrom {\re} {\re'})   }
& & 
\inferruleSD{[\sc{Ext-3}]}
	{   
	M \vdash A \rightarrow A'
	 
	}
	{   M \vdash (A \wedge \inside {\re} {) \rightarrow (A' \wedge \inside{\re}})   }
\end{array}
\\ \\
\inferruleSD{[\sc{Ext-4}]}
	{   
	\
	}
	{   M \vdash {((A_1 \vee A_2) \wedge \protectedFrom {\re} {\re'})) \leftrightarrow 
	                     ((A_1  \wedge  \protectedFrom {\re} {\re'})  \vee (A_2  \wedge \protectedFrom {\re} {\re'}))  }}
\\ \\
\inferruleSD{[\sc{Ext-5}]}
	{   \ 
	}
	{   M \vdash {((A_1 \vee A_2) \wedge \inside{\re}) \leftrightarrow 
	                     ((A_1  \wedge  \inside{\re})  \vee (A_2  \wedge  \inside{\re}))  }}
\end{array}
$
\caption{From $M \vdash_{ul} A$ to $M \vdash  A$}
\label{assert:logic:extend}
\end{figure}

\noindent
The extension shown in  in Fig.  \ref{assert:logic:extend}  preserves soundness of the logic:

\begin{lemma}
Assume a logic $\vdash_{ul}$, such that 
\begin{center}
$M \vdash_{ul} A   \ \ \ \  \Longrightarrow  \ \ \ \  \forall \sigma.[\ M, \sigma \models A\ ]$.\\
 \end{center}

\noindent
Extend this logic according to the rules in Fig.  \ref{assert:logic:extend} and in Fig \ref{f:protection:conseq:ext}, and obtain $M \vdash  A$. Then, we have:
\begin{center}
$M \vdash A   \ \ \ \  \Longrightarrow  \ \ \ \  \forall \sigma.[\ M, \sigma \models A\ ]$.\\
 \end{center}

\end{lemma}

\begin{proof}
By induction over the derivation that $M \vdash A$.
\end{proof}

Note that the rules in  in Fig.  \ref{assert:logic:extend} allow the derivation  of $M\vdash A$, 
for which $\Pos A$ does not hold -- \eg we can derive  $M \vdash \inside {\re}  \rightarrow  \inside{\re}$ through application of rule {\sc{Ext-3}}.
However, this does not affect soundness of our logic --    $\Pos {\_}$ is required only in specifications.

\subsection{\Strong satisfaction of assertions}
\label{s:scoped:mean}

\begin{definition}
\label{def:restrict}
For a state $\sigma$, and a number $i\in \mathbb{N}$ with $i \leq \DepthSt {\sigma}$,   module $M$, and assertions $A$, $A'$ we define: 
\begin{itemize}
\item
$  \satDAssertFrom M  \sigma k   A  \ \  \ \triangleq \  \ \  
  k\leq  \DepthSt {\sigma} \ \wedge \  \forall i\!\in\![k...\DepthSt {\sigma}].[\ M,{\RestictTo {\sigma}{i}} \models A[\overline{ {{\interpret \sigma z}/ z}}]\ ] \ \  \mbox{where} \ \
  \overline z=\fv(A).$ 
\end{itemize}
\end{definition}
 
 Remember the definition of  $\RestictTo  \sigma k$, which returns a new state whose top frame is the $k$-th frame from $\sigma$. Namely, $\RestictTo {(\phi_1...\phi_i...\phi_n,\chi)} {i}\ \ \ \triangleq \ \ \ (\phi_1...\phi_i,\chi)$

\begin{lemma}
\label{l:shallow:scoped}
For a states $\sigma$, $\sigma'$, numbers $k,k'\in \mathbb{N}$, assertions  $A$, $A'$, frame $\phi$ and variables $\overline z$, $\overline u$:
\begin{enumerate}
\item
$ \satDAssertFrom M  \sigma { \DepthSt \sigma}   A \ \ \Longleftrightarrow\ \ M,\sigma \models A\ $
\item
$ \satDAssertFrom M  \sigma {k} A \ \wedge\  k\leq k'\  \  \   \Longrightarrow \ \ \satDAssertFrom M  \sigma {k'} A$ 
\item 
\label{shallow:to:scoped}
$ M,\sigma \models A \ \wedge\ \Stable A \  \ \Longrightarrow \  \  \forall k  \leq  \DepthSt \sigma.[ \ \satDAssertFrom M  \sigma k   A \ ]$
\item
\label{fourSD}
$ M  \models A \rightarrow A'\  \  \   \Longrightarrow \ \ \forall \sigma. \forall k\leq  { \DepthSt \sigma}.[ \ \satDAssertFrom M  \sigma {k} A
\ \Longrightarrow \  \satDAssertFrom M  \sigma {k} A'\ ]$

\end{enumerate}
\end{lemma}

\noindent
\vspace{.1cm}
{\textbf{Proof Sketch}} 

\begin{enumerate}
\item
By unfolding and folding the definitions.
\item
By unfolding and folding the definitions.
\item
By induction on the definition of $\Stable {\_}$.
\item
  By contradiction: Assume that there exists a $\sigma$,  and a  $k\leq \DepthSt \sigma$,    such that  \\
$\strut \hspace{2cm}   \satDAssertFrom M  \sigma {k} A$ \ \ \ and \ \ \ $\neg (\satDAssertFrom M  \sigma {k} A')$\\
 The  above implies that \\
$\strut \hspace{2cm} \forall i\geq k.[\ \ M,{\RestictTo {\sigma}{i}} \models A[\overline{ {{\interpret \sigma z}/ z}}]\ \ ]$, \ \ \ and\\
$\strut \hspace{2cm} \exists j\geq k.[\ \  M,{\RestictTo {\sigma}{j}} \not\models A'[ \overline{ {{\interpret \sigma z}/ z}}]\ \ ]$,\\
 where  $\overline z \triangleq \fv(A)\cup \fv(A')$.\\
 Take $\sigma''\triangleq  \RestictTo {\sigma}{j}$. Then we have that\\
$\strut \hspace{2cm} M, \sigma'' \models A[\overline{ {{\interpret \sigma z}/ z}}]$,  and  $M,  \sigma'' \not\models A'[\overline{ {{\interpret \sigma z}/ z}}]$.\\
 This contradicts $ M  \models A \rightarrow A'$.\\
 {\footnoteSD{NOTE TO AUTHORS  proof hinges on the fact that we consider the "restricted" state, $\sigma''$ a "dully-fledged" state, and the fact that we no longer require "Arising".}}
 {\footnoteSD{SD  wondered whether  \ref{l:shallow:scoped}.\ref{four} would still hold if we allowed the assertions to "reflect" on the frame, to say things eg like "this and x are the only local variables". But such assertions would not have the property \ref{l:shallow:scoped}.\ref{fiveSD}}}
\end{enumerate}
\noindent
{\textbf{End Proof Sketch}}

Finally, the following lemma allows us to combine shallow and \Strong satisfaction:

\begin{lemma}
\label{l:shallow:scoped:scoped}
For states  $\sigma$,  $\sigma'$, frame $\phi$ such that $\sigma'=\sigma  \pushSymbol \phi$, and for  
assertion $A$, such that $fv(A)=\emptyset$:
\begin{itemize}
\item
$\satDAssertFrom M  \sigma k A\   \wedge \ M, \sigma' \models A \ \ \ \Longleftrightarrow \ \ \   \satDAssertFrom M  {\sigma'} k  A$ 
\end{itemize}
\end{lemma}

\begin{proof}
By structural induction on $A$, and unfolding/folding the definitions. Using also lemma \ref{l:prtFrom} from later.
\end{proof}

\subsection{Shallow and \Strong Semantics of Hoare tuples}
\label{s:shallow:deep:appendix}
Another example demonstrating that assertions at the end of a method execution might not hold after the call:

\begin{example}[$Stb^+$ not always preserved by Method Return]
\label{ex:motivate:scopes}
Assume state $\sigma_a$, such that $\interpret {\sigma_a} {\prg{this}}=o_1$, $\interpret {\sigma} {\prg{this}.f}=o_2$, $\interpret {\sigma} {x}=o_3$, $\interpret {\sigma} {x.f}=o_2$,  
and $\interpret {\sigma} {x.g}=o_4$, where $o_2$ is external and all other objects are internal. 
We then have $..,\sigma_a \models  \inside {o_4}$.
Assume 
 the continuation of $\sigma_a$   consists of a method $x.m()$. Then,
upon entry to that method, when we push the new frame, we have  state $\sigma_b$, which also satisfies $..,\sigma_b \models  \inside {o_4}$.
Assume 
 the   body of $m$ is $\prg{this}.f.m1(\prg{this}.g); \prg{this}.f := \prg{this};  \prg{this}.g := \prg{this}$, and 
 the external method $m1$ stores in the 
receiver a reference to the argument.
Then, at the end of method execution, and before popping the stack, we have   state $\sigma_c$, which also satisfies $..,\sigma_c \models  \inside {o_4}$.
However, after we pop the stack, we obtain $\sigma_d$, for which $..,\sigma_d \not\models  \inside {o_4}$.
\end{example}

\begin{definition}[\Strong Satisfaction of Quadruples by States]
\label{def:restrict:sat}
For modules $\Mtwo$, $M$, state $\sigma$,  
and assertions $A$, $A'$ and  $A''$
\begin{itemize}
\item
$ {\satDAssertQuadrupleFrom \Mtwo  M  \sigma   {A} {A'} {A''} } \ \ \triangleq \ \ $  \\
$\strut \hspace{.5cm} \forall k, \overline{z}, \sigma',\sigma''.[
  \satDAssertFrom M  \sigma k   {A}  \  
  \ \Longrightarrow$\\
$\strut \hspace{4.5cm}    [ \ {\leadstoBoundedStarFin {M\madd \Mtwo}{\sigma}  {\sigma'} }\  \Rightarrow \    \satDAssertFrom M  {\sigma'} k   {\sdN{A'}}   \ ]
 \ \wedge$\\
$\strut \hspace{4.5cm}    [ \ {\leadstoBoundedStar  {M\madd \Mtwo}{\sigma}  {\sigma''} }\  \Rightarrow\      \satDAssertFrom M  {\sigma''}  k  {(\externalexec \rightarrow A''[\sdN{\overline{\interpret \sigma z/z}}])} \ ] $\\
$\strut \hspace{2.3cm}\ \ \ \ ]  $ \\
$\strut \hspace{2.3cm}\ \ \ \  \mbox{where }  \sdN{ \overline z= \fv(A)}$ 
\end{itemize}
\end{definition}

\begin{lemma} 
For all $M$, $\Mtwo$ $A$, $A'$, $A''$ and $\sigma$:
\begin{itemize}
\item
$ {\satDAssertQuadrupleFrom \Mtwo  M  \sigma   {A} {A'} {A''} } \ \Longrightarrow \ \
  {\satAssertQuadruple  \Mtwo  M   {A}  \sigma  {A'} {A''} } $
\end{itemize}
\end{lemma}

\label{sect:HLmeans}

We  define the {\emph {meaning}} of  our Hoare triples, $\triple {A} {stmt} {A'}$,  in the usual way, \ie that execution of $stmt$ in a state that satisfies $A$ leads to a state which satisfies $A'$.  
In addition to that, Hoare quadruples, $\quadruple {A} {stmt} {A'} {A''}$, promise that any external future states scoped by $\sigma$ will satisfy $A''$.
We give both a weak and a shallow version of the semantics

 \begin{definition}[\Strong Semantics of Hoare triples]
For modules $M$, and assertions $A$, $A'$   we define:
\begin{itemize}
\item
\label{def:hoare:sem:one}
$\satisfies  {M} {  \{\, A \,  \}\ stmt\  \{\, A' \, \} }\ \ \ \triangleq$\\
{$\strut  \ \ \ \forall    \Mtwo. \forall  \sigma.[ \ \   
 \sigma.\prg{cont}\txteq stmt\   \Longrightarrow \ 
{\satAssertQuadruple  \Mtwo  M   {\  A\ }  \sigma  { A'  } {true} } \ \ ]$
}
 \item
 \label{def:hoare:sem:two}
$\satisfies {M} {\quadruple {A} {stmt} {A'} {A''}}  \ \ \  \triangleq$ \\
{$\strut  \ \ \  \forall    \Mtwo. \forall  \sigma.[ \  \  \sigma.\prg{cont}\txteq stmt\   \Longrightarrow \ 
{\satAssertQuadruple  \Mtwo  M    {\  A\ }  \sigma  { A'  } {A''} } \ \ ]$
}
\item
\label{def:hoare:sem:three}
$\satisfiesD {M} {  \{\, A \,  \}\ stmt\  \{\, A' \, \} }\ \ \ \triangleq$\\
{$\strut  \ \ \ \forall    \Mtwo. \forall  \sigma.[ \ \   
 \sigma.\prg{cont}\txteq stmt\   \Longrightarrow \ 
{\satDAssertQuadrupleFrom \Mtwo  M  \sigma   {\  A\ } { A'  } {true} } \ \ ]$
}
 
 \item
 \label{def:hoare:sem:four}
$\satisfiesD {M} {\quadruple {A} {stmt} {A'} {A''}}  \ \ \  \triangleq$ \\
{$\strut  \ \ \  \forall    \Mtwo. \forall  \sigma.[ \ \    \sigma.\prg{cont}\txteq stmt\   \Longrightarrow \ 
{\satDAssertQuadrupleFrom \Mtwo  M  \sigma   {\  A\ } { A'  } {A''} } \ \ ]$
}

\end{itemize}
\end{definition}

 \begin{lemma}[\Strong   vs Shallow Semantics of Quadruples]
For all $M$, $A$, $A'$, and $stmt$:
\begin{itemize}
\item
$\satisfiesD {M} {\quadruple {A} {stmt} {A'} {A''}}   \ \ \ \Longrightarrow \ \ \ \satisfies  {M} {\quadruple {A} {stmt} {A'} {A''}} $
\end{itemize}
\end{lemma}
 \begin{proof}
 By unfolding and folding the definitions
 \end{proof}

\subsection{\Strong satisfaction of specifications} 
\label{sect:HLmeans:scoped}

We now give a \Strong meaning to specifications: 

\begin{definition} [\Strong Semantics of  Specifications]

We define $\satisfiesD{M}{{S}}$ by cases: 

\label{def:necessity-semantics:strong}

\begin{enumerate}
\item
{ 
$\satisfiesD{M}{\TwoStatesN {\overline {x:C}} {A}} \ \  \ \triangleq   \ \ \ 
\forall \sigma.[\  \satisfiesD {M} {\quadruple {\externalexec \wedge \overline {x:C} \wedge A} {\sigma} {A} {A} } \ ] $
}
  \item
 {$\satisfiesD{M} { \mprepostN {A_1}{p\ D}{m}{y}{D}{A_2} {A_3}    }\  \ \ \   \triangleq  \ \ \ $}\\ 
 {
$\strut  \ \ \   \ \ \ \ \ \ \ \ \   \    \forall   y_0,\overline y, \sigma[ \ \ \ \sigma\prg{cont}\txteq {u:=y_0.m(y_1,..y_n)} \ \ \Longrightarrow \ \ 
\satisfiesD {M} {\quadruple  {A_1'} }   {\sigma}   {A_2' } {A_3' }  \  \ \  ]$ } \\
$\strut \ \ \   \ \ \ \ \ \ \ \ \   \  \mbox{where}$\\
$\strut  \ \ \   \ \ \ \ \ \ \ \ \   \  \ A_1' \triangleq   y_0:D,{\overline {y:D}}   \wedge   A[y_0/\prg{this}],\  \  A_2' \triangleq A_2[u/res,y_0/\prg{this}],\ \ A_3' \triangleq A_3[y_0/\prg{this}] $
 \item
 $\satisfiesD{M}{S\, \wedge\, S'}$\ \ \  \ \ \  $\triangleq$  \  \ \  \   $\satisfiesD{M}{S}\ \wedge \ \satisfiesD{M}{S'}$
\end{enumerate}
\end{definition}

 \begin{lemma}[\Strong   vs Shallow Semantics of Quadruples]
For all $M$, $S$:
\begin{itemize}
\item
$\satisfiesD {M} {S}   \ \ \ \Longrightarrow \ \ \ \satisfies  {M} {S} $
\end{itemize}
\end{lemma}

\subsection{Soundness of the Hoare Triples Logic}
\label{s:sound:app:triples}

\begin{auxLemma}
\label{l:no:call}
For any module $M$, assertions $A$, $A'$ and $A''$, such that $\Pos A$, and $\Pos {A'}$, and a statement $stmt$ which does not contain any method calls:
\begin{center}
$  \satisfiesD {M} {\triple {A} {stmt} {A'} }  \ \ \Longrightarrow\ \  \satisfiesD {M} {\quadruple {A} {stmt} {A'} {A''}}$
\end{center}
\end{auxLemma}

\begin{proof}

\end{proof}


\subsubsection{Lemmas about protection}
\label{s:app:protect:lemmas}

\begin{definition}

$\LRelevantKO  {\sigma} {k}\ \ \ \triangleq\ \ \  \{ \alpha \mid  \exists i.[ \ k \leq i \leq \depth \sigma \ \wedge \ \alpha \in \LRelevantO  {\RestictTo \sigma i}\ ]$
\end{definition}
 
Lemma \ref{exec:rel} guarantees that  program execution reduces the locally reachable objects, unless it allocates new ones.
That is, any objects locally reachable in the $k$-th frame of the new state ($\sigma'$), are either new, or were locally reachable in the $k$-th frame of the previous state ($\sigma$).

{
\begin{lemma} For all $\sigma$, $\sigma'$, and $\alpha$, where $\models \sigma$, and where $k\leq \DepthSt {\sigma}$:
\label{exec:rel}
\begin{itemize}
\item
$\leadstoBounded  {\Mtwo\cdot M}  {\sigma}  {\sigma'}\ \  \Longrightarrow\ \ 
  \LRelevantKO   {\sigma'} {k}\cap \sigma   \subseteq\   \LRelevantKO {\sigma} {k}  $
\item
$\leadstoBoundedStar  {\Mtwo\cdot M}  {\sigma}  {\sigma'}\ \  \Longrightarrow\ \ 
  \LRelevantKO   {\sigma'} {k}\cap \sigma   \subseteq\   \LRelevantKO {\sigma} {k} $
\end{itemize}
\end{lemma}
 }
 
\begin{proof} $~ $

\begin{itemize}
\item
If the step is a method call, then the assertion follows by construction.
If the steps is a   local execution in a method, we proceed by case analysis. If it is an assignment to a local variable, then 
$\forall k.[\ \LRelevantKO   {\sigma'} {k}= \LRelevantKO   {\sigma} {k}\ ]$.
If the step is the creation of a new object, then the assertion holds by construction.
If it it is a field assignment, say, $\sigma'=\sigma[\alpha_1,f \mapsto \alpha_2]$, then we have that 
$\alpha_1, \alpha_2 \  \LRelevantKO {\sigma} {\DepthSt \sigma}$. 
And therefore, by Lemma \ref{rel:smaller}, we also have that $\alpha_1, \alpha_2 \  \LRelevantKO {\sigma} {k}$
All locally reachable objects in $\sigma'$ were either already reachable in $\sigma$ or reachable through $\alpha_2$,
Therefore, we also have that  $\LRelevantKO   {\sigma'} {k}\subseteq  \LRelevantKO   {\sigma} {k}$
 And by definition of $\leadstoBounded  {\_}  {\_}  {\_}$, it is not a method return.
 
\item
By induction on the number of steps in $\leadstoBoundedStar  {\Mtwo\cdot M}  {\sigma}  {\sigma'}$. 
For the steps that correspond to method calls, the assertion follows by construction.
For the steps that correspond to local execution in a method, the assertion follows from the bullet above.
For the steps that correspond to method returns, the assertion follows by lemma \ref{rel:smaller}.
\end{itemize}
\end{proof}

Lemma \ref{change:external} guarantees that any change to the contents of an external object can only happen during execution of an external method.
 
 {
 \begin{lemma} For all $\sigma$, $\sigma'$:
 \label{change:external}
\begin{itemize}
\item
$\leadstoBounded  {\Mtwo\cdot M}  {\sigma}  {\sigma'}\ \wedge \  \sigma \models \external{\alpha} \ \wedge\  {\interpret{\sigma} {\alpha.f}} \neq {\interpret{\sigma'} {\alpha.f}}
 \ \ \Longrightarrow\ \  M,\sigma \models \extThis$
\end{itemize}
\end{lemma}
  } 
  \begin{proof}
  Through inspection of the operational semantics in Fig. \ref{f:loo-semantics}, and in particular rule {\sc{Write}}.
  \end{proof}


Lemma \ref{lemma:inside:preserved}  guarantees that internal code which does not include method calls preserves absolute protection. 
It is used in the proof of soundness of the inference rule {\sc{Prot-1}}.

  {
 \begin{lemma} For all $\sigma$, $\sigma'$, and $\alpha$:
 \label{lemma:inside:preserved} 
\begin{itemize}
\item
$ \satDAssertFrom M  {\sigma} k   {\inside{\alpha}}  \, \wedge \,M, \sigma \models \intThis \, \wedge \, \sigma.\prg{cont} \mbox{ contains no method calls } \ \wedge\ \leadstoBounded   {\Mtwo\cdot M}  {\sigma}  {\sigma'}\  \ \ \Longrightarrow\ \ \satDAssertFrom M  {\sigma'} k   {\inside{\alpha}}$
\item
$ \satDAssertFrom M  {\sigma} k   {\inside{\alpha}}  \ \wedge \ M, \sigma \models \intThis \ \wedge \ \sigma.\prg{cont}  \mbox{ contains no method calls } \ \wedge\ \leadstoBoundedStar  {\Mtwo\cdot M}  {\sigma}  {\sigma'}\  \ \ \Longrightarrow\ \ \satDAssertFrom M  {\sigma'} k   {\inside{\alpha}}$
\end{itemize}
\end{lemma}
}

\begin{proof} $~ $

\begin{itemize}
\item
Because $\sigma.\prg{cont}$   contains no method calls, we also have that $\DepthSt {\sigma'}=\DepthSt {\sigma}$. Let us take
$m=\DepthSt {\sigma}$.

We continue by contradiction.
 Assume that $\satDAssertFrom M  {\sigma} k   {\inside{\alpha}}$ and $\notSatDAssertFrom M  {\sigma} k   {\inside{\alpha}}$  

Then:\\
(*)\ $\forall f.\forall i\in [k..m].\forall \alpha_o\in  \LRelevantKO {\sigma} {i}.[\ M, \sigma \models \external {\alpha_o} \Rightarrow \interpret {\sigma} {\alpha_o.f} \neq \alpha\ \wedge \alpha_o \neq \alpha \ ] $.
\\
(**) $\exists f.\exists  j\in [k..m].\exists \alpha_o\in  \LRelevantKO {\sigma'} {j}.[\ M, \sigma' \models \external {\alpha_o} \wedge \interpret {\sigma'} {\alpha_o.f} = \alpha\ \vee \alpha_o = \alpha\ ] $
\\
We proceed by cases
\begin{description}
\item{1st Case} $\alpha_o \notin  \sigma$, \ie $\alpha_o$ is a new object.
  Then, by our operational semantics, it cannot have a field pointing to an already existing object ($\alpha$), nor can it be equal with $\alpha$. Contradiction.
\item{2nd Case}  $\alpha_o \in  \sigma$. Then, by Lemma \ref{exec:rel}, we obtain that $\alpha_o \in \LRelevantKO {\sigma} {j}$.
Therefore, using (*), we obtain that $\interpret {\sigma} {\alpha_o.f} \neq \alpha$, and therefore 
$ \interpret {\sigma} {\alpha_o.f} \neq \interpret {\sigma'} {\alpha_o.f}$.
By lemma \ref{change:external}, we obtain $M, \sigma \models \extThis$. Contradiction!
\end{description}

\item
By induction on the   number of steps, and using the bullet above.
\end{itemize}

\end{proof}


  {
 \begin{lemma} For all $\sigma$, $\sigma'$, and $\alpha$:
 \label{lemma:protTwo}
\begin{itemize}
\item
$  M, \sigma  \models    { \protectedFrom \alpha {\alpha_o}}  \ \wedge\   \sigma.\prg{heap}= \sigma'.\prg{heap} \ \ \Longrightarrow\ \  M, \sigma' \models      { \protectedFrom \alpha {\alpha_o}} $
\end{itemize}
\end{lemma}
}

\begin{proof}
By unfolding and folding the definitions.
\end{proof}

{
 \begin{lemma} For all $\sigma$,  and $\alpha$, $\alpha_o$, $\alpha_1$, $\alpha_2$:
 \label{l:prtFrom}
\begin{itemize}
\item
$ M, \sigma  \models    {\protectedFrom \alpha  {\alpha_o}}  \  \wedge \ \  M, \sigma  \models    {\protectedFrom \alpha  {\alpha_1}}    \   \ 
\Longrightarrow\ \ M, \sigma[\alpha_2,f \mapsto \alpha_1] \models  \protectedFrom\alpha   {\alpha_o}$
\end{itemize}
\end{lemma}
}

{
\begin{definition}
\begin{itemize}
\item
$M, \sigma \models \internalPaths{\re} \ \ \triangleq \ \ \forall \overline{f}.[\  M, \sigma \models \internal{\re.\overline{f}}\ ]$
\end{itemize}
\end{definition}
}

{
 \begin{lemma} For all $\sigma$, and $\alpha_o$ and $\alpha$:
\begin{itemize}
\item
$M, \sigma \models \internalPaths{\alpha_o}  \    \ \ \Longrightarrow\ \ M, \sigma \models {\protectedFrom \alpha {\alpha_o}}$
\end{itemize}
\end{lemma}
}

\noindent
\vspace{.2cm}
\textbf{Proof Sketch Theorem \ref{l:triples:sound}} 
The proof goes by case analysis over the rules from Fig. \ref{f:underly}  applied to obtain $M \vdash \{ A \}\ stmt \  \{ A' \} $:

\begin{description} 

\item[{\sc{embed\_ul}}] 
By  soundness of the underlying Hoare logic (axiom \ref{ax:ul:sound}), we obtain that  $M\ \models\ \triple {A} {stmt}   {A'}$.
By axiom \ref{ax:ul} we also obtain that $\Stable{A}$ and  $\Stable{A'}$. 
This, together with   Lemma \ref{l:shallow:scoped}, part \ref{shallow:to:scoped}, gives us that
$\satisfiesD {M} {\triple {A} {stmt} {A'} }$. 
By the assumption of {\sc{extend}}, $stmt$ does not contain any method call. Rest follows by lemma \ref{l:no:call}.

\item[{\sc{Prot-New}}] By operational semantics, no field of another object will point to $u$, and therefore $u$ is protected, and protected from all variables $x$.

\item[{\sc{Prot-1}}] by Lemma \ref{lemma:inside:preserved}. The rule premise
${\hproves{M}  	{\ z=e   } {\ stmt } {\ z=e\  }}$ allows us to consider addresses, $\alpha$,   rather than expressions, $e$.

\item[{\sc{Prot-2}}] by Lemma \ref{lemma:protTwo}. The rule premise
${\hproves{M}  	{\ z=e \wedge z=e'\ } {\ stmt } {\ z=e \wedge z=e'\ }}$ allows us to consider addresses $\alpha$, $\alpha'$ rather than expressions $e$, $e'$.

\item[{\sc{Prot-3}}] also by Lemma \ref{lemma:protTwo}. Namely, the rule does not change, and $y.f$ in the old state has the same value as $x$ in the new state.

\item[{\sc{Prot-4}}] by Lemma \ref{l:prtFrom}.

\item[{\sc{types-1}}] 

Follows from type system, the assumption of {\sc{types-1}} and lemma \ref{l:no:call}.

\end{description}
\noindent
\vspace{.1cm}
{\textbf{End Proof Sketch}}

\subsection{Well-founded ordering}

 \begin{definition}
\label{def:measure}
For a module $M$, and modules $\Mtwo$,   we define a measure, $\measure {A, \sigma,A',A''} {M,\Mtwo} $, and based on it, a well founded ordering $(A_1,\sigma_1,A_2, A_3) \ll_{M,\Mtwo}  (A_4,\sigma_2,A_5,A_6)$
as follows:
\begin{itemize}
\item
 $\measure {A, \sigma,A',A''} {M,\Mtwo} \  \ \triangleq \ \  (m, n)$,  \ \ \  where
\begin{itemize}
\item
$m$ is the minimal number of execution steps so that $ \leadstoBoundedStarFin {M\cdot \Mtwo} {\sigma}    {\sigma'}$  for some $\sigma'$, {and $\infty$ otherwise}.
 \item
  $n$ is minimal depth of all proofs of $M \vdash \quadruple {A} {\sigma.\prg{cont}} {A'} {A''} $.
\end{itemize}
 \item
 $(m,n) \ll (m',n')$\ \  \ \ $\triangleq$ \ \  \ \ $ m<m'\vee  (m=m'  \wedge n < n')   $.
\item
$(A_1,\sigma_1,A_2, A_3) \ll_{M,\Mtwo}  (A_4,\sigma_2,A_5, A_6)$  \  \  $\triangleq$ \ \ 
$\measure {A_1, \sigma_1,A_2, A_3} {M,\Mtwo}  \ll \measure {A_4, \sigma_2,A_5. A_6} {M,\Mtwo} $
\end{itemize}
\end{definition}


\begin{lemma}
\label{lemma:normal:two}
For any modules $M$ and $\Mtwo$,  the relation $\_ \ll_{M,\Mtwo}  \_$ is well-founded.
\end{lemma}


\subsection{Public States, properties of executions consisting of several steps}

We t define a state to be public, if the currently executing method is public.

\begin{definition}
We use the form
$M, \sigma \models \pubMeth$ to express that the currently executing method is public.\footnote{This can be done by looking in the caller's frame -- ie the one right under the topmost frame -- the class of the current receiver and the name of the currently executing method, and then looking up the method definition in the module $M$; if not defined there, then it is not \prg{public}. }
Note that $\pubMeth$ is not part of the assertion language.
\end{definition}

 \begin{auxLemma}[Enclosed Terminating Executions]\footnoteSD{TODO find better name for the aux lemma}
 \label{lemma:encl:tem}
 For   modules $\Mtwo$,   states $\sigma$, $\sigma'$, $\sigma_1$:
\begin{itemize}
\item
$  \leadstoBoundedStarFin {\Mtwo}  {\sigma}  {\sigma'} \  \wedge \  \leadstoBoundedStar  {\Mtwo}  {\sigma}  {\sigma_1} 
\ \ \  \Longrightarrow\ \ \  
 \exists \sigma_2.[\ \ \leadstoBoundedStarFin {\Mtwo} {\sigma_1}  {\sigma_2}  
\ \wedge\ 
\leadstoBoundedStarThree  {\Mtwo}  {\sigma_2}  {\sigma}   {\sigma'} \ \ ]$
\end{itemize}

\end{auxLemma} 

\begin{auxLemma}[Executing  sequences]
\label{lemma:subexp}
For modules $\Mtwo$, statements $s_1$, $s_2$,  states $\sigma$, $\sigma'$, $\sigma'''$:
\begin{itemize}
\item
$ \sigma.\texttt{cont} = s_1; s_2 \ \ \wedge\ \  \leadstoBoundedStarFin {\Mtwo}  {\sigma}  {\sigma'}\ \ 
\wedge \ \
\leadstoBoundedStar {\Mtwo}  {\sigma}  {\sigma''}\
$\\
$  \Longrightarrow$\\
$   \exists \sigma''.[\ \ \ \ \   \leadstoBoundedStarFin {\Mtwo} {\sigma[\texttt{cont}\mapsto s_1]}  {\sigma''}  
\ \wedge\ 
\leadstoBoundedStarFin {\Mtwo} {\sigma''[\texttt{cont}\mapsto s_2]}   {\sigma'} \  \wedge$
\\
$\strut \hspace{1.2cm}  [ \ \ \leadstoBoundedStar {\Mtwo} {\sigma[\texttt{cont}\mapsto s_1]}   {\sigma''}\ \vee \ \leadstoBoundedStarFin {\Mtwo}  {\sigma''[\texttt{cont}\mapsto s_2]}   {\sigma'''}\ ]\ \ \ \ \ \ \ \  \ ] $
\end{itemize}
\end{auxLemma}

\subsection{Summarised Executions}

We repeat the two diagrams given in \S \ref{s:summaized}.

\begin{tabular}{lll}
\begin{minipage}{.45\textwidth}
The diagram opposite  shows such an execution:
  $ \leadstoBoundedStarFin {\Mtwo\cdot M}    {\sigma_2}  {\sigma_{30}}$ consists of 4 calls to external objects,
and 3 calls to internal objects.
The calls to external objects are from $\sigma_2$ to $\sigma_3$,  from $\sigma_3$ to $\sigma_4$, from $\sigma_9$ to $\sigma_{10}$, 
and  from $\sigma_{16}$ to $\sigma_{17}$.
 The calls to internal objects are from $\sigma_5$ to $\sigma_6$, rom $\sigma_7$ to $\sigma_8$, and from $\sigma_{21}$ to $\sigma_{23}$. 
\end{minipage}
& \ \  &
\begin{minipage}{.4\textwidth}
\resizebox{6.2cm}{!}
{
\includegraphics[width=\linewidth]{diagrams/summaryA.png}
} \end{minipage}
\end{tabular}

\begin{tabular}{lll}
\begin{minipage}{.45\textwidth}
 In terms of our example, we want to summarise the execution of the two ``outer'' internal, public methods into the 
 ``large'' steps $\sigma_6$ to $\sigma_{19}$ and $\sigma_{23}$ to $\sigma_{24}$.
 And are not concerned with the states reached from these two public method executions.  
\end{minipage}
& \ \  &
\begin{minipage}{.4\textwidth}
\resizebox{6.2cm}{!}
{
\includegraphics[width=\linewidth]{diagrams/summaryB.png}
} \end{minipage}
\end{tabular} 

\noindent 
In order to express such summaries, Def. \ref{def:exec:sum} introduces the following concepts:
\begin{itemize}
\item
 ${\leadstoBoundedThreeStarExt {\Mtwo\cdot M} {\sigma\bd}  {\sigma}  {\sigma'}}$ \ \ \  execution from $\sigma$ to $\sigma'$ scoped by $\sigma\bd$, involving  external states only.
\item
${\WithPub {\Mtwo\cdot M}    {\sigma}  {\sigma'} {\sigma_1}}$ \  \ \  \ \ \ \ \ \ \ \  ${\sigma}$ is an external state  calling an internal public method, and \\
$\strut \hspace{3.25cm}$ $\sigma'$ is the state after return from the public method, and \\
$\strut \hspace{3.25cm}$  $\sigma_1$ is the first state upon entry to the public method.  
\end{itemize}
  
  \noindent
Continuing with our example, we have the following execution summaries:
\begin{enumerate}
\item
${\leadstoBoundedThreeStarExt {\Mtwo\cdot M} {\sigma_3}  {\sigma_3}  {\sigma_5}}$\ \ \ 
Purely external execution from $\sigma_3$ to $\sigma_5$, scoped by $\sigma_3$.
\item
${\WithPub {\Mtwo\cdot M}    {\sigma_5}  {\sigma_{20}} {\sigma_{6}}}$. \ \ \ 
Public method call from external state $\sigma_5$ into  nternal state $\sigma_6$ returning to $\sigma_{20}$. 
Note that this   summarises two  internal method executions ($\sigma_{6}-\sigma_{19}$, and $\sigma_8-\sigma_{14}$),
and two external method executions ($\sigma_{6}-\sigma_{19}$, and $\sigma_8-\sigma_{14}$).
\item
 ${\leadstoBoundedThreeStarExt {\Mtwo\cdot M} {\sigma_3}  {\sigma_{20}}  {\sigma_{21}}}$.
 \item
${\WithPub {\Mtwo\cdot M}    {\sigma_{21}}  {\sigma_{25}} {\sigma_{23}}}$. \ \ \ 
Public method call from  external state ${\sigma_{21}}$ into internal state $\sigma_{23}$, and returning to external state $\sigma_{25}$.
 \item
  ${\leadstoBoundedThreeStarExt {\Mtwo\cdot M} {\sigma_3}  {\sigma_{25}}  {\sigma_{28}}}$.
\ \ \ 
  Purely external execution from $\sigma_{25}$ to $\sigma_{28}$, scoped by ${\sigma_3}$.
\end{enumerate}

\begin{definition}
\label{def:exec:sum}
For any module $M$  where $M$ is the internal module, external modules $\Mtwo$, and states $\sigma\bd$,  $\sigma$,  $\sigma_1$, ... $\sigma_n$, and $\sigma'$, we define:

\begin{enumerate}

\item
 ${\leadstoBoundedThreeStarExt {\Mtwo\cdot M} {\sigma\bd}  {\sigma}  {\sigma'}}$ \ \ \ \ \   $\triangleq$ \ \ 
$
\begin{cases}
M, \sigma  \models  \extThis\  \wedge\  \\
[ \ \ \ 
\sigma=\sigma' \, \wedge\,  \EarlierS  {\sigma\bd}  {\sigma} \, \wedge\,  \EarlierS  {\sigma\bd}  {\sigma''}\ \ \ \ \ \vee\\
\ \ \ \exists \sigma''[\,  \leadstoBoundedThree {\Mtwo\cdot M} {\sigma}  {\sigma\bd}   {\sigma''} \  \wedge\  
{\leadstoBoundedThreeStarExt {\Mtwo\cdot M} {\sigma\bd}  {\sigma''}  {\sigma'}}\, ] \ \ \ ]
\end{cases}
$

\item
${\WithPub {\Mtwo\cdot M}    {\sigma}  {\sigma'} {\sigma_1}}$ \  \ \  \ \ \ \ \ \ \ \ $\triangleq$ \ \ 
$\begin{cases}
M, \sigma  \models \extThis \ \wedge \\
\exists   \sigma_1'\ [ \ \   \leadstoBoundedThree  {\Mtwo\cdot M} {\sigma} {\sigma}  {\sigma_{1}}\, \wedge\,  M, \sigma_1 \models \pubMeth \ \wedge \\ 
\strut \ \ \ \ \  \ \ \ \ \ \   \leadstoBoundedStarFin {\Mtwo\cdot M} {\sigma_1}  {\sigma_1'}  \ \wedge \   \leadstoBounded  {\Mtwo\cdot M} {\sigma_1'}      {\sigma'} \ \ ] 
\end{cases}
$

\item
$\WithExtPub {\Mtwo\cdot M} {\sigma\bd}  {\sigma}  {\sigma'} {\epsilon}$ \ \      \  $\triangleq$ \ \ 
$\leadstoBoundedThreeStarExt {\Mtwo\cdot M} {\sigma\bd}  {\sigma}  {\sigma''}$

\item
\label{four:defg23a}
$\WithExtPub {\Mtwo\cdot M} {\sigma\bd}  {\sigma}  {\sigma'} {\sigma_1}$  \ \ \  $\triangleq$ \ \ 
$\exists \sigma_1',\sigma_2'.  
\begin{cases}
 \ \   {\leadstoBoundedThreeStarExt {\Mtwo\cdot M} {\sigma\bd}  {\sigma}  {\sigma_1'}}\ \wedge\ 
{\WithPub {\Mtwo\cdot M}    {\sigma_1'}  {\sigma_2'} {\sigma_1}}  \ \ \wedge \\
 \ \  {\leadstoBoundedThreeStarExt {\Mtwo\cdot M} {\sigma\bd}  {\sigma_2'}  {\sigma'}}   \\
  \end{cases}$

\item
\label{four:defg23}
$\WithExtPub {\Mtwo\cdot M} {\sigma\bd}  {\sigma}  {\sigma'} {\sigma_1...\sigma_n}$   \ \  $\triangleq$ \ \ 
$\exists \sigma_1'.[ \  \
 \WithExtPub {\Mtwo\cdot M} {\sigma\bd}  {\sigma}  {\sigma_1'} {\sigma_1} 
  \ \wedge \ 
    {\WithExtPub {\Mtwo\cdot M} {\sigma\bd}  {\sigma_1'}  {\sigma'} {\sigma_2...\sigma_n} }   \  \ ]
$

\item
\label{six:g23}
$\leadstoBoundedExtPub {\Mtwo\cdot M}    {\sigma}  {\sigma'} $    \ \ \   \ \ \  \ \ \ \   \ \ \ \  $\triangleq$   \ \ 
 $ \exists n\!\in\! \mathbb{N}. \exists \sigma_1,...\sigma_n. \ \WithExtPub {\Mtwo\cdot M} {\sigma}  {\sigma}  {\sigma'} {\sigma_1...\sigma_n} 
$
\end{enumerate}
\end{definition}

\vspace{.1cm}

Note   that 
${\leadstoBoundedThreeStarExt {\Mtwo\cdot M} {\sigma\bd}  {\sigma}  {\sigma'}}$ implies that $\sigma$ is external, but does not
imply that $\sigma'$ is external.
${\leadstoBoundedThreeStarExt {\Mtwo\cdot M} {\sigma}  {\sigma}  {\sigma'}}$. 
On the other hand, $\WithExtPub {\Mtwo\cdot M} {\sigma\bd}  {\sigma}  {\sigma'} {\sigma_1...\sigma_n}$ implies that $\sigma$ and $\sigma'$ are external, and  $\sigma_1$, ... $\sigma_1$  are internal and public.
Finally, note that   in part (\ref{six:g23}) above it is possible that $n=0$, and so 
$\leadstoBoundedExtPub {\Mtwo\cdot M}    {\sigma}  {\sigma'} $  also holds when
Finally, note that the decomposition used in (\ref{four:defg23}) is not unique, but since we only care for the public states this is of no importance.

\vspace{.2cm}

Lemma \ref{lemma:external_breakdown:term} says that\\
\begin{enumerate}
\item
Any terminating execution which starts at an external state ($\sigma$) consists of a number of external states interleaved with another number of terminating calls to public methods.
\item
Any execution execution which starts at an external state ($\sigma$) and reaches another state ($\sigma'$) also consists of a number of external states interleaved with another number of terminating calls to public methods, which may be followed by a call to some public method (at $\sigma_2$), and from where another execution, scoped by $\sigma_2$  reaches $\sigma'$.
\end{enumerate}

 \begin{auxLemma}
\label{lemma:external_breakdown:term}[Summarised Executions]
For   module $M$, modules $\Mtwo$, and states $\sigma$, $\sigma'$:
\\
\\
If $M,\sigma \models \extThis$, then
\begin{enumerate}
\item
\label{lemma:external_breakdown:term:one}
$\leadstoBoundedStarFin {M\cdot \Mtwo}  {\sigma}  {\sigma'}  \ \ \  \ 
\Longrightarrow \ \ \  \ \leadstoBoundedExtPub {\Mtwo\cdot M}    {\sigma}  {\sigma'}$
\item
\label{lemma:external_breakdown:two}
$\leadstoBoundedStar  {M\cdot \Mtwo}  {\sigma}  {\sigma'}  \ \ \  \ \ \  
\Longrightarrow$ 

\begin{enumerate}
\item
$\strut \ \ \ \ \ \ \ \    \leadstoBoundedExtPub {\Mtwo\cdot M}    {\sigma}  {\sigma'}\ \ \ \  \vee$
\item
$\strut \ \ \ \ \ \ \ \    \exists \sigma_c,\sigma_d.[\ 
\leadstoBoundedExtPub {\Mtwo\cdot M}    {\sigma}  {\sigma_c} 
\wedge\ \leadstoBounded  {\Mtwo\cdot M}    {\sigma_c}  {\sigma_d} 
\wedge \ M, \sigma_c \models \pubMeth \wedge \leadstoBoundedStar  {\Mtwo\cdot M}    {\sigma_d}  {\sigma'} \ ]
$
\end{enumerate}
\end{enumerate}
\end{auxLemma}

\begin{auxLemma}
\label{lemma:external_exec_preserves_more}[Preservation of Encapsulated Assertions]
For any module $M$, modules $\Mtwo$,  assertion  $A$, and 
 states $\sigma\bd$, $\sigma$, $\sigma_1$ ... $\sigma_n$, $\sigma_a$, $\sigma_b$ and $\sigma'$:

\noindent
If

\noindent
 $\strut \hspace{.5cm} M \vdash \encaps A \   \wedge   \ fv(A)=\emptyset \  \wedge \ 
\satDAssertFrom M {\sigma} k A \ \wedge \ k \leq \DepthSt {\sigma\bd}$. 

\noindent
Then

\begin{enumerate}
\item
\label{lemma:external_exec_preserves_more:one}
$  M, \sigma  \models \extThis \ \wedge \  \leadstoBoundedThree  {\Mtwo\cdot M} {\sigma} {\sigma\bd}  {\sigma'} 
\ \ \Longrightarrow \ \ \ \satDAssertFrom M {\sigma'} k A$

\item
$   \leadstoBoundedThreeStarExt {\Mtwo\cdot M} {\sigma\bd}  {\sigma}  {\sigma'} 
\ \ \Longrightarrow \ \ \ \satDAssertFrom M {\sigma'} k A$

\item
\label{lemma:external_exec_preserves_more:three}
$ \WithExtPub {\Mtwo\cdot M} {\sigma\bd}  {\sigma}  {\sigma'} {\sigma_1...\sigma_n}\ \ \wedge $\\
 $\strut \ \ \ \  \  \forall i\in [1..n]. \forall \sigma_{f}.[ \ \  \satDAssertFrom M {\sigma_i} k A  \ \wedge \  \leadstoBoundedStarFin {M\cdot \Mtwo}  {\sigma_i}  {\sigma_{f}} \ 
\ \ \Longrightarrow \ \  \satDAssertFrom M {\sigma_f} k A \ ]$\\
$\Longrightarrow $
\\
 $\strut \ \ \ \  \ \satDAssertFrom M {\sigma'} k A $ 
 \\
  $\strut \ \ \ \  \  \wedge $
  \\
 $\strut \ \ \ \  \  \forall i\in [1..n].   \satDAssertFrom M {\sigma_i} k A $
 \\
 $\strut \ \ \ \  \  \wedge $
  \\
 $\strut \ \ \ \  \  \forall i\in [1..n]. \forall \sigma_{f}.[ \ \    \leadstoBoundedStarFin {M\cdot \Mtwo}  {\sigma_i}  {\sigma_{f}} \ 
\ \ \Longrightarrow \ \  \satDAssertFrom M {\sigma_f} k A \ ]$

\end{enumerate}

\end{auxLemma}

\noindent
\textbf{Proof Sketch}

\begin{enumerate}
\item
 is proven by Def. of $\encaps{\_}$ and the fact $\DepthSt {\sigma'} \geq \DepthSt {\sigma\bd}$ and therefore $k\leq  \DepthSt {\sigma'}$.
In particular, the step $\leadstoBoundedThree  {\Mtwo\cdot M} {\sigma} {\sigma\bd}  {\sigma'}$ may push or pop a frame onto $\sigma$.
If it pops a frame, then $\satDAssertFrom M {\sigma'} k A $ holds by definition.
If is pushes a frame, then $M, \sigma' \models A$, by lemma \ref{lem:encap-soundness}; this gives that $\satDAssertFrom M {\sigma'} k A $.

\item   by induction on the number of steps in $ \leadstoBoundedThreeStarExt {\Mtwo\cdot M} {\sigma\bd}  {\sigma}  {\sigma'} $, and using (1).

\item
 by induction on the number of states   appearing in ${\sigma_1...\sigma_n}$, and using (2).
\end{enumerate}

\textbf{End Proof Sketch}

\subsection{Sequences, Sets, Substitutions and Free Variables}

Our system makes heavy use of textual substitution,   textual inequality, and the concept of free variables in assertions. 
 
In this subsection we introduce some notation and some lemmas to deal with these concepts.
These concepts and lemmas are by no means novel; we list them here so as to use them more easily in the subsequent proofs.

\begin{definition}[Sequences,   Disjointness, and Disjoint Concatenation]
For any variables $v$, $w$, and sequences of variables $\overline v$, $\overline w$ we define:
\begin{itemize}
\item
 $v \in \overline w \ \ \triangleq \ \  \exists  \overline {w_1},  \overline {w_1}[\  {\overline w} = \overline {w_1}, v, \overline {w_2} \ ]$
\item
$v \# w \ \ \triangleq \ \ \neg(v \txteq w)$.
\item
$\overline v \subseteq \overline w \ \ \triangleq \ \ \forall v.[\ v \in  \overline v\ \Rightarrow\ v \in  \overline w\ ]$
\item
$\overline v \#  \overline w \ \ \triangleq \ \ \forall v \in  \overline v. \forall w \in  \overline w.  [ \ v \# w\  ]$
\item
$ \overline v \cap \overline w \ \ \triangleq \ \  \overline u, \ \ \ \ \mbox{such that}\   \forall u.[ \ \ u  \in   \overline v \cap \overline w \ \ \Leftrightarrow\ \  [ \ u\in \overline v\ \wedge\ u\in \overline w\  ]$
\item
$ \overline v \setminus \overline w \ \ \triangleq \ \  \overline u, \ \ \ \ \mbox{such that}\   \forall u.[ \ \ u  \in  \overline v \setminus \overline w \ \ \Leftrightarrow\ \  [ \ u\in \overline v\ \wedge\ u\notin \overline w\  ]$
\item
$\overline v; \overline w \ \ \triangleq \ \ \overline v$, $\overline w$ \ \ if $\overline v \#  \overline w $ \ \ \ and  undefined otherwise.
\end{itemize}
\end{definition}

\begin{lemma}[Substitutions and Free Variables]
\label{l:sfs}
For any sequences of variables $\overline x$, $\overline y$, $\overline z$, $\overline v$, $\overline w$, a variable $w$, any assertion $A$, we have
\begin{enumerate}
\item
\label{l:sfs:zero}
$ \overline x[\overline{y/x} ] = \overline y $
\item
\label{l:sfs:zero:one}
$ \overline {x} \# \overline y \ \   \Rightarrow \ \  \overline y[\overline{z / x} ] = \overline y $
\item
\label{l:sfs:one}
$\overline {z} \subseteq \overline y \ \   \Rightarrow \ \  \overline y[\overline{z / x} ] \subseteq \overline y $
\item
\label{l:sfs:two}
$\overline {y} \subseteq \overline z \ \   \Rightarrow \ \  \overline y[\overline{z / x} ] \subseteq \overline z $
\item
\label{l:sfs:three}
 $\overline x \# \overline y \ \ \Rightarrow \ \ {\overline z}[\overline{y / x}]  \# \overline x $ 
 \item
 \label{l:sfs:four}
 $\fv(A[\overline{y / x}] )\, =\, \fv(A)[\overline{y / x}] $
 \item
 \label{l:sfs:five}
 $\fv(A)\, =\, \overline x; \overline v, \ \ \   \fv(A[\overline{y / x}] )\, = \,   \overline y; \overline w 
 \ \ \ \Longrightarrow\ \ \ 
 \overline v \, = \, (\overline {y}\cap\overline{v}); \overline w $
  \item
  \label{l:sfs:sixa}
 $ \overline v \# \overline x   \# \overline y   \# \overline u  \ \ \    \ \ \ \Longrightarrow\ \ \ 
w[ \overline {u/x} ][ \overline {v/y} ]  \txteq w[ \overline {v/y} ][ \overline {u/x} ]  $

 \item
  \label{l:sfs:six}
 $ \overline v \# \overline x   \# \overline y   \# \overline u  \ \ \    \ \ \ \Longrightarrow\ \ \ 
A[ \overline {u/x} ][ \overline {v/y} ]  \txteq A[ \overline {v/y} ][ \overline {u/x} ]  $
\item
  \label{l:sfs:seven}
 $( fv (A[ \overline {y/x} ])\setminus \overline y)\, \# \, \overline x$
  \item
  \end{enumerate}

\end{lemma}

\noindent 
\begin{proof} 

\begin{enumerate}
\item
 by induction on  the number of elements in $\overline x$ 
\item
 by induction on  the number of elements in $\overline y$ 

\item
 by induction on  the number of elements in $\overline y$ 
\item
 by induction on  the number of elements in $\overline y$ 
\item
 by induction on the structure of $A$ 
 \item
 by induction on the structure of $A$ 
\item
Assume that\\
$(ass1)\ \ \  \fv(A)\, =\, \overline x; \overline v,$
\\
$(ass2)\ \ \  \fv(A[\overline{y / x}] )\, = \,   \overline y; \overline w$\\
We define:
\\
$(a) \ \ \  \overline {y_0} \triangleq \overline v \cap \overline y, \ \ \  \overline {v_2} \triangleq \overline v \setminus \overline y, \ \ \ \overline {y_1} = \overline {y_0}[\overline {x/y}]$
\\
This gives:\\
$(b) \ \ \ \overline {y_0} \#   \overline {v_2}$
\\
$(c)\ \ \ \overline v =  \overline {y_0}; \overline {v_2}$\
\\
$(d) \ \ \  \overline {y_1}  \subseteq \overline y$
\\
$(e) \ \ \ \overline {v_2}[\overline{y / x}] = \overline {v_2}$, \ \ \ from assumption and (a) we have $\overline x \# \overline v_2$ and by Lemma \ref{l:sfs}) part (\ref{l:sfs:zero:one})
\\
We now calculate \\
$\begin{array}{lcll}
\ \ \  \fv(A[\overline{y / x}] )   & = &  (\overline x; \overline v)[\overline{y / x}] & \mbox{ by (ass1), and Lemma \ref{l:sfs} part (\ref{l:sfs:three}).}
\\ 
& = &  (\overline x; \overline {y_0}; \overline {v_2})[\overline{y / x}] & \mbox{ by (c) above }
\\
& = &   \overline x[\overline{y / x}], \, \overline {y_0}[\overline{y / x}], \overline {v_2}[\overline{y / x}] & \mbox{ by distributivity of $[../..]$ }
\\
& = &   \overline y, \, \overline {y_1}, \overline {v_2}  & \mbox{ by Lemma \ref{l:sfs} part (\ref{l:sfs:zero}), (a), and (e). }
\\
& = &   \overline y; \overline {v_2}  & \mbox{ because (d), and  $ \overline y \# \overline {v_2}$ }
\\
\ \ \  \fv(A[\overline{y / x}] )   & = &  \overline y; \overline w  & \mbox{ by (ass2)}
\end{array}
$
\\
The above gives that $\overline {v_2} = \overline {w}$. This, together with (a) and (c) give that $\overline {v} = (\overline {y}\cap\overline{v});\overline{w}$ 
\item
By case analysis on whether $w \in \overline x$ ... etc
\item
By induction on the structure of $A$, and the guarantee from (\ref{l:sfs:sixa}).
\item
We take a variable sequence $\overline z$ such that \\
$(a) \  \  \fv(A  ) \subseteq \overline{x}; \overline z$
\\
This gives that\\
$(b) \  \   \overline{x} \# \overline z$
\\
Part (\ref{l:sfs:four}) of our lemma and (a) give\\
$(c) \  \  \fv(A[\overline{y / x}] ) \subseteq \overline{y}, \overline z$
\\
Therefore
\\
$(d) \  \  \fv(A[\overline{y / x}] ) \setminus {\overline y } \subseteq  \overline z$
\\
The above, together with (b) conclude the proof
 \end{enumerate}
\end{proof} 

\begin{lemma}[Substitutions and Adaptations]
\label{l:sybbs:adapt}
For any sequences of variables $\overline x$, $\overline y$, sequences of expressions $\overline e$, and   any assertion $A$, we have
\begin{itemize}
\item
$ \overline x \# \overline y \ \ \ \Longrightarrow \ \ \  \PushAS {y} {(A[\overline {e/x}])}  \txteq   (\PushAS {y} {A})[\overline{ e/x}] $
\end{itemize}

\end{lemma}

\begin{proof}

\noindent 
We first consider $A$ to be $\inside e_0$, and just take one variable. Then, \\
$\strut \ \ \ \ { \PushAS  {y} { (\inside {e_0}[e/x] ) } }
 \ \txteq\  {\protectedFrom {e_0[ {e/x}]} {y}}$, \\
and\\
$\strut \ \ \ \ (\PushAS {y} {\inside {e_0}})[e/x] \  \txteq\  \protectedFrom {e_0[{e/x}]} {y[ {e/x}]}$. \\
When $x \# y$  then the two assertions from above  are textually equal.
The rest follows by induction on the length of $\overline x$ and the structure of $A$.
\end{proof}

\begin{lemma}
 \label{l:sfs:eight}
For   assertion  $A$, variables $\overline {x}$, $\overline v$, $\overline y$, $\overline w$, $\overline {v_1}$,
addresses  $\overline {\alpha_x}$, $\overline {\alpha_y}$, $\overline {\alpha_v}$, and $\overline {\alpha_{v_1}}$

\noindent
If
\begin{enumerate}[a.]
\item
\label{l:sfs:eight:one}
 $\fv(A)\, \txteq \, \overline {x}; \overline v$, \ \ \ $\fv(A[\overline{y / x}] )\, \txteq  \,   \overline y; \overline w$, 
\item
\label{l:sfs:eight:two}
$ \forall x\in {\overline x}.[\ x[\overline{y/x}]  [\overline{\alpha_y/y}]\ =\  x[\overline{\alpha_x/x} ]  \ ]$
\item
\label{l:sfs:eight:three}
$\overline v \txteq  \overline {v_1}; \overline w$, \ \ \ $\overline {v_1}\txteq  \overline y\cap \overline v$, \ \ \ $ \overline {\alpha_{v,1} } = \overline {v_1}[\overline{\alpha_v/v} ]$
\end{enumerate}

\noindent
then

\begin{itemize}
\item
  $ A [\overline{y/x}] [\overline{\alpha_y/y}]
  \txteq A [\overline{\alpha_x/x}] [\overline{\alpha_{v,1}/v_1}] $
\end{itemize}

\end{lemma}

\begin{proof}
$ ~ $

From Lemma \ref{l:sfs}, part \ref{l:sfs:five}, we obtain\  \ \ $(*)\ \  \overline v \, = \, (\overline {y}\cap\overline{v}); \overline w $

We first prove that \\
$\strut \hspace{3cm} \ \  (**) \ \  \forall z\in \fv(A)[\ z [\overline{y/x}] [\overline{\alpha_y/y}]   \txteq z [\overline{\alpha_x/x}] [\overline{\alpha_{v,1}/v_1}] $.
 
Take any arbitrary $z\in \fv(A)$. \\
Then, by assumptions \ref{l:sfs:eight:one} and \ref{l:sfs:eight:three}, and (*) we have that either $z\in \overline x$, or $z\in \overline {v_1}$, or 
$z\in \overline {w}$.
 
  \begin{description}
 \item
 [  {1st Case}]  
 $z\in \overline x$. 
 Then, there exists some $y_z\in \overline y$, and some $\alpha_z \in \overline {\alpha_y}$, such that    $z[\overline{y/x}]=y_z$  and
  $y_z [\overline{\alpha_y/y}] = \alpha_z$. 
  On the other hand,   by part \ref{l:sfs:eight:two} we obtain, that  $z[\overline{\alpha_x/x}] =\alpha_z$.
  And because $\overline {v_1}\# \overline {\alpha x}$ we also have that $\alpha_z[\overline{\alpha_{v,1}/v_1}] $=$\alpha_z$. 
  This concludes the case.
  \item
   [  {2nd Case}]  
$ z\in \overline {v_1}$, which means that $ z\in \overline {y}\cap\overline{v}$. Then, because $\overline x \# \overline  v$, we have that $z[\overline{y/x}]=z$. 
And because $z\in \overline y$, we obtain that there exists a $\alpha_z$, so that  $z [\overline{\alpha_y/y}] = \alpha_z$.
Similarly, because  $\overline x \# \overline  v$, we also obtain that $z[\overline{\alpha_x/x}]=z$.
And because $\overline {v_1} \subseteq \overline y$, we also obtain that $z[\overline {\alpha_{v,1}/v_1}]$=$z[\overline {\alpha_{y}/y}]$.
This concludes the case.
  \item
     [   {3rd Case}]  
  $z\in \overline {w}$. From part  \ref{l:sfs:eight:one} of the assumptions and from (*) we obtain  $\overline {w} \#  \overline {y}\# \overline {x}$, which implies  that $z [\overline{y/x}] [\overline{\alpha_y/y}]$=$z$.
  Moreover, (*) also gives that $\overline {w} \#  \overline {v_1}$, and this gives that
  $z[\overline{\alpha_x/x}] [\overline{\alpha_{v,1}/v_1}] $=$z$.
  This concludes the case
\end{description}

The lemma follows from (*) and structural induction on $A$.
\end{proof}

\subsection{Reachability, Heap Identity, and their properties}
We consider states with the same heaps ($\sigma \sim \sigma'$) and properties about  reachability of an address from another address ($\Reach {\alpha} {\alpha'} {\sigma}$). 

\begin{definition}
For any state  $\sigma$,  addresses $\alpha$, $\alpha'$, we define

\begin{itemize}
\item
$\Reach {\alpha} {\alpha'} {\sigma} \ \ \triangleq\ \ \exists \overline f. [\ \interpret \sigma {\alpha.\overline f} = \alpha' \ ]$
\item
$\sigma \sim \sigma' \ \ \triangleq\ \ \exists \chi, \overline {\phi_1}, \overline {\phi_2}.[\ \sigma=( \overline {\phi_1}, \chi) \ \wedge \ \sigma'=( \overline {\phi_1}, \chi) \ ]$
\end{itemize}
\end{definition}

\begin{lemma}
\label{l:sim}
For any module $M$, state  $\sigma$,  addresses $\alpha$, $\alpha'$, $\alpha''$

\begin{enumerate}
\item
\label{l:reaches:protected}
$M, \sigma \models {\protectedFrom {\alpha} {\alpha'} } \ \wedge\   \Reach {\alpha'} {\alpha''} {\sigma}\ \   \ \Longrightarrow\ \ \ M, \sigma \models {\protectedFrom {\alpha} {\alpha''}  }$ 

\item
\label{l:sim:reaches}
$\sigma \sim \sigma'  \ \Longrightarrow \ \ [\ \Reach {\alpha} {\alpha'} {\sigma}  \ \Longleftrightarrow \ \Reach {\alpha} {\alpha'} {\sigma'}\ ] $

\item
\label{l:sim:protectedFrom}
$\sigma \sim \sigma'   \ \  \Longrightarrow\ \  [\  M, \sigma \models {\protectedFrom {\alpha} {\alpha''}}  \ \Longleftrightarrow \  M, \sigma' \models {\protectedFrom {\alpha} {\alpha''}} \  ]   $ 
 \item
 \label{l:sim:valid}
$\sigma \sim \sigma'  \   \wedge \ \fv(A)=\emptyset \wedge \ \Stable A \   \  \Longrightarrow\ \  [\  M, \sigma \models A  \ \Longleftrightarrow \  M, \sigma' \models A \  ]   $ 

 \end{enumerate}
\end{lemma}

\begin{proof} $ ~ $

\begin{enumerate}
\item
By  unfolding/folding the definitions
\item
By  unfolding/folding the definitions
\item
By unfolding/folding definitions.
\item
By structural induction on $A$, and Lemma \ref{l:sim} part \ref{l:sim:protectedFrom}.
\end{enumerate}

\end{proof}

\subsection{Preservation of assertions when pushing or popping frames}

In this section we  discuss the preservation of satisfaction of assertions when calling methods or returning from methods -- \ie when pushing or popping  frames. 
Namely, since  pushing/popping frames  does not modify the heap, these operations should preserve satisfaction of some assertion $A$, up to the fact that a) passing an object as a parameter of a a result might break its protection, and 
b) the bindings of variables change  with pushing/popping frames.
To deal with a)  upon method call, 
we   require that the fame being pushed or the frame to which we return is internal ($M, \sigma' \models \intThis$), or require the adapted version of an assertion (\ie  ${\PushAS  {v} { A}}$ rather than $A$).
To deal with b) we either require that there are no free variables in $A$, or we break the free variables of $A$ into two parts, \ie $\fv(\Ain) =  \overline{v_1};\overline{v_2}$, where the value of $\overline{v_3}$ in the caller is the same as that of  $\overline{v_1}$ in the called frame.
This   is described in  lemmas \ref{l:calls} -  \ref{l:calls:return}.

We have four lemmas: Lemma \ref{l:calls} describes preservation from a caller to an internal called, lemma \ref{l:calls:external}
describes preservation from a caller to any called, Lemma \ref{l:calls:return} describes preservation from an internal called to a caller, and  Lemma \ref{l:calls:return:ext} describes preservation from an any called to a caller, 
These four lemmas are used in the soundness proof for the four Hoare rules   about method calls, as given in Fig. \ref{f:calls}. 

In the rest of this  section we will first introduce some further auxiliary concepts and lemmas, 
and then state, discuss  and  prove Lemmas \ref{l:calls} -  \ref{l:calls:return}.

\vspace{1cm}

\textbf{Plans for next three subsections}
 Lemmas \ref{l:calls}-\ref{l:calls:external} are quite complex, because they deal with substitution of some of the assertions' free variables.
 We therefore approach the proofs gradually: 
 We  first state and prove a  very simplified version of   Lemmas \ref{l:calls}-\ref{l:calls:external}, where the assertion ($\Ain$ or $\Aout$)    is only about protection and only contains addresses; this is the only basic assertion which is not $Stbl$.
 We then state a slightly more general version of  Lemmas \ref{l:calls}-\ref{l:calls:external}, where the assertion ($\Ain$ or $\Aout$)  is variable-free.

\subsection{Preservation of variable-free simple protection when   pushing/popping frames}

We now move to consider preservation of variable-free assertions about protection when pushing/popping frames

\begin{lemma}[From caller to  called - protected, and variable-free]
\label{l:aux:caller:called}

For any address $\alpha$, addresses $\overline \alpha$, states $\sigma$, $\sigma'$,  
and frame $\phi$.

\noindent
If   $\sigma'=\sigma  \pushSymbol \phi $  
\noindent
then

\begin{enumerate}[a.]
\label{l:aux:caller:called:one}
\item
$ \satDAssertFrom M  \sigma k    {\inside \alpha} \ \ \wedge \ \  
  M, \sigma' \models \intThis\ \ \wedge \ \  Rng(\phi) \subseteq  \LRelevantO \sigma
 \hfill \Longrightarrow  \ \ \  \   \ \satDAssertFrom M  {\sigma'} k   {\inside \alpha} $

\item
\label{l:aux:caller:called:two}
$\satDAssertFrom M  \sigma k    {\protectedFrom {\alpha} {\overline {\alpha}}}    \ \ \wedge \ \  Rng(\phi) \subseteq   {\overline \alpha}  
 \hfill \Longrightarrow  \ \ \  \    M,  {\sigma'} \models   {\inside \alpha}$

\item
\label{l:aux:caller:called:three}
$\satDAssertFrom M  \sigma k    {{\inside \alpha} \wedge  {\protectedFrom {\alpha} {\overline {\alpha}}}}    \ \ \wedge \ \  Rng(\phi) \subseteq    {\overline \alpha} 
 \hfill \Longrightarrow  \ \ \  \   \satDAssertFrom M  {\sigma'} k  {\inside \alpha}$

\end{enumerate}

\end{lemma}

\begin{proof} $ ~ $ 

\begin{enumerate}[a.]
 
\item
(1) Take any $\alpha'\in \LRelevantO {\sigma'}$. 
Then, by assumptions, we have   $\alpha'\in \LRelevantO {\sigma}$. 
This gives, again by assumptions, that $M, \sigma \models {\protectedFrom {\alpha} {\alpha'}}$. 
By the construction of $\sigma$, and lemma \ref{l:sim} part \ref{l:reaches:protected}, we obtain that (2)  $M, \sigma' \models {\protectedFrom {\alpha} {\alpha'}}$.  
From (1) and (2) and because $M, \sigma' \models \intThis$ we obtain that $M, \sigma' \models {\inside {\alpha}}$. 
Then apply lemma \ref{l:sim} part \ref{l:shallow:scoped:scoped}, and we are done.
 
\item 
By  unfolding and folding the definitions, and application of Lemma \ref{l:sim} part \ref{l:reaches:protected}.
\item
By part   \ref{l:sim} part \ref{l:aux:caller:called:two} and \ref{l:shallow:scoped:scoped}.
\end{enumerate}

Notice that   part  \ref{l:aux:caller:called:one} requires that the called ($\sigma'$) is internal, but   parts  \ref{l:aux:caller:called:two} and  \ref{l:aux:caller:called:three} do not.

Notice also that the conclusion in part  \ref{l:aux:caller:called:two} is  $ M,  {\sigma'} \models   {\inside \alpha}$  and not $ \satDAssertFrom M  {\sigma'} k   {\inside \alpha}$.
 This is so, because it is possible that $ M, \sigma \models    {\protectedFrom {\alpha} {\overline {\alpha}}}$ but $ M, \sigma \not\models    {\inside {\alpha}}$. 

\end{proof}

\begin{lemma}[From  called to caller -- protected, and variable-free]
\label{l:aux:called:caller}

For any   states $\sigma$, $\sigma'$,  variable  $v$,  address  $\alpha_v$,
addresses  $\overline{\alpha}$,     
and statement $stmt$.

\noindent
 If $\  \sigma'= (\sigma \popSymbol)[v\! \mapsto \alpha_v][\prg{cont}\!\mapsto\! stmt] $,\ 
  
\noindent
then

\begin{enumerate} [a.]
\item
\label{l:aux:called:caller:one}
$\satDAssertFrom M  \sigma k  {\inside {\alpha}} \ \  \wedge \ \ k < \DepthSt {\sigma} \ \  \wedge \ \ M, \sigma \models {\protectedFrom \alpha {\alpha_v}}   
 \hfill \Longrightarrow  \ \ \  \   \satDAssertFrom M  {\sigma'} k    {\inside {\alpha}} $ .

  \item
 \label{l:aux:called:caller:two}
 $M, \sigma \models  {\inside {\alpha}}\   \ \wedge\ \  \overline {\alpha} \subseteq \LRelevantO \sigma
 \hfill \Longrightarrow  \ \ \  \  %
 \satDAssertFrom M  {\sigma'} k   {\protectedFrom {\alpha}  {\overline \alpha}}$.

\end{enumerate}
\end{lemma}

\begin{proof} $~ $  

\begin{enumerate} [a.]
\item
(1) Take any $i\!\in\![k..\DepthSt {\sigma'})$. 
Then, by definitions  and assumption, we have   $M, {\RestrictTo \sigma i} \models {\inside \alpha}$. 
Take any $\alpha'\!\in\!\LRelevantO {\RestrictTo \sigma i}$. 
We obtain that $M, {\RestrictTo \sigma i} \models {\protectedFrom {\alpha} {\alpha'}}$. 
Therefore, $M, {\RestrictTo \sigma i} \models {\inside {\alpha}}$.
Moreover, ${\RestrictTo \sigma i}$=${\RestrictTo {\sigma'} i}$, and we therefore obtain (2) \ $M, {\RestrictTo {\sigma'}  i} \models {\inside \alpha}$.

\vspace{.05cm} 

(3) Now take a  $\alpha'\!\in\!\LRelevantO {\sigma'}$.\\
Then, we have that either (A):\ $\alpha'\!\in\!\LRelevantO {\RestrictTo \sigma {\DepthSt {\sigma'}}}$, or (B):\  $\Reach {\alpha_r} {\alpha'} {\sigma'}$. 

In the case of (A): Because $k,\DepthSt \sigma = \DepthSt {\sigma'} + 1$, and because $\satDAssertFrom M  \sigma k  {\inside {\alpha}}$
we have  $M, \sigma \models {\protectedFrom {\alpha} {\alpha'}}$.
Because $\sigma \sim \sigma'$ and  Lemma \ref{l:sim} part \ref{l:sim:protectedFrom}, we obtain (A') $M, \sigma'   \models {\protectedFrom {\alpha} {\alpha'}}$ 
 
 In the case of (B): Because $\sigma \sim \sigma'$ and  lemma \ref{l:sim} part \ref{l:sim:reaches}, we obtain  $\Reach {\alpha_r} {\alpha'} {\sigma}$. 
 Then, applying Lemma \ref{l:sim} part \ref{l:sim:protectedFrom} and assumptions, we obtain (B') $M, \sigma'   \models {\protectedFrom {\alpha} {\alpha'}}$.
 
 From (3), (A), (A'), (B) and (B') we obtain: (4)  $M, \sigma'   \models {\inside {\alpha}}$.
 
 \vspace{.05cm} 

 With (1), (2), (4) and Lemma \ref{l:sim} part \ref{l:sim:valid} we are done.
 \item
 From the definitions we obtain that $M, \sigma \models {\protectedFrom {\alpha} {\overline {\alpha}}}$.
 Because $\sigma \sim \sigma'$ and  lemma  \ref{l:sim} part \ref{l:sim:protectedFrom}, we obtain $M, \sigma' \models {\protectedFrom {\alpha} {\overline {\alpha}}}$.
 And because of lemma \ref{l:shallow:scoped}, part \ref{shallow:to:scoped},  we obtain $\satDAssertFrom M  {\sigma'} k   {\protectedFrom {\alpha}  {\overline \alpha}}$.
\end{enumerate}

\end{proof}

\subsection{Preservation of variable-free, $Stbl^+$, assertions when   pushing/popping frames}
We now move consider preservation of variable-free assertions when pushing/popping frames, and generalize the lemmas  \ref{l:aux:caller:called} and
 \ref{l:aux:called:caller}

\begin{lemma}[From caller to  called -  variable-free, and $Stbl^+$]
\label{l:aux:aux:caller:called}

For any assertion $A$,  addresses $\overline \alpha$, states $\sigma$, $\sigma'$,  and frame $\phi$.

\noindent
If \   $\sigma'=\sigma  \pushSymbol \phi $ \  and\  $\Pos A$, \ and \  $\fv(A)=\emptyset$,

\noindent
then

\begin{enumerate}[a.]

\item
\label{l:aux:aux:caller:called:one}
$ \satDAssertFrom M  \sigma k    A \ \ \wedge \ \  
  M, \sigma' \models \intThis\ \ \wedge \ \  Rng(\phi) \subseteq  \LRelevantO \sigma
  \hfill \Longrightarrow  \ \ \  \   \ \satDAssertFrom M  {\sigma'} k   A $

\item
\label{l:aux:aux:caller:called:two}
$\satDAssertFrom M  \sigma k    {\PushASLong  {(\overline {\alpha})} {A}}    \ \ \wedge \ \  Rng(\phi) \subseteq   {\overline \alpha}  
 \hfill \Longrightarrow  \ \ \  \    M,  {\sigma'} \models   {A}$

\item
\label{l:aux:aux:caller:called:three}
$\satDAssertFrom M  \sigma k    {A \wedge   {\PushASLong  {(\overline {\alpha})} {A}} }    \ \ \wedge \ \  Rng(\phi) \subseteq    {\overline \alpha} 
 \hfill \Longrightarrow  \ \ \  \   \satDAssertFrom M  {\sigma'} k  {A}$

\end{enumerate}

\end{lemma}

\begin{proof} $ ~ $ 

\begin{enumerate}[a.]
 
\item
By Lemma \ref{l:aux:caller:called}, part \ref{l:aux:caller:called:one}  and structural induction on the definition of $\Pos {\_}$.
\item
By Lemma \ref{l:aux:caller:called}, part \ref{l:aux:caller:called:one}  and structural induction on the definition of $\Pos {\_}$.
\item 
By part   \ref{l:aux:aux:caller:called:two} and Lemma  \ref{l:shallow:scoped}.
\end{enumerate}
 
\end{proof}

\begin{lemma}[From  called to caller -- protected, and variable-free]
\label{l:aux:aux:called:caller}

For any   states $\sigma$, $\sigma'$,  variable  $v$,  address  $\alpha_v$,
addresses  $\overline{\alpha}$,     
and statement $stmt$.

\noindent
 If $\  \sigma'= (\sigma \popSymbol)[v\! \mapsto \alpha_v][\prg{cont}\!\mapsto\! stmt] $,  \  and\  $\Pos A$, \ and \  $\fv(A)=\emptyset$
  
\noindent
then

\begin{enumerate} [a.]
\item
\label{l:aux:aux:called:caller:one}
$\satDAssertFrom M  \sigma k  {A} \ \  \wedge \ \ k < \DepthSt {\sigma} \ \  \wedge \ \ M, \sigma \models   {\PushASLong  {\alpha_v} {A}} 
 \hfill \Longrightarrow  \ \ \  \   \satDAssertFrom M  {\sigma'} k    {A} $ .

 \item
 \label{l:aux:aux:called:caller:two}
 $M, \sigma \models  {A}\   \ \wedge\ \  \overline {\alpha} \subseteq \LRelevantO \sigma
 \hfill \Longrightarrow  \ \ \  \  %
 \satDAssertFrom M  {\sigma'} k   {\PushASLong  {(\overline {\alpha})} {A}}$.

\end{enumerate}
\end{lemma}

\begin{proof} $~ $  

\begin{enumerate} [a.]
\item
By Lemma \ref{l:aux:called:caller}, part \ref{l:aux:called:caller:one}  and structural induction on the definition of $\Pos {\_}$.
 \item
 By Lemma \ref{l:aux:called:caller}, part \ref{l:aux:called:caller:two}  and structural induction on the definition of $\Pos {\_}$.  
\end{enumerate}

\end{proof}
 
 \subsection{Preservation of assertions when pushing or popping frames -- stated and proven}
\begin{lemma}[From caller to internal called]
\label{l:calls}

For any assertion $\Ain$, states $\sigma$, $\sigma'$,  
variables  $\overline{v_1}$,    $\overline{v_2}$,  $\overline{v_3}$,  $\overline{v_4}$,  $\overline{v_6}$,  
and frame $\phi$.

\noindent
If 
\begin{enumerate}[(i)]
\item 
\label{l:calls:r:one}
$ \Pos \Ain$,  
\item 
\label{l:calls:r:two}
$\fv(\Ain) =  \overline{v_1}; \overline{v_2} $\footnote{As we said earlier. this gives  also that the variable sequences  are pairwise disjoint, \ie $\overline{v_1}\#\overline{v_2}$.},
\ \ \ 
$\fv(\Ain[\overline {v_3/v_1}]) =  \overline{v_3}; \overline{v_4} $, \ \ \ \ 
$ \overline {v_6}\triangleq\overline{v_2}\cap\overline{v_3}; \overline{v_4} $, 
\item
\label{l:calls:r:three}
$\sigma'=\sigma  \pushSymbol \phi  \ \ \ \  \wedge\ \ \  Rng(\phi)= \overline{\interpret {\sigma} {v_3} }$
\item
\label{l:calls:r:four}
 $\overline {\interpret {\sigma'}  {v_1} } = \overline {\interpret {\sigma} {v_3} }$, 

\end{enumerate}

\noindent
then

\begin{enumerate}[a.]
\item
\label{l:calls:callee:one}
$\satDAssertFrom M  \sigma k   \Ain[\overline {v_3/v_1}] 
\ \ \wedge\ \ M, \sigma' \models \intThis
 \hfill \Longrightarrow  \ \ \  \   \ \satDAssertFrom M  {\sigma'} k  {\Ain[\overline { {\interpret {\sigma} {v_6}} / {v_6} } ] } $
 
\item

\label{l:calls:callee:two}
$ \satDAssertFrom M  \sigma k    {\PushASLong  {(\overline {v_3})} {(\Ain[\overline {v_3/v_1}])}} 
\hfill \Longrightarrow  \ \ \  \   \ \ \  \  
M, \sigma' \models   {\Ain\overline { {\interpret {\sigma} {v_6}} / {v_6} } ] }$.

\end{enumerate}

\end{lemma}

\paragraph{Discussion of Lemma} In lemma \ref{l:calls},   state $\sigma$ is the state right before pushing the new frame on the stack,   while 
state $\sigma'$ is the state right after pushing the frame on the stack.
That is, $\sigma$ is the last state before entering the method body, and $\sigma'$ is the first state after entering the method body.
$\Ain$ stands for the method's precondition, while the variables $\overline {v_1}$ stand for the formal parameters of the method,
and $\overline {v_3}$ stand for the actual parameters of the call.
Therefore, $\overline {v_1}$ is the domain of the new frame, and $\overline {\sigma} {v_3}$ is its range.
The variables $\overline {v_6}$ are the free variables of $\Ain$ which are not in  $\overline {v_1}$ -- \cf Lemma \ref{l:sfs} part
( \ref{l:sfs:five}).
Therefore if (\ref{l:calls:callee:one})  the callee is internal, and 
 $\Ain[\overline {v_3/v_1}]$ holds  at the call point, or
 if (\ref{l:calls:callee:two}) ${\PushASLong  {(\overline {v_3})} {(\Ain[\overline {v_3/v_1}])}}$
 holds  at the call point, 
 then $\Ain[\overline {.../v_61}]$  holds right after pushing $\phi$ onto the stack.
Notice the difference in the conclusion in (\ref{l:calls:callee:one}) and (\ref{l:calls:callee:two}): in the first case we have \strong satisfaction, while in the second case we only have shallow satisfaction.

\begin{proof} $ ~ $ 

We will use $\overline {\alpha_1}$ as short for $\{\interpret {\sigma'} {v_1}$, and $\overline {\alpha_3}$ as short for $\overline {\interpret \sigma {v_3} }$.
\\
We aslo define  $\overline{v_{6,1}} \triangleq \overline{v_2} \cap \overline{v_3}$, \ \ 
   $\overline {\alpha_{6,1}} \triangleq  \overline {v_{6,1}} [\overline {\interpret {\sigma} {v_6}/v_6 }]$
\\
We  establish that\\
$\ \strut \ \  (*) \ \ \  \Ain[\overline {v_3/v_1}][\overline { {\interpret \sigma {v_3} }/v_3 }] \txteq \Ain[\overline  {\alpha_1/v_1}][\overline { \alpha_{6,1}/v_{6,1}}] $\\
This holds by  By Lemma \ref{l:sfs:eight} and assumption \ref{l:calls:r:four} of the current lemma.
\\
And we define    $\overline{v_{6,2}} \triangleq \overline{v_2} \setminus \overline{v_3}$, \ \ 
       $\overline {\alpha_{6,2}} \triangleq  \overline {v_{6_2}} [\overline {\interpret {\sigma} {v_6}/v_6 }]$.

\begin{enumerate}[a.]
\item
Assume  \\
  $\satDAssertFrom M  \sigma k   \Ain[\overline {v_3/v_1}] $. 
      \ \ \ \ \ By Lemma \ref{l:assrt:unaffect}    this  implies that \\
$\satDAssertFrom M  \sigma k   \Ain[\overline {v_3/v_1}][\overline { {\alpha_3}/v_3 }] $   
      \ \ \ \ \ By (*) from above we have\\
$\satDAssertFrom M  \sigma k   \Ain[\overline  {\alpha_1 /v_1}  [\overline { \alpha_{6,1}/v_{6,1}}] $ \\    
$\strut \hspace{2cm}  $ The above, and   Lemma \ref{l:assrt:unaffect}     give that\\
$\satDAssertFrom M  \sigma k   \Ain[\overline { \alpha_1/v_1}]  [\overline { \alpha_{6,1}/v_{6,1}}]  [\overline { \alpha_{6,2}/v_{6,2}}]$ 
\\ The assertion above is variable-free.Therefore,  by Lemma \ref{l:aux:aux:caller:called} part \ref{l:aux:aux:caller:called:one} we also obtain\\ 
$\satDAssertFrom M   {\sigma'} k   {\Ain[\overline {\alpha_1/v_1}]  [\overline { \alpha_{6,1}/v_{6,1}}]  [\overline { \alpha_{6,2}/v_{6,2}}]}$ \\
With  \ref{l:assrt:unaffect}    the above gives \\
$\satDAssertFrom M  {\sigma'}  k  { \Ain[\overline{  \interpret {\sigma'} {v_1}/v_1}]  [\overline {  \interpret {\sigma} {v_6}/v_6 }]}$\\
 By Lemma \ref{l:assrt:unaffect}   , we obtain \\
$\satDAssertFrom M   {\sigma'}  k  { \Ain   [\overline {  \interpret {\sigma} {v_6}/v_6 }]}$
\item
Assume\\
$ \satDAssertFrom M  \sigma k    {\PushASLong  {(\overline {v_3})} {(\Ain[\overline {v_3/v_1}])}}$. 
 \ \ \ \ \ By Lemma \ref{l:assrt:unaffect}    this  implies that \\
 $ \satDAssertFrom M  \sigma k    {(\PushASLong  {(\overline {v_3})} {(\Ain[\overline {v_3/v_1} ])}) [\overline { {\alpha_3}/v_3 }] }$
 \ \ \ \ \ which implies that\\
$ \satDAssertFrom M  \sigma k   { \PushASLong  {(\overline {\alpha_3})} {(\Ain[\overline {v_3/v_1} ][\overline { {\alpha_3}/v_3 }]) }  }$
      \ \ \ \ \ By (*) from above we have\\
$\satDAssertFrom M  \sigma k  { \PushASLong  {(\overline {\alpha_3})}  {( \Ain[\overline  {\alpha_1 /v_1}]  [\overline { \alpha_{6,1}/v_{6,1}}] )}}$
 \\
$\strut \hspace{2cm}  $The above, and   Lemma \ref{l:assrt:unaffect}     give that\\
 $\satDAssertFrom M  \sigma k  {({ \PushASLong  {(\overline {\alpha_3})}  {( \Ain[\overline  {\alpha_1 /v_1} ] [\overline { \alpha_{6,1}/v_{6,1}}] )}}
 ) [\overline  {\alpha_{6,2} /v_{6,2}} ] }$
  \\
$\strut \hspace{2cm}  $ And the above gives\\
 $\satDAssertFrom M  \sigma k  { \PushASLong  {(\overline {\alpha_3})}  {( \Ain[\overline  {\alpha_1 /v_1} ] [\overline { \alpha_{6,1}/v_{6,1}}]  [\overline  {\alpha_{6,2} /v_{6,2}} ] ) }} $
  \\
The assertion above is variable-free.Therefore,  by Lemma \ref{l:aux:aux:caller:called} part \ref{l:aux:aux:caller:called:two} we also obtain\\ 
$ M  \sigma' \models  {   {  \Ain[\overline  {\alpha_1 /v_1} ] [\overline { \alpha_{6,1}/v_{6,1}}]  [\overline  {\alpha_{6,2} /v_{6,2}} ] }} $  
  \\
We apply Lemma  \ref{l:assrt:unaffect}   , and  Lemma \ref{l:assrt:unaffect}   , and obtain \\
$ M,   {\sigma'}  \models  { \Ain   [\overline {  \interpret {\sigma} {v_6}/v_6 }]}$

\end{enumerate}

\end{proof}

\begin{lemma}[From caller to any called]
\label{l:calls:external}

For any assertion $\Ain$, states $\sigma$, $\sigma'$,  
variables   $\overline{v_3}$  
statement $stmt$, and frame $\phi$.

\noindent
If 
\begin{enumerate}[(i)]
\item 
\label{l:calls:ext:re:one}
$ \Pos \Ain$,  
\item 
\label{l:callsext:re:two}
$\fv(\Ain) =  \emptyset$,
\item
\label{l:calls:ext:re:three}
$\sigma'=\sigma  \pushSymbol \phi  \ \ \ \  \wedge\ \ \  Rng(\phi)= \overline{\interpret {\sigma} {v_3} },$
\end{enumerate}

\noindent
then

\begin{enumerate}[a.]
\item
\label{l:calls:callee:three}
$ \satDAssertFrom M  \sigma k    {\PushASLong  {(\overline {v_3})} {\Ain} }  \hfill \Longrightarrow  \ \ \  \   \ \ \  \   M, \sigma' \models   {\Ain} $.
 
\item
\label{l:calls:callee:four}
$\satDAssertFrom M  \sigma k    {(\Ain  \wedge  ({\PushASLong  {(\overline {v_3})} {\Ain} } ) )}    \hfill \Longrightarrow  \ \ \  \   \satDAssertFrom M  {\sigma'} k  {\Ain}$

\end{enumerate}

\end{lemma}

\begin{proof}
\begin{enumerate}[a.]
\item
Assume that \\
$ \satDAssertFrom M  \sigma k    {\PushASLong  {(\overline {v_3})} {\Ain} }$\\
	By Lemma \ref{l:assrt:unaffect}    this  implies that \\
$ \satDAssertFrom M  \sigma k    {\PushASLong  {(\interpret {\sigma} {v_3})} {\Ain} }$.\\
We now have a variable-free assertion, and by Lemma \ref{l:aux:aux:caller:called}, part \ref{l:aux:aux:caller:called:two}, we obtain\\
$ M,  \sigma'    \models \Ain$.
\item
Assume that \\
$ \satDAssertFrom M  \sigma k    {\Ain \wedge \PushASLong  {(\overline {v_3})} {\Ain} }$\\
	By Lemma \ref{l:assrt:unaffect}    this  implies that \\
$ \satDAssertFrom M  \sigma k    {\Ain \wedge \PushASLong  {(\interpret {\sigma} {v_3})} {\Ain} }$.\\
We now have a variable-free assertion, and by Lemma \ref{l:aux:aux:caller:called}, part \ref{l:aux:aux:caller:called:two}, we obtain\\
$ \satDAssertFrom M  {\sigma'} k  \models \Ain$.
\end{enumerate}

\end{proof}

\paragraph{Discussion of Lemma \ref{l:calls:external}} In this lemma,  as in  lemma \ref{l:calls}, 
  $\sigma$ stands for the last state before entering the method body, and $\sigma'$ for the first state after entering the method body.
$\Ain$ stands for a module invariant in which all free variables have been substituted by addresses.
The lemma is intended for external calls, and therefore we have no knowledge of the method's formal parameters.
The variables   $\overline {v_3}$ stand for the actual parameters of the call, and therefore 
 $\overline {\interpret {\sigma} {v_3}}$ is the range of the new frame.
Therefore if (\ref{l:calls:callee:three})   the adapted version,
 ${\PushASLong  {(\overline {v_3})} {\Ain} }$, holds  at the call point,
 then the unadapted version, $\Ain$  holds right after pushing $\phi$ onto the stack.
 Notice that  even though the premise of (\ref{l:calls:callee:three}) requires \strong satisfaction, the conclusion promises
 only weak satisfaction.
 Moreover, if (\ref{l:calls:callee:four})   the adapted as well as the unadapted version,
 $\Ain \wedge {\PushASLong  {(\overline {v_3})} {\Ain} }$
 holds  at the call point,
 then the unadapted version, $\Ain$  holds right after pushing $\phi$ onto the stack.
Notice the difference in the conclusion in (\ref{l:calls:callee:three}) and (\ref{l:calls:callee:four}): in the first case we have  shallow satisfaction, while in the second case we only have \strong satisfaction.


\begin{lemma}[From internal called to caller]
\label{l:calls:return}
\label{l:calls:return:deep}

For any assertion $\Aout$, states $\sigma$, $\sigma'$, variables $res$, $u$ variable sequences  $\overline{v_1}$, $\overline{v_3}$,   $\overline{v_5}$,  and statement $stmt$.

\noindent
 If 
 
\begin{enumerate}[(i)]
\item 
\label{l:calls:retrun:one}
$ \Pos \Aout$,  
\item 
\label{l:callers:r:two}
$\fv(\Aout) \subseteq  \overline{v_1} $,
\item
$  \overline {\interpret {\sigma'} {v_5}}, {\interpret {\sigma} {res}} \subseteq \LRelevantO \sigma  \ \ \wedge \ \  M, \sigma' \models \intThis$.
 \item
\label{l:callers:three}
\label{l:calls:return:three}
$\sigma'= (\sigma \popSymbol)[ u  \mapsto  {\interpret {\sigma} {res}}][\prg{cont}\!\mapsto\! stmt]
\ \ \wedge \ \   \overline {\interpret {\sigma} {v_1} }\, =\,  \overline {\interpret {\sigma'}  {v_3} }$.
  \end{enumerate}
  
\noindent
then

\begin{enumerate} [a.]
\item
\label{l:calls:caller:one}
$\satDAssertFrom M  \sigma k   {\Aout \, \wedge\,  {(\PushASLong  {res} {\Aout})} }\ \  \wedge\ \ \DepthSt {\sigma'} \geq k  \ 
 \hfill \Longrightarrow  \ \ \  \   \satDAssertFrom M  {\sigma'} k   {\Aout[\overline {v_3 / v_1}]  }$ .
 \item
 \label{l:calls:caller:two}
 $M, \sigma \models  {\Aout}\ 
 \hfill \Longrightarrow  \ \ \  \  %
 \satDAssertFrom M  {\sigma'} k   {\PushASLong  {\overline {v_5}}    {(\Aout[\overline { v_3/v_1}])} }$. 
\end{enumerate}

\end{lemma}

\paragraph{Discussion of Lemma \ref{l:calls:return}}  
 State  $\sigma$ stands for the last state in the method body, and $\sigma'$ for the first state after exiting  the method call.
$\Aout$ stands for a method postcondition.
The lemma is intended for internal calls, and therefore we know the method's formal parameters.
The variables   $\overline {v_1}$ stand for the formal  parameters of the method, and  $\overline {v_3}$ stand for the actual parameters of the call.
Therefore the   formal parameters of the called have the same values as the actual parameters in the caller
  $  \overline {\interpret {\sigma} {v_1} }\, =\,  \overline {\interpret {\sigma'}  {v_3} }$.
Therefore   ( \ref{l:calls:caller:one})  and  (\ref{l:calls:caller:two})  
promise that if the postcondition $\Aout$ holds before popping the frame, then it also holds after popping frame after replacing the 
the formal parameters by the actual parameters $\Aout[\overline{v_3/v_1}]$.
As in earlier lemmas, there is an important difference between  (\ref{l:calls:caller:one}) and (\ref{l:calls:caller:two}):
In (\ref{l:calls:caller:one}), we require \emph{deep satisfaction at the called}, 
and obtain at the deep satisfaction of the \emph{unadapted} version ($\Aout[\overline{v_3/v_1}]$) at the return point;
while in (\ref{l:calls:caller:two}), we only require \emph{shallow satisfaction at the called}, 
and obtain deep satisfaction of the \emph{adapted} version (${\PushASLong  {\overline {v_5}}    {(\Aout[\overline { v_3/v_1}])} }$),
at the return point.

\begin{proof}$ ~ $

We use the following short hands: $\alpha$ as   $\overline {\interpret {\sigma} {res}}$, 
\ \ \ \ 
$\overline {\alpha_1}$   for $\overline {\interpret {\sigma} {v_1}}$, \ \ \ 
$\overline {\alpha_5}$ as short for $\overline {\interpret {\sigma'} {v_5}}$.

\newcommand{\substOne}{[\overline{\alpha_1/v_1}]}
\newcommand{\substFive}{[\overline{\alpha_5/v_5}]}

\begin{enumerate}[a.]
\item
Assume that \\
$\satDAssertFrom M  \sigma k   {\Aout  \wedge  {\PushASLong  {res} {\Aout} } } $\\
	By Lemma \ref{l:assrt:unaffect}    this  implies that \\
$ \satDAssertFrom M  \sigma k   {\Aout\substOne  \wedge  {\PushASLong  {\alpha} {(\Aout\substOne)} } }$.\\
We now have a variable-free assertion,
and by Lemma \ref{l:aux:aux:called:caller}, part \ref{l:aux:aux:called:caller:one}, we obtain\\
$ \satDAssertFrom M  \sigma k   {\Aout\substOne} $.\\
By Lemma \ref{l:assrt:unaffect}   , and because $ \overline {\interpret {\sigma} {v_1} }\, =\,  \overline {\interpret {\sigma'}  {v_3} }$  this  implies that \\
$ \satDAssertFrom M  \sigma k   {\Aout[\overline { v_3/v_1}] } $.
\item
Assume that \\
$M,  \sigma\models \Aout  $\\
By Lemma \ref{l:assrt:unaffect}    this  implies that \\
$M,  \sigma\models \Aout\substOne$\\
We now have a variable-free assertion,
and by Lemma \ref{l:aux:aux:called:caller}, part \ref{l:aux:aux:called:caller:two}, we obtain\\
$ \satDAssertFrom M  {\sigma'} k    {\PushASLong {\overline \alpha_5} {\Aout\substOne}}$\\
By Lemma \ref{l:assrt:unaffect}   , and because $ \overline {\interpret {\sigma} {v_1} }\, =\,  \overline {\interpret {\sigma'}  {v_3} }$  and $\alpha_5 = \overline {\interpret {\sigma'} {v_5}}$, we obtain \\
$ \satDAssertFrom M  {\sigma'} k    {\PushASLong {\overline v_5} {\Aout[\overline {v_3/v_1}]}}$
\end{enumerate}

\end{proof}

\begin{lemma}[From any called to caller]
\label{l:calls:return:ext}

For any  assertion $\Aout$, states $\sigma$, $\sigma'$, variables $res$, $u$ variable sequence   $\overline{v_5}$,  and statement $stmt$.

\noindent
 If 
 
\begin{enumerate}[(i)]
\item 
$ \Pos \Aout$,  
\item 
 $\fv(\Aout) = \emptyset$,
\item
$  \overline {\interpret {\sigma'} {v_5}}, {\interpret {\sigma} {res}} \subseteq \LRelevantO \sigma $.
 \item
$\sigma'= (\sigma \popSymbol)[ u  \mapsto  {\interpret {\sigma} {res}}][\prg{cont}\!\mapsto\! stmt]$.
  \end{enumerate}
 \noindent
then

\begin{enumerate}[a.]
\item
 \label{l:ext:return:one}
 $M, \sigma \models  {\Aout}\  
 \hfill \Longrightarrow  \ \ \  \  \satDAssertFrom M  {\sigma'} k  {\PushASLong  {(\overline {v_5})}    {\Aout}}$.

\item
\label{l:ext:return:two}
$\satDAssertFrom M  \sigma k   \Aout  \ \  \wedge\ \ \DepthSt {\sigma'} \geq k  \ 
 \hfill \Longrightarrow  \ \ \  \   \satDAssertFrom M  {\sigma'} k   {\Aout \, \wedge\, {\PushASLong  {(\overline {v_5})}    {\Aout}}} $

\end{enumerate}

\end{lemma}

\begin{proof}$ ~ $

\begin{enumerate}[a.]
\item
Assume that \\
$\strut \ \ \ M, \sigma \models  {\Aout}$\\
Since $\Aout$ is a variable-free assertion,  by   Lemma \ref{l:aux:aux:called:caller}, part \ref{l:aux:aux:called:caller:one}, we obtain\\  
$\strut \ \ \ \satDAssertFrom M  {\sigma'} k   {\PushASLong  {(\overline {\interpret {\sigma'} {v_5}})}    {\Aout}}$.\\
By Lemma \ref{l:assrt:unaffect}   ,  we obtain \\
$\strut \ \ \ \satDAssertFrom M  {\sigma'} k  {\PushASLong  {(\overline {v_5})}    {\Aout}}$
\item
Similar argument to the proof of Lemma \ref{l:calls:return}, part (b).
\end{enumerate}

\end{proof}

\paragraph{Discussion of lemma  \ref{l:calls:return:ext}}, Similarly to  lemma \ref{l:calls:return},  
in this lemma, state  $\sigma$ stands for the last state in the method body, and $\sigma'$ for the first state after exiting  the method call.
$\Aout$ stands for a method postcondition.
The lemma is meant to apply to external calls, and therefore, we do not know the method's formal parameters, 
$\Aout$ is meant to stand for a module invariant where all the free variables have been substituted by addresses --
\ie $\Aout$ has no free variables.
The variables $\overline {v_3}$ stand for the actual parameters of the call.
Parts    (\ref{l:ext:return:one})  and  (\ref{l:ext:return:two})
promise that if the postcondition $\Aout$ holds before popping the frame, then it its adapted version 
also holds after popping frame (${\PushASLong  {\overline {v_5}}    {\Aout}}$).
 As in earlier lemmas, there is an important difference between   (\ref{l:ext:return:one})  and  (\ref{l:ext:return:two})
In  (\ref{l:ext:return:one}),  we require \emph{shallow satisfaction at the called}, 
and obtain   deep satisfaction of the \emph{adapted} version (${\PushASLong  {\overline {v_5}}    {\Aout}}$) at the return point;
while in (\ref{l:ext:return:two}), we  require \emph{deep satisfaction at the called}, 
and obtain deep satisfaction of the   \emph{conjuction}  of the \emph{unadapted} with the \emph{adapted} version (${\Aout \, \wedge\, {\PushASLong  {\overline {v_5}}    {\Aout}}}$),
at the return point.

\subsection{\textbf{Use of Lemmas \ref{l:calls}-\ref{l:calls:external} }}

As we said earlier, Lemmas \ref{l:calls}-\ref{l:calls:external} are used to prove the soundness of the Hoare logic rules for method calls.

In the proof of soundness of {\sc{Call\_Int}}. we will use Lemma \ref{l:calls} part (\ref{l:calls:callee:one}) and Lemma \ref{l:calls:return} part (\ref{l:calls:caller:one}).  
In the proof of soundness of {\sc{Call\_Int\_Adapt}} we will use  Lemma \ref{l:calls} part (\ref{l:calls:callee:two}) and Lemma \ref{l:calls:return} part (\ref{l:calls:caller:two}).
In the proof of soundness of {\sc{Call\_Ext\_Adapt}} we will use  Lemma \ref{l:calls:external} part (\ref{l:calls:callee:three}) and Lemma \ref{l:calls:return:ext} part (\ref{l:ext:return:one}).
And finally, in the proof of soundness of {\sc{Call\_Ext\_Adapt\_Strong}} we will use  Lemma \ref{l:calls:external} part (\ref{l:calls:callee:four}) and Lemma \ref{l:calls:return:ext} part (\ref{l:ext:return:two}).

\newcommand{\SPSF}{\SPS \ \ \ \ \ \ \ }
\newcommand{\Apr}{A_{pr}}

\subsection{Proof of Theorem \ref{t:quadruple:sound} -- part (A) }
\label{s:app:proof:sketch;quadruples}
\noindent
\vspace{.2cm}
  {\textbf{Begin Proof}}

\noindent
Take any $M$, $\Mtwo$, with\\ 
$\strut \ \ \hspace{2.3cm} \ \ (1) \ \ \vdash M $.
\\
We will prove that\\
$\strut \ \ \hspace{2.3cm} \ \ (*)\ \ \forall \sigma, A, A', A''.$\\
$\strut \ \ \hspace{2.3cm} \ \ \ \ \  \ \ [ \ M\ \vdash\  \quadruple {A} {\sigma.\prg{cont}} {A'} {A''}  \ \ \Longrightarrow \ \    M\ \modelsD\  \quadruple {A} {stmt} {A'} {A''}\ ]$.\\
by induction on the well-founded ordering  $\_ \ll_{M,\Mtwo}  \_$.
\\
Take $\sigma$, $A$, $A'$, $A''$, $\overline z$, $\overline w$, $\overline \alpha$, $\sigma'$, $\sigma''$  arbitrary. Assume that\\
$\strut \ \ \hspace{2.3cm} \ \ (2) \ \ M\ \vdash\  \quadruple {A} {\sigma.\prg{cont}} {A'} {A''}$\\
$\strut \ \ \hspace{2.3cm} \ \ (3) \ \ \overline {w}=\fv(A)\cap dom(\sigma), \  \ \ \  \overline z = {\fv(A)\setminus dom(\sigma)}$\footnote{Remember that $dom(\sigma)$ is the set of variables defined in the top frame of $\sigma$} \\
$\strut \ \ \hspace{2.3cm} \ \ (4) \ \ \satDAssertFrom M  \sigma k   {A{[{\overline { \alpha/z}]}}}$\\
To show\\
$\strut \ \ \hspace{2.3cm} \ \ (**)\ \ \  \ \    \leadstoBoundedStarFin {\Mtwo\cdot M}  {\sigma}  {\sigma'}\ \ \Longrightarrow\ \     \satDAssertFrom M  {\sigma'} k   {{A'}{[{\overline {\alpha/z}]}}}$\\
$\strut \ \ \hspace{2.3cm} \ \ (*\!*\!*) \ \ \,    \leadstoBoundedStar  {\Mtwo\cdot M}  {\sigma}  {\sigma''}\ \, \ \ \Longrightarrow\ \     \satDAssertFrom M  {\sigma''}  k  \extThis \rightarrow {{A''[{\overline {\alpha/z}}][\overline{\interpret \sigma {w}/w}]}}$
 
 \vspace{.2cm}
\noindent
We proceed by case analysis on the  rule applied in the last step of the proof of (2). We only describe some cases.

\begin{description} 
 
 \item[{\sc{mid}}] 
 
 By Theorem \ref{l:triples:sound}.

\newcommand{\SPS}{\strut \ \ \hspace{0.5cm} \ \ }
 
\item[{\sc{sequ}}] 
Therefore, there exist statements $stmt_1$ and $stmt_2$, and assertions  $A_1$, $A_2$ and $A''$, so that $A_1\txteq A$, and $A_2 \txteq A'$, and $\sigma.\prg{cont}\txteq  stmt_1; stmt_2$,.
We apply lemma \ref{lemma:subexp}, and obtain that there exists an intermediate state $\sigma_1$. 
The proofs for  $stmt_1$ and $stmt_2$, and the intermediate state $\sigma_1$ are in the $\ll$ relation. 
Therefore, we can apply the inductive hypothesis.
 
  \item[{\sc{combine}}]  by induction hypothesis, and unfolding and folding the definitions
  
 \item[{\sc{consequ}}]  using Lemma \ref{l:shallow:scoped} part \ref{fourSD}  and axiom \ref{lemma:axiom:enc:assert:ul}

\renewcommand{\Apr}{A_{pr}}
\newcommand{\Apra}{A_{pra}}
\newcommand{\Apoa}{A_{poa}}
\item[{\sc{Call\_Int}}]
 
 Therefore, there exist $u$, $y_o$, $C$, $\overline y$,  $A_{pre}$, $A_{post}$, and $A_{mid}$, such that \\
 $\SPS (5) \ \ \sigma.\prg{cont}\txteq u:=y_0.m(\overline y)$,\\
$\SPS (6) \  \ \promises  M {\mprepostN {A_{pre}}{D}{m}{x}{D}{A_{post}} {A_{mid}}}$, \\
$\SPS (7) \  \ A \txteq y_0:D ,\overline {y:D}\ \wedge \  A_{pre}[y_0,\overline y/\prg{this},\overline x],$\\
$\SPSF  A'  \txteq A_{post}[y_0,\overline y, u/\prg{this},\overline x, \prg{res}],$\\ 
$\SPSF  A'' \txteq  A_{mid}$. 
\\
Also, \\
$\SPS (8) \ \ \leadstoBounded  {\Mtwo\cdot M}  {\sigma}  {\sigma_1}$, \\
 where \\
$\SPS (8a) \ \ \ \sigma_1\triangleq (\PushSLong { (\prg{this}\mapsto {\interpret{\sigma} {y_0}},{\overline{x \mapsto {\interpret{\sigma} {y}}}})}\sigma [\prg{cont}\mapsto stmt_m]$, \\ 
$\SPS (8b) \ \ \   \prg{mBody}(m,D,M)=\overline{y:D}\{\    stmt_m\ \}$ .\\
We define the shorthands:\\
$\SPS (9) \ \  \Apr \triangleq  \prg{this}:D,\overline{x:D}\, \wedge\, A_{pre}$.
\\
$\SPS (9a) \ \  \Apra \triangleq  \prg{this}:D,\overline{x:D}\, \wedge\, A_{pre}\, \wedge\, {\PushASLong  {(y_0,\overline {y})}  {A_{pre}} } $.
\\
$\SPS (9b) \ \  \Apoa \triangleq  A_{post}\, \wedge\, {\PushASLong {\prg{res}}  {A_{post}}} $.
\\
By (1), (6), (7), (9),  and definition of $\vdash M$ in Section \ref{sect:wf}  rule {\sc{Method}} and we obtain\\
$\SPS (10) \ \  M \vdash  \quadruple { \ \Apra \  } {\ stmt_m } {\ \Apoa\ } {A_{mid}}$.\\
From (8) we obtain\\  
$\SPS (11) \ \ (\Apra,\sigma_1,\Apoa, A_{mid}) \ll_{M,\Mtwo} (A,\sigma,A', A'')$
\\
In order to be able to apply the induction hypothesis, we need to prove something of the form $... \sigma_1 \models \Apr[../fv(\Apr)\setminus dom(\sigma_1)]$. To that aim we will apply Lemma   \ref{l:calls} part \ref{l:calls:callee:one} on (4), (8a) and (9). For this, we take
\\
$\SPS (12) \ \ \ \overline {v_1} \triangleq \prg{this},\overline x$, \ \ \ 
$\overline {v_2} \triangleq  \fv(\Apr)\setminus \overline {v_1}$, \ \ \ 
$\overline {v_3} \triangleq  y_0,\overline y$,\ \ \ 
$\overline {v_4} \triangleq  \fv(A)\setminus \overline {v_3}$
\\
These definitions give that
\\
$\SPS (12a) \ \  A \txteq \Apr[\overline {v_3/v_1}]$,
\\
$\SPS (12b)  \ \  \fv(\Apr)\ =\overline{v_1}; \overline {v_2}$.
\\ 
$\SPS (12c)  \ \  \fv(A)\ =\overline {v_3};\overline {v_4} $.
\\
With (12a), (12b), (12c), (  and Lemma \ref{l:sfs}  part (\ref{l:sfs:five}), we obtain that \\ 
$\SPS (12d)\ \    \overline {v_2} =   \overline{y_r}; \overline{v_4}$, \ \ \ \ where $\overline{y_r}\triangleq {\overline {v_2}}\cap  \overline {v_3}$
\\
Furthermore, (8a), and (12) give that:\\
$\SPS (12e) \ \  \overline { {\interpret {\sigma_1} {v_1}} = {\interpret {\sigma} {v_3}} }$
\\
Then, (4),  (12a),   (12c) and (12f)  give that\\
$\SPS (13) \ \ \satDAssertFrom M  {\sigma} {k}   {\Apr[\overline{v_3/v_1}][\overline{\alpha/z}]}$\\
Moreover, we have that $\overline z \# \overline {v_3}$. From   Lemma \ref{l:sfs}  part (\ref{l:sfs:seven}) we obtain $\overline z \# \overline {v_1}$. 
And, because $\overline \alpha$ are addresses wile $\overline {v_1}$ are variables, we also have that $\overline \alpha\#\overline {v_1}$.
These facts, together with   Lemma \ref{l:sfs}  part (\ref{l:sfs:six}) give that\\
$\SPS (13a) \ \    \Apr[\overline{v_3/v_1}][\overline{\alpha/z}] \txteq   \Apr[\overline{\alpha/z}]  [\overline{v_3/v_1}] $\\
From (13a) and (13), we obtain
\\
$\SPS (13b) \ \     \satDAssertFrom M  {\sigma} {k}   { \Apr[\overline{\alpha/z}]  [\overline{v_3/v_1}] }$
\\
From (4), (8a), (12a)-(12e)  we see that the requirements  
 of Lemma \ref{l:calls}  part  \ref{l:calls:callee:one} are satisfied where we take  $\Ain$ to be $\Apr[\overline{\alpha/z}]$.
 We use the definition of $y_r$ in (12d), and define\\
$\SPS (13c) \ \ \overline {v_6} \triangleq  y_r;  (\overline {v_4}\setminus \overline{z})$\ \ \ \ which, with (12d) also gives:\ $ \overline {v_2} = \overline{v_6}; \overline{z}$
\\
We apply   Lemma   \ref{l:calls} part \ref{l:calls:callee:one} on (13b), (13c)  and  obtain\\
$\SPS (14a) \ \  \satDAssertFrom M  {\sigma_1} {k}   { \Apr[\overline{\alpha/z}][ \overline {  \interpret{\sigma} {v_6}/v_6}] } $. \\
Moreover, we have the $M, \sigma_1 \models \intThis$. We apply lemma \ref{l:calls:return}.(\ref{l:calls:return}), and obtain\\
$\SPS (14b) \ \  \satDAssertFrom M  {\sigma_1} {k}   
      { \Apr[\overline{\alpha/z}][ \overline {  \interpret{\sigma} {v_6}/v_6}] 
          \, \wedge\,  
          { \PushASLong {(\prg{this}, \overline y)} {\Apr}   
       } } $. \\
 With similar re-orderings to earlier, we obain\\
$\SPS (14b) \ \  \satDAssertFrom M  {\sigma_1} {k}   {\Apra[\overline{\alpha/z}][\overline{\interpret{\sigma}{v_6}/v_6}]}$. 

\vspace{.1cm}

For the proof of $(**)$ as well as for the proof of $(*\!*\!*)$, we will want to apply the inductive hypothesis.   For this, we need to determine the value of  
$\fv(\Apr)  \setminus dom(\sigma_1)$, as well as the value of $\fv(\Apr)  \cap dom(\sigma_1)$. 
  This is what we do next. From (8a) we have that\\
$\SPS (15a) \ \  dom(\sigma_1)=\{ \prg{this}, \overline x \}$.\\
This, with (12)  and (12b)  gives  that\\
$\SPS (15b) \ \ \fv(\Apra)\cap dom(\sigma_1)= \overline {v_1}$.\\
$\SPS (15c) \ \  \fv(\Apra)  \setminus dom(\sigma_1) = \overline {v_2}$. \\
Moreover, (12d) and (13d) give that\\
$\SPS (15d) \ \  \fv(\Apra)  \setminus dom(\sigma_1) = \overline{z_2} = \overline {z}; \overline{v_6}$.

 \vspace{.3cm}
Proving $(**)$. Assume that   $\leadstoBoundedStarFin  {\Mtwo\cdot M}  {\sigma}  {\sigma'}$. Then, by the operational semantics, we obtain that 
there exists state $\sigma_1'$, such that \\
$\SPS (16) \ \ \leadstoBoundedStarFin  {\Mtwo\cdot M}  {\sigma_1}  {\sigma_1'}$ \\
$\SPS (17) \ \ \sigma'=(\sigma_1'\popSymbol)[u \mapsto {\interpret {\sigma_1'} {\prg{res}}}][\prg{cont}\mapsto \epsilon]$.
\\
We now apply the induction hypothesis on (14), (16), (15d), and obtain
\\ 
$\SPS (18) \ \  \satDAssertFrom M  {\sigma_1'} {k}   {(A_{post})[\overline{\alpha/z}][\overline{\interpret {\sigma} {v_6} /v_6}]}$.
\\
We now want to obtain something of the form $...\sigma' \models ...A'$. We now want to be able to apply   Lemma \ref{l:calls:return},  part \ref{l:calls:caller:one} on (18). Therefore, we  define
\\
$\SPS (18a) \ \  \Aout \triangleq  \Apoa[\overline{\alpha/z}][\overline{\interpret {\sigma} {v_6} /v_6}]$
\\
$\SPS (18b) \ \  \overline {v_{1,a}} \triangleq  \overline {v_1}, \prg{res}$,   \ \ \ \ 
$\overline {v_{3,a}} \triangleq  \overline {v_3}, u$.
\\
The wellformedness condition for specifications requires that $\fv (A_{post}) \subseteq  \fv (\Apr) \cup \{ \prg{res} \}$. 
This, together with  (9), (12d) and (18b) give  \\
$\SPS (19a) \ \  \fv(\Aout) \subseteq  \overline {v_{1,a}}$
\\
Also, by (18b), and (17), we have that   
\\
$\SPS (19b) \ \ \overline {\interpret {\sigma'} {v_{3,a}}  =  \interpret {\sigma_1'} {v_{1,a}} }$.
\\
From (4) we obtain that  $ k \leq \DepthSt {\sigma}$. From  (8a) we obtain that $ \DepthSt {\sigma_1} = \DepthSt {\sigma} +1$. From (16) we obtain that  $ \DepthSt {\sigma_1'} = \DepthSt {\sigma_1}$, and from (17) we obtain that $ \DepthSt {\sigma'} = \DepthSt {\sigma_1'} -1$. All this gives that:\\
$\SPS (19c) \ \ k \leq \DepthSt {\sigma'}$
\\
We now apply  Lemma \ref{l:calls:return},  part \ref{l:calls:caller:one}, and obtain  
\\
$\SPS (20) \ \  \satDAssertFrom M  {\sigma'} {k}   {\Aout[\overline{v_{3,a}/v_{1,a}}]}$.
\\
We expand the definition from (18a), and re-order the substitutions by a similar argument as in in step (13a), using Lemma    part (\ref{l:sfs:six}), and obtain\\
$\SPS (20a) \ \  \satDAssertFrom M  {\sigma'} {k}   {\Apoa[\overline{v_{3,a}/v_{1,a}}][\overline{\alpha/z}][\overline{\interpret {\sigma} {v_6} /v_6}]}$.
\\
By  (20a),   (18b),  and because  by Lemma \ref{l:var:unaffect} we have that $\overline{\interpret {\sigma} {v_6}}$=$\overline{\interpret {\sigma'} {v_6}}$, we  obtain
\\
$\SPS (21) \ \ \satDAssertFrom M  {\sigma'} {k}   {(\Apoa)[y_0,\overline y, u/\prg{this},\overline x, \prg{res}][\overline{\alpha/z}]}. $. 
\\
With (7) we conclude.\\
~ \\

 \vspace{.3cm}
Proving $(*\!*\!*)$. Take a $\sigma''$. Assume that\\
$\SPS (15) \ \ \leadstoBoundedStar   {\Mtwo\cdot M}  {\sigma}  {\sigma''}$\\
$\SPS (16) \ \ \ {\Mtwo\cdot M}, \sigma'' \models \extThis$.\\
Then, from (8) and (15) we also obtain that\\
$\SPS (15) \ \ \leadstoBoundedStar   {\Mtwo\cdot M}  {\sigma_1}  {\sigma''}$\\
By (10), (11) and application of the induction hypothesis on (13),  (14c), and (15), we obtain that\\
$\SPS (\beta') \ \  \satDAssertFrom M  {\sigma''} {k}   {A_{mid}[\overline{\alpha/z}][\overline{\interpret \sigma {w}/w}]}$.\\
and using (7) we are done.
\\
~ \\

\newcommand{\Ainv}{A_{inv}}
\newcommand{\Ainva}{A_{inva}}

\item[{\sc{Call\_Ext\_Adapt}}] is  in some parts, similar to {\sc{Call\_Int}}. 
We highlight the differences in \npgreen{green}.

Therefore, there exist $u$, $y_o$, $\overline C$, $D$, $\overline y$, ands $\Ainv$, such that \\
 $\SPS (5) \ \ \sigma.\prg{cont}\txteq u:=y_0.m(\overline y)$,\\
$\SPS (6) \  \ \npgreen{\promises  M {\TwoStatesN {\overline {x:C}} {\Ainv} } }$, \\
$\SPS (7) \  \ A \txteq y_0:\prg{external} ,\overline {x:C}\ \wedge \  \PushASLongG {({y_0, \overline y})} {\npgreen{\Ainv}} ,$ \\
$\SPS \ \ \ \ \ \ \ A'  \txteq \PushASLongG {(y_0, \overline y)} {\npgreen{\Ainv}}$,\\
$\SPS\ \ \ \ \ \  \  A'' \txteq \npgreen{\Ainv}$. 
\\
Also, \\
$\SPS (8) \ \ \leadstoBounded  {\Mtwo\cdot M}  {\sigma}  {\sigma_a}$, \\
where \\
$\SPS (8a) \ \ \ \sigma_a\triangleq (\PushSLong { (\prg{this}\mapsto {\interpret{\sigma} {y_0}},{\overline{p \mapsto {\interpret{\sigma} {y}}}})}\sigma [\prg{cont}\mapsto stmt_m]$, \\ 
$\SPS (8b) \ \ \   \prg{mBody}(m,D,\overline {M})=\overline{p:D}\{\    stmt_m\ \}$ .\\
$\SPS (8c) \ \ \  \npgreen{D\ \mbox{is the class of}\, \interpret  {\sigma} {y_0}, \ \mbox{and $D$ is external.}}$
\\
By (7), and well-formedness of module invariants, we obtain\\
$\SPS (9a) \ \ \ \npgreen{\fv(\Ainv)\subseteq  \overline x}$,\\
$\SPS (9a) \ \ \ \npgreen{\fv(A)=y_0,\overline y, \overline x}$\\
\npgreen{By Barendregt, we also obtain that}\\
$\SPS (10) \ \ \ \npgreen{dom(\sigma)\ \#  \overline{x}}$
\\
This, together with (3) gives that
\\
$\SPS (10) \ \ \ \npgreen{\overline z = \overline x}$
\\
From (4), (7) and the definition of satisfaction we obtain\\
$\SPS (10) \ \ \ \npgreen{\satDAssertFrom M  {\sigma} k {(\overline{x:C} \wedge {\PushSLong {y_0,\overline y} {\Ainv}})[\overline{\alpha/z}]}}$.\\
The above gives that\\
$\SPS (10a) \ \ \ \npgreen{\satDAssertFrom M  {\sigma} k {\PushSLong {y_0,\overline y} {(\ (\overline{x:C})[\overline{\alpha/z}]\, \wedge\,  ( {\Ainv[\overline{\alpha/z}]}) \ ) } } }$  .
\\
We take $\Ain$ to be $\ (\overline{x:C})[\overline{\alpha/z}]\, \wedge\,  ( {\Ainv[\overline{\alpha/z}]})$, and apply Lemma \ref{l:calls:external},
part \ref{l:calls:callee:one}.
This gives that
\\
$\SPS (11)\ \ \npgreen{ M,  {\sigma_a} \models   {(\overline{x:C})[\overline{\alpha/z}]\, \wedge\,    {\Ainv[\overline{\alpha/z}]}  } }$

\vspace{.3cm}

Proving $(**)$. We shall use the short hand\\
$\SPS (12) \ \  \npgreen{\Ao \triangleq  \overline {\alpha:C} \wedge {\Ainv[\overline{\alpha/z}]}}$. 
\\
Assume that   $\leadstoBoundedStarFin  {\Mtwo\cdot M}  {\sigma}  {\sigma'}$. Then, by the operational semantics, we obtain that 
there exists state $\sigma_b'$, such that \\
$\SPS (16) \ \ \leadstoBoundedStarFin  {\Mtwo\cdot M}  {\sigma_a}  {\sigma_b}$ \\
$\SPS (17) \ \ \sigma'=(\sigma_b\popSymbol)[u \mapsto {\interpret {\sigma_1'} {\prg{res}}}][\prg{cont}\mapsto \epsilon]$.
\\
\npgreen{By Lemma \ref{lemma:external_breakdown:term} part \ref{lemma:external_breakdown:term:one}, and Def. \ref{def:exec:sum}, we obtain that there exists a sequence of states $\sigma_1$, ... $\sigma_n$, such that }\\
$\SPS (17) \ \ \npgreen{\WithExtPub {\Mtwo\cdot M} {\sigma_a}  {\sigma_a}  {\sigma_b} {\sigma_1...\sigma_n}}$
\\
\npgreen{By Def.   \ref{def:exec:sum}, the states $\sigma_1$, ... $\sigma_n$ are all public, and  correspond to the execution of a public method.
Therefore, by rule {\sc{Invariant}} for well-formed modules, we obtain that}\\
\newcommand{\Ainvr}{A_{inv,r}}
$\SPS (18) \ \ \npgreen{\forall i\in 1..n. }$\\
$ \npgreen{\strut\ \ \ \  \ \ \ \ \ \ \ \ 
    [ \  \hprovesN {M}  
               {   \prg{this}: {D_i}, \overline{p_i:D_i},  \overline{x:C} \, \wedge\,  \Ainv  }   
  	      { \sigma_i.\prg{cont}  }   
	      {  \Ainvr  }
	      { \Ainv  }
	     \  ] 
	     }
	     $ 
\\
\npgreen{where $D_i$ is the class of the receiver, $\overline{p_i}$ are the formal parameters, and $\overline {D_i}$ are the types of the formal parameters of the $i$-th public method, and where we use the shorthand $\Ainvr \triangleq  \PushASLong {\prg{res}} { \Ainv } $.}
\\ 
\npgreen{Moreover, (17) gives that}\\
$\SPS (19) \ \ \npgreen{\forall i\in 1..n. [ \ \  \ \leadstoBoundedStar {M\cdot\Mtwo} {\sigma} {\sigma_i} \ \ \ ]}$
\\
\npgreen{From (18) and (19) we obtain}\\
 $\SPS (20) \ \ \npgreen{\forall i\in [1..n]. }$\\
 ${\npgreen{
 \strut \ \ \ \ \ \ \ \ \ \ \ \ \ \ \ \ \ \ \ \ \ \ [ \ \     (  \prg{this}:D_i, \overline{p_i:D_i},\overline{x:C} \, \wedge\,  \Ainv,\sigma_i,\Ainvr, \Ainv ) }}$\\
 ${\npgreen{
 \strut \ \ \ \ \ \ \ \ \ \ \ \ \ \ \ \ \ \ \ \ \  \ \ \ \ \ \ \ \ \    \ll_{M,\Mtwo} }}$\\
 ${\npgreen{
 \strut \ \ \ \ \ \ \ \ \ \ \ \ \ \ \ \ \ \   \ \ \ \ \ \ \ \ \   ( A,\sigma,A',A'') \ \
                                        ]}}
                                      $
 \\
We take \\
$\SPS (21) \ \ k=\DepthSt {\sigma_a}$
 By application of the induction hypothesis on (20) we obtain that
\\
 $\SPS (22) \ \ \npgreen{\forall i\in [1..n]. \forall \sigma_{f}.[ \ \  \satDAssertFrom M {\sigma_i} k \Ao  \ \wedge \  \leadstoBoundedStarFin {M\cdot \Mtwo}  {\sigma_i}  {\sigma_{f}} \ 
\ \ \Longrightarrow \ \  \satDAssertFrom M {\sigma_f} k \Ao \ ]}$ 
\\
\npgreen{We can now apply Lemma \ref{lemma:external_exec_preserves_more}, part \ref{lemma:external_exec_preserves_more:three}, and because  $\DepthSt {\sigma_a}=\DepthSt {\sigma_b}$, we obtain that}
\\
 $\SPS (23) \ \ \npgreen{M, \sigma_b \models \Ainv[\overline{\alpha/x} ]} $ 
 \\
We apply Lemma  \ref{l:calls:return:ext} part \ref{l:ext:return:one}, and obtain
 \\
  $\SPS (24) \ \  M, \sigma' \models {\PushASLong {y_0,\overline y} {\Ainv[\overline{\alpha/x} ] }} $
  \\
  And since  ${\PushASLong {y_0,\overline y} {\Ainv[\overline{\alpha/x} ] }}$ is stable, and by rearranging, and applying (10), we obtain
 \\
   $\SPS (25) \ \   \satDAssertFrom M  {\sigma'} k   { (\PushASLong {y_0,\overline y} {\Ainv})[\overline{\alpha/z}] }$
\\  
Apply (7), and  we are done.
 
 \vspace{.3cm}

Proving   $(*\!*\!*)$. Take a $\sigma''$. Assume that\\
$\SPS (12) \ \ \leadstoBoundedStar   {\Mtwo\cdot M}  {\sigma}  {\sigma''}$\\
$\SPS (13) \ \ \ {\Mtwo\cdot M}, \sigma'' \models \extThis$.\\
We apply lemma \ref{lemma:external_breakdown:term:one}, part \ref{lemma:external_breakdown:two}
on (12) and see that there are two cases\\
\textbf{1st Case} $\leadstoBoundedExtPub {\Mtwo\cdot M}    {\sigma_a}  {\sigma''}$\\
That is, the execution from $\sigma_a$ to $\sigma''$ goes only through external states. 
We use (11), and that 
 $\Ainv$ is encapsulated, and are done with lemma \ref{lemma:external_exec_preserves_more}, part 
 \ref{lemma:external_exec_preserves_more:one}.
\\
\textbf{2nd Case} for some  $\sigma_c$, $\sigma_d.$ we have\\
$
\strut \ \ \ \ \ \leadstoBoundedExtPub {\Mtwo\cdot M}    {\sigma_a}  {\sigma_c} 
\wedge\ \leadstoBounded  {\Mtwo\cdot M}    {\sigma_c}  {\sigma_d} 
\wedge \ M, \sigma_d \models \pubMeth \wedge \leadstoBoundedStar  {\Mtwo\cdot M}    {\sigma_d}  {\sigma'}$\\
We apply similar arguments as in steps (17)-(23) and obtain \\
$\SPS (14) \ \ \ \ {M, \sigma_c \models \Ainv[\overline{\alpha/x} ]} $ 
\\
State $\sigma_c$ is a public, internal state;  therefore there exists a Hoare proof that it preserves the invariant.
By applying the inductive hypothesis, and the fact that $\overline z = \overline x$, we obtain:
\\
$\SPS (14) \ \ \ \ {M, \sigma'' \models \Ainv[\overline{\alpha/z} ]} $
\\
~ \\

\item[{\sc{Call\_Ext\_Adapt\_Strong}}] is  very similar to {\sc{Call\_Ext\_Adaprt}}.
We will summarize the similar steps, and highlight the differences in \npgreen{green}.

Therefore, there exist $u$, $y_o$, $\overline C$, $D$, $\overline y$, ands $\Ainv$, such that \\
 $\SPS (5) \ \ \sigma.\prg{cont}\txteq u:=y_0.m(\overline y)$,\\
$\SPS (6) \  \ \npgreen{\promises  M {\TwoStatesN {\overline {x:C}} {\Ainv} } }$, \\
$\SPS (7) \  \ A \txteq y_0:\prg{external} ,\overline {x:C}\ \wedge \ \npgreen{\Ainv}\  \wedge \PushASLongG {({y_0, \overline y})} {{\Ainv}} ,$ \\
$\SPS \ \ \ \ \ \ \ A'  \txteq  \npgreen{\Ainv}\  \wedge \PushASLong {(y_0, \overline y)} {\npgreen{\Ainv}}$,\\
$\SPS\ \ \ \ \ \  \  A'' \txteq  {\Ainv}$. 
\\
Also, \\
$\SPS (8) \ \ \leadstoBounded  {\Mtwo\cdot M}  {\sigma}  {\sigma_a}$, \\
By similar steps to (8a)-(10) from the previous case, we obtain\\
$\SPS (10a) \ \ \ \   \satDAssertFrom M  {\sigma} k {\npgreen{{\Ainv[\overline{\alpha/z}]} \wedge \ \PushSLong {y_0,\overline y} {(\ (\overline{x:C})[\overline{\alpha/z}]\, \wedge\,  ( {\Ainv[\overline{\alpha/z}]}) \ ) } } }$.\\
We now apply lemma  apply Lemma \ref{l:calls:external},
part \ref{l:calls:callee:two}.
This gives that
\\
$\SPS (11)\ \   \satDAssertFrom M  {\sigma_a} {\npgreen{k}} { (\ \ (\overline{x:C})[\overline{\alpha/z}] \, \wedge\,    \npgreen{\Ainv[\overline{\alpha/z}]} \wedge \ \PushSLong {y_0,\overline y} {(\ (\overline{x:C})[\overline{\alpha/z}] )} \ \ ) }$.
\\
the rest is similar to earlier cases
 
\vspace{.3cm}

\end{description}
\noindent
\vspace{.1cm}
  {\textbf{End Proof}} 

\subsection{Proof Sketch of Theorem  \ref{thm:soundness} -- part (B)}

\label{s:app:proof:sketch;overall}
 {\textbf{Proof Sketch}}
 By induction on the  cases for the specification $S$. If it is a method spec, then the theorem follows from \ref{t:quadruple:sound}. If it is a conjunction, then by inductive hypothesis.
 \\
 The interesting case is $S \txteq {\TwoStatesN {\overline {x:C}} {A}}$.
 \\
 Assume that 
 $ \satDAssertFrom M  {\sigma} k A[\overline {\alpha/x} ]$, that  $M,\sigma \models \extThis$,
 that $\leadstoBoundedStar  {M\cdot \Mtwo}  {\sigma}  {\sigma'}$, and that $M,\sigma \models \extThis$,
 \\ 
 We want to show that $ \satDAssertFrom M  {\sigma'} k A[\overline {\alpha/x} ]$.
 \\
 Then, by lemma
 \ref{lemma:external_breakdown:term}, we obtain that either \\
$\strut \ \ \ \ \ $ (1)\  $\leadstoBoundedExtPub {\Mtwo\cdot M}    {\sigma}  {\sigma'}$, or\\
$\strut \ \ \ \ \ $  (2)\  $\exists \sigma_1,\sigma_2.[\
\leadstoBoundedExtPub {\Mtwo\cdot M}    {\sigma}  {\sigma_1}
\wedge\ \leadstoBounded  {\Mtwo\cdot M}    {\sigma_1}  {\sigma_2}
\wedge \ M, \sigma_2 \models \pubMeth \wedge \leadstoBoundedStar  {\Mtwo\cdot M}    {\sigma_2}  {\sigma'} \ ]$
\\
In Case (1), we apply  \ref{lemma:external_exec_preserves_more}, part (3).  In order to fulfill the second premise of Lemma  \ref{lemma:external_exec_preserves_more}, part (3), we make use of the fact that $\vdash M$,   apply the rule {\sc{Method}}, and theorem \ref{t:quadruple:sound}.
This gives us $ \satDAssertFrom M  {\sigma'} k A[\overline {\alpha/x} ]$
\\
In Case (2), we proceed as in (1) and obtain that $ \satDAssertFrom M  {\sigma_1} k A[\overline {\alpha/x} ]$. Because $M \vdash \encaps A$, we also obtain that 
$ \satDAssertFrom M  {\sigma_2} k A[\overline {\alpha/x} ]$.
Since we are now executing a public method, and because $\vdash M$, we can apply {\sc{Invariant}}, and theorem \ref{t:quadruple:sound}, and obtain $ \satDAssertFrom M  {\sigma'} k A[\overline {\alpha/x} ]$\\
 \vspace{.1cm}
  {\textbf{End Proof Sketch}} 

\clearpage

\newcommand{\STwo}{\ensuremath{S_2}}
\newcommand{\STwoStrong}{\ensuremath{S_{2,strong}}}
\newcommand{\SPT}{~ \strut \hspace{.9cm}}
\newcommand{\Alocals}{\prg{A}_{buy}}
\newcommand{\Alocalsb}{\prg{A}_{1}}
\newcommand{\Ids}{\prg{Ids}_{buy}}
\newcommand{\Alocalstr}{\prg{A}_{trns}}
\newcommand{\Idstr}{\prg{Ids}_{trns}}
\newcommand{\Alocalsset}{\prg{A}_{set}}
\newcommand{\Alocalssets}{\prg{A}_{2}}
\newcommand{\Idsset}{\prg{Ids}_{set}}
\newcommand{\stmtsP}{\prg{stmts}_{10,11,12}}
\newcommand{\step}[1]{ \vspace{.1cm} \noindent {\textbf{#1}}}

\section{  {Proving Limited Effects for the Shop/Account Example}}

\label{s:app:example}

In Section \ref{s:outline} we introduced a \verb|Shop| that allows clients to make purchases through the
\verb|buy| method.
The body if this method  includes a method call to an unknown external object (\verb|buyer.pay(...)|).

In this section  we use our Hoare logic from Section \ref{sect:proofSystem} to {outline the proof} that the \verb|buy| method
does not expose the \verb|Shop|'s  \verb|Account|, its \verb|Key|, or allow the \verb|Account|'s balance to be illicitly modified. 

We {outline the proof} that $M_{good} \vdash \STwo$, and that $M_{fine} \vdash \STwo$.
{We  also show why $M_{bad} \not\vdash \STwo$.}

{We   rewrite the code of $M_{good}$ and so $M_{fine}$
so that it adheres to the syntax as defined in Fig. \ref{f:loo-syntax} (\S \ref{s:app:syntax:transform}). 
We  extend the specification $\STwo$, so that is also makes a specification for the private method \prg{set} (\S \ref{s:extend:spec}). 
After that, we outline the proofs  that $M_{good} \vdash \STwo$, and that $M_{fine} \vdash \STwo$ (in \S \ref{s:app:example:proofs}),
and that  $M_{good} \vdash \SThree$, and that $M_{fine} \vdash \SThree$ (\S \ref{s:app:example:proofs:S3}).
These proofs have been mechanized in Coq, and the source code will be
submitted as an artefact. 
We also discuss why $M_{bad} \not\vdash \STwo$ (\S \ref{s:bad:not:S2}).}

\subsection{Expressing the \prg{Shop} example in the syntax from Fig. \ref{f:loo-syntax}}
\label{s:app:syntax:transform}

{
We now express our example in the syntax of Fig. \ref{f:loo-syntax}. 
For this, we  add a return type to each of the methods; 
We turn all local variables to parameter; We add an explicit assignment to the variable \prg{res}: and We   add a temporary variable \prg{tmp} to which we assign the result of our \prg{void} methods.
For simplicity, we allow 
the shorthands \prg{+=} and \prg{-=}.
And we also allow definition of local variables, \eg  \prg{int price := ..} }

\begin{lstlisting}[mathescape=true, language=Chainmail, frame=lines]
module M$_{good}$
  ...   
  class Shop
    field accnt : Account, 
    field invntry : Inventory, 
    field clients: ..
  
    public method buy(buyer:external, anItem:Item, price: int, 
            myAccnt: Account, oldBalance: int, newBalance: int, tmp:int) : int
      price := anItem.price;
      myAccnt := this.accnt;
      oldBalance := myAccnt.blnce;
      tmp := buyer.pay(myAccnt, price)     // $\red{\mbox{external\ call}}!$
      newBalance := myAccnt.blnce;
      if (newBalance == oldBalance+price) then
          tmp := this.send(buyer,anItem)
      else
         tmp := buyer.tell("you have not paid me") ; 
      res := 0
     
      private method send(buyer:external, anItem:Item) : int
       ... 
  class Account
    field blnce : int 
    field key : Key
    
    public method transfer(dest:Account, key':Key, amt:nat) :int
      if (this.key==key') then
        this.blnce-=amt;
        dest.blnce+=amt
      else
        res := 0
      res := 0
	  
     public method set(key':Key) : int
      if (this.key==null)  then
      		this.key:=key'
      else 
        res := 0
      res := 0
\end{lstlisting}

\noindent
Remember that $M_{fine}$ is identical to $M_{good}$, except for the method \prg{set}. We describe the module below.

\begin{lstlisting}[mathescape=true, language=Chainmail, frame=lines]
module M$_{fine}$
  ...   
  class Shop
     ...  $ as\  in\  M_{good}$
  class Account
    field blnce : int 
    field key : Key
    
    public method transfer(dest:Account, key':Key, amt:nat) :int
       ...  $as\ in\ M_{good}$
	  
     public method set(key':Key, k'':Key) : int
      if (this.key==key')  then
      		this.key:=key''
      else 
        res := 0
      res := 0
\end{lstlisting}

\subsection{Proving that $M_{good}$ and $M_{fine}$ satisfy  $\STwo$}
\label{s:extend:spec}

We redefine $\STwo$ so that it also describes the behaviour of method \prg{send}. and have:
\\
{\sprepost
		{\strut \ \ \ \ \ \ \ \ \ S_{2a}} 
		{ \prg{a}:\prg{Account} \wedge  \prg{e}:\prg{external}  \wedge \protectedFrom{\prg{a.key}} {\prg{e} } } 
		 {\prg{private Shop}}
		 {\prg{send}}
		 {\prg{buyer}:\prg{external},\prg{anItem}:\prg{Item} }
		 { \protectedFrom{\prg{a.key}} {e} }
		 { \protectedFrom{\prg{a.key}} {e} }
}
\\
{\sprepost
		{\strut \ \ \ \ \ \ \ \ \ S_{2b}} 
		{ \prg{a}:\prg{Account} \wedge \prg{a.blnce} =\prg{b} }
		 {\prg{private Shop}}
		 {\prg{send}}
		 {\prg{buyer}:\prg{external}, \prg{anItem}:\prg{Item} }
		 { \prg{a.blnce} =\prg{b} }
		{   \prg{a.blnce} =\prg{b} }
}
\\
$\strut  \SPSP  \STwoStrong \ \  \triangleq \ \ \ \STwo \ \wedge \ S_{2a} \ \wedge \ S_{2b} $


 \label{s:app:example:proofs}

For brevity we only show that \verb|buy| satisfies our scoped invariants, as the all other methods of 
the \verb|M|$_{good}$ interface are relatively simple, and do not make any external calls. 

{ To write our proofs more succinctly, we will use \prg{ClassId}::\prg{methId}.\prg{body} as a shorthand for the method body of \prg{methId} defined in \prg{ClassId}.}

\begin{lemma}[$M_{good}$ satisfies $\STwoStrong$]
\label{lemma:exampleKeyProtect}
\label{l:Mgood:S2}
$M_{good} \vdash \STwoStrong$
\end{lemma}
\begin{proofO}
In order to prove that 
$$M_{good} \vdash \TwoStatesN {\prg{a}:\prg{Account}}  {\inside{\prg{\prg{a.key}}}}$$
we have to apply  \textsc{Invariant} from Fig. \ref{f:wf}.
 That is, for each  class $C$  defined in $M_{good}$, and for each public method $m$ in $C$, with parameters $\overline{y:D}$, we have to prove that
 \small
\begin{align*}
M_{good}\ \vdash \ \ &   \{ \ \prg{this}:\prg{C},\, \overline{y:D},\, \prg{a}:\prg{Account}\, \wedge\,
		             {\inside{\prg{a.key}}}\ \wedge\       \protectedFrom {\prg{a.key}} {(\prg{this},\overline y)} \  \} \\
		& \SPT  \prg{C}::\prg{m}.\prg{body}\  \\
		&
                   \{\ {\inside{\prg{a.key}}}\ \wedge\ \ \protectedFrom {\prg{a.key}} {\prg{res}}   \ \}\ ||\ \{\ {\inside{\prg{a.key} } } \ \} \\
\end{align*}

\normalsize
Thus, we need to prove  three Hoare quadruples: one for \prg{Shop::buy}, one for  \prg{Account::transfer}, and one for  \prg{Account::set}.  That is, we have to prove that
 \small
\begin{align*}
\text{(1?)}  \ \ \ \ M_{good}\ \vdash  \  \ 
		&	\{  \ \Alocals, \, \prg{a}:\prg{Account} \, \wedge\, {\inside{\prg{a.key}}} \, \wedge \, \protectedFrom {\prg{a.key}} {\Ids}  \  \} \\
		& \SPT \prg{Shop}::\prg{buy}.\prg{body}\ \\  
		& \{ {\inside{\prg{a.key}}}\ \wedge\ {\PushASLong {\prg{res}} {\inside{\prg{a.key}}}}  \} \ \ \  || \ \ \ 
		   \{ {\inside{\prg{a.key}}} \}
\\
\text{(2?)}  \ \ \ \ M_{good} \vdash \ 
		&	\{  \ \Alocalstr, \, \prg{a}:\prg{Account}\, \wedge\,  {\inside{\prg{a.key}}} \, \wedge \, \protectedFrom {\prg{a.key}} {\Idstr}  \  \} \\
		& \SPT \prg{Account}::\prg{transfer}.\prg{body}\ \\  
		& \{ {\inside{\prg{a.key}}}\ \wedge\ {\PushASLong {\prg{res}} {\inside{\prg{a.key}}}}  \} \ \ \  || \ \ \ 
		   \{ {\inside{\prg{a.key}}} \}
\\
\text{(3?)}  \ \ \ \ M_{good} \vdash \ 
		&	\{  \ \Alocalsset, \, \prg{a}:\prg{Account}\, \wedge\,  {\inside{\prg{a.key}}} \, \wedge \, \protectedFrom {\prg{a.key}} {\Idsset}  \  \} \\
		& \SPT \prg{Account}::\prg{set}.\prg{body}\ \\  
		& \{ {\inside{\prg{a.key}}}\ \wedge\ {\PushASLong {\prg{res}} {\inside{\prg{a.key}}}}  \} \ \ \  || \ \ \ 
		   \{ {\inside{\prg{a.key}}} \}
\end{align*}
\normalsize
where we are using ? to indicate that this needs to be proven, and 
where we are using the shorthands\\
\small
$
\begin{array}{c}
\begin{array}{lcl}
 \Alocals\ &  \triangleq   \  &   \prg{this}:\prg{Shop}, \prg{buyer} : \prg{external}, \prg{anItem} : \prg{Item},\, \prg{price} : \prg{int},  \\
&  & \prg{myAccnt} : \prg{Account},\, \prg{oldBalance}:  \prg{int},  \prg{newBalance}:  \prg{int},  \prg{tmp}:  \prg{int}.\\
  \Ids\ &   \triangleq   & \prg{this},  \prg{buyer}, \prg{anItem}, \prg{price}, \prg{myAccnt},\, \prg{oldBalance},\  \prg{newBalance},  \prg{tmp}.\\ 
\Alocalstr\  & \triangleq \  &   \prg{this}:\prg{Account}, \prg{dest} : \prg{Account}, \prg{key'} : \prg{Key},\prg{amt}:\prg{nat} \\
  \Idstr\  &  \triangleq  & \prg{this},\, \prg{dest} ,\, \prg{key'},\, \prg{amt} \\
  \Alocalsset\  & \triangleq \  &   \prg{this}:\prg{Account}, \, \prg{key'} : \prg{Key},\, \prg{key''} : \prg{Key}.\\
    \Idsset\  &  \triangleq  & \prg{this},\, \prg{key'} ,\, \prg{key''}. \\
\end{array}
\end{array}
$
\normalsize

We will also need to prove that \prg{Send} satisfies specifications $S_{2a}$ and $S_{2b}$.

We outline the proof of (1?) in Lemma \ref{l:buy:sat},
and the proof of (2) in Lemma \ref{l:transfer:S2}.
We do not prove (3), but will prove that \prg{set} from $M_{fine}$
satisfies $S_2$; shown in  Lemma \ref{l:set:sat} -- ie for module $M_{fine}$.

\end{proofO}

We also want to prove that $M_{fine}$ satisfies the specification $\STwoStrong$.

\begin{lemma} [$M_{fine}$ satisfies $\STwoStrong$]
\label{l:Mfine:S2}

$M_{fine} \vdash \STwoStrong$
\end{lemma}
\begin{proofO}
The proof of
$$M_{fine} \vdash \TwoStatesN {\prg{a}:\prg{Account}}  {\inside{\prg{\prg{a.key}}}}$$
goes along similar lines to the proof of lemma \ref{l:Mgood:S2}.
Thus, we need to prove the following  three Hoare quadruples: 
 \small
\begin{align*}
\text{(4?)}  \ \ \ \ M_{fine}\ \vdash  \  \ 
		&	\{  \ \Alocals, \, \prg{a}:\prg{Account} \, \wedge\, {\inside{\prg{a.key}}} \, \wedge \, \protectedFrom {\prg{a.key}} {\Ids}  \  \} \\
		& \SPT \prg{Shop}::\prg{buy}.\prg{body}\ \\  
		& \{ {\inside{\prg{a.key}}}\ \wedge\ {\PushASLong {\prg{res}} {\inside{\prg{a.key}}}}  \} \ \ \  || \ \ \ 
		   \{ {\inside{\prg{a.key}}} \}
\\
\text{(5?)}  \ \ \ \ M_{fine} \vdash \ 
		&	\{  \ \Alocalstr, \, \prg{a}:\prg{Account}\, \wedge\,  {\inside{\prg{a.key}}} \, \wedge \, \protectedFrom {\prg{a.key}} {\Idstr}  \  \} \\
		& \SPT \prg{Account}::\prg{transfer}.\prg{body}\ \\  
		& \{ {\inside{\prg{a.key}}}\ \wedge\ {\PushASLong {\prg{res}} {\inside{\prg{a.key}}}}  \} \ \ \  || \ \ \ 
		   \{ {\inside{\prg{a.key}}} \}
\\
\text{(6?)}  \ \ \ \ M_{fine} \vdash \ 
		&	\{  \ \Alocalsset, \, \prg{a}:\prg{Account}\, \wedge\,  {\inside{\prg{a.key}}} \, \wedge \, \protectedFrom {\prg{a.key}} {\Idsset}  \  \} \\
		& \SPT \prg{Account}::\prg{set}.\prg{body}\ \\  
		& \{ {\inside{\prg{a.key}}}\ \wedge\ {\PushASLong {\prg{res}} {\inside{\prg{a.key}}}}  \} \ \ \  || \ \ \ 
		   \{ {\inside{\prg{a.key}}} \}
\end{align*}

\normalsize

The proof of (4?) is identical to that of (1?); the proof of (5?) is identical to that of (2?). 
We outline the proof (6?)    in Lemma \ref{l:set:sat} in \S \ref{s:set:sat}.

\end{proofO}

\label{s:buy:sat}

\begin{lemma}[\prg{Shop::buy} satisfies $S_2$]
\label{l:buy:sat}
 
\begin{align*}
\text{(1)}  \ \ \ \ M_{good} \vdash 
		&	\{  \ \Alocals\,\prg{a}:\prg{Account}\ \wedge\, {\inside{\prg{a.key}}} \, \wedge \, \protectedFrom {\prg{a.key}} {\Ids}  \  \} \\
		& \SPT \prg{Shop}::\prg{buy}.\prg{body}\ \\  
		& \{ {\inside{\prg{a.key}}}\ \wedge\ {\PushASLong {\prg{res}} {\inside{\prg{a.key}}}}  \} \ \ \  || \ \ \ 
		   \{ {\inside{\prg{a.key}}} \}
\end{align*}

\end{lemma}

\begin{proofO}
We will use the shorthand $\Alocalsb \triangleq \Alocals, \,\prg{a}:\prg{Account}$.
We will split the proof into 1) proving that statements 10, 11, 12 preserve the protection of \prg{a.key} from the \prg{buyer}, 2) proving that the external call 

\step{1st Step: proving statements 10, 11, 12}

We apply the underlying Hoare logic and prove that the statements on lines 10, 11, 12 do not affect the value of \prg{a.key}, ie that for a $z\notin \{ \prg{price}, \prg{myAccnt}, \prg{oldBalance} \}$, we have 

\begin{align*}
\text{(10)}  \ \ \ \ {M_{good} \vdash_{ul}} 
		&	\{  \ \Alocalsb\  \wedge\ z=\prg{a.key} \} \\
		&   \SPT \prg{price:=anItem.price}; \\  
		&   \SPT \prg{myAccnt:=this.accnt}; \\  
                 &   \SPT \prg{oldBalance := myAccnt.blnce};\\
		& \{ z=\prg{a.key} \}
\end{align*}

We then apply {\sc{Embed\_UL}}, {\sc{Prot-1}} and {\sc{Prot-2}} and {\sc{Combine}} and and {\sc{Types-2}} on (10) and use the shorthand $\stmtsP$ for the statements on lines 10, 11 and 12, and obtain: 
\\
\begin{align*}
\text{(11)}  \ \ \ \ M_{good} \vdash 
		&	\{  \ \Alocalsb\  \wedge\ {\inside{\prg{a.key}}} \, \wedge\, \protectedFrom{\prg{buyer}} {\prg{a.key}}  \} \\
		& \SPT \stmtsP\ \\  
		& \{ \ {\inside{\prg{a.key}}}  \, \wedge\, \protectedFrom{\prg{buyer}} {\prg{a.key}}   \}
\end{align*}

We apply  {\sc{Mid}}  on (11) and obtain 
\begin{align*}
\text{(12)}  \ \ \ \ M_{good} \vdash 
		&	\{  \ \Alocalsb\, \wedge\, \protectedFrom {\prg{a.key}} {\prg{buyer}}\  \} \\
		& \SPT \stmtsP\ \\  
		& \{ \ \Alocalsb\, \wedge \  {\inside{\prg{a.key}}} \, \wedge\, \protectedFrom{\prg{buyer}} {\prg{a.key}}  \ \} \ \ || \\
		& \{ \ {\inside{\prg{a.key}}}\  \}
\end{align*}

\step{2nd Step: Proving the External Call}

We now need to prove that the external method call \prg{buyer.pay(this.accnt, price)} protects the \prg{key}. i.e.
\begin{align*}
\text{(13?)} \ \ \ M_{good} \vdash & \{ \ \Alocalsb \   \wedge\    {\inside{\prg{a.key}}},\, \wedge\, \protectedFrom {\prg{a.key}} {\prg{buyer}}  \} \\
		  		& \SPT  \prg{tmp := buyer.pay(myAccnt, price)}\ \\  
		  		& \{ \ \ \ \Alocalsb \ \wedge\ {\inside{\prg{a.key}}} \, \wedge\, \protectedFrom{\prg{buyer}} {\prg{a.key}} \} \ \ \  || \ \ \  \\
		  		&   \{ \   {\inside{\prg{a.key}}}\  \}
\end{align*}
\normalsize

We use that $M \vdash \TwoStatesN  {\prg{a}:\prg{Account}}  {\inside{\prg{a.key}}}$
 and  obtain
 \\
 \small
\begin{align*}
\text{(14)} \ \ \ M_{good} \vdash & \{ \ \prg{buyer}:\prg{external},\,  {\inside{\prg{a.key}}} \, \wedge\, 
\protectedFrom {\prg{a.key}} {(\prg{buyer},\prg{myAccnt},\prg{price})} \  \} \\
		  		& \SPT  \prg{tmp := buyer.pay(myAccnt, price)}\ \\  
		  		& \{ \ \inside{\prg{a.key}} \, \wedge\, 
\protectedFrom {\prg{a.key}} {(\prg{buyer},\prg{myAccnt},\prg{price})}\ \} \ \ \  || \ \ \  \\
		  		&   \{ \   {\inside{\prg{a.key}}}\  \}
\end{align*}
\normalsize

 Moreover by the type declarations and the type rules, we obtain that all objects directly or indirectly accessible accessible from \prg{myAccnt} are internal or scalar.
 This, together with  \textsc{Prot-Intl}, gives that
\\
$
\begin{array}{llll}
& (15) & & M_{good} \vdash \Alocalsb   \longrightarrow \protectedFrom {\prg{a.key}} {\prg{myAccnt}} 
\end{array}
$
\\
Similarly, because \prg{price} is a \prg{nat}, and because of \textsc{Prot-Int}$_1$, we obtain 
\\
$
\begin{array}{llll}
& (16) & & M_{good} \vdash \Alocalsb \   \longrightarrow \protectedFrom {\prg{a.key}} {\prg{price}} 
\end{array}
$

We apply {\textsc{Consequ}} on (15), (16) and (14) and obtain (13)!

\normalsize

\end{proofO}

\begin{lemma} [ \prg{transfer} satisfies $S_2$]
\label{l:transfer:S2}
\small
\begin{align*}
\text{(2)}  \ \ \ \ M_{good} \vdash \ 
		&	\{  \ \Alocalstr, \, \prg{a}:\prg{Account}\, \wedge\,  {\inside{\prg{a.key}}} \, \wedge \, \protectedFrom {\prg{a.key}} {\Idstr}  \  \} \\
		& \SPT \prg{Account}::\prg{transfer}.\prg{body}\ \\  
		& \{ {\inside{\prg{a.key}}}\ \wedge\ {\PushASLong {\prg{res}} {\inside{\prg{a.key}}}}  \} \ \ \  || \ \ \ 
		   \{ {\inside{\prg{a.key}}} \}
\end{align*}
\normalsize

\end{lemma}

\begin{proofO}

To prove (2), we will need to prove that

\small
\begin{align*}
\text{(21?)}  \ \ \ \ M_{good} \vdash \ 
		&	\{  \ \Alocalstr, \, \prg{a}:\prg{Account}\, \wedge\,  {\inside{\prg{a.key}}} \, \wedge \, \protectedFrom {\prg{a.key}} {\Idstr}  \  \} \\
		&  \SPT   \prg{if (this.key==key') then }\\
		& \SPT \SPT   \SPT\SPT  \prg{this.\balance:=this.\balance-amt} \\
	        & \SPT \SPT   \SPT\SPT  \prg{dest.\balance:=dest.\balance+amt} \\
		& \SPT   \prg{else }\\
		& \SPT\SPT   \SPT\SPT  \prg{res:=0} \\
		& \SPT \prg{res:=0} \\
		& \{ {\inside{\prg{a.key}}}\ \wedge\ {\PushASLong {\prg{res}} {\inside{\prg{a.key}}}}  \} \ \ \  || \ \ \ 
		   \{ {\inside{\prg{a.key}}} \}
\end{align*}
\normalsize

Using the underlying Hoare logic we can prove that no account's \prg{key} gets modified, namely

\small
\begin{align*}
\text{(22)}  \ \ \ \ M_{good}\ \vdash_{ul} \ 
		&	\{  \ \Alocalstr, \, \prg{a}:\prg{Account}\, \wedge\,   {\inside{\prg{a.key}}} \\
		&  \SPT   \prg{if (this.key==key') then }\\
		& \SPT \SPT   \SPT\SPT  \prg{this.\balance:=this.\balance-amt} \\
	        & \SPT \SPT   \SPT\SPT  \prg{dest.\balance:=dest.\balance+amt} \\
		& \SPT   \prg{else }\\
		& \SPT\SPT   \SPT\SPT  \prg{res:=0} \\
		& \SPT \prg{res:=0} \\
		& \{    \inside{\prg{a.key}}   \} \ \ \ 
\end{align*}
\normalsize

Using (22) and {\sc{[Prot-1]}},  we obtain

\small
\begin{align*}
\text{(23)}  \ \ \ \ M_{good}\ \vdash  \ 
		&	\{  \ \Alocalstr, \, \prg{a}:\prg{Account}\, \wedge\,  z=\prg{a.key} \} \\
		&   \SPT   \prg{if (this.key==key') then }\\
		& \SPT \SPT   \SPT\SPT  \prg{this.\balance:=this.\balance-amt} \\
	        & \SPT \SPT   \SPT\SPT  \prg{dest.\balance:=dest.\balance+amt} \\
		& \SPT   \prg{else }\\
		& \SPT\SPT   \SPT\SPT  \prg{res:=0} \\
		& \SPT \prg{res:=0} \\
		& \{  z=\prg{a.key} \} \ \ \ 
\end{align*}
\normalsize

Using (23) and   {\sc{[Embed-UL]}}, we obtain 

\small
\begin{align*}
\text{(24)}  \ \ \ \ M_{good}\ \vdash \ 
		&	\{  \ \Alocalstr, \, \prg{a}:\prg{Account}\, \wedge\,  z=\prg{a.key} \} \\
		&  \SPT   \prg{if (this.key==key') then }\\
		& \SPT \SPT   \SPT\SPT  \prg{this.\balance:=this.\balance-amt} \\
	        & \SPT \SPT   \SPT\SPT  \prg{dest.\balance:=dest.\balance+amt} \\
		& \SPT   \prg{else }\\
		& \SPT\SPT   \SPT\SPT  \prg{res:=0} \\
		& \SPT \prg{res:=0} \\
		& \{  z=\prg{a.key} \}  \ \ \  || \ \ \  \{  z=\prg{a.key} \} 
\end{align*}
\normalsize

 {\sc{[Prot\_Int]}} and the fact that $z$ is an \prg{int}  gives us that ${\PushASLong {\prg{res}} {\inside{\prg{a.key}}}}$.
Using  {\sc{[Types]}},  and  {\sc{[Prot\_Int]}} and   {\sc{[Consequ]}}   on (24) we obtain (21?).

\end{proofO}

\label{s:set:sat}

We want to prove that this public method satisfies the specification  $\STwoStrong$, namely

\begin{lemma}[$\prg{set}$ satisfies $\STwo$]
\label{l:set:sat}
\label{l:satisfies:Mfine:pec2}
 
\begin{align*}
\text{(6)}  \ \ \ \ M_{fine} \vdash 
		&	\{  \ \Alocalsset\ \wedge\  {\inside{\prg{a.key}}} \, \wedge \, \protectedFrom {\prg{a.key}} {\Idsset}  \  \} \\
		& \SPT   \prg{if (this.key==key') then }\\
		& \SPT \SPT   \SPT\SPT  \prg{this.key:=key''} \\
	        & \SPT   \prg{else }\\
		& \SPT\SPT   \SPT\SPT  \prg{res:=0} \\
		& \SPT \prg{res:=0} \\
& \{ {\inside{\prg{a.key}}}\ \wedge\ {\PushASLong {\prg{res}} {\inside{\prg{a.key}}}}  \} \ \ \  || \ \ \ 
	   \{ {\inside{\prg{a.key}}} \}
\end{align*}

\end{lemma}

\begin{proofO}
We will be  using the shorthand 
 $\SPT  \Alocalssets\ \triangleq \  \prg{a}:\prg{Account},\  \Alocalsset$.\\

To prove (6), we will use the  {\sc{Sequence}} rule, and we want to prove
\\
\begin{align*}
\text{(61?)}  \ \ \ \ M_{fine} \vdash 
		&	\{  \ \Alocalssets\ \wedge\  {\inside{\prg{a.key}}} \, \wedge \, \protectedFrom {\prg{a.key}} {\Idsset} \  \} \\
		& \SPT   \prg{if (this.key==key') then }\\
		& \SPT \SPT   \SPT\SPT  \prg{this.key:=key''} \\
	        & \SPT   \prg{else }\\
		& \SPT\SPT   \SPT\SPT  \prg{res:=0} \\
		& \{ \ \Alocalssets\,\wedge\ {\inside{\prg{a.key}}}\    \} \ \ \  || \ \ \ 
		   \{ {\inside{\prg{a.key}}} \}
\end{align*}
and that
\begin{align*}
\text{(62?)}  \ \ \ \ M_{fine} \vdash
          &  \{ \ \Alocalssets\,\wedge \inside{\prg{a.key}} \  \}  \\
		& \SPT\SPT   \SPT\SPT  \prg{res:=0} \\
		& \{ {\inside{\prg{a.key}}}\ \wedge\ {\PushASLong {\prg{res}} {\inside{\prg{a.key}}}}  \} \ \ \  || \ \ \ 
		   \{ {\inside{\prg{a.key}}} \}
\end{align*}

(62?) follows   from the types, and {\sc{Prot-Int}}$_1$, the fact that \prg{a.key} did not change, and  \sdN{ {\sc{Prot-1}}}.

\vspace{.5cm}
We now  want to  prove (61?). For this, will apply the {\sc{If-Rule}}. That is, we need to prove that

\begin{align*}
\text{(63?)}  \ \ \ \ M_{fine} \vdash 
		&	\{  \ \Alocalssets\,\wedge\, {\inside{\prg{a.key}}} \, \wedge \, \protectedFrom {\prg{a.key}} {\Idsset} \, \wedge  \,  \prg{this.key}=\prg{key'}\  \} \\
		& \SPT \SPT   \SPT\SPT  \prg{this.key:=key''} \\
		& \{ {\inside{\prg{a.key}}}  \} \ \ \  || \ \ \ 
		   \{ {\inside{\prg{a.key}}} \}
\end{align*}
 
and that
 
\begin{align*}
\text{(64?)}  \ \ \ \ M_{fine} \vdash 
		&	\{  \ \Alocalssets\,\wedge\, {\inside{\prg{a.key}}} \, \wedge \, \protectedFrom {\prg{a.key}} {\Idsset} \, \wedge  \,  \prg{this.key}\neq\prg{key'}\  \} \\
		& \SPT\SPT   \SPT\SPT  \prg{res:=0} \\
		& \{ {\inside{\prg{a.key}}}\   \} \ \ \  || \ \ \ 
		   \{ {\inside{\prg{a.key}}} \}
\end{align*}

(64?) follows easily from  the fact that \prg{a.key} did not change, and  {\sc{Prot-1}}.

\vspace{.5cm}
We look at the proof of (63?).  We will apply the {\sc{Cases}} rule, and distinguish on whether \prg{a.key}=\prg{this.key}. That is, we want to prove that\\
\small{
\begin{align*}
\text{(65?)}  \ \ \ \ M_{fine} \vdash 
		&	\{  \ \Alocalssets\,\wedge\, {\inside{\prg{a.key}}} \, \wedge \, \protectedFrom {\prg{a.key}} {\Idsset} \, \wedge  \,  \prg{this.key}=\prg{key'}\ \wedge\ \prg{this.key}=\prg{a.key}  \} \\
			& \SPT \SPT   \SPT\SPT  \prg{this.key:=key''} \\
	       	& \{ {\inside{\prg{a.key}}}\   \} \ \ \  || \ \ \ 
		   \{ {\inside{\prg{a.key}}} \}
\end{align*}
}
\\
and that
\\
\small{
\begin{align*}
\text{(66?)}  \ \ \ \ M_{fine} \vdash 
		&	\{  \ \Alocalssets\,\wedge\, {\inside{\prg{a.key}}} \, \wedge \, \protectedFrom {\prg{a.key}} {\Idsset} \, \wedge  \,   \,  \prg{this.key}=\prg{key'}\  \wedge \prg{this.key}\neq\prg{a.key'}\  \} \\
		& \SPT \SPT   \SPT\SPT  \prg{this.key:=key''} \\
		& \{ {\inside{\prg{a.key}}}\   \ \ \  || \ \ \ 
		   \{ {\inside{\prg{a.key}}} \}
\end{align*}
}
\vspace{.2cm}
\normalsize
We can prove (65?) through application of {\sc{Absurd}}, {\sc{ProtNeq}}, and {\sc{Consequ}}, as follows

\begin{align*}
\text{(61c)}  \ \ \ \ M_{fine} \vdash 
		&	\{  \ false  \} \\
		& \SPT \SPT   \SPT\SPT  \prg{this.key:=key''} \\
		& \{ {\inside{\prg{a.key}}}\   \} \ \ \  || \ \ \ 
		   \{ {\inside{\prg{a.key}}} \}
\end{align*}

By  {\sc{ProtNeq}}, we have $M_{fine} \vdash  \protectedFrom {\prg{a.key}} {\prg{key'}} \, \longrightarrow\, {\prg{a.key}}\neq {\prg{key'}}$, and therefore obtain

\begin{align*}
\text{(61d)}  \ \ \ \ M_{fine} \vdash  ... \wedge \, \protectedFrom {\prg{a.key}} {\Idsset} \, \wedge  \, \prg{this.key}=\prg{a.key}\, \wedge\,  \prg{this.key}=\prg{key'}\ \longrightarrow \ false 
\end{align*}

We apply  {\sc{Consequ}} on (61c) and (61d) and obtain (61aa?).

\vspace{.5cm}
We can prove (66?) by proving that \prg{this.key}$\neq$\prg{a.key} implies that $\prg{this}\neq \prg{a}$ (by the underlying Hoare logic), which again implies that the assignment \prg{this.key := ... } leaves the value of \prg{a.key} unmodified. We apply {\sc{Prot-1}}, and are done.

\end{proofO}

\subsection{Showing that $M_{bad}$ does not satisfy $S_2$ nor $S_3$}

\subsubsection{$M_{bad}$ does not satisfy $S_2$}
$M_{bad}$ does not satisfy $S_2$. We can argue this semantically (as in \S \ref{s:bad:not:S2}), and also in terms of the proof system (as in \ref{s:bad:not:S2:proof}).

\subsubsection{$M_{bad}\nvDash S_2$}
\label{s:bad:not:S2}
 The reason is that $M_{bad}$ exports the public method \prg{set}, which updates the key without any checks. 
So, it could start in a state where the key of the account was protected, and then update it to something not protected.

In more detail: Take any state $\sigma$, where $M_{bad},\sigma \models a_0:\prg{Account}, k_0:\prg{Key} \wedge \inside{a_0.\prg{key}}$. 
Assume also that $M_{bad},\sigma \models \extThis$.  
Finally, assume that the continuation in $\sigma$ consists of $a_0.\prg{set}(k_0)$.
Then we obtain that $M_{bad}, \sigma \leadsto^* \sigma'$, where $\sigma'=\sigma[a_0.\prg{key}\mapsto k_0]$.
We also  have that $M_{bad},\sigma' \models \extThis$, and because $k_0$ is a local variable, we also have that $M_{bad},\sigma' \nvDash \inside{k_0}$.
Moreover, $M_{bad}, \sigma' \models a_0.\prg{krey}=k_0 $.
Therefore, $M_{bad},\sigma' \nvDash \inside{a_0.\prg{key}}$.

\subsubsection{$M_{bad}  \nvdash S_2$}
\label{s:bad:not:S2:proof}

In order to prove that $M_{bad}  \vdash S_2$, we would have needed to prove, among other things,  that \prg{set} satisfied $S_2$, which means proving that

\small{
\begin{align*}
\text{(ERR\_1?)}  \ \ \ \ M_{bad}\ \vdash \ 
		&	\{  \ \prg{this}:\prg{Account}, \prg{k'}:\prg{Key}, a:\prg{Account}\, \wedge\, \inside{a.\prg{key}}  \, \wedge \, \protectedFrom {\prg{a.key}} { \{\prg{this},\prg{k'}\} }\   \} \\
			& \SPT \SPT   \SPT\SPT  \prg{this.key:=k'}; \\
			& \SPT \SPT   \SPT\SPT \prg{res}:=0 \\ 
	       	& \{ \  \inside{a.\prg{key}}  \, \wedge \, \protectedFrom {a.\prg{key}} {\prg{res}} \   \} \ \ \  || \ \ \ 
		   \{ ... \}
\end{align*}
}

However, we cannot  establish $\text{(ERR\_1?)}$.
Namely, when we  take the case where $\prg{this}=a$,  we would need to establish, that

\small{
\begin{align*}
\text{(ERR\_2?)}  \ \ \ \ M_{bad}\ \vdash \ 
		&	\{  \ \prg{this}:\prg{Account}, \prg{k'}:\prg{Key}\, \wedge\, \inside{\prg{this}.\prg{key}}  \, \wedge \, \protectedFrom {\prg{this.key}} { \{\prg{this},\prg{k'}\} }\   \} \\
			& \SPT \SPT   \SPT\SPT  \prg{this.key:=k'}  \\
	       	& \{ \ \inside{\prg{this}.\prg{key}}  \   \} \ \ \  || \ \ \ 
		   \{ ... \}
\end{align*}
} 

And there is no way to prove $\text{(ERR\_2?)}$. In fact, $\text{(ERR\_2?)}$  is not sound, for the reasons outlined in \S \ref{s:bad:not:S2}.

\subsubsection{$M_{bad}$ does not satisfy $S_3$}

We have already argued in Examples \ref{e:versions} and \ref{e:public} that $M_{bad}$ does not satisfy $S_3$, by giving a semantic argument -- ie we are in state where $ \inside{a_0.\prg{key}}$, and execute $\prg{a}_0.\prg{set(k1)}; \prg{a}_0.\prg{transfer}(...\prg{k1})$. 

Moreover, if we attempted to prove that \prg{set} satisfies $S_3$, we would have to show that

\small{
\begin{align*}
\text{(ERR\_3?)}  \ \ \ \ M_{bad}\ \vdash \ 
		&	\{  \ \prg{this}:\prg{Account},\ \prg{k'}:\prg{Key}, a:\prg{Account},\, b:\prg{int}\ \wedge \\
  	 &	  \ \ \     \inside{a.\prg{key}}  \, \wedge \,    \protectedFrom {\prg{a.key}} { \{\prg{this},\prg{k'} \} }   
 	  \ \wedge\ a.\prg{\balance}\geq b \ \}  \\
 					& \SPT \SPT   \SPT\SPT  \prg{this.key:=k'}; \\ 
 			& \SPT \SPT   \SPT\SPT \prg{res}:=0 \\  
 	       	& \{ \  \inside{a.\prg{key}}  \, \wedge \, \protectedFrom {a.\prg{key}} {\prg{res}} \ \wedge\  a.\prg{\balance}\geq b \ \}  \ \ \  || \ \ \ 
		   \{ ... \}
\end{align*}
}

which, in the case of $a=\prg{this}$ would imply that

\small{
\begin{align*}
\text{(ERR\_4?)}  \ \ \ \ M_{bad}\ \vdash \ 
		&	\{  \ \prg{this}:\prg{Account},\ \prg{k'}:\prg{Key}, \, b:\prg{int}\ \wedge \\
  	 &	  \ \ \     \inside{\prg{this}.\prg{key}}  \, \wedge \,    \protectedFrom {\prg{\prg{this}.key}} { \{\prg{this},\prg{k'} \} }   
 	  \ \wedge\ \prg{this}.\prg{\balance}\geq b \ \}  \\
 					& \SPT \SPT   \SPT\SPT  \prg{this.key:=k'} \\
 	& \{ \  \inside{\prg{this}.\prg{key}}     \ \}  \ \ \  || \ \ \ 
		   \{ ... \}
\end{align*}
}

And  $\text{(ERR\_4?)}$ cannot be proven and does not hold.

\subsection{Demonstrating that $M_{good} \vdash \SThree$, and that $M_{fine} \vdash \SThree$}
 \label{s:app:example:proofs:S3}

 \subsection{Extending the specification $\SThree$}
\label{s:extend:spec:three}

As in \S \ref{s:extend:spec}, we redefine $\SThree$ so that it also describes the behaviour of method \prg{send}. and have:
\\
\\
$\strut  \SPSP  \SThreeStrong \ \  \triangleq \ \ \ \SThree \ \wedge \ S_{2a} \ \wedge \ S_{2b} $


\begin{lemma}[module $M_{good}$  satisfies $\SThreeStrong$]
\label{lemma:S3}
\label{l:Mgood:S3}
$M_{good} \vdash \SThreeStrong$
\end{lemma}
\begin{proofO}
In order to prove that 
$$M_{good} \vdash \TwoStatesN {\prg{a}:\prg{Account}, b:\prg{int} }{\  \inside{\prg{\prg{a.key}}} \wedge \prg{a.\balance}\geq b\ }$$
we have to apply  \textsc{Invariant} from Fig. \ref{f:wf}.
 That is, for each  class $C$  defined in $M_{good}$, and for each public method $m$ in $C$, with parameters $\overline{y:D}$, we have to prove that they satisfy the corresponding quadruples.


\normalsize
Thus, we need to prove  three Hoare quadruples: one for \prg{Shop::buy}, one for  \prg{Account::transfer}, and one for  \prg{Account::set}.  That is, we have to prove that
 \small
\begin{align*}
\text{(31?)}  \ \ \ \ M_{good}\ \vdash  \  \ 
		&	\{  \ \Alocals, \, \prg{a}:\prg{Account}, b:\prg{int} \, \wedge\, {\inside{\prg{a.key}}} \, \wedge \, \protectedFrom {\prg{a.key}} {\Ids} \wedge \prg{a.\balance}\geq b \  \} \\
		& \SPT \prg{Shop}::\prg{buy}.\prg{body}\ \\  
		& \{ {\inside{\prg{a.key}}}\ \wedge\ {\PushASLong {\prg{res}} {\inside{\prg{a.key}}}}  \wedge \prg{a.\balance}\geq b \} \ \ \  || \ \ \ 
		   \{ {\inside{\prg{a.key}}}  \wedge \prg{a.\balance}\geq b \}
\\
\text{(32?)}  \ \ \ \ M_{good} \vdash \ 
		&	\{  \ \Alocalstr, \, \prg{a}:\prg{Account}\, , b:\prg{int} \, \wedge\,  {\inside{\prg{a.key}}} \, \wedge \, \protectedFrom {\prg{a.key}} {\Idstr}  \ \wedge \prg{a.\balance}\geq b \  \} \\
		& \SPT \prg{Account}::\prg{transfer}.\prg{body}\ \\  
		& \{ {\inside{\prg{a.key}}}\ \wedge\ {\PushASLong {\prg{res}} {\inside{\prg{a.key}}} } \wedge \prg{a.\balance}\geq b \  \} \ \ \  || \ \ \ 
		   \{ {\inside{\prg{a.key}}} \wedge \prg{a.\balance}\geq b \  \}
\\
\text{(33?)}  \ \ \ \ M_{good} \vdash \ 
		&	\{  \ \Alocalsset, \, \prg{a}:\prg{Account}\, , b:\prg{int} \, \wedge\,  {\inside{\prg{a.key}}} \, \wedge \, \protectedFrom {\prg{a.key}} {\Idsset}  \ \wedge \prg{a.\balance}\geq b \  \} \\
		& \SPT \prg{Account}::\prg{set}.\prg{body}\ \\  
 		& \{ {\inside{\prg{a.key}}}  \wedge\ {\PushASLong {\prg{res}} {\inside{\prg{a.key}}}}   \wedge \prg{a.\balance}\geq b \  \}\ \ \  || \ \ \ 
		   \{ {\inside{\prg{a.key}}} \wedge \prg{a.\balance}\geq b \  \} 
\end{align*}
\normalsize
where we are using ? to indicate that this needs to be proven, and 
where we are using the shorthands $\Alocals$,   $\Ids$, $\Alocalstr$, $\Idstr$, $\Alocalsset$ as defined earlier.

 \end{proofO}
 
The proofs for $M_{fine}$ are similar.

 We outline the proof of (31?) in Lemma \ref{l:buy:sat:S3}. 
 We outline the proof of (32?) in  Lemma \ref{l:transfer:sat:S3}.

\subsubsection{Proving that \prg{Shop::buy} from $M_{good}$ satisfies $\SThreeStrong$ and also $S_4$}
\label{s:buy:sat:S3}

\begin{lemma}[function $M_{good}::\prg{Shop}::\prg{buy}$  satisfies $\SThreeStrong$ and also $S_4$]
\label{l:buy:sat:S3}
 
\begin{align*}
\text{(31)}  \ \ \ \ M_{good}\ \vdash  \  \ 
		&	\{  \ \Alocals, \, \prg{a}:\prg{Account}, b:\prg{int}, \, \wedge\, {\inside{\prg{a.key}}} \, \wedge \, \protectedFrom {\prg{a.key}} {\Ids} \wedge \prg{a.\balance}\geq b \  \} \\
		& \SPT \prg{Shop}::\prg{buy}.\prg{body}\ \\  
		& \{ {\inside{\prg{a.key}}}\ \wedge\ {\PushASLong {\prg{res}} {\inside{\prg{a.key}}}}  \wedge \prg{a.\balance}\geq b \} \ \ \  || \ \ \ 
		   \{ {\inside{\prg{a.key}}}  \wedge \prg{a.\balance}\geq b \}
\end{align*}

\end{lemma}

\begin{proofO}
Note that (31) is a proof that $M_{good}::\prg{Shop}::\prg{buy}$  satisfies $\SThreeStrong$ and also  hat $M_{good}::\prg{Shop}::\prg{buy}$  satisfies $S_4$. This is so, because application of {\sc{[Method]}} on $S_4$ gives us exactly the proof obligation from (31).

This proof is similar to the proof of lemma \ref{l:buy:sat}, with the extra requirement here that we need to argue about balances not decreasing.
To do this, we will leverage the assertion about balances given in $S_3$.

We will use the shorthand $\Alocalsb \triangleq \Alocals, \,\prg{a}:\prg{Account}, b:\prg{int}$.
We will split the proof into 1) proving that statements 10, 11, 12 preserve the protection of \prg{a.key} from the \prg{buyer}, 2) proving that the external call 

\step{1st Step: proving statements 10, 11, 12}

We apply the underlying Hoare logic and prove that the statements on lines 10, 11, 12 do not affect the value of \prg{a.key} nor that of \prg{a.\balance}.  Therefore, for a $z,z'\notin \{ \prg{price}, \prg{myAccnt}, \prg{oldBalance} \}$, we have 

\begin{align*}
\text{(40)}  \ \ \ \ {M_{good} \vdash_{ul}} 
		&	\{  \ \Alocalsb\  \wedge\ z=\prg{a.key} \wedge\ z'=\prg{a.\balance}  \} \\
		&   \SPT \prg{price:=anItem.price}; \\  
		&   \SPT \prg{myAccnt:=this.accnt}; \\  
                 &   \SPT \prg{oldBalance := myAccnt.blnce};\\
		& \{ z=\prg{a.key} \wedge\ z'=\prg{a.\balance} \}
\end{align*}

We then apply {\sc{Embed\_UL}}, {\sc{Prot-1}} and {\sc{Prot-2}} and {\sc{Combine}} and and {\sc{Types-2}} on (10) and use the shorthand $\stmtsP$ for the statements on lines 10, 11 and 12, and obtain: 
\\
\begin{align*}
\text{(41)}  \ \ \ \ M_{good} \vdash 
		&	\{  \ \Alocalsb\  \wedge\ {\inside{\prg{a.key}}} \, \wedge\, \protectedFrom{\prg{buyer}} {\prg{a.key}} \wedge\ z'=\prg{a.\balance}  \} \\
		& \SPT \stmtsP\ \\  
		& \{ \ {\inside{\prg{a.key}}}  \, \wedge\, \protectedFrom{\prg{buyer}} {\prg{a.key}} \wedge\ z'=\prg{a.\balance}   \}
\end{align*}

We apply  {\sc{Mid}}  on (11) and obtain 
\begin{align*}
\text{(42)}  \ \ \ \ M_{good} \vdash 
		&	\{  \ \Alocalsb\, \wedge\, \protectedFrom {\prg{a.key}} {\prg{buyer}}\ \wedge\ z'=\prg{a.\balance}  \} \\
		& \SPT \stmtsP\ \\  
		& \{ \ \Alocalsb\, \wedge \  {\inside{\prg{a.key}}} \, \wedge\, \protectedFrom{\prg{buyer}} {\prg{a.key}}  \ \wedge\ z'=\prg{a.\balance}  \} \ \ || \\
		& \{ \ {\inside{\prg{a.key}}}\  \wedge\ z'=\prg{a.\balance} \}
\end{align*}

\step{2nd Step: Proving the External Call}

We now need to prove that the external method call \prg{buyer.pay(this.accnt, price)} protects the \prg{key}, and does nit decrease the balance, i.e.
\small
\begin{align*}
\text{(43?)} \ \ \ M_{good} \vdash & \{ \ \Alocalsb \   \wedge\    {\inside{\prg{a.key}}} \, \wedge\, \protectedFrom {\prg{a.key}} {\prg{buyer}} \wedge\ z'= \prg{a.\balance}  \} \\
		  		& \SPT  \prg{tmp := buyer.pay(myAccnt, price)}\ \\  
		  		& \{ \ \ \ \Alocalsb \ \wedge\ {\inside{\prg{a.key}}} \, \wedge\, \protectedFrom{\prg{buyer}} {\prg{a.key}} \wedge \ \prg{a.\balance}\geq z'\  \} \ \ \  || \ \ \  \\
		  		&   \{ \   {\inside{\prg{a.key}}}\ \wedge \  \prg{a.\balance}\geq z'  \}
\end{align*}
\normalsize

We use that $M \vdash \TwoStatesN  {\prg{a}:\prg{Account},\prg{b}:\prg{int}, }  {\inside{\prg{a.key}} \wedge \prg{a.\balance}\geq z'}   $
 and  obtain
 \\
 \small
\begin{align*}
\text{(44)} \ \ \ M_{good} \vdash & \{ \ \prg{buyer}:\prg{external},\,  {\inside{\prg{a.key}}} \, \wedge\, 
\protectedFrom {\prg{a.key}} {(\prg{buyer},\prg{myAccnt},\prg{price})} \  \wedge\ z'\geq \prg{a.\balance}  \} \\
		  		& \SPT  \prg{tmp := buyer.pay(myAccnt, price)}\ \\  
		  		& \{ \ \inside{\prg{a.key}} \, \wedge\, 
\protectedFrom {\prg{a.key}} {(\prg{buyer},\prg{myAccnt},\prg{price})}\ \wedge\ z'\geq \prg{a.\balance}  \} \ \ \  || \ \ \  \\
		  		&   \{ \   {\inside{\prg{a.key}}}\  \wedge\ z'\geq \prg{a.\balance}  \}
\end{align*}
\normalsize 
 
In order to obtain (43?) out of (44), we use the type system and type declarations and obtain that all objects transitively reachable from
\prg{myAccnt} or \prg{price} are scalar or internal. Thus, we 
apply \textsc{Prot-Intl}, 
and obtain\\
$
\begin{array}{llll}
& (45) & & M_{good} \vdash \Alocalsb \wedge  {\inside{\prg{a.key}}}\  \longrightarrow\ \protectedFrom {\prg{a.key}} {\prg{myAccnt}} 
\\
& (46) & & M_{good} \vdash \Alocalsb \wedge  {\inside{\prg{a.key}}}\  \longrightarrow\ \protectedFrom {\prg{a.key}} {\prg{price}} 
\\
& (47) & & M_{good} \vdash \Alocalsb \wedge  z'= \prg{a.\balance}\   \longrightarrow\  z'\geq \prg{a.\balance} 
\end{array}
$

We apply {\textsc{Consequ}} on (44), (45), (46) and (47) and obtain (43)!

\normalsize

\step{3nd Step: Proving the Remainder of the Body}

 We now need to prove that lines 15-19 of the method preserve the protection of \prg{a.key}, and do not decrease \prg{a.balance}.
 We outline the  remaining proof in less detail.
 
 We prove the internal call on line 16, using the method specification for \prg{send}, using $S_{2a}$ and $S_{2b}$, and applying rule {\sc{[Call\_Int]}}, and obtain

 \small
\begin{align*}
\text{(48)} \ \ \ M_{good}\  \vdash \ & \{ \ \prg{buyer}:\prg{external},\, \prg{item}:\prg{Intem} \wedge {\inside{\prg{a.key}}} \, \wedge\, 
\protectedFrom {\prg{a.key}} {(\prg{buyer}}  \  \wedge\ z'= \prg{a.\balance}  \} \\
		  		& \SPT  \prg{tmp := this.send(buyer,Item)}\ \\  
		  		& \{ \ \inside{\prg{a.key}} \, \wedge\, 
\protectedFrom {\prg{a.key}} {\prg{buyer}} \ \wedge\ z'= \prg{a.\balance}  \} \ \ \  || \ \ \  \\
		  		&   \{ \   {\inside{\prg{a.key}}}\  \wedge\ z'= \prg{a.\balance}  \}
\end{align*}
\normalsize

We now need to prove that the external method call \prg{buyer.tell("You have not paid me")} also protects the \prg{key}, and does nit decrease the balance. We can do this by applying the rule about protection from strings  {\sc{[Pror\_Str]}}, the fact that $M_{good} \vdash S_{3}$, and rule  {\sc{[Call\_Extl\_Adapt]}} and obtain:

 \small
\begin{align*}
\text{(49)} \ \ \ M_{good}\  \vdash \ & \{ \ \prg{buyer}:\prg{external},\, \prg{item}:\prg{Intem} \wedge {\inside{\prg{a.key}}} \, \wedge\, 
\protectedFrom {\prg{a.key}} {(\prg{buyer}}  \  \wedge\ z'/geq  \prg{a.\balance}  \} \\
		  		& \SPT  \prg{tmp:=buyer.tell("You have not paid me")}\ \\  
		  		& \{ \ \inside{\prg{a.key}} \, \wedge\, 
\protectedFrom {\prg{a.key}} {\prg{buyer}} \ \wedge\ z'\geq \prg{a.\balance}  \} \ \ \  || \ \ \  \\
		  		&   \{ \   {\inside{\prg{a.key}}}\  \wedge\ z'\geq  \prg{a.\balance}  \}
\end{align*}
\normalsize 

We can now apply  {\sc{[If\_Rule}}, and {\sc{[Conseq}} on (49) and (50),  and obtain

 \small
\begin{align*}
\text{(50)} \ \ \ M_{good}\  \vdash \ & \{ \ \prg{buyer}:\prg{external},\, \prg{item}:\prg{Intem} \wedge {\inside{\prg{a.key}}} \, \wedge\, 
\protectedFrom {\prg{a.key}} {(\prg{buyer}}  \  \wedge\ z'\geq  \prg{a.\balance}  \} \\
		  		& \SPT  \prg{if} ... \prg{then}\\
				& \SPT \SPT  \prg{tmp:=this.send(buyer,anItem)}\ \\  
				& \SPT  \prg{else}\\
				& \SPT \SPT  \prg{tmp:=buyer.tell("You have not paid me")}\ \\  
		  		& \{ \ \inside{\prg{a.key}} \, \wedge\, 
\protectedFrom {\prg{a.key}} {\prg{buyer}} \ \wedge\ z'\geq \prg{a.\balance}  \} \ \ \  || \ \ \  \\
		  		&   \{ \   {\inside{\prg{a.key}}}\  \wedge\ z'\geq  \prg{a.\balance}  \}
\end{align*}
\normalsize 

The rest follows through application of {\sc{[Prot\_Int}}, and {\sc{[Seq]}}.

\end{proofO}

\begin{lemma}[function $M_{good}::\prg{Account}::\prg{transfer}$ satisfies $S_3$]
\label{l:transfer:sat:S3}
 \small
 \begin{align*}
\text{(32)}  \ \ \ \ M_{good} \vdash \ 
		&	\{  \ \Alocalstr, \, \prg{a}:\prg{Account}\, , b:\prg{int} \, \wedge\,  {\inside{\prg{a.key}}} \, \wedge \, \protectedFrom {\prg{a.key}} {\Idstr}  \ \wedge \prg{a.\balance}\geq b \  \} \\
		& \SPT \prg{Account}::\prg{transfer}.\prg{body}\ \\  
		& \{ {\inside{\prg{a.key}}}\ \wedge\ {\PushASLong {\prg{res}} {\inside{\prg{a.key}}} } \wedge \prg{a.\balance}\geq b \  \} \ \ \  || \ \ \ 
		   \{ {\inside{\prg{a.key}}} \wedge \prg{a.\balance}\geq b \  \}
\end{align*}
\normalsize

\end{lemma}

\begin{proofO}
We will use   the shorthand $stmts_{28-33}$ for the statements in the body of \prg{transfer}. 
We will prove   the preservation of protection, separately from the balance not decreasing when the key is protcted.

For the former, applying the steps in the proof of Lemma \ref{l:transfer:S2},  we obtain

\small
 \begin{align*}
\text{(21)}  \ \ \ \ M_{good} \vdash \ 
		&	\{  \ \Alocalstr, \, \prg{a}:\prg{Account}\, \wedge\,  {\inside{\prg{a.key}}} \, \wedge \, \protectedFrom {\prg{a.key}} {\Idstr}  \  \} \\
		&  \SPT   stmts_{28-33} \\
		& \{ {\inside{\prg{a.key}}}\ \wedge\ {\PushASLong {\prg{res}} {\inside{\prg{a.key}}}}  \} \ \ \  || \ \ \ 
		   \{ {\inside{\prg{a.key}}} \}
\end{align*}
\normalsize

For the latter, we  rely on the underlying Hoare logic to ensure that no balance decreases, except perhaps that of the receiver, in which case its key was not protected.
Namely, we have that

\small
 \begin{align*}
 \text{(71)}  \ \ \ \ M_{good} \vdash_ul \ 
		&	\{  \ \Alocalstr, \, \prg{a}:\prg{Account}\, \wedge\,  \prg{a.\balance}=b\ \wedge \ (\prg{this}\neq \prg{a} \vee prg{this}.\prg{key}\neq \prg{key}' ) \  \} \\
		&  \SPT   stmts_{28-33} \\
		& \{ \prg{a.\balance}\geq b \}
\end{align*}
\normalsize

We apply rules {\sc{Embed\_UL}} and   {\sc{Mid}} on (71), and obtain

\small
 \begin{align*}
 \text{(72)}  \ \ \ \ M_{good} \vdash  \ 
		&	\{  \ \Alocalstr, \, \prg{a}:\prg{Account}\, \wedge\,  \prg{a.\balance}=b\ \wedge \ (\prg{this}\neq \prg{a} \vee prg{this}.\prg{key}\neq \prg{key}' ) \  \} \\
		&  \SPT   stmts_{28-33} \\
		& \{ \prg{a.\balance}\geq b \} \ \ \  || \ \ \  \{ \prg{a.\balance}\geq b \}
\end{align*}
\normalsize

Moreover, we have 

\small
$\begin{array}{llll}
 \text{(73)} \ \ \  & M_{good}\  & \vdash  \ & \protectedFrom {\prg{a.key}} {\Idstr}\ \  \rightarrow \ \ \protectedFrom {\prg{a.key}} {\prg{key'}} \\
 \text{(74)} &  M_{good}  & \vdash  & \protectedFrom {\prg{a.key}} {\prg{key'}} \ \  \rightarrow\ \  \prg{a.key} \neq \prg{key'}  \\
  \text{(75)} &  M_{good}  & \vdash  &    \prg{a.key} \neq  \prg{key'}  \ \ \rightarrow \ \ \prg{a}\neq \prg{this} \vee \prg{this.key}\neq \prg{key}'
\end{array}
$

normalsize

Applying (73), (74), (75) and {\sc{Conseq}} on (72) we obtain:

\small
 \begin{align*}
 \text{(76)}  \ \ \ \ M_{good} \vdash  \ 
		&	\{  \ \Alocalstr, \, \prg{a}:\prg{Account}\, \wedge\,  \prg{a.\balance}=b\ \wedge \ \protectedFrom {\prg{a.key}} {\Idstr}\ \  \} \\
		&  \SPT   stmts_{28-33} \\
		& \{ \prg{a.\balance}\geq b \} \ \ \  || \ \ \  \{ \prg{a.\balance}\geq b \}
\end{align*}
\normalsize

We combine  (72) and (76) through {\sc{Combine}} and obtain (32).

\end{proofO}

\subsection{Dealing with polymorphic function calls}
\label{app:polymorphic}

The case split rules together with the rule of consequence allow our Hoare logic to formally reason about polymorphic calls, where the receiver may be internal or external.

We demonstrate this through an example where we may have an external receiver, or a receiver from a class $C$. Assume we had a module $M$ with a scoped invariant (as in A), and an internal method specification as in (B). 

$\begin{array}{lclcl}
& (A)& M & \vdash & \TwoStatesN{y_1:D}{A} \\
& (B) & M & \vdash & \mprepostLong {A_1}{ \prg{p}\ C} {m}   {y_1: D} {A_2} { \parallel \{A_3\}}  \\
\\
  \end{array}
$
\\
Here \prg{p} may be \prg{private} or \prg{ublic}; the argument apples either way.

\vspace{.1cm}

\noindent 
Assume also implications as in (C)-(H)

$\begin{array}{lclcl}
&  (C) & M & \vdash & A_0 \ \rightarrow \ \PushASLong {(y_0,y_1)}{A}\\
&  (D) & M & \vdash &  \PushASLong {(y_0,y_1)}{A} \rightarrow A_4\\
 & (E) &  M & \vdash & A \rightarrow A_5\\
  & (F) &  M & \vdash & A_0 \rightarrow A_1[y_0/\prg{this}]\\
  & (G) &  M & \vdash & A_2[y_0,u/\prg{this},res] \rightarrow A_4 \\
    & (H) &  M & \vdash & A_3  \rightarrow A_5 
\\
  \end{array}
$

Then, by application of {\sc{Call\_Ext\_Adapt}}  on (A) we obtain (I)

$
\SPT (I)\ \ {  \hprovesN {M} 
						{ \  y_0:external, y_1: D \wedge  {\PushASLong {(y_0,y_1)} {A}}\ }  
						 { \ u:=y_0.m(y_1)\    }
					         { \   \PushASLong {(y_0,y_1)}{A} \ } 
						{  \ A \ }  }
						\\
						$
						
By application of the rule of consequence on (I) and (C), (D), and (E), we obtain

$
\SPT (J)\ \ {  \hprovesN {M} 
						{ \  y_0:external, y_1: D \wedge  A_0\ }  
						 { \ u:=y_0.m(y_1)\    }
					         { \   A_4 \ } 
						{  \ A_5\ }  }
						\\
						$

Then, by application of {\sc{[Call\_Intl]}}  on (B) we obtain (K)

$
\SPT (K)\ \ {  \hprovesN {M} 
						{ \  y_0:C, y_1: D \wedge  A_1[y_0/\prg{this}]\ }  
						 { \ u:=y_0.m(y_1)\    }
					         { \  A_2[y_0,u/\prg{this},res] \ } 
						{  \ A_3 \ }  }
						\\
						$
						
By application of the rule of consequence on (K) and (F), (G), and (H), we obtain

$
\SPT (L)\ \ {  \hprovesN {M} 
						{ \  y_0:C, y_1: D \wedge  A_0\ }  
						 { \ u:=y_0.m(y_1)\    }
					         { \  A_4 \ } 
						{  \ A_5 \ }  }
						\\
						$

By case split, {\sc{[Cases]}}, on (J) and (L), we obtain

$
\SPT (polymoprhic)\ \ {  \hprovesN {M} 
						{ \  (y_0: external \vee y_0:C), y_1: D \wedge  A_0\ }  
						 { \ u:=y_0.m(y_1)\    }
					         { \  A_4 \ } 
						{  \ A_5 \ }  }
						\\
						$

\bibliography{Case,more,Response1}   
\clearpage

\end{document}